\documentclass[a4paper,reqno]{amsart}
\usepackage{amsmath,amssymb,mathrsfs}
\usepackage{txfonts,bm,pifont}
\usepackage{cite}
\usepackage{float}
\usepackage{stackrel}
\usepackage{subcaption}
\usepackage{graphicx,xcolor}
\usepackage{comment}
\usepackage{enumitem}
\usepackage{longtable}
\usepackage{hyperref}

\newtheorem{theorem}{Theorem}[section]
\newtheorem{proposition}[theorem]{Proposition}

\newtheorem{lemma}[theorem]{Lemma}
\newtheorem{remark}[theorem]{Remark}
\theoremstyle{definition}
\newtheorem{definition}[theorem]{Definition}

\newtheorem{condition}[theorem]{Condition}

\numberwithin{equation}{section}
\numberwithin{figure}{section}
\numberwithin{table}{section}
\newcommand{\wutilde}[1]{\vrule depth 0pt width 0pt%
{\raise0.8pt\hbox{$\smash{{\mathop{#1} \limits_{\displaystyle\widetilde{}}}}$}}}
\newcommand{\wuhat}[1]{\vrule depth 0pt width 0pt%
{\raise0.6pt\hbox{$\smash{{\mathop{#1} \limits_{\displaystyle\widehat{}}}}$}}}
\newcommand{\ol}[1]{\overline{#1}}

\newcommand{\ul}[1]{\underline{#1}}

\newcommand{\al}{\alpha}
\newcommand{\be}{\beta}
\newcommand{\de}{\delta}

\newcommand{\si}{\sigma}

\newcommand{\ep}{\bm{\epsilon}}

\newcommand{\PDE}{P$\Delta$E}
\newcommand{\bbZ}{\mathbb{Z}}
\newcommand{\bbR}{\mathbb{R}}
\newcommand{\bbC}{\mathbb{C}}

\newcommand{\ii}{{\rm i}}

\newcommand{\bml}{{\bm l}}

\newcommand{\set}[2]{\left\{\left. #1 ~\right|~ #2 \right\}}

\newcommand{\sn}[1]{{\rm sn}\left(#1\right)}
\newcommand{\polyQ}{\mathrm Q}
\newcommand{\polyH}{\mathrm H}
\newcommand{\polyD}{\mathrm D}
\newcommand{\polyN}{\mathrm N}

\makeatletter
\long\def\@makecaption#1#2{
 \vskip 10pt
 \setbox\@tempboxa\hbox{#1. #2}
 \ifdim \wd\@tempboxa >\hsize #1. #2\par \else \hbox
to\hsize{\hfil\box\@tempboxa\hfil}
 \fi}
\makeatother
\begin{document}
\allowdisplaybreaks

\title[]{Classification of quad-equations on a cuboctahedron}
\author{Nalini Joshi}
\address{School of Mathematics and Statistics F07, The University of Sydney, New South Wales 2006, Australia.}
\email{nalini.joshi@sydney.edu.au}
\author{Nobutaka Nakazono}
\address{Institute of Engineering, Tokyo University of Agriculture and Technology, 2-24-16 Nakacho Koganei, Tokyo 184-8588, Japan.}
\email{nakazono@go.tuat.ac.jp}
\begin{abstract}
In this paper, we consider polynomials associated with faces and internal quadrilaterals of a cuboctahedron and classify them under the requirement that they are consistent. These polynomials give rise to a system of partial difference equations on a face-centred cubic lattice. Our results were motivated by $\tau$-functions related to discrete Painlev\'e equations.
\end{abstract}


\maketitle
\setcounter{tocdepth}{1}
\tableofcontents
\section{Introduction}\label{Introduction}
This paper is concerned with consistent polynomials associated to a face-centered cubic (FCC) lattice, given by
\begin{equation}\label{eqn:def_Omega}
 \Omega=\set{\bml=\sum_{i=1}^3l_i\ep_i}{l_i\in\bbZ,~l_1+l_2+l_3\in2\bbZ},
\end{equation}
where $\{\ep_1,\ep_2,\ep_3\}$ is a standard basis of $\bbR^3$. Each cell in the FCC lattice shown in Figure \ref{fig:fcc} contains a cuboctahedron, and the polynomials we consider here are associated with quadrilateral faces and internal quadrilaterals of such a cuboctahedron. These polynomials give rise to partial difference equations on the FCC lattice. In this paper, we report on the classification of such polynomials. 
\begin{figure}[H]
  \centering
  \includegraphics[width=0.4\textwidth]{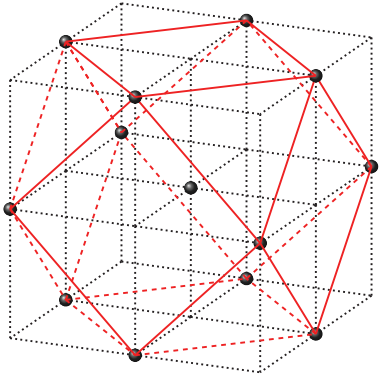}
  \caption{A cell of the $\Omega$ lattice, showing an embedded cuboctahedron.}
\label{fig:fcc}
  \end{figure}
  \noindent
 
   Earlier classification of polynomials were found for simple cubic lattices \cite{ABS2003:MR1962121,ABS2009:MR2503862,BollR2011:MR2846098,BollR2012:MR3010833,BollR:thesis} and led to \PDE s that have properties closely related to those of integrable PDEs, such as the Korteweg-de Vries equation. 
   We note recent extension of such results to the FCC lattice \cite{kels2020interaction}, which involves polynomials in 5 variables and edge relations motivated by Yang-Baxter equations. 

We here consider polynomials of 4 variables, each associated to a vertex of a quadrilateral that forms a face or internal slice of a cuboctahedron. Our motivation came from new examples of \PDE s that arise naturally in our study of discrete Painlev\'e equations.
(An example is given in \cite{JN:Preparing}.) 

\subsection{Main Results} Our main result is expressed by Theorem \ref{theo:classification_CACO}, which gives a classification of polynomials that lead to quad-equations (defined in Definition \ref{def:quad-eqn}) around a cuboctahedron. 
These quad-equations lead to constrained \PDE s that are consistent on a 3D lattice of quadrilaterals given by overlapping cuboctahedra as described in Definition \ref{def:CACO_PDEs}. 
\subsection{Notation and Definitions}
In this section, we define notation and useful terminology that will be used throughout the paper. 
\begin{definition}\label{def:quad-eqn}\rm
Let $Q=Q(x,y,z,w)$ be a multivariable polynomial over $\bbC$.
\begin{enumerate}[leftmargin=1.4cm,label=(\alph*)]
\item%
The polynomial $Q$ is defined to be a multi-affine polynomial, if it is linear in each variable.
\item%
The multi-affine polynomial $Q$ is said to be irreducible over $\bbC$, 
if the equation $Q=0$ can be solved for each argument, and the solution is a rational function of the other three arguments.
\item%
The equation $Q=0$ will be referred to as a quad-equation, if $Q$ is an irreducible multi-affine polynomial. 
\item%
If $Q'=Q'(x,y,z,w)$ is rational in each argument, and written as $Q'=Q/R$, where $R$ is a polynomial and $Q$ is an irreducible multi-affine polynomial, then $Q'=0$ is also called a quad-equation.
\end{enumerate}
\end{definition}

\begin{remark}
In what follows, for simplicity, we sometimes use the term ``quad-equation $Q$" instead of ``quad-equation $Q=0$".
Moreover, when we use this term, it will mean that $Q$ is given as an irreducible multi-affine polynomial.
\end{remark}

\begin{remark}
Notice that the condition of irreducibility implies non-vanishing of certain coefficients 
and non-factoring occurring in a multi-affine polynomial $Q(x,y,z,w)$. 
We do not restate such conditions separately in the body of the paper.
\end{remark}

\begin{definition}
Let $P$ and $Q$ be quad-equations, which are respectively given by irreducible multi-affine polynomials $P'$ and $Q'$.
The quad-equation $P$ is said to be equivalent to a quad-equation $Q$,
 denoted by 
\begin{equation}
 P\equiv Q,
\end{equation}
if $P'$ is a constant multiple of $Q'$. 
Moreover, given two sets of quad-equations $\{P_1,\dots,P_n\}$ and $\{Q_1,\dots,Q_n\}$ such that there exists a permutation $\si\in\mathfrak{S}_n$ with
\begin{equation}
 P_1\equiv Q_{\si(1)}\,,\dots,P_n\equiv Q_{\si(n)},
\end{equation}
we write
\begin{equation}
 \{P_1,\dots,P_n\}=\{Q_1,\dots,Q_n\}.
\end{equation}
\end{definition}

\begin{definition}\rm
Let $r$ be the M\"obius transformation of the variables $x$, $y$, $z$, $w$, given by
\begin{equation}
 r.x=\frac{A_1x+B_1}{C_1x+D_1},\quad
 r.y=\frac{A_2y+B_2}{C_2y+D_2},\quad
 r.z=\frac{A_3z+B_3}{C_3z+D_3},\quad
 r.w=\frac{A_4w+B_4}{C_4w+D_4},
\end{equation}
where $A_i,B_i,C_i,D_i\in\bbC$ and $A_iD_i-B_iC_i\neq 0$ for $i=1,\dots,4$. We refer to such transformations of 4 variables as $\mathfrak M_4$. 
Then, we define the action of $r$ on a quad-equation $Q=Q(x,y,z,w)$ by
\begin{equation}
 r.Q=(C_1x+D_1)(C_2y+D_2)(C_3z+D_3)(C_4w+D_4)Q(r.x,r.y,r.z,r.w),
\end{equation}
which we write for conciseness and simplicity as 
\begin{equation}
 r.Q=Q(r.x,r.y,r.z,r.w),
\end{equation}
as these are equivalent as quad-equations.
\end{definition}

\begin{definition}\rm
A quad-equation $P=P(x,y,z,w)$ is said to be equivalent up to M\"obius transformations to a quad-equation $Q=Q(x,y,z,w)$, 
if there exists an $\mathfrak M_4$ transformation $r$ such that
\begin{equation}
 P\equiv r.Q.
\end{equation}
In this case, the equivalence is denoted by 
\begin{equation}
 P\sim_m Q.
\end{equation}
\end{definition}

\subsection{Background}
\label{subsection:Background}
Integrable systems are widely applicable models of science, occurring in fluid dynamics, particle physics and optics. 
The prototypical example is the famous Korteweg-de Vries (KdV) equation whose solitary wave-like solutions interact elastically like particles, leading to the invention of the term \textit{soliton}. 
It is then natural to ask which discrete versions of such equations are also integrable. 
This question turns out to be related to consistency conditions for polynomials associated to faces of cubes as we explain below.

Integrable discrete systems were discovered \cite{NQC1983:MR719638,NCWQ1984:MR763123,QNCL1984:MR761644,nimmo1998integrable} from mappings that turn out to be consistent on multi-dimensional cubes. (We note that there are additional systems that do not fall into this class, see e.g., \cite[Chapter 3]{HJN2016:MR3587455}.) 
These are quad-equations in the sense of Definition \ref{def:quad-eqn}. 
In \cite{ABS2003:MR1962121,ABS2009:MR2503862,BollR2011:MR2846098,BollR2012:MR3010833,BollR:thesis}, 
Adler-Bobenko-Suris {\it et al.} classified quad-equations satisfying the following properties.
Consider a cube as shown in Figure \ref{fig:cube}. 
\begin{figure}[H]
\begin{center}
\includegraphics[width=0.4\textwidth]{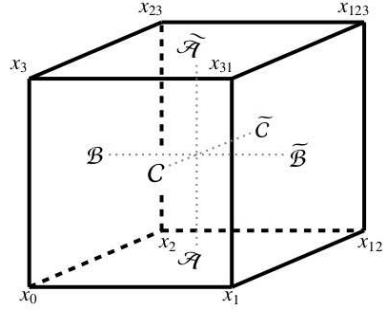}
\end{center}
\caption{A cube with variables $x_0$, \dots, $x_{123}$ associated with vertices and 
 equations $\mathcal A$, $\mathcal B$, $\mathcal C$, $\widetilde{\mathcal A}$, $\widetilde{\mathcal B}$, $\widetilde{\mathcal C}$ associated with faces.
}
\label{fig:cube}
\end{figure}
\noindent We assume that each face is associated with a quad-equation, and label these by
$\mathcal A$, $\mathcal B$, $\mathcal C$, $\widetilde{\mathcal A}$, $\widetilde{\mathcal B}$, $\widetilde{\mathcal C}$ as indicated in Figure \ref{fig:cube}, where the arguments of each polynomial are given in terms of the vertices as follows:
\begin{subequations}\label{eq:sixfaces}
\begin{align}
 {\mathcal A}(x_0,x_1,x_2,x_{12})&=0,\ 
 &\mathcal B(x_0,x_2,x_3,x_{23})&=0,\ 
 &\mathcal C(x_0,x_3,x_1,x_{31})&=0,\\
 \widetilde{\mathcal A}(x_3,x_{31},x_{23},x_{123})&=0,\ 
 &\widetilde{\mathcal B}(x_1,x_{12},x_{31},x_{123})&=0,\ 
 &\widetilde{\mathcal C}(x_2,x_{23},x_{12},x_{123})&=0.
\end{align}
\end{subequations}
Note that the six equations \eqref{eq:sixfaces} provide 3 different ways of evaluating $x_{123}$ in terms of $x_0$, $x_1$, $x_2$ and $x_3$. 
The requirement of consistency of these values gives rise to the following definition.
\begin{definition}[CAC and tetrahedron properties]~\\[-1.5em]
\begin{enumerate}[leftmargin=1.4cm,label={\rm (\alph*)}]
\item 
When the 3 results for $x_{123}$ given by equations \eqref{eq:sixfaces} are equal, this system of quad-equations is said to be \emph{3D consistent} or \emph{consistent around a cube} (CAC).
\item 
When the result for $x_{123}$ turns out to depend only on $x_1$, $x_2$, $x_3$, and $x_0$ depends only on $x_{12}$, $x_{23}$, and $x_{31}$, the system \eqref{eq:sixfaces} is said to have the \emph{tetrahedron property}.
\end{enumerate}
\end{definition}

A system of equations \eqref{eq:sixfaces} with the CAC property can be interpreted as a system of {\PDE}s on a cubic lattice $\bbZ^3$, by identifying iterates of a variable $U_{l,m,n}$ with the values at each vertex. 
We refer to such \PDE s as ABS equations. 
It turns out that ABS equations contain many well known integrable \PDE s \cite{NC1995:MR1329559,NCWQ1984:MR763123,NQC1983:MR719638,HirotaR1977:MR0460934}.

Reductions of integrable PDEs lead to Painlev\'e equations, which first arose in the search for new transcendental functions in the early 1900's\cite{PainleveP1902:MR1554937,GambierB1910:MR1555055,FuchsR1905:quelques}. 
Again a natural question is to ask whether discrete versions exist with analogous properties. 
This question led to the discovery of second-order difference equations called the discrete Painlev\'e equations\cite{GR2004:MR2087743,KNY2017:MR3609039,QNCL1984:MR761644,J2020:rnoti,joshi2019discrete}). It is now well-known that discrete Painlev\'e equations have initial value spaces with geometric structures that can be identified with root systems and affine Weyl groups \cite{SakaiH2001:MR1882403}. 
Sakai showed that there are 22 types of initial value spaces as shown in Table \ref{tab:sakai}.

\begin{table}[H]
\begin{center}
\caption{Types of spaces of initial values.}~\\
\label{tab:sakai}
\begin{tabular}{|l|l|}
\hline
Discrete type&Type of space of initial values\\
\hline
Elliptic&$A_0^{(1)}$\rule[-.5em]{0em}{1.6em}\\
\hline
Multiplicative&$A_0^{(1)\ast}$, $A_1^{(1)}$, $A_2^{(1)}$, $A_3^{(1)}$, \dots, $A_8^{(1)}$, $A_7^{(1)'}$\rule[-.5em]{0em}{1.6em}\\
\hline
Additive&$A_0^{(1)\ast\ast}$, $A_1^{(1)\ast}$, $A_2^{(1)\ast}$, $D_4^{(1)}$, \dots, $D_8^{(1)}$, $E_6^{(1)}$, $E_7^{(1)}$, $E_8^{(1)}$\rule[-.5em]{0em}{1.6em}\\
\hline
\end{tabular}
\end{center}
\end{table}

\subsection{Outline of the paper}
In this paper, we list systems of quad-equations arising from consistency conditions on a cuboctahedron. 
The setting is described in further detail and definitions are given in \S \ref{section:preliminary}. 
In \S \ref{section:classification_Cuboctahedron}, we provide a sequence of lemmas that lead to the main result Theorem \ref{theo:classification_CACO}. 
Detailed proofs of lemmas are proved in Appendices \ref{section:proof_lemma_CAO}, \ref{section:proof_lemma_r}, \ref{section:proof_lemma_CQ123456789_1} and \ref{section:proof_lemma_CACO_CO1CO2CO3}. 
Conditions on the parameters involved in these steps are proved in Appendix \ref{section:list_theorem_conditions}.

\section{Preliminary steps}\label{section:preliminary}
We start by recalling polynomials classified by Adler {\it et al.} \cite{ABS2003:MR1962121,ABS2009:MR2503862} and Boll \cite{BollR2011:MR2846098,BollR2012:MR3010833,BollR:thesis}, which give rise to quad-equations.
We then define the consistency property for quad-equations on octahedra before defining consistency for quad-equations defined on cuboctahedra.
In \S \ref{subsection:CACOPropPDEs}, we extend these polynomials to {\PDE}s on the FCC lattice before providing an example of such a {\PDE} system in \S \ref{subsection:example_PDEs}.

\subsection{Classifications of quad-equations modulo M\"obius transformations}
We show here that several ABS polynomials can be rewritten in terms of another set of polynomials (called $\rm N3$, $\rm N2$, $\rm N1$ below), which play a prominent role in the study of quad-equations on a cuboctahedron.

\begin{lemma}[Adler {\it et al.} \cite{ABS2009:MR2503862,ABS2003:MR1962121,BollR2012:MR3010833,BollR2011:MR2846098,BollR:thesis}]
\label{lemma:classification_ABS_polys}
Under the equivalence relation $\sim_m$, any quad-equation $Q(x,y,z,w)$ is equivalent to one of the following quad-equations:
{\allowdisplaybreaks
\begin{align}
 &\polyQ4(x,y,z,w;\al,\be)
 =\sn\al\sn\be\sn{\al+\be}(k^2xyzw+1)\notag\\
 &\hspace{8em}-\sn\al(xy+zw)-\sn\be(xw+yz)\notag\\
 &\hspace{8em}+\sn{\al+\be}(xz+yw),
 \tag{$\polyQ4$}\label{eqn:Q4}\\
 &\polyQ3(x,y,z,w;\al,\be;\de)
 =(\al-\al^{-1})(xy+zw)+(\be-\be^{-1})(xw+yz)\notag\\
 &\hspace{9em}-(\al\be-\al^{-1}\be^{-1})(xz+yw)\notag\\
 &\hspace{9em}+\frac{\de(\al-\al^{-1})(\be-\be^{-1})(\al\be-\al^{-1}\be^{-1})}{4},
 \tag{$\polyQ3$}\label{eqn:Q3}\\
 &\polyQ2(x,y,z,w;\al,\be)
 =\al(x-w)(y-z)+\be(x-y)(w-z)\notag\\
 &\hspace{8em}-\al\be(\al+\be)(x+y+z+w)\notag\\
 &\hspace{8em}+\al\be(\al+\be)(\al^2+\al\be+\be^2),
 \tag{$\polyQ2$}\label{eqn:Q2}\\
 &\polyQ1(x,y,z,w;\al,\be;\de)=\al(x-w)(y-z)+\be(x-y)(w-z)\notag\\
 &\hspace{9em}-\de\al\be(\al+\be),
 \tag{$\polyQ1$}\label{eqn:Q1}\\
 &\polyH3(x,y,z,w;\al,\be;\de_1,\de_2)=\al(xy+zw)-\be(xw+yz)\notag\\
 &\hspace{10.5em}+(\al^2-\be^2)\left(\de_1+\dfrac{\de_2}{\al\be}yw\right),
 \tag{$\polyH3$}\label{eqn:H3}\\
 &\polyH2(x,y,z,w;\al,\be;\de)=(x-z)(y-w)+(\be-\al)(x+y+z+w)\notag\\
 &\hspace{9em}+\be^2-\al^2+\de(\be-\al)(2y+\al+\be)(2w+\al+\be)\notag\\
 &\hspace{9em}+\de(\be-\al)^3,
 \tag{$\polyH2$}\label{eqn:H2}\\
 &\polyH1(x,y,z,w;\al,\be;\de)=(x-z)(y-w)+(\be-\al)(1+\de yw),\tag{$\polyH1$}\label{eqn:H1}\\[0.2em]
 &\polyD4(x,y,z,w;\de_1,\de_2,\de_3)=xz+yw+\de_1yz+\de_2zw+\de_3,\tag{$\polyD4$}\label{eqn:D4}\\[0.2em]
 &\polyD3(x,y,z,w)=y+xz+xw+zw,\tag{$\polyD3$}\label{eqn:D3}\\[0.2em]
 &\polyD2(x,y,z,w;\de)=x+z+y(w+\de x),\tag{$\polyD2$}\label{eqn:D2}\\[0.2em]
 &\polyD1(x,y,z,w)=x+y+z+w,\tag{$\polyD1$}\label{eqn:D1}
\end{align}
}\noindent
where $\al,\be\in\bbC$, $\de,\de_1,\de_2,\de_3\in\{0,1\}$ and the parameter $k$ $(k\neq 0,1)$ is the modulus of the Jacobi elliptic function $\sn{t}$.
The first four types are collectively called $Q$-type equations, while the remaining equations are called $H$-type, but separated into $H^4$- {\rm (for \eqref{eqn:H3}--\eqref{eqn:H1})} and $H^6$-types {\rm (for \eqref{eqn:D4}--\eqref{eqn:D1})} respectively.
Note that current form of $\polyQ4$ given by the Jacobi elliptic function was first found in \cite{HietarintaJ2005:MR2217106}.
\end{lemma}
In the proof of Lemma \ref{lemma:classification_ABS_polys}, certain biquadratics and discriminants played a prominent role. We recall them here in order to prove that the polynomials in that lemma are equivalent to another set of polynomials below.

\begin{remark}
Let the irreducible multi-affine polynomial be given by $Q(x_1, x_2, x_3, x_4)$. 
Then, the six biquadratics $h^{ij}$ are given by
\begin{equation}
 h^{ij}(x_i, x_j)=\frac{\partial Q}{\partial x_k}\frac{\partial Q}{\partial x_l}-Q\,\frac{\partial^2 Q}{\partial x_k\partial x_l},
\end{equation}
where $\{k,l\}$ is the complement of $\{i, j\}$ in $\{1,2,3,4\}$,
and, furthermore, the four discriminants $r_i$, $i=1,2,3,4$, are given by
\begin{equation}
r_i(x_i)=\left(\frac{\partial h^{ij}}{\partial x_j}\right)^2-2 h^{ij}\left(\frac{\partial^2 h^{ij}}{\partial x_j^2}\right).
\end{equation}
Note that $r_i$ does not depend on the choice of $j$.
The polynomials listed in Lemma \ref{lemma:classification_ABS_polys} are categorized according to whether $h^{ij}$ is degenerate or not and subsequently by the multiplicity of the zeroes of $r_i$. 
Here, $h^{ij}$ is called degenerate if the solution of $h^{ij}=0$ includes points (rather than lines).
\end{remark}

For later consideration on cuboctahedra, it is convenient to define the following three polynomials:
\begin{align}
 &\polyN3(x,y,z,w;\al_1,\al_2,\al_3,\al_4)=\al_1 xy+\al_2 zw+\al_3 xw+\al_4 yz,\tag{$\polyN3$}\\
 &\polyN2(x,y,z,w;\al_1,\al_2,\al_3,\al_4)=(x+z)(\al_1 y+\al_2 w)+\al_3 y+\al_4 w,\tag{$\polyN2$}\\
 &\polyN1(x,y,z,w;\al_1,\al_2,\al_3,\al_4)=\al_1(x+z)(y+w)+\al_2(x+z)+\al_3(y+w)+\al_4,\tag{$\polyN1$}
\end{align}
where $\al_i\in\bbC$, $i=1,2,3,4$.
We will refer to these polynomials as $N$-type. 

We show below that the quad-equations $\polyH3_{\de_1=\de_2=0}$, $\polyH1$, $\polyD4_{\de_1=\de_2=0}$, $\polyD2_{\de=0}$ and $\polyD1$ are equivalent to certain cases of $N$-type quad-equations. 
In other words, Lemma \ref{lemma:classification_ABS_polys} can be rewritten as the following.
\begin{lemma}
\label{lemma:classification_JN_polys}
Under the equivalence relation $\sim_m$, any quad-equation $Q(x,y,z,w)$ is equivalent to one of the following quad-equations:
\begin{equation*}
\begin{array}{lll}
 \polyQ4(x,y,z,w;\al,\be),
 &\polyQ3(x,y,z,w;\al,\be;\de),
 &\polyQ2(x,y,z,w;\al,\be),\\
 \polyQ1(x,y,z,w;\al,\be;\de),
 &\polyH3(x,y,z,w;\al,\be;\de_1,\de_2),
 &\polyH2(x,y,z,w;\al,\be;\de),\\
 \polyD4(x,y,z,w;\de_1,\de_2,\de_3),
 &\polyD3(x,y,z,w),
 &\polyD2(x,y,z,w;1),\\
 \polyN3(x,y,z,w;\al_1,\al_2,\al_3,\al_4),
 &\multicolumn{2}{l}{\polyN2(x,y,z,w;\al_1,\al_2,\al_3,\al_4),}\\
 \polyN1(x,y,z,w;\al_1,\al_2,\al_3,\al_4),
\end{array}
\end{equation*}
where $(\de_1,\de_2)\neq(0,0)$ for $\polyH3$ and $\polyD4$.
\end{lemma}
\begin{proof}
The only assertions we need to prove concern the quad-equations $\polyH3_{\de_1=\de_2=0}$, $\polyH1$, $\polyD4_{\de_1=\de_2=0}$, $\polyD2_{\de=0}$ and $\polyD1$. 
Our argument relies on the properties of their respective biquadratics and discriminants.
  
Notice that for $\al_i\neq0$, $i=1,2,3,4$, and $\al_1\al_2\neq \al_3\al_4$, the biquadratics and discriminants of $\polyN3(x,y,z,w;\al_1,\al_2,\al_3,\al_4)$ are given by 
\begin{subequations}\label{eqns:biquadratics_discriminants_N3}
\begin{align}
 &h^{ij}=x_ix_j,~i\in\{1,3\},~j\in\{2,4\},
 &&h^{13}=(\al_3x_1+\al_2x_3)(\al_1x_1+\al_4x_3),\\
 &h^{24}=(\al_4x_2+\al_2x_4)(\al_1x_2+\al_3x_4),
 &&r_k={x_k}^2,~k=1,2,3,4.
\end{align}
\end{subequations}
We note that these are equivalent to the corresponding biquadratics and discriminants for $\polyH3_{\de_1=\de_2=0}$ under the condition $\al^2\neq \be^2$. 
Therefore, by the ABS classification results, these two polynomials ($\polyN3$ and $\polyH3_{\de_1=\de_2=0}$ with conditions) give equivalent quad-equations.

Now we consider the case where one of the parameters $\al_i$ is zero, while the remaining parameters $\al_j$ are not zero. 
Then, the biquadratics and discriminants of $\polyN3$ are given by \eqref{eqns:biquadratics_discriminants_N3} with $\al_i=0$.
This case corresponds to the biquadratic and discriminant of $\polyD4_{\de_1=\de_2=0,\de_3=1}$. 
Moreover, when $\{\al_1=\al_2=0, \al_3,\al_4\neq0\}$ or $\{\al_1,\al_2\neq0, \al_3=\al_4=0\}$,
the biquadratics and discriminants of $\polyN3$ are equivalent to the corresponding biquadratics and discriminants for $\polyD4_{\de_1=\de_2=\de_3=0}$. 

In other words, we gave shown that $\polyN3$ is equivalent to $\polyH3_{\de_1=\de_2=0}$ and $\polyD4_{\de_1=\de_2=0}$ with appropriate changes of variables and transformations of parameters.

The remaining cases can be proved in exactly the same way. 
We find that $\polyN2$ is equivalent to $\polyH1_{\de=1}$ and $\polyD2_{\de=0}$, while $\polyN1$ is equivalent to $\polyH1_{\de=0}$ and $\polyD1$. 
This completes the proof.
\end{proof}

\subsection{Consistency around an octahedron property}\label{subsection:def_CAO}
In this section, we give a definition of a consistency around an octahedron.

Consider the octahedron centered around the origin whose six vertices are given by
\begin{equation}
 V=\{\pm(\ep_1+\ep_2),\pm(\ep_2+\ep_3),\pm(\ep_1+\ep_3)\},
\end{equation}
where $\{\ep_1,\ep_2,\ep_3\}$ form the standard basis of $\bbR^3$.
We assign the variables $u(\bml)$ to the vertices $\bml\in V$
and impose the following relations:
\begin{equation}\label{eqn:octahedron_Q1Q2Q3}
 Q_1\left(u_4,u_2,u_1,u_5\right)=0,\quad
 Q_2\left(u_2,u_6,u_5,u_3\right)=0,\quad
 Q_3\left(u_6,u_4,u_3,u_1\right)=0,
\end{equation}
where $Q_i$, $i=1,2,3$, are quad-equations and
\begin{equation}
\begin{split}
 &u_1=u(\ep_2+\ep_3),\quad
 u_2=u(-\ep_1-\ep_3),\quad
 u_3=u(\ep_1+\ep_2),\quad
 u_4=u(-\ep_2-\ep_3),\\
 &u_5=u(\ep_1+\ep_3),\quad
 u_6=u(-\ep_1-\ep_2).
\end{split}
\end{equation}
The planes that pass through the vertices $\{u_4,u_2,u_1,u_5\}$, $\{u_2,u_6,u_5,u_3\}$ and $\{u_6,u_4,u_3,u_1\}$ give 3 quadrilaterals that lie in the interior of the octahedron (see Figure \ref{fig:octahedron_3D}).
The quad-equations $Q_i$, $i=1,2,3$, are assigned to the quadrilaterals.
The consistency around an octahedron property is defined by the following.

\begin{figure}[htbp]
 \begin{center}
 \includegraphics[width=0.4\textwidth]{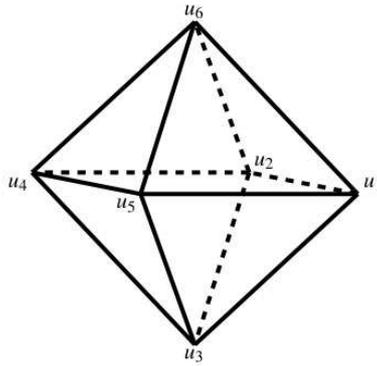}
 \end{center}
\caption{An octahedron labelled with vertices $u_1$,\dots,$u_6$.}
\label{fig:octahedron_3D}
\end{figure}

\begin{definition}\rm
The octahedron with quad-equations $\{Q_1,Q_2,Q_3\}$ is said to have a {\it consistency around an octahedron (CAO) property}, 
if each quad-equation can be obtained from the other two equations.
An octahedron is said to be a {\it CAO octahedron}, if it has the CAO property.
\end{definition}

\subsection{Consistency around a cuboctahedron property}\label{subsection:def_CACO}
In this section, we give a definition of a consistency around a cuboctahedron.

We consider the cuboctahedron centered around the origin whose twelve vertices are given by
\begin{equation}
 V=\set{\pm\ep_i\pm\ep_j}{i,j\in\bbZ,~1\leq i<j\leq 3},
\end{equation}
where $\{\ep_1,\ep_2,\ep_3\}$ form the standard basis of $\bbR^3$.
We assign the variables $u(\bml)$ to the vertices $\bml\in V$
and impose the following relations:
\begin{subequations}\label{eqns:V_caco_general}
\begin{align}
 &Q_1\left(u_5,u_1,v_5,v_4\right)=0,\quad
 Q_2\left(v_2,v_1,u_2,u_4\right)=0,\quad
 Q_3\left(u_3,u_5,v_3,v_2\right)=0,
 \label{eqn:V_caco_general_1}\\
 &Q_4\left(v_6,v_5,u_6,u_2\right)=0,\quad
 Q_5\left(u_1,u_3,v_1,v_6\right)=0,\quad
 Q_6\left(v_4,v_3,u_4,u_6\right)=0,
 \label{eqn:V_caco_general_2}\\
 &Q_7\left(u_4,u_2,u_1,u_5\right)=0,\quad
 Q_8\left(u_2,u_6,u_5,u_3\right)=0,\quad
 Q_9\left(u_6,u_4,u_3,u_1\right)=0,
 \label{eqn:V_caco_general_3}
\end{align}
\end{subequations}
where $Q_i$, $i=1,\dots,9$, are quad-equations and
\begin{equation}\label{eqn:ui_vi_ufunction}
\begin{array}{llll}
 u_1=u(\ep_2+\ep_3),
 &u_2=u(-\ep_1-\ep_3),
 &u_3=u(\ep_1+\ep_2),
 &u_4=u(-\ep_2-\ep_3),\\
 u_5=u(\ep_1+\ep_3),
 &u_6=u(-\ep_1-\ep_2),
 &v_1=u(\ep_2-\ep_3),
 &v_2=u(\ep_1-\ep_3),\\
 v_3=u(\ep_1-\ep_2),
 &v_4=u(-\ep_2+\ep_3),
 &v_5=u(-\ep_1+\ep_3),
 &v_6=u(-\ep_1+\ep_2).
\end{array}
\end{equation}
Note that quad-equations $Q_i$, $i=1,\dots,6$, are assigned to the faces of the cuboctahedron.
Moreover, $u_i$, $i=1,\dots,6$, collectively form the vertices of an octahedron
and quad-equations $Q_i$, $i=7,8,9$, are assigned to the quadrilaterals that appear as sections passing through four vertices of the octahedron.
(See Figure \ref{fig:cuboctahedron_3D}.)

\begin{figure}[htbp]
\centering
\begin{subfigure}[b]{0.43\textwidth}
\centering
 \includegraphics[width=\textwidth]{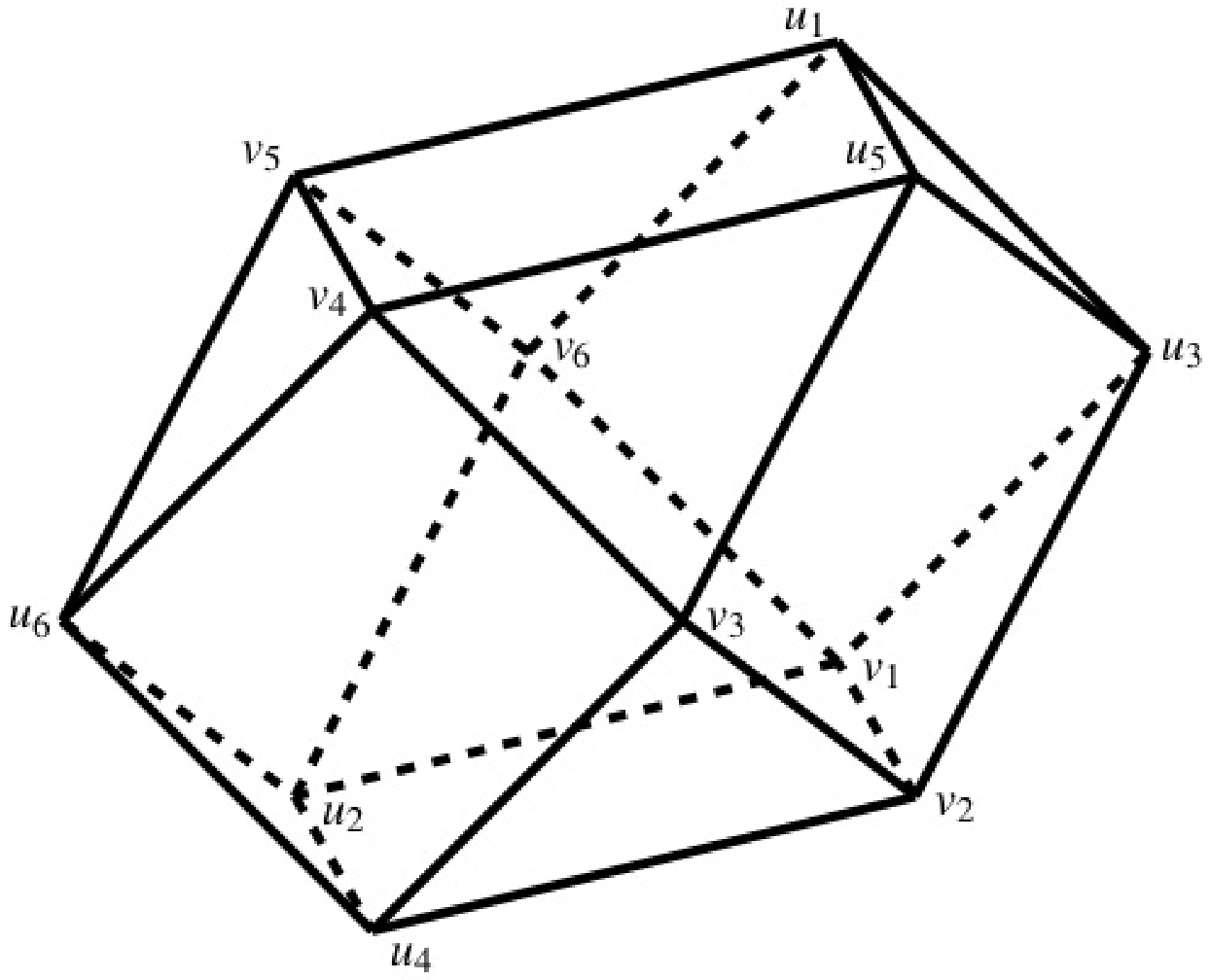}
 \caption{A cuboctahedron labelled with vertices $v_1$,\dots,$v_6$, $u_1$,\dots,$u_6$.}
 \label{fig:cuboctahedron_3D_a}
\end{subfigure}
\hspace{2em}
\begin{subfigure}[b]{0.43\textwidth}
\centering
 \includegraphics[width=\textwidth]{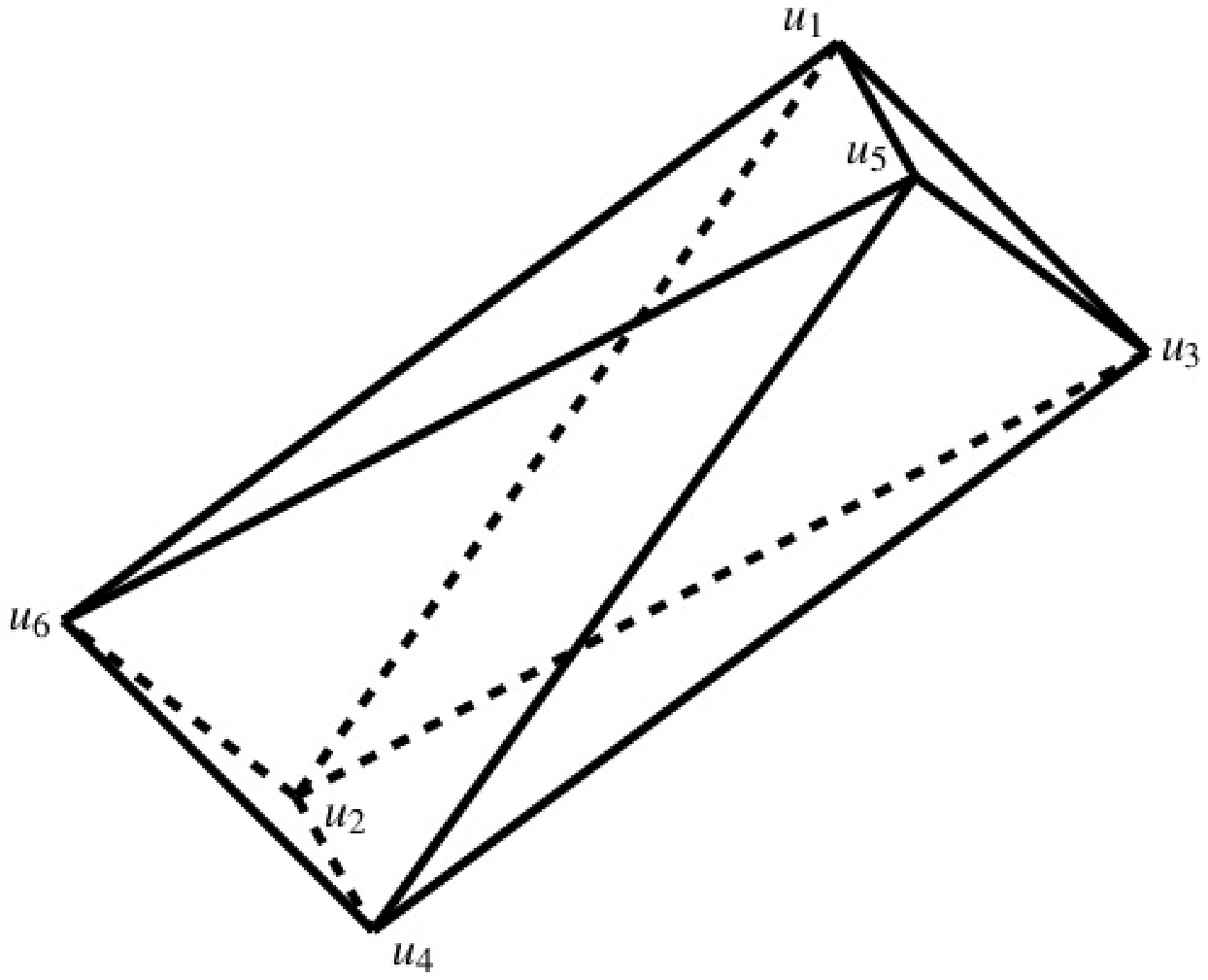}
 \caption{An octahedron labelled with vertices $u_1$,\dots,$u_6$.}
\end{subfigure}
 \caption{A cuboctahedron and an interior octahedron.}
 \label{fig:cuboctahedron_3D}
\end{figure}

\begin{definition}[CACO property]\label{def:CACO_cuboctahedron}\rm
The cuboctahedron with quad-equations $\{Q_1,\dots,Q_9\}$ is said to have a {\it consistency around a cuboctahedron (CACO) property}, 
if the following properties hold.
\begin{description}
\item[(i)] 
The octahedron with quad-equations $\{Q_7,Q_8,Q_9\}$ has the CAO property.
\item[(ii)]
Assume that $u_1,\dots,u_6$ are given so as to satisfy $Q_i=0$, $i=7,8,9$, and, in addition, $v_k$ is given, for some $k\in\{1,\dots,6\}$.
Then, quad-equations $Q_i$, $i=1\dots,6$, determine the variables $v_j$, $j\in\{1,\dots,6\}\backslash\{k\}$, uniquely.
\end{description}
A cuboctahedron is said to be a {\it CACO cuboctahedron}, if it has the CACO property.
\end{definition}

\begin{definition}[Square property]\rm
The CACO cuboctahedron with quad-equations $\{Q_1,\dots,Q_9\}$ is said to have a {\it square property}, 
if there exist polynomials $K_i=K_i(x,y,z,w)$, $i=1,2,3$, where
\begin{equation}
 \deg_x{K_i}=\deg_w{K_i}=1,\quad
 1\leq\deg_y{K_i},\,\deg_z{K_i},
\end{equation}
satisfying
\begin{equation}
 K_1(v_1,u_1,u_4,v_4)=0,\quad
 K_2(v_2,u_2,u_5,v_5)=0,\quad
 K_3(v_3,u_3,u_6,v_6)=0. 
\end{equation}
Then, each equation $K_i=0$ is called a {\it square equation}.
\end{definition}

To explain Definition \ref{def:CACO_cuboctahedron} in plainer language, we use an orthogonal projection of a cuboctahedron to two dimensions, shown in Figure \ref{fig:cuboctahedron2D}. Note that the vertices labelled $u_1,\ldots, u_6$ form a hexagram, while the convex hull of the projection which connects the vertices $v_1,\ldots, v_6$ form a hexagon. Below, we refer to the collection of all such vertices in terms of $u({\bm p})$, by using the notation of Equation \eqref{eqn:ui_vi_ufunction}.
\begin{remark}\label{remark:CACO_2D}
Assume that all vertices in the inner hexagram and a vertex in the outer hexagon of an orthogonal projection of a cuboctahedron are given. (Refer to Figure \ref{fig:cuboctahedron2D}.) 

Label the given vertex in the outer hexagon by $u({\bm p})$, for some ${\bm p}$.
Then, there are two ways of obtaining the value of $u(-{\bm p})$.
One is obtained by using quad-equations that occur in the walk from $u({\bm p})$ to $u(-{\bm p})$ clockwise around the outer hexagon, while the other is by using quad-equations in the anti-clockwise walk around the hexagon.
The CACO property ensures that $u(-{\bm p})$ is determined uniquely regardless of the direction.

For example, if ${\bm p}=\ep_2-\ep_3$, then
$u({\bm p})=v_1$, while
$u(-{\bm p})=v_4$. Walking in a clockwise direction means using quad-equations $Q_i$, $i=2,3,6$, while the anti-clockwise direction means using quad-equations $Q_i$, $i=5,4,1$.
\end{remark}
 
\begin{figure}[H]
 \begin{minipage}{0.45\hsize}
 \begin{center}
  \hspace*{-2.em}\includegraphics[width=1\textwidth]{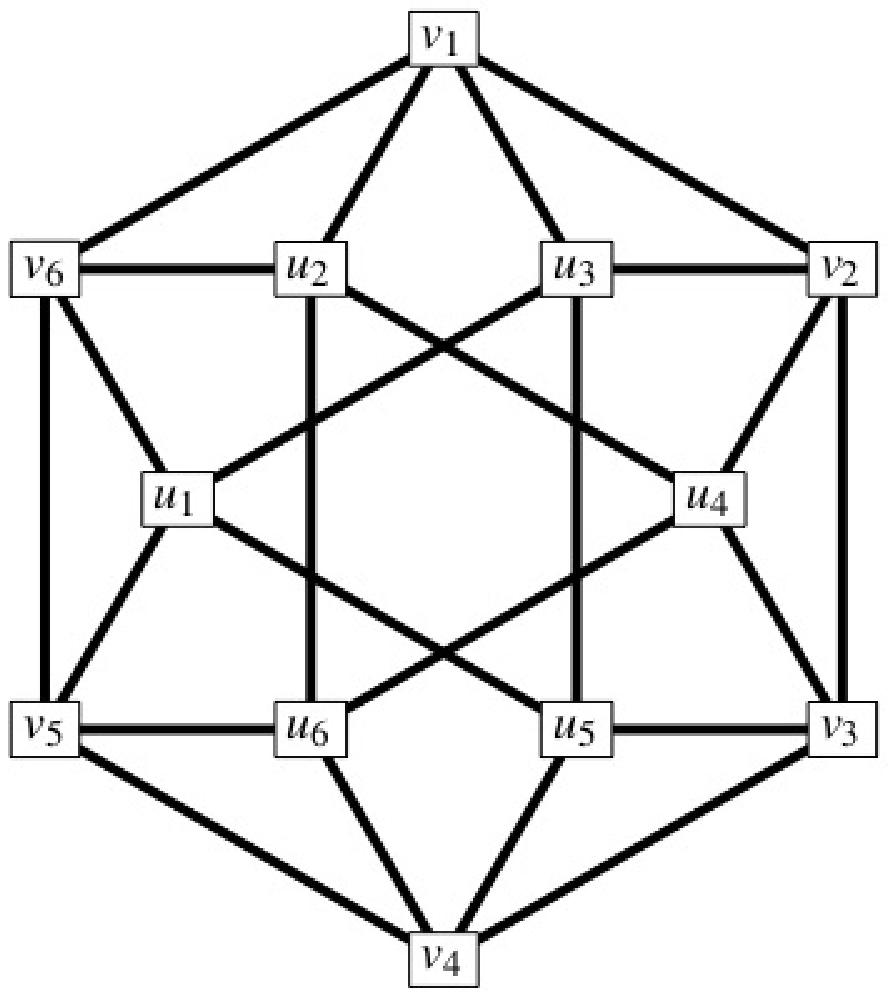}
 \end{center}
 \end{minipage}
 \begin{minipage}{0.5\hsize}
\begin{flushleft}
 \hspace*{-1.5em}\includegraphics[width=0.57\textwidth]{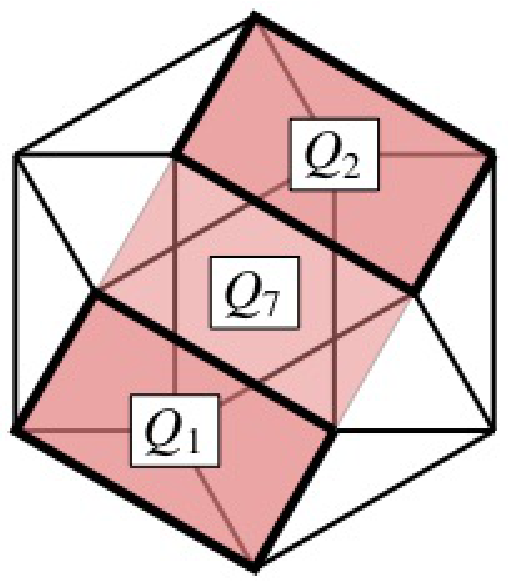}\\[-4.5em]
 \hspace*{8.5em}\includegraphics[width=0.57\textwidth]{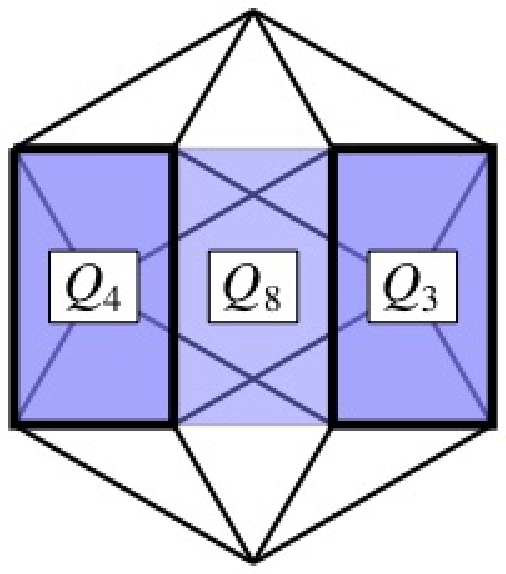}\\[-4.5em]
 \hspace*{-1.5em}\includegraphics[width=0.57\textwidth]{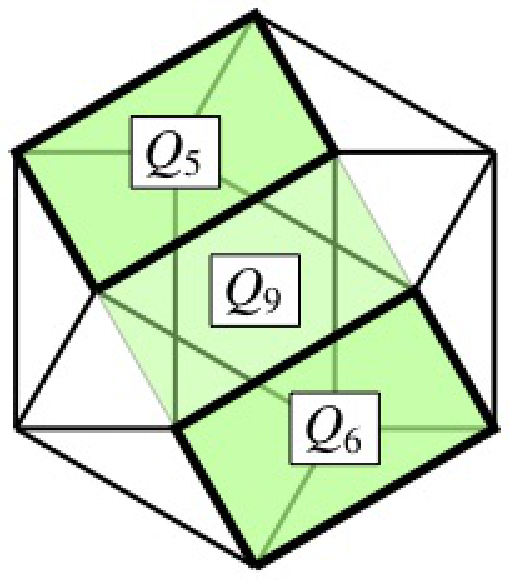}
\end{flushleft}
 \end{minipage}
\caption{Orthogonal projection of the cuboctahedron centered on a triangular face. Figures on the right show where quad-equations $Q_i$, $i=1,\dots,9$, lie in the projection.}
\label{fig:cuboctahedron2D}
\end{figure}

\subsection{CACO property to \PDE s}\label{subsection:CACOPropPDEs}
We now explain how to associate quad-equations with \PDE s in three-dimensional space. This requires us to consider overlapping cuboctahedra that lead to two-dimensional tessellations consisting of quadrilaterals. 
For each given cuboctahedron, there are twelve overlapping cuboctahedra. One such pair of overlapped cuboctahedra is shown in Figure \ref{fig:2cuboctahedra3D}.

The two cuboctahedra in Figure \ref{fig:2cuboctahedra3D} are coloured yellow and blue to distinguish them. Note that one of the vertices of the yellow cuboctahedron forms the centre of the blue one, and vice versa, one of the vertices of the blue cuboctahedron forms the centre of the yellow one. Note also that there are four shared vertices, which are coloured green in the figure.

Notice that there are pairs of quadrilateral faces that form 2D planes in any sequence of overlapped cuboctahedra. For example, in Figure \ref{fig:2cuboctahedra3D}, the bottom faces form such a plane and so do the top faces. To fix notation, we take each cuboctahedron to be a copy of the one drawn in Figure \ref{fig:cuboctahedron_3D_a}. So, for example, the bottom face of each cuboctahedron corresponds to the quadrilateral $Q_2$ and the top face corresponds to $Q_1$. 

The twelve overlapping cuboctahedra around a given one provide six directions of tiling by quadrilaterals. For later convenience, we label directions by $\ep_i\pm \ep_j$, $1\leq i<j\leq 3$. Vertices labelled in this way form the FCC lattice:
\begin{equation}\label{eqn:def_Omega}
 \Omega=\set{\sum_{i=1}^3l_i\ep_i}{l_i\in\bbZ,~l_1+l_2+l_3\in2\bbZ}.
\end{equation}
Such vertices are interpreted as being iterated on each successive cuboctahedron. 
We are then led to 6 partial difference equations that occur on these iterated cuboctahedra:
\begin{subequations}\label{eqns:PDEs_P123456}
\begin{align}
 &P_1\left(u_{\ol{13}},u_{\ol{23}},u_{\ul{1}\ol{3}},u_{\ul{2}\ol{3}}\right)=0,\qquad
 P_2\left(u_{\ol{12}},u_{\ol{13}},u_{\ol{1}\ul{2}},u_{\ol{1}\ul{3}}\right)=0,\\
 &P_3\left(u_{\ol{23}},u_{\ol{12}},u_{\ol{2}\ul{3}},u_{\ul{1}\ol{2}}\right)=0,\qquad
 P_4\left(u_{\ul{23}},u_{\ul{13}},u_{\ol{23}},u_{\ol{13}}\right)=0,\\
 &P_5\left(u_{\ul{13}},u_{\ul{12}},u_{\ol{13}},u_{\ol{12}}\right)=0,\qquad
 P_6\left(u_{\ul{12}},u_{\ul{23}},u_{\ol{12}},u_{\ol{23}}\right)=0,
\end{align}
\end{subequations}
where $u=u(\bml)$ and $\bml\in\Omega$.
Here, $P_i$, $i=1,\dots,6$, are quad-equations,
and subscripts $\bar{i}$ and $\ul{j}$ mean $\bml\to\bml+\ep_i$ and $\bml\to\bml-\ep_j$, respectively.

\begin{remark}
It is important to note that $P_i$, $i=1,\dots,6$,  are not necessarily autonomous, because its coefficients may vary with $\bml$. This possibility is embedded in our notation. For example,
$P_1\left(u_{\ol{1}\ul{3}},u_{\ol{2}\ul{3}},u_{\ul{13}},u_{\ul{23}}\right)$
means 
$P_1\left(u_{\ol{13}},u_{\ol{23}},u_{\ul{1}\ol{3}},u_{\ul{2}\ol{3}}\right)$
after a shift $\bml\to\bml-2\ep_3$. In this case, the coefficients will also be shifted appropriately.
\end{remark}

\begin{figure}[htbp]
 \begin{center}
 \includegraphics[width=0.8\textwidth]{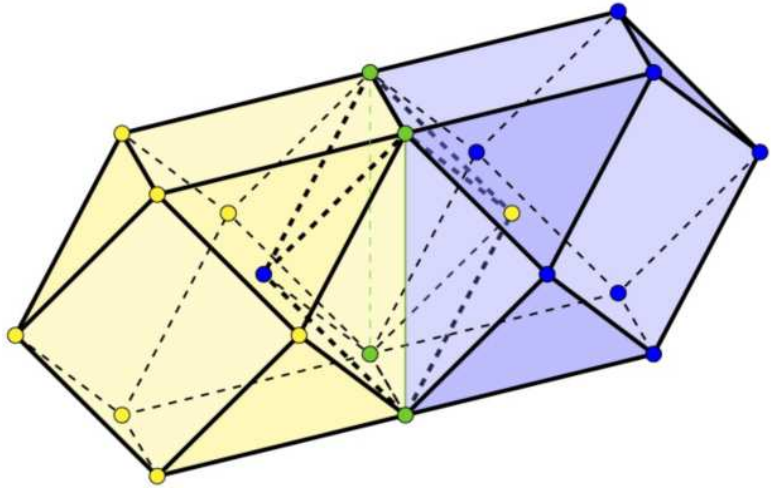}
 \end{center}
\caption{The two overlapping cuboctahedra, one with yellow and one with blue. The green vertices and edges correspond to where the two cuboctahedra overlap.}
\label{fig:2cuboctahedra3D}
\end{figure}

Conversely, given $\bml\in\Omega$, we obtain the cuboctahedron centered around $\bml$. We refer to its quad-equations as before by
$\{Q_1(\bml),\dots,Q_9(\bml)\}$. Moreover, the overlapped region gives an octahedron centred around $\bml+\ep_3$, and we label its quad-equations by $\{\hat{Q}_1(\bml),\hat{Q}_2(\bml),\hat{Q}_3(\bml)\}$.

Each such quad-equation is identified with the 6 partial difference equations given in Equations \eqref{eqns:PDEs_P123456} in the following way. 
Firstly, for $Q_1,\ldots, Q_9$, we use
\begin{subequations}
\begin{align}
 &Q_1(\bml)=P_1\left(u_{\ol{13}},u_{\ol{23}},u_{\ul{1}\ol{3}},u_{\ul{2}\ol{3}}\right)=0,\quad
 Q_2(\bml)=P_1\left(u_{\ol{1}\ul{3}},u_{\ol{2}\ul{3}},u_{\ul{13}},u_{\ul{23}}\right)=0,\\
 &Q_3(\bml)=P_2\left(u_{\ol{12}},u_{\ol{13}},u_{\ol{1}\ul{2}},u_{\ol{1}\ul{3}}\right)=0,\quad
 Q_4(\bml)=P_2\left(u_{\ul{1}\ol{2}},u_{\ul{1}\ol{3}},u_{\ul{12}},u_{\ul{13}}\right)=0,\\
 &Q_5(\bml)=P_3\left(u_{\ol{23}},u_{\ol{12}},u_{\ol{2}\ul{3}},u_{\ul{1}\ol{2}}\right)=0,\quad
 Q_6(\bml)=P_3\left(u_{\ul{2}\ol{3}},u_{\ol{1}\ul{2}},u_{\ul{23}},u_{\ul{12}}\right)=0,\\
 &Q_7(\bml)=P_4\left(u_{\ul{23}},u_{\ul{13}},u_{\ol{23}},u_{\ol{13}}\right)=0,\quad
 Q_8(\bml)=P_5\left(u_{\ul{13}},u_{\ul{12}},u_{\ol{13}},u_{\ol{12}}\right)=0,\\
 &Q_9(\bml)=P_6\left(u_{\ul{12}},u_{\ul{23}},u_{\ol{12}},u_{\ol{23}}\right)=0,
\end{align}
\end{subequations}
and for $\hat{Q}_1,\hat{Q}_2,\hat{Q}_3$, we use
\begin{subequations}
\begin{align}
 &\hat{Q}_1(\bml)=P_1\left(u_{\ol{13}},u_{\ol{23}},u_{\ul{1}\ol{3}},u_{\ul{2}\ol{3}}\right)=0,\quad
 \hat{Q}_2(\bml)=P_2\left(u_{\ol{23}},u_{\ol{33}},u_{\ul{2}\ol{3}},u\right)=0,\\
 &\hat{Q}_3(\bml)=P_3\left(u_{\ol{33}},u_{\ol{13}},u,u_{\ul{1}\ol{3}}\right)=0. 
\end{align}
\end{subequations}
We are now in a position to define the CACO property for \PDE s.

\begin{definition}[CACO property for {\PDE}s]\label{def:CACO_PDEs}\rm
We transfer the definition of CACO and square properties to the system of \PDE s \eqref{eqns:PDEs_P123456} in an obvious way, as follows. 
The system of \PDE s \eqref{eqns:PDEs_P123456} is said to have the CACO property, 
if the following properties hold.
\begin{description}
\item[(i)]
The cuboctahedra with quad-equations $\{Q_1(\bml),\dots,Q_9(\bml)\}$ have the CACO and square properties,
and, furthermore, the square equations 
\begin{align*}
 &K_1\left(u_{\ol{2}\ul{3}},u_{\ol{23}},u_{\ul{23}},u_{\ul{2}\ol{3}}\right)=0,\qquad
 K_2\left(u_{\ul{1}\ol{3}},u_{\ol{13}},u_{\ul{13}},u_{\ol{1}\ul{3}}\right)=0,\\
 &K_3\left(u_{\ol{1}\ul{2}},u_{\ol{12}},u_{\ul{12}},u_{\ul{1}\ol{2}}\right)=0,
\end{align*}
are consistent with the {\PDE}s $P_i=0$, $i=1,2,3$.
\item[(ii)]
The octahedra with quad-equations $\{\hat{Q}_1(\bml),\hat{Q}_2(\bml),\hat{Q}_3(\bml)\}$ have the CAO property.
\end{description}
\end{definition}

\begin{remark}
Property {\bf (i)} in Definition \ref{def:CACO_PDEs}
imposes constraints on the system \eqref{eqns:PDEs_P123456}. We can verify the number of degrees of freedom in the constrained initial values for this system by considering the planes 
where each pair of equations $P_1$ and $K_3$, $P_2$ and $K_1$, and $P_3$ and $K_2$ holds. 
\end{remark}

\begin{remark}
Property {\bf (ii)} in Definition \ref{def:CACO_PDEs} means that the octahedra arising from the overlaps between cuboctahedra have the CAO property.
Note that these octahedra are different from the interior octahedra arising inside each individual cuboctahedron.
\end{remark}

\subsection{Example: a system of \PDE s which has the CACO property}\label{subsection:example_PDEs}
In this section, we provide an example of a system of \PDE s with the CACO property.

Consider the following system of \PDE s:
\begin{subequations}\label{eqns:samplePDEs_P123456}
\begin{align}
 &P_1\left(u_{\ol{13}},u_{\ol{23}},u_{\ul{1}\ol{3}},u_{\ul{2}\ol{3}}\right)
 =\polyN3\left(u_{\ol{13}},u_{\ol{23}},u_{\ul{1}\ol{3}},u_{\ul{2}\ol{3}};a_2,a_1,a_3,a_4\right)=0,
 \label{eqn:samplePDEs_P123456_1}\\
 &P_2\left(u_{\ol{12}},u_{\ol{13}},u_{\ol{1}\ul{2}},u_{\ol{1}\ul{3}}\right)
 =\polyN3\left(u_{\ol{12}},u_{\ol{13}},u_{\ol{1}\ul{2}},u_{\ol{1}\ul{3}};a_6,a_5,a_7,a_8\right)=0,
 \label{eqn:samplePDEs_P123456_2}\\
 &P_3\left(u_{\ol{23}},u_{\ol{12}},u_{\ol{2}\ul{3}},u_{\ul{1}\ol{2}}\right)
 =\polyN3\left(u_{\ol{23}},u_{\ol{12}},u_{\ol{2}\ul{3}},u_{\ul{1}\ol{2}};a_{10},a_9,a_{11},a_{12}\right)=0,
 \label{eqn:samplePDEs_P123456_3}\\
 &P_4\left(u_{\ul{23}},u_{\ul{13}},u_{\ol{23}},u_{\ol{13}}\right)
 =\polyN3\left(u_{\ul{23}},u_{\ul{13}},u_{\ol{23}},u_{\ol{13}};a_1,a_2,a_3,a_4\right)=0,\\
 &P_5\left(u_{\ul{13}},u_{\ul{12}},u_{\ol{13}},u_{\ol{12}}\right)
 =\polyN3\left(u_{\ul{13}},u_{\ul{12}},u_{\ol{13}},u_{\ol{12}};a_5,a_6,a_7,a_8\right)=0,\\
 &P_6\left(u_{\ul{12}},u_{\ul{23}},u_{\ol{12}},u_{\ol{23}}\right)
 =\polyN3\left(u_{\ul{12}},u_{\ul{23}},u_{\ol{12}},u_{\ol{23}};a_9,a_{10},a_{11},a_{12}\right)=0,
\end{align}
\end{subequations}
where $u=u(\bml)$, $\bml=\sum_{i=1}^3l_i\ep_i\in\Omega$
(recall that $\Omega$ is defined by \eqref{eqn:def_Omega})
and
\begin{subequations}
\begin{align}
 &a_1=\al_{12}+(-1)^{l_2+l_3}\de_2-(-1)^{l_1+l_3}\de_3,
 &&a_2=\al_{12}-(-1)^{l_2+l_3}\de_2+(-1)^{l_1+l_3}\de_3,\\
 &a_3=\al_{21}-c+(-1)^{l_1+l_2}\de_1,
 &&a_4=\al_{21}+c-(-1)^{l_1+l_2}\de_1,\\
 &a_5=\al_{23}+(-1)^{l_1+l_3}\de_3-(-1)^{l_1+l_2}\de_1,
 &&a_6=\al_{23}-(-1)^{l_1+l_3}\de_3+(-1)^{l_1+l_2}\de_1,\\
 &a_7=\al_{32}-c+(-1)^{l_2+l_3}\de_2,
 &&a_8=\al_{32}+c-(-1)^{l_2+l_3}\de_2,\\
 &a_9=\al_{31}+(-1)^{l_1+l_2}\de_1-(-1)^{l_2+l_3}\de_2,
 &&a_{10}=\al_{31}-(-1)^{l_1+l_2}\de_1+(-1)^{l_2+l_3}\de_2,\\
 &a_{11}=\al_{13}-c+(-1)^{l_1+l_3}\de_3,
 &&a_{12}=\al_{13}+c-(-1)^{l_1+l_3}\de_3,
\end{align}
\end{subequations}
where
\begin{equation}
 \al_{ij}=\al_i(l_i)-\al_j(l_j),\quad i,j\in\{1,2,3\},\quad
  \al_i(k)=\al_i(0)+k,\quad i\in\{1,2,3\},~k\in\mathbb{Z}.
\end{equation}
Here, $\al_i(0)$, $i=1,2,3$, $c$, and $\de_j$, $j=1,2,3$, are complex parameters. 

We can easily verify that the system \eqref{eqns:samplePDEs_P123456} has the CACO property (see Definition \ref{def:CACO_PDEs}).
Remark \ref{rem:sqexample} considers the square property of the system. 
Finally, we state where this example fits in our classification in Remark \ref{rem:typeexample}.
\begin{remark}\label{rem:sqexample}
The system \eqref{eqns:samplePDEs_P123456} has the square property (see Definition \ref{def:CACO_PDEs}). 
Then the following square equations exist:
\begin{subequations}
\begin{align}
 &K_1=K\left(u_{\ol{2}\ul{3}},u_{\ol{23}},u_{\ul{23}},u_{\ul{2}\ol{3}}\,;\,\al_{12},\al_{13},(-1)^{l_1+l_2}\de_1,(-1)^{l_2+l_3}\de_2,c,(-1)^{l_1+l_3}\de_3\right)=0,\\
 &K_2=K\left(u_{\ul{1}\ol{3}},u_{\ol{13}},u_{\ul{13}},u_{\ol{1}\ul{3}}\,;\,\al_{23},\al_{21},(-1)^{l_2+l_3}\de_2,(-1)^{l_1+l_3}\de_3,c,(-1)^{l_1+l_2}\de_1\right)=0,\\
 &K_3=K\left(u_{\ol{1}\ul{2}},u_{\ol{12}},u_{\ul{12}},u_{\ul{1}\ol{2}}\,;\,\al_{31},\al_{32},(-1)^{l_1+l_3}\de_3,(-1)^{l_1+l_2}\de_1,c,(-1)^{l_2+l_3}\de_2\right)=0,\label{eqn:sample_K3}
\end{align}
\end{subequations}
where the polynomial $K=K(x,y,z,w\,;A_1,A_2,A_3,A_4,A_5,A_6)$ is given by 
\begin{align}\label{eqn:example_def_squareK}
 K=&\Big({A_1}^2(x-z)(y-w)+{A_2}^2(x-y)(w-z)\Big)(y-z)\notag\\
 &+(A_3-A_4-A_5+A_6)\Big(A_1(x-z)(y-w)+A_2(x-y)(w-z)\Big)(y+z)\notag\\
 &+(A_3-A_6)(A_4-A_5)(y-z)(yz-xy-xz+xw)\notag\\
 &+\Big((A_3-A_4)y-(A_4-A_6)z\Big)\Big((A_3-A_5)z-(A_5-A_6)y\Big)(x-w)\notag\\
 &+(A_4-A_5)^2(x-w)yz.
\end{align}
Note that we have 
\begin{equation}\label{eqn:degree_xyzw_K}
 \deg_x{K}=\deg_w{K}=1,\quad
 \deg_y{K}=\deg_z{K}=2.
\end{equation}
\end{remark}

\begin{remark}\label{rem:typeexample}
The system of equations \eqref{eqns:samplePDEs_P123456} is classified in {\rm Type I} and satisfies the condition \eqref{eqn:typeI_cond1}.
(See Lemma \ref{lemma:classification_CACO_CO1CO2CO3_1}.)
\end{remark}

\begin{proposition}[\cite{JN:Preparing}]\label{prop:red_A2_Painleve}
The system of \PDE s \eqref{eqns:samplePDEs_P123456}
can be reduced to the $A_2^{(1)\ast}$-type discrete Painlev\'e equations by imposing the $(1,1,1)$-periodic condition $u_{\ol{123}}=u$.
\end{proposition}

\section{Classification of quad-equations on a cuboctahedron}\label{section:classification_Cuboctahedron}
\subsection{Classification of quad-equations on an octahedron}\label{subsection:class_CAO}
In this section, we give a classification of the quad-equations on an octahedron which have the consistency around an octahedron property.

In the following, we fix the form of the quad-equation $Q_1$ considering the symmetries of an octahedron, M\"obius transformations and Lemma \ref{lemma:classification_JN_polys} (the result is given in Remark \ref{rem:QHDN}).

Consider the symmetries of an octahedron.
Letting
\begin{equation}
 G^{(o)}=\Big\langle (1\,4),~(2\,5),~(3\,6),~(2\,3\,5\,6),~(1\,3\,4\,6),~(1\,2\,4\,5)\Big\rangle
\end{equation}
be the subgroup of the symmetric group $\mathfrak{S}_6$ and
lifting the action of $\si\in G^{(o)}$ upon the variables $u_i$, $i=1,\dots,6$, by 
\begin{equation}
 \si.u_i=u_{\si(i)}.
\end{equation}
The permutations $(1\,4)$, $(2\,5)$ and $(3\,6)$ are reflections 
with the diagrams of a section passing through four vertices as the axes,
and the permutations $(2\,3\,5\,6)$, $(1\,3\,4\,6)$ and $(1\,2\,4\,5)$ are rotations 
about axes passing through the vertices $\{u_1,u_4\}$, $\{u_2,u_5\}$ and $\{u_3,u_6\}$, respectively.
(See Figure \ref{fig:octahedron_3D}.)
Therefore, we can regard $G^{(o)}$ as providing symmetries of the octahedron.

Using the M\"obius transformations and the transformation group $G^{(o)}$, 
we define the equivalence relation $\sim_o$ for quad-equations on an octahedron as the following.

\begin{definition}\rm
Let $\{P,Q,R\}$ and $\{P',Q',R'\}$ denote quad-equations on an octahedron, identified as follows
\begin{equation}
 P=P(u_4,u_2,u_1,u_5),\quad
 Q=Q(u_2,u_6,u_5,u_3),\quad
 R=R(u_6,u_4,u_3,u_1),
\end{equation}
and
\begin{equation}
 P'=P'(u_4,u_2,u_1,u_5),\quad
 Q'=Q'(u_2,u_6,u_5,u_3),\quad
 R'=R'(u_6,u_4,u_3,u_1).
\end{equation}
We say that the quad-equations $\{P,Q,R\}$ are equivalent, modulo M\"obius transformations and transformation group $G^{(o)}$, to the quad-equations $\{P',Q',R'\}$, or
\begin{equation}
 \{P,Q,R\}\sim_o\{P',Q',R'\},
\end{equation} 
if there exist
a M\"obius transformation $r$ 
and a permutation $\si\in G^{(o)}$ such that 
\begin{equation}
 \{P',Q',R'\}=\Big\{r.\si.P,\,r.\si.Q,\,r.\si.R\Big\}.
\end{equation}
\end{definition}

Not all permutations of the 4 vertices in a quad-equation are allowed because some violate the choice of a quadrilateral in the octahedron. For example, the equatorial quadrilateral in Figure \ref{fig:octahedron_3D} allows $(1245)$ but not $(12)$ without also including $(45)$. Therefore, the $4!$ permutations of $\mathfrak{S}_4$ reduce to $\bigm|\mathfrak{S}_4\cap G^{(o)}\bigm|=8$. 

These considerations allow us to prove the following lemma, which fixes the form of the quad-equation $Q_1$ in \eqref{eqn:octahedron_Q1Q2Q3} using the equivalence relation $\sim_o$. 
\begin{lemma}\label{lemma:Q1_3type}
Assume that the quad-equation $Q_1=Q_1(u_4,u_2,u_1,u_5)$ is given by a polynomial $Q(x,y,z,w)$.
Then, under the equivalence relation $\sim_o$, 
it is sufficient to consider the following three cases:
\begin{equation}
 Q_1=Q(u_4,u_2,u_1,u_5),\quad
 Q_1=Q(u_4,u_2,u_5,u_1),\quad
 Q_1=Q(u_4,u_1,u_2,u_5).
\end{equation}
\end{lemma}

Given a generic polynomial $Q(x,y,z,w)$, Lemma \ref{lemma:Q1_3type} states that there are 3 inequivalent cases of corresponding quad-equations. 
However, for certain cases of quad-equations in Lemma \ref{lemma:classification_JN_polys}, there are further equivalences between these 3 types. 
This is clarified in the following lemma.

\begin{lemma}\label{lemma:Q4321_D432_symmetry}
Consider the cases when $Q_1$ is given by $\polyQ4$, $\polyQ3$, $\polyQ2$, $\polyQ1$, $\polyD4$, $\polyD3$ or $\polyD2_{\de=1}$. 
Then, the 3 cases identified in Lemma \ref{lemma:Q1_3type} are equivalent under the equivalence relation $\sim_o$.
\end{lemma}
\begin{proof}
The definition of the polynomial $\polyQ4$ shows that the following relation holds:
\begin{equation}
 \polyQ4(x,y,w,z;\al,\be)=-\polyQ4(x,y,z,w;-\al,\al+\be),
\end{equation}
which gives 
\begin{subequations}
\begin{align}
 &\polyQ4(u_4,u_2,u_5,u_1;\al,\be)=-\polyQ4(u_4,u_2,u_1,u_5;-\al,\al+\be),\\
 &\polyQ4(u_4,u_1,u_2,u_5;\al,\be)=-\polyQ4(u_4,u_1,u_5,u_2;-\al,\al+\be).
\end{align}
\end{subequations}
Note that the second case appears to be different to those listed in Lemma \ref{lemma:Q1_3type}. However, the right side of the second equation is the case
$Q_1=\polyQ4(u_4,u_1,u_5,u_2;\al,\be)$, which is equivalent to the case
$Q_1=\polyQ4(u_4,u_2,u_1,u_5;\al,\be)$, under the equivalence relation $\sim_o$.
Therefore, whether we take $Q_1$ to be $\polyQ4(u_4,u_2,u_1,u_5;\al,\be)$, or 
$\polyQ4(u_4,u_2,u_5,u_1;\al,\be)$, or $\polyQ4(u_4,u_1,u_2,u_5;\al,\be)$ leads to the same result. 

The remaining cases can be also verified, in a similar manner, by using the following relations:
\begin{align*}
 &\polyQ3(x,y,w,z;\al,\be;\de)=-\polyQ3(x,y,z,w;\al^{-1},\al\be;\de),\\
 &\polyQ2(x,y,w,z;\al,\be)=-\polyQ2(x,y,z,w;-\al,\al+\be),\\
 &\polyQ1(x,y,w,z;\al,\be;\de)=\polyQ1(x,y,z,w;-\al,\al+\be;\de),\\
 &\polyD4(x,y,w,z;\de_1,\de_2,\de_3)=yz\,\polyD4(x,y^{-1},z^{-1},w;\de_3,\de_2,\de_1),\\
 &\polyD3(x,y,w,z)=\polyD3(x,y,z,w),\\
 &\polyD2(x,y,w,z;1)=y\,\polyD2(x,y^{-1},z,w;1).
\end{align*}
\end{proof}

\begin{remark}\label{rem:QHDN}
The above lemma shows that we need only consider one case for each of $\polyQ4$, $\polyQ3$, $\polyQ2$, $\polyQ1$, $\polyD4$, $\polyD3$, $\polyD2_{\delta=1}$, $\polyD1$. 
On the other hand, there still remain 3 cases for each of $\polyH3$, $\polyH2$, $\polyN3$, $\polyN2$, and $\polyN1$. 
Therefore, Lemmas \ref{lemma:classification_JN_polys},\ref{lemma:Q1_3type} and \ref{lemma:Q4321_D432_symmetry} leave a total of 22 cases to be considered, listed below:
\[\arraycolsep=3.pt
\begin{array}{lll}
 \polyQ4(u_4,u_2,u_1,u_5;\al,\be),
 &\polyQ3(u_4,u_2,u_1,u_5;\al,\be;\de),
 &\polyQ2(u_4,u_2,u_1,u_5;\al,\be),\\
 \polyQ1(u_4,u_2,u_1,u_5;\al,\be;\de),
 &\polyD4(u_4,u_2,u_1,u_5;\de_1,\de_2,\de_3),
 &\polyD3(u_4,u_2,u_1,u_5),\\
 \polyD2(u_4,u_2,u_1,u_5;1),
 &\polyH3(u_4,u_2,u_1,u_5;\al,\be;\de_1,\de_2),
 &\polyH3(u_4,u_2,u_5,u_1;\al,\be;\de_1,\de_2),\\
 \polyH3(u_4,u_1,u_2,u_5;\al,\be;\de_1,\de_2),
 &\polyH2(u_4,u_2,u_1,u_5;\al,\be;\de),
 &\polyH2(u_4,u_2,u_5,u_1;\al,\be;\de),\\
 \polyH2(u_4,u_1,u_2,u_5;\al,\be;\de),
 &\polyN3(u_4,u_2,u_1,u_5;\{\al_i\}),
 &\polyN3(u_4,u_2,u_5,u_1;\{\al_i\}),\\
 \polyN3(u_4,u_1,u_2,u_5;\{\al_i\}),
 &\polyN2(u_4,u_2,u_1,u_5;\{\al_i\}),
 &\polyN2(u_4,u_2,u_5,u_1;\{\al_i\}),\\
 \polyN2(u_4,u_1,u_2,u_5;\{\al_i\}),
 &\polyN1(u_4,u_2,u_1,u_5;\{\al_i\}),
 &\polyN1(u_4,u_2,u_5,u_1;\{\al_i\}),\\
 \polyN1(u_4,u_1,u_2,u_5;\{\al_i\}),
\end{array}
\]
where $(\de_1,\de_2)\neq(0,0)$ for $\polyH3$ and $\polyD4$.
\end{remark}

Recall that we use the convention adopted in Equation \eqref{eqn:octahedron_Q1Q2Q3}, which fixes the labels $\{Q_1,Q_2,Q_3\}$ to quadrilaterals with particular choices of variables. The following two lemmas lead to the classification of CAO quad-equations presented in Lemma \ref{lemma:classification_Q123_CAO}.
\begin{lemma}\label{lemma:classification_CAO_1}
Under the equivalence relation $\sim_o$, any CAO octahedron with quad-equations $\{Q_1,Q_2,Q_3\}$ is equivalent to one of the following:
\begin{subequations}
\begin{align}
&\begin{cases}\label{eqn:Q123_N3}
 \polyN3(u_4,u_2,u_1,u_5;\al_1,\al_2,\al_3,\al_4),\\
 u_5 \,q(u_6,u_3;B_1,B_2,B_3,B_4)+u_2 \,q(u_6,u_3;B_5,B_6,B_7,B_8),\\
 (\al_4 u_1+\al_1 u_4)\,q(u_6,u_3;B_1,B_2,B_3,B_4)\\
 \hspace{1em}-(\al_2 u_1+\al_3 u_4)\,q(u_6,u_3;B_5,B_6,B_7,B_8),
\end{cases}\\
&\begin{cases}\label{eqn:Q123_N2}
 \polyN2(u_4,u_2,u_1,u_5;\al_1,\al_2,\al_3,\al_4),\\
 u_5 \,q(u_6,u_3;B_1,B_2,B_3,B_4)+u_2 \,q(u_6,u_3;B_5,B_6,B_7,B_8),\\
 \Big(\al_1(u_1+u_4)+\al_3\Big)\,q(u_6,u_3;B_1,B_2,B_3,B_4)\\
 \hspace{1em}-\Big(\al_2(u_1+u_4)+\al_4\Big)\,q(u_6,u_3;B_5,B_6,B_7,B_8),
\end{cases}\\
&\begin{cases}\label{eqn:Q123_N1}
 \polyN1(u_4,u_2,u_1,u_5;\al_1,\al_2,\al_3,\al_4),\\
 q(u_6,u_3;B_1,B_2,B_3,B_4)+(u_5+u_2) q(u_6,u_3;B_5,B_6,B_7,B_8),\\
 \Big(\al_1(u_1+u_4)+\al_3\Big)\,q(u_6,u_3;B_1,B_2,B_3,B_4)\\
 \hspace{1em}-\Big(\al_2(u_1+u_4)+\al_4\Big)\,q(u_6,u_3;B_5,B_6,B_7,B_8),
\end{cases}
\end{align}
\end{subequations}
where $B_i\in\bbC$, $i=1\dots,8$, and
\begin{equation}\label{eqn:def_q}
 q(x,y;a,b,c,d)=axy+bx+cy+d.
\end{equation}
\end{lemma}

The proof of Lemma \ref{lemma:classification_CAO_1} is given in Appendix \ref{section:proof_lemma_CAO}. Note that this result still contains many parameters, which can be reduced by using M\"obius transformations of the variables $u_3$ and $u_6$, as shown in the following lemma.

\begin{lemma}\label{lemma:classification_rational_r}
Given parameters $B_i\in\bbC$, $i=1\dots,8$, and for $q$ defined by \eqref{eqn:def_q}, assume that the rational function
\begin{equation}\label{eqn:def_r_qq}
 r(x,y;\{B_i\})=\dfrac{q(x,y;B_1,B_2,B_3,B_4)}{q(x,y;B_5,B_6,B_7,B_8)},
\end{equation}
satisfies $\partial r/\partial x\not=0$ and $\partial r/\partial y\not=0$. 
Then, it can be transformed to one of two canonical forms
\begin{subequations}
\begin{align}
 &r_1(x,y;\{A_i\})=\dfrac{A_1x+A_2y}{A_3x+A_4y},\label{eqn:r_type1}\\
 &r_2(x,y;A)=\dfrac{1}{x+y}+A,\label{eqn:r_type2}\\
 &r_3(x,y)=x+y,\label{eqn:r_type3}
\end{align}
\end{subequations}
by M\"obius transformations of $x$ and $y$, for some $A_i,A\in\bbC$ with $A_1A_4\not=A_2A_3$.
\end{lemma}

The proof of Lemma \ref{lemma:classification_rational_r} is given in Appendix \ref{section:proof_lemma_r}.
The main classification of CAO quad-equations is given by the following lemma.
\begin{lemma}[CAO classification]\label{lemma:classification_Q123_CAO}
Under the equivalence relation $\sim_o$, any CAO octahedron with quad-equations $\{Q_1,Q_2,Q_3\}$ is equivalent to one of the following:
\begin{subequations}
\begin{align}
&\begin{cases}\label{eqn:Q123_N333}
 \polyN3(u_4,u_2,u_1,u_5;\al_1,\al_2,\al_3,\al_4),\\
 \polyN3(u_2,u_6,u_5,u_3;\al_5,\al_6,\al_7,\al_8),\\
 \polyN3(u_6,u_4,u_3,u_1;\al_{1,8}-\al_{3,5},\al_{4,6}-\al_{2,7},\al_{4,8}-\al_{2,5},\al_{1,6}-\al_{3,7}),
\end{cases}\\
&\begin{cases}\label{eqn:Q123_N223}
 \polyN2(u_2,u_4,u_5,u_1;\al_1,\al_2,\al_3,\al_4),\\
 \polyN2(u_2,u_6,u_5,u_3;\al_5,\al_6,\al_7,\al_8),\\
 \polyN3(u_6,u_4,u_3,u_1;\al_{3,5}-\al_{1,7},\al_{4,6}-\al_{2,8},\al_{4,5}-\al_{2,7},\al_{3,6}-\al_{1,8}),
\end{cases}\\
&\begin{cases}\label{eqn:Q123_N221}
 \polyN2(u_4,u_2,u_1,u_5;\al_1,\al_2,\al_3,\al_4),\\
 \polyN2(u_6,u_2,u_3,u_5;\al_5,\al_6,\al_7,\al_8),\\
 \polyN1(u_6,u_4,u_3,u_1;\al_{2,5}-\al_{1,6},\al_{4,5}-\al_{3,6},\al_{2,7}-\al_{1,8},\al_{4,7}-\al_{3,8}),
\end{cases}\\
&\begin{cases}\label{eqn:Q123_N111}
 \polyN1(u_4,u_2,u_1,u_5;\al_1,\al_2,\al_3,\al_4),\\
 \polyN1(u_2,u_6,u_5,u_3;\al_5,\al_6,\al_7,\al_8),\\
 \polyN1(u_6,u_4,u_3,u_1;\al_{2,5}-\al_{1,7},\al_{4,5}-\al_{3,7},\al_{2,6}-\al_{1,8},\al_{4,6}-\al_{3,8}),
\end{cases}
\end{align}
\end{subequations}
where $\al_{i,j}=\al_i\al_j$.
\end{lemma}
\begin{proof}
Our starting point is the result of Lemma \ref{lemma:classification_CAO_1}, which gives 3 sets of quad-equations $\{Q_1,Q_2,Q_3\}$ to consider. 
Each set has three possible canonical forms of $r$, given by $r_1$, $r_2$ and $r_3$, which therefore results in nine cases to consider in total. 
We show below that some of these are equivalent, or occur as special cases of others.
  
We start by considering the \eqref{eqn:Q123_N3} and \eqref{eqn:r_type1}. Recall that $Q_2$ gives 
\begin{equation}
 u_5 \,q(u_6,u_3;B_1,B_2,B_3,B_4)+u_2 \,q(u_6,u_3;B_5,B_6,B_7,B_8)=0.
\end{equation}
By using M\"obius transformations of $u_3$ and $u_6$, we can rewrite this equation as
\begin{equation}
 \polyN3(u_2,u_6,u_5,u_3;A_3,A_2,A_4,A_1)=0.
\end{equation}
Then, $\{Q_1,Q_2,Q_3\}$ becomes
\begin{equation}\label{eqn:CAO_proof_1}
\begin{cases}
 \polyN3(u_4,u_2,u_1,u_5;\al_1,\al_2,\al_3,\al_4),\\
 \polyN3(u_2,u_6,u_5,u_3;A_3,A_2,A_4,A_1),\\
 \polyN3(u_6,u_4,u_3,u_1;A_1\al_1-A_3\al_3,A_2\al_4-A_4\al_2,A_1\al_4-A_3\al_2,A_2\al_1-A_4\al_3).
\end{cases}
\end{equation}
Similarly, the pair \eqref{eqn:Q123_N3} and \eqref{eqn:r_type2} gives
\begin{equation}\label{eqn:CAO_proof_2}
\begin{cases}
 \polyN3(u_4,u_2,u_1,u_5;\al_1,\al_2,\al_3,\al_4),\\
 \polyN2(u_6,u_2,u_3,u_5;1,A_5,0,1),\\
 \polyN2(u_6,u_4,u_3,u_1;A_5\al_1-\al_3,A_5\al_4-\al_2,\al_1,\al_4),
\end{cases}
\end{equation}
and the pair \eqref{eqn:Q123_N3} and \eqref{eqn:r_type3} gives
\begin{equation}\label{eqn:CAO_proof_7}
\begin{cases}
 \polyN3(u_5,u_4,u_2,u_1;\al_3,\al_4,\al_2,\al_1),\\
 \polyN2(u_6,u_5,u_3,u_2;1,0,0,1),\\
 \polyN2(u_6,u_4,u_3,u_1;\al_1,\al_4,-\al_3,\al_2).
\end{cases}
\end{equation}
We can easily verify that the following hold:
\begin{align*}
 &\eqref{eqn:Q123_N333}\sim_o\eqref{eqn:CAO_proof_1},\\
 &\eqref{eqn:Q123_N223}\sim_o\eqref{eqn:CAO_proof_2},~\eqref{eqn:CAO_proof_7}.
\end{align*}

In a similar manner, we obtain the following results:
\begin{align*}
 &\eqref{eqn:Q123_N223}\sim_o
 \big\{\eqref{eqn:Q123_N2},\eqref{eqn:r_type1}\big\},\\
 &\eqref{eqn:Q123_N221}\sim_o
 \big\{\eqref{eqn:Q123_N2},\eqref{eqn:r_type2}\big\},~
 \big\{\eqref{eqn:Q123_N2},\eqref{eqn:r_type3}\big\},~
 \big\{\eqref{eqn:Q123_N1},\eqref{eqn:r_type1}\big\},\\
 &\eqref{eqn:Q123_N111}\sim_o
 \big\{\eqref{eqn:Q123_N1},\eqref{eqn:r_type2}\big\},~
 \big\{\eqref{eqn:Q123_N1},\eqref{eqn:r_type3}\big\}.
\end{align*}
Therefore, we have completed the proof.
\end{proof}
%
%

\subsection{Classification of quad-equations on a cuboctahedron}\label{subsection:classification_CACO}
In this section, we classify the quad-equations on a cuboctahedron by using the conditions \ref{cond:CO1}--\ref{cond:CO3} defined below.
Note that the following description is given only on one individual cuboctahedron. 

Consider the system of quad equations on a cuboctahedron as described by Equations \eqref{eqns:V_caco_general}.
We assume that the cuboctahedron with quad-equations $\{Q_1,\dots,Q_9\}$ satisfies the following properties.
\begin{condition}\label{condition:CO}
The quad-equations $\{Q_1,\dots,Q_9\}$ satisfy the following properties:
\begin{enumerate}[leftmargin=1.4cm,label={\rm (CO\arabic*)}]
\item\label{cond:CO1}
the CACO property;
\item\label{cond:CO2}
the square property;
\item\label{cond:CO3}
there exist irreducible multi-affine polynomials $P_i$, $i=1,2,3$, such that the system of equations is given by
\begin{subequations}\label{eqns:Q123456789_CO3}
\begin{align}
 &Q_1=P_1\left(u_5,u_1,v_5,v_4;A^{(1)}_1,\dots,A^{(1)}_{k_1}\right),
 &&Q_2=P_1\left(v_2,v_1,u_2,u_4;A^{(1)}_1,\dots, A^{(1)}_{k_1}\right),\label{eqn:Q1_P1_CO3}\\
 &Q_3=P_2\left(u_3,u_5,v_3,v_2;A^{(2)}_1,\dots,A^{(2)}_{k_2}\right),
 &&Q_4=P_2\left(v_6,v_5,u_6,u_2;A^{(2)}_1,\dots,A^{(2)}_{k_2}\right),\\
 &Q_5=P_3\left(u_1,u_3,v_1,v_6;A^{(3)}_1,\dots,A^{(3)}_{k_3}\right),
 &&Q_6=P_3\left(v_4,v_3,u_4,u_6;A^{(3)}_1,\dots,A^{(3)}_{k_3}\right).\\
 &Q_7=P_1\left(u_4,u_2,u_1,u_5;B^{(1)}_1,\dots,B^{(1)}_{k_1}\right),
 &&Q_8=P_2\left(u_2,u_6,u_5,u_3;B^{(2)}_1,\dots, B^{(2)}_{k_2}\right),\\
 &Q_9=P_3\left(u_6,u_4,u_3,u_1;B^{(3)}_1,\dots,B^{(3)}_{k_3}\right),\label{eqn:Q9_P3_CO3}
\end{align}
\end{subequations}
where $k_1,k_2,k_3\in\bbZ_{\geq0}$, and $A_i^{(j)}$ and $B_i^{(j)}$ are complex parameters.
If $A^{(i)}_{j}\equiv0$, then $B^{(i)}_{j}\equiv0$, and vice versa.
\end{enumerate}
Note that the parameters $A_i^{(j)}$, $B_i^{(j)}$ in \ref{cond:CO3} are not necessarily independent. 
Conditions \ref{cond:CO1} and \ref{cond:CO2} give rise to relations between them as shown in Lemma \ref{lemma:classification_CACO_CO1CO2CO3_1}--\ref{lemma:classification_CACO_CO1CO2CO3_3} below.
\end{condition}

Clearly, if the cuboctahedron undergoes a reflection, analogous equations should hold. This leads to the consideration of a symmetry group, generated by the reflections on the vertices of the cuboctahedron. 
This motivates the definition of the group:
\begin{equation}\label{eqn:Gco}
 G^{(CO)}=\langle s_{12},s_{23},s_{13},\iota\rangle,
\end{equation}
where $s_{12}$, $s_{23}$, $s_{13}$ and $\iota$ are transformations defined by the following actions:
\begin{align*}
 s_{12}:&~u_1\leftrightarrow u_5,\quad u_2\leftrightarrow u_4,\quad 
 v_1\leftrightarrow v_2,\quad v_3\leftrightarrow v_6,\quad v_4\leftrightarrow v_5,\\
 s_{23}:&~u_2\leftrightarrow u_6,\quad u_3\leftrightarrow u_5,\quad 
 v_1\leftrightarrow v_4,\quad v_2\leftrightarrow v_3,\quad v_5\leftrightarrow v_6,\\
 s_{13}:&~u_1\leftrightarrow u_3,\quad u_4\leftrightarrow u_6,\quad 
 v_1\leftrightarrow v_6,\quad v_2\leftrightarrow v_5,\quad v_3\leftrightarrow v_4,\\
 \iota:&~
 u_1\leftrightarrow u_4,\quad u_2\leftrightarrow u_5,\quad u_3\leftrightarrow u_6,\quad 
 v_1\leftrightarrow v_4,\quad v_2\leftrightarrow v_5,\quad v_3\leftrightarrow v_6,
\end{align*}
or, equivalently, the following actions on the standard basis:
\begin{equation*}
 s_{12}:\ep_1\leftrightarrow \ep_2,\quad
 s_{23}:\ep_2\leftrightarrow \ep_3,\quad
 s_{13}:\ep_1\leftrightarrow \ep_3,\quad
 \iota:\{\ep_1,\ep_2,\ep_3\}\leftrightarrow \{-\ep_1,-\ep_2,-\ep_3\}.
\end{equation*}

The transformation $s_{12}$ can be visualized as the reflection across a slice of the cuboctahedron that divides it into two equal pieces. The slice in this case is taken through the vertices $\{u_3,u_6\}$ and the center of the edge connecting the vertices $\{u_1,u_5\}$. Similarly, the transformations $s_{23}$ and $s_{13}$ are reflections across two other ways of slicing the cuboctahedron. The slice corresponding to $s_{23}$ goes through the vertices $\{u_1,u_4\}$ and the center of the edge connecting the vertices $\{u_2,u_6\}$, while the one for $s_{13}$ goes through the vertices $\{u_2,u_5\}$ and the center of the edge connecting the vertices $\{u_1,u_3\}$.
Moreover, the transformation $\iota$ corresponds to the reflection of vertices about the centre of the cuboctahedron. 
(See Figure \ref{fig:cuboctahedron_3D}.)

\begin{remark}\label{remark:symmetry_GCO}
In Appendix \ref{section:proof_lemma_CACO_CO1CO2CO3}, we will apply these transformations to paticular cases of quad-equations.
The invariance of such quad-equations under $G^{(CO)}$ will need conditions on parameters.
We do not list them here, as different cases studied in Appendix \ref{section:proof_lemma_CACO_CO1CO2CO3} will lead to different conditions on parameters.
\end{remark}

We now collect the transformations under which quad-equations on a cuboctahedron remain invariant. Using the M\"obius transformations and the transformation group $G^{(CO)}$, 
we define the equivalence relation $\sim_{co}$ as follows.

\begin{definition}\rm
The cuboctahedron with quad-equations $\{Q_1,\dots,Q_9\}$ is said to be equivalent under $\sim_{co}$ to the cuboctahedron with quad-equations $\{Q_1',\dots,Q_9'\}$, i.e.
\begin{equation}
 \{Q_1,\dots,Q_9\}\sim_{co}\{Q_1',\dots,Q_9'\},
\end{equation} 
if the following properties hold.
Let a M\"obius transformation of the variables $u_i$ and $v_j$, where $i,j=1,\dots,6$, be denoted by $r$ and elements of $G^{(CO)}$ be denoted $\si$. 
Then we have
\begin{equation}
 \{Q_1,\dots,Q_9\}=\Big\{r.\si.Q_1',\dots,r.\si.Q_9'\Big\}.
\end{equation}
In a similar manner, we also define the equivalence relation $\sim_{co}$ for the quad-equations on the octahedron with quad-equations $\{Q_7,Q_8,Q_9\}$.
\end{definition}

In the following, we fix the form of the quad-equation $Q_7$.
Assume that the quad-equation $Q_7=Q_7(u_4,u_2,u_1,u_5)$ is given by an irreducible multi-affine polynomial polynomial $Q(x,y,z,w)$.
Then, there are 24 ways to define $Q_7$ by the polynomial $Q(x,y,z,w)$ up to M\"obius transformations.
Lemma \ref{lemma:classification_Q123_CAO} shows that 
if the quad-equations $\{Q_7,Q_8,Q_9\}$ has the CAO property,
then the quad-equation $Q_7$ is given by one of the polynomials $\polyN3$, $\polyN2$ and $\polyN1$, that is, the polynomial $Q(x,y,z,w)$ is given by one of the polynomials $\polyN3$, $\polyN2$ and $\polyN1$.
Because of the following relations: 
\begin{align*}
 &\polyN3(z,y,x,w;\al_4,\al_3,\al_2,\al_1)
 =\polyN3(x,w,z,y;\al_3,\al_4,\al_1,\al_2)
 =\polyN3(x,y,z,w;\al_1,\al_2,\al_3,\al_4),\\
 &\polyN2(z,y,x,w;\al_1,\al_2,\al_3,\al_4)
 =\polyN2(x,w,z,y;\al_2,\al_1,\al_4,\al_3)
 =\polyN2(x,y,z,w;\al_1,\al_2,\al_3,\al_4),\\
 &\polyN1(z,y,x,w;\al_1,\al_2,\al_3,\al_4)
 =\polyN1(x,w,z,y;\al_1,\al_2,\al_3,\al_4)
 =\polyN1(x,y,z,w;\al_1,\al_2,\al_3,\al_4),
\end{align*} 
it is sufficient to consider the following six cases for each $N$-type polynomial:
\begin{align*}
 &Q_7=Q(u_4,u_2,u_1,u_5),\quad
 Q_7=Q(u_2,u_4,u_5,u_1),\quad
 Q_7=Q(u_4,u_2,u_5,u_1),\\
 &Q_7=Q(u_2,u_4,u_1,u_5),\quad
 Q_7=Q(u_4,u_1,u_2,u_5),\quad
 Q_7=Q(u_1,u_2,u_5,u_4).
\end{align*}
Therefore, it is sufficient to consider the eighteen cases.
Lemma \ref{lemma:classification_CAO_1} (or, Lemmas \ref{lemma:CAO_N3}--\ref{lemma:CAO_N1}) shows that $Q_7$ must be given by one of the following:
\begin{align*}
 &\polyN3(u_4,u_2,u_1,u_5;\al_1,\al_2,\al_3,\al_4),\quad
 \polyN3(u_2,u_4,u_5,u_1;\al_1,\al_2,\al_3,\al_4),\\
 &\polyN2(u_4,u_2,u_1,u_5;\al_1,\al_2,\al_3,\al_4),\quad
 \polyN2(u_2,u_4,u_5,u_1;\al_1,\al_2,\al_3,\al_4),\\
 &\polyN1(u_4,u_2,u_1,u_5;\al_1,\al_2,\al_3,\al_4),\quad
 \polyN1(u_2,u_4,u_5,u_1;\al_1,\al_2,\al_3,\al_4),
\end{align*}
and the others are inadmissible cases or special cases of the cases above.
Moreover, since the polynomials $\polyN3$ and $\polyN1$ satisfy the following relations:
\begin{subequations}
\begin{align}
 \polyN3(y,x,w,z;\al_1,\al_2,\al_4,\al_3)=\polyN3(x,y,z,w;\al_1,\al_2,\al_3,\al_4),\label{eqn:relation_N3_2}\\
 \polyN1(y,x,w,z;\al_1,\al_3,\al_2,\al_4)=\polyN1(x,y,z,w;\al_1,\al_2,\al_3,\al_4),\label{eqn:relation_N1_2}
\end{align} 
\end{subequations}
it is sufficient to consider the four cases:
\begin{align*}
 &Q_7=\polyN3(u_4,u_2,u_1,u_5;\al_1,\al_2,\al_3,\al_4),\quad
 Q_7=\polyN2(u_4,u_2,u_1,u_5;\al_1,\al_2,\al_3,\al_4),\\
 &Q_7=\polyN2(u_2,u_4,u_5,u_1;\al_1,\al_2,\al_3,\al_4),\quad
 Q_7=\polyN1(u_4,u_2,u_1,u_5;\al_1,\al_2,\al_3,\al_4).
\end{align*}
Therefore, we obtain the following lemma.
\begin{lemma}\label{lemma:classification_Q789_CAO}
Under the equivalence relation $\sim_{co}$, any CAO octahedron with quad-equations $\{Q_7,Q_8,Q_9\}$ is equivalent to one of the systems of quad-equations \eqref{eqn:Q123_N333}--\eqref{eqn:Q123_N111}.
\end{lemma}
\begin{proof}
Firstly, we let $Q_7$ be one of the following:
\begin{align*}
 &\polyN3(u_4,u_2,u_1,u_5;\al_1,\al_2,\al_3,\al_4),\quad
 \polyN2(u_4,u_2,u_1,u_5;\al_1,\al_2,\al_3,\al_4),\\
 &\polyN1(u_4,u_2,u_1,u_5;\al_1,\al_2,\al_3,\al_4).
\end{align*}
Then, in a similar manner to the proof of Lemma \ref{lemma:classification_Q123_CAO},
we obtain the nine CAO octahedra, each of which is equivalent to one of the systems of quad-equations \eqref{eqn:Q123_N333}--\eqref{eqn:Q123_N111} under the equivalence relation $\sim_{co}$.

Next, let
\begin{equation}
 Q_7=\polyN2(u_2,u_4,u_5,u_1;\al_1,\al_2,\al_3,\al_4),
\end{equation}
and $Q_8$ be the generic form \eqref{eqn:Q2_general_2}.
Then, in a similar manner to the proofs in Appendix \ref{section:proof_lemma_CAO},
we obtain the following CAO octahedron with quad-equations:
\begin{equation}\label{eqn:hatQ123_N2}
\begin{cases}
 \polyN2(u_2,u_4,u_5,u_1;\al_1,\al_2,\al_3,\al_4),\\
 q(u_6,u_3;B_1,B_2,B_3,B_4)+(u_5+u_2) \,q(u_6,u_3;B_5,B_6,B_7,B_8),\\
 (\al_2 u_1+\al_1u_4)\,q(u_6,u_3;B_1,B_2,B_3,B_4)\\
 \hspace{1em}-(\al_4u_1+\al_3u_4)\,q(u_6,u_3;B_5,B_6,B_7,B_8).
\end{cases}
\end{equation}
In a similar manner to the proof of Lemma \ref{lemma:classification_Q123_CAO},
we obtian the followin results:
\begin{align*}
 &\eqref{eqn:Q123_N223}\sim_{co}
 \big\{\eqref{eqn:hatQ123_N2},\eqref{eqn:r_type1}\big\},\\
 &\eqref{eqn:Q123_N221}\sim_{co}
 \big\{\eqref{eqn:hatQ123_N2},\eqref{eqn:r_type2}\big\},~
 \big\{\eqref{eqn:hatQ123_N2},\eqref{eqn:r_type3}\big\}.
\end{align*}
Therefore, we have completed the proof.
%
%
%
%
%
%
\end{proof}

\begin{lemma}
Under the equivalence relation $\sim_{co}$,
the polynomials $P_i=P_i(x,y,z,w;A_1,A_2,A_3,A_4)$, $i=1,2,3$, in the condition \ref{cond:CO3} can be given by one of
\begin{subequations}
\begin{align}
 &\begin{cases}\label{eqn:P123_Mobius_1}
 P_1=\polyN3\Big(r_4(x), r_2(y), r_1(z), r_5(w);A_1,A_2,A_3,A_4\Big),\\
 P_2=\polyN3\Big(r_2(x), r_6(y), r_5(z), r_3(w);A_1,A_2,A_3,A_4\Big),\\
 P_3=\polyN3\Big(r_6(x), r_4(y), r_3(z), r_1(w);A_1,A_2,A_3,A_4\Big),
 \end{cases}\\
 &\begin{cases}\label{eqn:P123_Mobius_2}
 P_1=\polyN2\Big(r_2(y), r_4(x), r_5(w), r_1(z);A_1,A_2,A_3,A_4\Big),\\
 P_2=\polyN2\Big(r_2(x), r_6(y), r_5(z), r_3(w);A_1,A_2,A_3,A_4\Big),\\
 P_3=\polyN3\Big(r_6(x), r_4(y), r_3(z), r_1(w);A_1,A_2,A_3,A_4\Big),
 \end{cases}\\
 &\begin{cases}\label{eqn:P123_Mobius_3}
 P_1=\polyN2\Big(r_4(x), r_2(y), r_1(z), r_5(w);A_1,A_2,A_3,A_4\Big),\\
 P_2=\polyN2\Big(r_6(y), r_2(x), r_3(w), r_5(z);A_1,A_2,A_3,A_4\Big),\\
 P_3=\polyN1\Big(r_6(x), r_4(y), r_3(z), r_1(w);A_1,A_2,A_3,A_4\Big),
 \end{cases}\\
 &\begin{cases}\label{eqn:P123_Mobius_4}
 P_1=\polyN1\Big(r_4(x), r_2(y), r_1(z), r_5(w);A_1,A_2,A_3,A_4\Big),\\
 P_2=\polyN1\Big(r_2(x), r_6(y), r_5(z), r_3(w);A_1,A_2,A_3,A_4\Big),\\
 P_3=\polyN1\Big(r_6(x), r_4(y), r_3(z), r_1(w);A_1,A_2,A_3,A_4\Big).
 \end{cases}
\end{align}
\end{subequations}
Here, $r_i$, $i=1,\dots,6$, are M\"obius transformations for arbitrary variable $X$ defined by
\begin{equation}\label{eqn:mobius_r_dij}
 r_i(X)=\dfrac{d_{i1}X+d_{i2}}{d_{i3}X+d_{i4}},
\end{equation}
where $d_{ij}$, $i=1,\dots,6$, $j=1,\dots,4$, are complex parameters.
\end{lemma}
\begin{proof}
Let $P_i=P_i(x,y,z,w)$, $i=1,2,3$, be the polynomials given by the condition \ref{cond:CO3}.
The octahedron with quad-equations $\{Q_7,Q_8,Q_9\}$, where
\begin{align*}
 &Q_7=P_1\left(u_4,u_2,u_1,u_5;B^{(1)}_1,\dots,B^{(1)}_{k_1}\right),\quad
 Q_8=P_2\left(u_2,u_6,u_5,u_3;B^{(1)}_1,\dots, B^{(1)}_{k_1}\right),\\
 &Q_9=P_3\left(u_6,u_4,u_3,u_1;B^{(2)}_1,\dots,B^{(2)}_{k_2}\right),
\end{align*}
satisfies the CAO property.
Then, from Lemma \ref{lemma:classification_Q789_CAO} there exist the pair of a M\"obius transformation of $u_i$, $i=1,\dots,6$,
and a permutation $\si\in G^{(co)}$ such that
$\{Q_7,Q_8,Q_9\}$ maps to one of \eqref{eqn:Q123_N333}--\eqref{eqn:Q123_N111}.
Therefore, we have completed the proof.
\end{proof}

\begin{lemma}\label{lemma:classification_Q123456789_1}
Under the equivalence relation $\sim_{co}$, the square property leads to three canonical classes of polynomials $P_i$, $i=1,2,3$, in Condition \ref{condition:CO}, which are given by
\begin{subequations}
\begin{align}
 &\begin{cases}
 P_1=\polyN3\Big(x, y, z, w;A_1,A_2,A_3,A_4\Big),\\
 P_2=\polyN3\Big(x, y, z, w;A_1,A_2,A_3,A_4\Big),\\
 P_3=\polyN3\Big(x, y, z, w;A_1,A_2,A_3,A_4\Big),
 \end{cases}\label{eqn:P123_caseN3_1}\\
 &\begin{cases}
 P_1=\polyN3\Big(x, y, z, w^{-1};A_1,A_2,A_3,A_4\Big),\\
 P_2=\polyN3\Big(x, y, z^{-1}, w^{-1};A_1,A_2,A_3,A_4\Big),\\
 P_3=\polyN3\Big(x, y, z^{-1}, w;A_1,A_2,A_3,A_4\Big),
 \end{cases}\label{eqn:P123_caseN3_3}\\
 &\begin{cases}
 P_1=\polyN1\Big(x, y, d_1 z, d_2 w;A_1,A_2,A_3,A_4\Big),\\
 P_2=\polyN1\Big(x, y, d_2 z, d_3 w;A_1,A_2,A_3,A_4\Big),\\
 P_3=\polyN1\Big(x, y, d_3 z, d_1 w;A_1,A_2,A_3,A_4\Big),
 \end{cases}\label{eqn:P123_caseN1}
\end{align}
\end{subequations}
where $d_j$, $j=1,2,3$, are complex parameters satisfying
\begin{equation}
 d_1 d_2 d_3=-1.
\end{equation}
\end{lemma}

The proof of Lemma \ref{lemma:classification_Q123456789_1} is given in Appendix \ref{section:proof_lemma_CQ123456789_1}.

These canonical classes lead to corresponding quad-equations $\{Q_1,\dots,Q_9\}$ given by \eqref{eqns:Q123456789_CO3} with $k_1=k_2=k_3=4$
and $\{P_1,P_2,P_3\}$ being one of \eqref{eqn:P123_caseN3_1}--\eqref{eqn:P123_caseN1}. The resulting three types of quad-equations are expressed by Lemmas \ref{lemma:classification_CACO_CO1CO2CO3_1}, 
\ref{lemma:classification_CACO_CO1CO2CO3_2}, and
\ref{lemma:classification_CACO_CO1CO2CO3_3} respectively.

Each type leads to lengthy expressions satisfied by the parameters in the quad-equations. Due to the length of these expressions, the details are collected in Appendix \ref{section:list_theorem_conditions}, to enable concise references to the conditions on parameters in the statements of the following lemmas.

\begin{lemma}[Type I]
\label{lemma:classification_CACO_CO1CO2CO3_1}
Consider the set of polynomials $\{P_1,P_2,P_3\}$ defined by \eqref{eqn:P123_caseN3_1}.
Then, the cuboctahedron with quad-equations \eqref{eqns:Q123456789_CO3} has the CACO property and the square property, 
if and only if the parameters satisfy 
\begin{equation}\label{eqn:caseN3_CAO_B2}\tag{I,II}
\begin{split}
 B^{(2)}_1=B^{(1)}_4 B^{(3)}_1-B^{(1)}_1 B^{(3)}_3,\quad
 B^{(2)}_2=B^{(1)}_2 B^{(3)}_4-B^{(1)}_3 B^{(3)}_2,\\
 B^{(2)}_3=B^{(1)}_4 B^{(3)}_4-B^{(1)}_1 B^{(3)}_2,\quad
 B^{(2)}_4=B^{(1)}_2 B^{(3)}_1-B^{(1)}_3 B^{(3)}_3,
\end{split}
\end{equation}
and one of the conditions \eqref{eqn:typeI_cond1} and \eqref{eqn:typeI_cond2} 
(or equivalent conditions under the equivalence relation $\sim_{co}$).
Note that we can without loss of generality replace $B^{(2)}_i$, $i=1,\dots,4$, in the condition \eqref{eqn:caseN3_CAO_B2} with $cB^{(2)}_i$, $i=1,\dots,4$, where $c$ is an arbitrary non-zero complex parameter.
\end{lemma}

\begin{lemma}[Type II]
\label{lemma:classification_CACO_CO1CO2CO3_2}
Consider the set of polynomials $\{P_1,P_2,P_3\}$ defined by \eqref{eqn:P123_caseN3_3}.
Then, the cuboctahedron with quad-equations \eqref{eqns:Q123456789_CO3} has the CACO property and the square property, 
if and only if the parameters satisfy
\begin{align}
 &C^{(13)}=C^{(22)}=C^{(31)}=C^{(44)}=0,\label{eqn:typeII_cond0}\tag{II.a}\\
 &C^{(12)},C^{(14)},C^{(21)},C^{(23)},C^{(32)},C^{(34)},C^{(41)},C^{(43)}\neq0,\label{eqn:typeII_C_nonzero}\tag{II.b}\\
 &A^{(i)}_1,A^{(i)}_2,A^{(i)}_3,A^{(i)}_4,
 B^{(i)}_1,B^{(i)}_2,B^{(i)}_3,B^{(i)}_4\neq0\quad (i=1,3),\label{eqn:typeII_AB_nonzero}\tag{II.c}
\end{align}
along with condition \eqref{eqn:caseN3_CAO_B2} and one of the conditions \eqref{eqn:typeII_cond1} and \eqref{eqn:typeII_cond2}
(or equivalent conditions under the equivalence relation $\sim_{co}$).
Here, the parameters $C^{(ij)}$ are given by \eqref{eqn:typeII_defC11}--\eqref{eqn:typeII_defC46}.
\end{lemma}

\begin{lemma}[Type III]
\label{lemma:classification_CACO_CO1CO2CO3_3}
Consider the set of polynomials $\{P_1,P_2,P_3\}$ defined by \eqref{eqn:P123_caseN1}.
Then, the cuboctahedron with quad-equations \eqref{eqns:Q123456789_CO3} has the CACO property and the square property, 
if and only if the parameters satisfy 
\begin{equation}\label{eqn:caseN1_CAO_B2}\tag{III}
\begin{split}
 B^{(2)}_1=B^{(1)}_3 B^{(3)}_1-B^{(1)}_1 B^{(3)}_2,\quad
 B^{(2)}_2=B^{(1)}_3 B^{(3)}_3-B^{(1)}_1 B^{(3)}_4,\\
 B^{(2)}_3=B^{(1)}_4 B^{(3)}_1-B^{(1)}_2 B^{(3)}_2,\quad
 B^{(2)}_4=B^{(1)}_4 B^{(3)}_3-B^{(1)}_2 B^{(3)}_4,
 \end{split}
\end{equation}
and one of the conditions {\rm(Type III-1-1)}--{\rm(Type III-3-16)} in Appendix \ref{subsection:proof_cond_para_N1}
(or equivalent conditions under the equivalence relation $\sim_{co}$).
Note that we can without loss of generality replace $B^{(2)}_i$, $i=1,\dots,4$, in the condition \eqref{eqn:caseN1_CAO_B2} with $cB^{(2)}_i$, $i=1,\dots,4$, where $c$ is an arbitrary non-zero complex parameter.
\end{lemma}

The proofs of Lemmas \ref{lemma:classification_CACO_CO1CO2CO3_1}--\ref{lemma:classification_CACO_CO1CO2CO3_3} are given in Appendices \ref{subsection:caseN3_1}--\ref{subsection:caseN1}, respectively.

The results of Lemmas \ref{lemma:classification_CACO_CO1CO2CO3_1}--\ref{lemma:classification_CACO_CO1CO2CO3_3} lead to the following theorem.

\begin{theorem}\label{theo:classification_CACO}
Under the equivalence relation $\sim_{co}$,
any cuboctahedron with quad-equations $\{Q_1,\dots,Q_9\}$ satisfying the properties \ref{cond:CO1}--\ref{cond:CO3} is equivalent to the cuboctahedron with quad-equations \eqref{eqns:Q123456789_CO3},
where the set of polynomials $\{P_1,P_2,P_3\}$ is given by one of \eqref{eqn:P123_caseN3_1}--\eqref{eqn:P123_caseN1},
and the conditions of the parameters are listed in Lemmas \ref{lemma:classification_CACO_CO1CO2CO3_1}--\ref{lemma:classification_CACO_CO1CO2CO3_3}.
\end{theorem}

\section{Concluding remarks}
\label{ConcludingRemarks}
In this paper, we presented new definitions of consistency around an octahedron and a cuboctahedron.
Moreover, we showed a classification of quad-equations which have the CACO property (see Theorem \ref{theo:classification_CACO}).

In a separate paper \cite{JN:Preparing}, we will give the details of Proposition \ref{prop:red_A2_Painleve}, that is,
the reduction from the system \eqref{eqns:samplePDEs_P123456} to the $A_2^{(1)\ast}$-type discrete Painlev\'e equations.

There are many open questions. One is to consider different polyhedra that may lead to other consistent arrangements of quad-equations. Another important question is to ask the same question for polytopes in higher dimensions. Such consistent arrangements of quad-equations leave open the intriguing possibility that all discrete Painlev\'e equations in Sakai's diagram \cite[Tables 3 and 4]{SakaiH2001:MR1882403} may be found by reduction of \PDE s.
\subsection*{Acknowledgment}
N. Nakazono would like to thank Profs M. Noumi, Y. Ohta and Y. Yamada  and Drs M. Kanki and T. Mase for inspiring and fruitful discussions.
This research was supported by an Australian Laureate Fellowship \# FL120100094 and grant \# DP160101728 from the Australian Research Council and JSPS KAKENHI Grant Numbers JP19K14559 and JP17J00092.
N. Nakazono gratefully acknowledges support from Nordita that enabled his attendance at the conference on Elliptic integrable systems, special functions and quantum field theory where this work was presented in June 16-20, 2019.

\appendix
\section{Proof of Lemma \ref{lemma:classification_CAO_1}}
\label{section:proof_lemma_CAO}
Here we give a proof of Lemma \ref{lemma:classification_CAO_1}. 
The proof requires several steps. Before we consider the 22 cases, we deduce some conditions on the quad-equations $\{Q_1,Q_2,Q_3\}$ placed on a CAO octahedron, where
\begin{equation}
 Q_1=Q_1(u_4,u_2,u_1,u_5),\quad
 Q_2=Q_2(u_2,u_6,u_5,u_3),\quad
 Q_3=Q_3(u_6,u_4,u_3,u_1).
\end{equation}
We can assume that $Q_1$ is given by one of the polynomials in Remark \ref{rem:QHDN}, without loss of generality.

We assume that the polynomial $Q_2$ takes the following generic form:
\begin{align}\label{eqn:Q2_general_1}
 &A_1u_2u_6u_5u_3+A_2u_2u_6u_5+A_3u_2u_6u_3
 +A_4u_2u_5u_3+A_5u_6u_5u_3\notag\\
 &\quad+A_6u_2u_6+A_7u_2u_5+A_8u_2u_3+A_9u_6u_5+A_{10}u_6u_3+A_{11}u_5u_3\notag\\
 &\quad+A_{12}u_2+A_{13}u_6+A_{14}u_5+A_{15}u_3+A_{16},
\end{align} 
or, equivalently,
\begin{equation}\label{eqn:Q2_general_2}
 q_1+u_5q_2+u_2q_3+u_2u_5q_4,
\end{equation}
where
\begin{subequations}
\begin{align}
 &q_1=q(u_6,u_3;A_{10},A_{13},A_{15},A_{16}),
 &&q_2=q(u_6,u_3;A_5,A_9,A_{11},A_{14}),\\
 &q_3=q(u_6,u_3;A_3,A_6,A_8,A_{12}),
 &&q_4=q(u_6,u_3;A_1,A_2,A_4,A_7).
\end{align}
\end{subequations}
Here, $A_i$, $i=1,\dots,16$, are complex parameters and the polynomial $q(x,y;a,b,c,d)$ is defined by Equation \eqref{eqn:def_q}.

Note that the above description includes cases that violate the condition that $Q_2$ must be a quad-equation. Such inappropriate cases include
\begin{subequations}
\begin{align}
 &q_1=q_2=0\quad \Rightarrow\quad \text{$Q_2$ is reducible},\label{eqn:proof_q1_q2_0}\\
 &q_1=q_3=0\quad \Rightarrow\quad \text{$Q_2$ is reducible},\label{eqn:proof_q1_q3_0}\\
 &q_2=q_4=0\quad \Rightarrow\quad \deg_{u_5}{Q_2}=0,\label{eqn:proof_q2_q4_0}\\
 &q_3=q_4=0\quad \Rightarrow\quad \deg_{u_2}{Q_2}=0.\label{eqn:proof_q3_q4_0}
\end{align}
\end{subequations}

In general, $Q_1$ and $Q_2$ involve 6 variables, $u_i$, $1\le i\le 6$. We start by eliminating a variable common to both, $u_5$, say, from the equations $Q_1=0$ and $Q_2=0$. The result is an algebraic equation with variables $\{u_1,u_2,u_3,u_4,u_6\}$, which we can write in the following form:
\begin{equation}\label{eqn:powerseries_u2}
 \sum_{k=0}^2C^{(k)}(u_1,u_3,u_4,u_6){u_2}^k=0,
\end{equation}
where the coefficients $C^{(k)}$, $k=0,1,2$, are polynomials, linear in each variable $u_1,u_3,u_4,u_6$, i.e.,
\begin{equation}
 C^{(k)}\in\bbC[u_1,u_3,u_4,u_6],\quad
 \deg_{u_i}C^{(k)}\leq1,\quad i=1,3,4,6.
\end{equation}
Below, we explicitly find $C^{(k)}$ in each case.

Now we use $Q_3=0$ to express the variable $u_6$ as a rational function of $\{u_1,u_3,u_4\}$. Letting 
\begin{equation}
 u_6=f(u_1,u_3,u_4)\in\bbC(u_1,u_3,u_4),
\end{equation}
we obtain
\begin{equation}\label{eqn:Ck_u134}
 C^{(k)}=C^{(k)}\Big(u_1,u_3,u_4,f(u_1,u_3,u_4)\Big)\in\bbC(u_1,u_3,u_4),\quad k=0,1,2.
\end{equation}
The CAO property relies on the fact that $\{u_1,u_2,u_3,u_4\}$ are initial values, i.e., $u_2$ does not depend on $\{u_1,u_3,u_4\}$. In other words, the coefficients $C^{(k)}$ in Equation \eqref{eqn:powerseries_u2} must vanish identically, i.e.,
\begin{equation}\label{eqn:Ckequals0}
 C^{(k)}\Big(u_1,u_3,u_4,u_6\Big)=0,\quad
 k=0,1,2.
\end{equation}

Since the arguments of $C^{(k)}$ are the same as those of $Q_3$, we are led to conclude that one of the following properties must hold:
\begin{enumerate}[leftmargin=1.4cm,label={\rm (\roman*)}]
\item
$C^{(k)}\equiv0$, that is, $C^{(k)}$ vanishes identically; or,
\item
$C^{(k)}\equiv Q_3$, that is, $C^{(k)}$ is equal to a constant multiple of $Q_3$.
\end{enumerate}
Below, we shall prove 
Lemmas \ref{lemma:CAO_Q4321}--\ref{lemma:CAO_N1}
which together prove Lemma \ref{lemma:classification_CAO_1}.

\begin{lemma}\label{lemma:CAO_Q4321}
Let $Q_1$ be one of 
\begin{subequations}
\begin{align}
 &\polyQ4(u_4,u_2,u_1,u_5;\al,\be),\label{eqn:Q1_Q4}\\
 &\polyQ3(u_4,u_2,u_1,u_5;\al,\be;\de),\label{eqn:Q1_Q3}\\
 &\polyQ2(u_4,u_2,u_1,u_5;\al,\be),\label{eqn:Q1_Q2}\\
 &\polyQ1(u_4,u_2,u_1,u_5;\al,\be;\de).\label{eqn:Q1_Q1}
\end{align}
\end{subequations}
Then, the octahedron with quad-equations $\{Q_1,Q_2,Q_3\}$ does not have the CAO property.
\end{lemma}
\begin{proof}
Consider the case \eqref{eqn:Q1_Q4}. 
(The argument is similar for the remaining cases.) 
Recall that $Q_2$ is given by Equation \eqref{eqn:Q2_general_2}.
Then, the coefficients in Equation \eqref{eqn:powerseries_u2} become
\begin{equation}
\begin{cases}
 C^{(0)}=-\Big(\sn{\al}u_1+\sn{\be}u_4\Big)q_1
 -\sn{\al+\be}\Big(\sn{\al} \sn{\be}+u_1 u_4\Big)q_2,\\
 C^{(1)}=\sn{\al+\be}\Big(1+k^2 \sn{\al} \sn{\be} u_1 u_4\Big) q_1
 +\Big(\sn{\be} u_1+\sn{\al} u_4\Big) q_2\\
 \hspace{3em}-\Big(\sn{\al} u_1+\sn{\be} u_4\Big) q_3
 -\sn{\al+\be}\Big(\sn{\al} \sn{\be}+u_1 u_4\Big) q_4,\\
 C^{(2)}=\sn{\al+\be}\Big(1+k^2 \sn{\al} \sn{\be} u_1 u_4\Big)q_3
 +\Big(\sn{\be} u_1+\sn{\al} u_4\Big) q_4.
\end{cases}
\end{equation}
If $C^{(0)}\equiv 0$, then $q_1=q_2=0$, and so we obtain the inadmissible case \eqref{eqn:proof_q1_q2_0}. 
By the argument given above, this implies $C^{(0)}\equiv Q_3$. 
In a similar manner, we can deduce that $C^{(2)}\equiv Q_3$.

Recalling Equation \eqref{eqn:Ckequals0}, we have two equations $C^{(0)}=0$ and $C^{(2)}=0$, which are linear in $u_4$. We eliminate this variable between the two equations to obtain
\begin{align}
 &\sn{\be}\sn{\al+\be}\Big(q_2 q_4-k^2\sn{\al}^2q_1 q_3\Big){u_1}^2
 +\sn{\be} \sn{\al+\be}\Big(q_1 q_3-\sn{\al}^2q_2 q_4\Big)\notag\\
 &-\Big((\sn{\al}^2-\sn{\be}^2) q_1 q_4-\sn{\al+\be}^2 (1-k^2 \sn{\al}^2 \sn{\be}^2) q_2 q_3\Big)u_1=0.
 \label{eqn:Q4algebraiceqn}
\end{align}
But, the condition that $C^{(0)}\equiv Q_3\equiv C^{(2)}$ means Equation \eqref{eqn:Q4algebraiceqn} must vanish identically.
Consequently, each coefficient of a power of $u_1$ must vanish.
These lead to conditions \eqref{eqn:proof_q1_q2_0} or \eqref{eqn:proof_q3_q4_0}, which are inadmissible. 
Therefore, the CAO property cannot hold.

%
%
%
%
\end{proof}

\begin{lemma}\label{lemma:CAO_D432}
Let $Q_1$ be one of 
\begin{subequations}
\begin{align}
 &\polyD4(u_4,u_2,u_1,u_5;\de_1,\de_2,\de_3),\label{eqn:Q1_D4}\\
 &\polyD3(u_4,u_2,u_1,u_5),\label{eqn:Q1_D3}\\
 &\polyD2(u_4,u_2,u_1,u_5;1),\label{eqn:Q1_D2}
\end{align}
\end{subequations}
where $(\de_1,\de_2)\neq(0,0)$.
Then, the octahedron with quad-equations $\{Q_1,Q_2,Q_3\}$ does not have the CAO property.
\end{lemma}
\begin{proof}
We consider each case separately. 
Consider the first case \eqref{eqn:Q1_D4}. 
Then, the coefficients in Equation \eqref{eqn:powerseries_u2} become
\begin{equation}
\begin{cases}
 C^{(0)}=\de_2 u_1 q_1-(\de_3+u_1 u_4) q_2,\\
 C^{(1)}=q_1-\de_1 u_1 q_2+\de_2 u_1 q_3-(\de_3+u_1 u_4) q_4,\\
 C^{(2)}=q_3-\de_1 u_1 q_4.
\end{cases}
\end{equation}
Notice that $C^{(2)}$ does not involve $u_4$, therefore, $C^{(2)}\equiv0$.
(In what follows, we omit this element of the argument for simplicity.) 
We now have two possibilities.
\begin{description}
\item[(i)]
Let $\de_1=1$.
From $C^{(2)}\equiv0$ we obtain the condition \eqref{eqn:proof_q3_q4_0}, which is inadmissible.
\item[(ii)]
Let $\de_1=0$ and $\de_2=1$.
From $C^{(2)}\equiv0$ we obtain $q_3=0$.
Then, from
\begin{equation}
 C^{(0)}-u_1C^{(1)}=(\de_3+u_1 u_4)(u_1q_4-q_2)\equiv 0,
\end{equation}
we obtain the condition \eqref{eqn:proof_q2_q4_0}, which is inadmissible.
\end{description}
Therefore, the CAO property does not hold.

We next consider the second case \eqref{eqn:Q1_D3}.
Then, we obtain
\begin{equation}
\begin{cases}
 C^{(0)}=(u_1+u_4) q_1-u_1 u_4 q_2,\\
 C^{(1)}=-q_2+(u_1+u_4) q_3-u_1 u_4 q_4,\\
 C^{(2)}=-q_4.
\end{cases}
\end{equation}
From $C^{(2)}\equiv0$, we obtain $q_4=0$.
Then, from
\begin{equation}
 C^{(0)}+{u_1}^2C^{(1)}=(u_1+u_4)(q_1-q_2u_1+q_3{u_1}^2)\equiv 0,
\end{equation}
we obtain the condition \eqref{eqn:proof_q1_q2_0}, which is inadmissible.
Therefore, the CAO property does not hold.

We finally consider the case \eqref{eqn:Q1_D2}.
Then, we obtain
\begin{equation}
\begin{cases}
 C^{(0)}=-(u_1+u_4) q_2,\\
 C^{(1)}=q_1-u_4 q_2-(u_1+u_4) q_4,\\
 C^{(2)}=q_3-u_4 q_4.
\end{cases}
\end{equation}
From $C^{(2)}\equiv0$, we obtain the condition \eqref{eqn:proof_q3_q4_0}, which is inadmissible.
Therefore, the CAO property does not hold. 
\end{proof}

\begin{lemma}\label{lemma:CAO_H3}
Let $Q_1$ be one of 
\begin{subequations}
\begin{align}
 &\polyH3(u_4,u_2,u_1,u_5;\al,\be;\de_1,\de_2),\label{eqn:Q1_H3_1}\\
 &\polyH3(u_4,u_2,u_5,u_1;\al,\be;\de_1,\de_2),\label{eqn:Q1_H3_2}\\
 &\polyH3(u_4,u_1,u_2,u_5;\al,\be;\de_1,\de_2),\label{eqn:Q1_H3_3}
\end{align}
\end{subequations}
where $(\de_1,\de_2)\neq(0,0)$.
Then, the octahedron with quad-equations $\{Q_1,Q_2,Q_3\}$ does not have the CAO property.
\end{lemma}
\begin{proof}
We consider each case separately. 
Consider the first case \eqref{eqn:Q1_H3_1}.
Then, we obtain
\begin{equation}
\begin{cases}
 C^{(0)}=\al\be(\al u_1-\be u_4)q_1-\al\be\de_1(\al^2-\be^2)q_2,\\
 C^{(1)}=\de_2(\al^2-\be^2)q_1+\al\be(\be u_1-\al u_4)q_2+\al\be(\al u_1-\be u_4)q_3\\
 \hspace{3em}-\de_1\al\be(\al^2-\be^2)q_4,\\
 C^{(2)}=\de_2(\al^2-\be^2)q_3+\al\be(\be u_1-\al u_4)q_4.
\end{cases}
\end{equation}
We now have three possibilities.
\begin{description}
\item[(i)]
Let $\de_1=\de_2=1$. 
Eliminating $u_4$ from $C^{(0)}=0$ and $C^{(2)}=0$, we obtain
\begin{equation}
 \al q_1 q_4 u_1-\al^2 q_2 q_4-q_1 q_3\equiv 0,
\end{equation}
which gives $q_1q_4=0$.
If $q_1=0$, then from $C^{(0)}=0$ we obtain the condition \eqref{eqn:proof_q1_q2_0}, which is inadmissible.
Moreover, if $q_4=0$, then from $C^{(2)}=0$ we obtain the condition \eqref{eqn:proof_q3_q4_0}, which is inadmissible.
\item[(ii)]
Let $\de_1=1$ and $\de_2=0$.
Then, from $C^{(2)}\equiv0$ we obtain $q_4=0$,
and then from $C^{(0)}=0$ and $C^{(1)}=0$ we respectively obtain
\begin{subequations}
\begin{align}
 &(\al u_1-\be u_4) q_1-(\al^2-\be^2) q_2=0,\\
 &(\be u_1-\al u_4) q_2+(\al u_1-\be u_4) q_3=0.
\end{align}
\end{subequations}
Eliminating $u_4$ from the equations above, we obtain
\begin{equation}
 q_1q_2 u_1-(\al q_2+\be q_3)q_2\equiv0,
\end{equation}
which gives $q_1q_2=0$.
From $C^{(0)}=0$ and $q_1q_2=0$, we obtain the condition \eqref{eqn:proof_q1_q2_0}, which is inadmissible.
\item[(iii)]
Let $\de_1=0$ and $\de_2=1$.
From $C^{(0)}\equiv 0$ we obtain $q_1=0$,
and then from $C^{(1)}=0$ and $C^{(2)}=0$ we respectively obtain
\begin{subequations}
\begin{align}
 &(\be u_1-\al u_4) q_2+(\al u_1-\be u_4)q_3=0,\\
 &(\al^2-\be^2) q_3+\al\be(\be u_1-\al u_4) q_4=0.
\end{align}
\end{subequations}
Eliminating $u_4$ from the equations above, we obtain
\begin{equation}
 \al\be q_3q_4u_1-(\al q_2+\be q_3)q_3\equiv 0,
\end{equation}
which gives $q_3q_4=0$.
Then, from $C^{(2)}=0$ and $q_3q_4=0$, we obtain the condition \eqref{eqn:proof_q3_q4_0}, which is inadmissible.
\end{description}
Therefore, the CAO property does not hold.

We next consider the second case \eqref{eqn:Q1_H3_2}.
Then, we obtain
\begin{equation}
 \begin{cases}
 C^{(0)}=\al\be\left(\al u_1 q_1+\Big(\be u_1 u_4-(\al^2-\be^2)\de_1\Big)q_2\right),\\
 C^{(1)}=-\al\be^2 q_1-\Big(\al^2\be u_4+(\al^2-\be^2)\de_2 u_1\Big)q_2+\al^2\be u_1 q_3\\
 \hspace{3em}+\al\be\Big(\be u_1 u_4-(\al^2-\be^2)\de_1\Big)q_4,\\
 C^{(2)}=-\al\be^2 q_3-\Big(\al^2\be u_4+(\al^2-\be^2)\de_2 u_1\Big) q_4.
\end{cases}
\end{equation}
In the exact same manner as the case \eqref{eqn:Q1_Q4}, 
we can show that the CAO property does not hold.

We finally consider the case \eqref{eqn:Q1_H3_3}.
Then, we obtain
\begin{equation}
 \begin{cases}
 C^{(0)}=-\Big(\al\be^2 u_4-(\al^2-\be^2)\de_2 u_1\Big)q_1-\al\be\Big(\al u_1 u_4+(\al^2-\be^2)\de_1\Big)q_2,\\
 C^{(1)}=\al^2\be q_1+\al\be^2 u_1q_2-\Big(\al\be^2 u_4-(\al^2-\be^2)\de_2 u_1\Big)q_3\\
 \hspace{3em}-\al\be\Big(\al u_1 u_4+(\al^2-\be^2)\de_1\Big)q_4,\\
 C^{(2)}=\al\be(\al q_3+\be u_1 q_4).
\end{cases}
\end{equation}
From $C^{(2)}\equiv0$, we obtain the condition \eqref{eqn:proof_q3_q4_0}, which is inadmissible.
Therefore, the CAO property does not hold.
\end{proof}

\begin{lemma}\label{lemma:CAO_H2}
Let $Q_1$ be one of 
\begin{subequations}
\begin{align}
 &\polyH2(u_4,u_2,u_1,u_5;\al,\be;\de),\label{eqn:Q1_H2_1}\\
 &\polyH2(u_4,u_2,u_5,u_1;\al,\be;\de),\label{eqn:Q1_H2_2}\\
 &\polyH2(u_4,u_1,u_2,u_5;\al,\be;\de).\label{eqn:Q1_H2_3}
\end{align}
\end{subequations}
Then, the octahedron with quad-equations $\{Q_1,Q_2,Q_3\}$ does not have the CAO property.
\end{lemma}
\begin{proof}
We consider each case separately.
Consider the first case \eqref{eqn:Q1_H2_1}.
Then, we obtain
\begin{equation}
\begin{cases}
 C^{(0)}=\left((\al-\be)\Big(1+2(\al+\be)\de\Big)-u_1+u_4\right) q_1\\
 \hspace{3em}-(\al-\be)\Big(\al+\be+2(\al^2+\be^2)\de+u_1+u_4\Big) q_2,\\
 C^{(1)}=4 (\al-\be)\de q_1
 -\Big(\al-\be+2(\al^2-\be^2)\de+u_1-u_4\Big) q_2\\
 \hspace{3em}+\Big(\al-\be+2(\al^2-\be^2)\de-u_1+u_4\Big) q_3\\
 \hspace{3em}-(\al-\be)\Big(\al+\be+2(\al^2+\be^2)\de+u_1+u_4\Big)q_4,\\
 C^{(2)}=4 (\al-\be)\de q_3-\Big(\al-\be+2(\al^2-\be^2)\de+u_1-u_4\Big) q_4.
\end{cases}
\end{equation}
We now have two possibilities.
\begin{description}
\item[(i)]
Let $\de=1$.
Eliminating $u_4$ from $C^{(0)}=0$ and $C^{(2)}=0$, we obtain
\begin{equation}
 q_2q_4u_1+q_1\Big(2q_3-(1+2\al+2\be)q_4\Big)-q_2\Big(2(\al-\be)q_3-\al(1+2\al)q_4\Big)\equiv0,
\end{equation}
which gives $q_2q_4=0$. 
If $q_2=0$, then from $C^{(0)}=0$ we obtain the condition \eqref{eqn:proof_q1_q2_0}, which is inadmissible, 
and if $q_4=0$, then from $C^{(2)}=0$ we obtain the condition \eqref{eqn:proof_q3_q4_0}, which is inadmissible.
\item[(ii)]
Let $\de=0$.
From $C^{(2)}\equiv0$ we obtain $q_4=0$,
and then from $C^{(0)}=0$ and $C^{(1)}=0$ we respectively obtain
\begin{subequations}
\begin{align}
 &(\al-\be-u_1+u_4) q_1-(\al-\be) (\al+\be+u_1+u_4) q_2=0,\\
 &(\al-\be+u_1-u_4) q_2-(\al-\be-u_1+u_4) q_3=0.
\end{align}
\end{subequations}
Eliminating $u_4$ from the equations above, we obtain
\begin{equation}
 q_2(q_2+q_3)u_1-q_2\Big(q_1-\al q_2-\be q_3\Big)\equiv0,
\end{equation}
which gives $q_2(q_2+q_3)=0$.
Because of the condition \eqref{eqn:proof_q2_q4_0}, which is inadmissible, we obtain $q_2+q_3=0$.
However, from $C^{(1)}=0$ and $q_2+q_3=0$, we obtain the condition \eqref{eqn:proof_q2_q4_0}, which is inadmissible.
\end{description}
Therefore, the CAO property does not hold.

We next consider the second case \eqref{eqn:Q1_H2_2}.
Then, we obtain
\begin{equation}
\begin{cases}
 C^{(0)}=-(u_1-\al+\be)q_1\\
 \hspace{1.5em}-\left((\al-\be)\Big(u_4+\al+\be+2(\al^2+\be^2)\de\Big)
 +u_1\Big(u_4+\al-\be+2(\al^2-\be^2)\de\Big)\right)q_2,\\
 C^{(1)}=q_1+\left(u_4-(\al-\be)\Big(1+2\de(2u_1+\al+\be)\Big)\right)q_2-(u_1-\al+\be)q_3\\
 \hspace{1.5em}-\left((\al-\be)\Big(u_4+\al+\be+2(\al^2+\be^2)\de\Big)
 +u_1\Big(u_4+\al-\be+2(\al^2-\be^2)\de\Big)\right)q_4,\\
 C^{(2)}=q_3+\Big(u_4-\al+\be-2\de(\al-\be)(2u_1+\al+\be)\Big)q_4.
\end{cases}
\end{equation}
In the exact same manner as the case \eqref{eqn:Q1_Q4}, 
we can show that the CAO property does not hold.

We finally consider the case \eqref{eqn:Q1_H2_3}. 
Then, we obtain
\begin{equation}
\begin{cases}
 C^{(0)}=-\Big(u_4+\al-\be+2(\al-\be)\de(2u_1+\al+\be)\Big)q_1\\
 \hspace{1.5em}-\left(u_1\Big(u_4-\al+\be-2(\al^2-\be^2)\de\Big)-(\al-\be)\Big(u_4+\al+\be+2(\al^2+\be^2)\de\Big)\right)q_2,\\
 C^{(1)}=q_1+(u_1+\al-\be)q_2-\Big(u_4+\al-\be+2\de(\al-\be)(2u_1+\al+\be)\Big)q_3\\
 \hspace{1.5em}-\left(u_1\Big(u_4-\al+\be-2(\al^2-\be^2)\de\Big)-(\al-\be)\Big(u_4+\al+\be+2(\al^2+\be^2)\de\Big)\right)q_4,\\
 C^{(2)}=q_3+(u_1+\al-\be)q_4.
\end{cases}
\end{equation}
From $C^{(2)}\equiv0$, we obtain the condition \eqref{eqn:proof_q3_q4_0}, which is inadmissible.
Therefore, the CAO property does not hold.
\end{proof}

\begin{lemma}\label{lemma:CAO_N3}
Let $Q_1$ be one of 
\begin{subequations}
\begin{align}
 &\polyN3(u_4,u_2,u_1,u_5;\al_1,\al_2,\al_3,\al_4),\label{eqn:Q1_N3_1}\\
 &\polyN3(u_4,u_2,u_5,u_1;\al_1,\al_2,\al_3,\al_4),\label{eqn:Q1_N3_2}\\
 &\polyN3(u_4,u_1,u_2,u_5;\al_1,\al_2,\al_3,\al_4).\label{eqn:Q1_N3_3}
\end{align}
\end{subequations}
Then, we obtain quad-equations \eqref{eqn:Q123_N3}, which have the CAO property.
\end{lemma}
\begin{proof}
We consider each case separately. We show below that the case \eqref{eqn:Q1_N3_1} gives the generic system of quad-equations \eqref{eqn:Q123_N3}, while the remaining cases \eqref{eqn:Q1_N3_2} and \eqref{eqn:Q1_N3_3} give special cases of this system.
  
Consider the first case \eqref{eqn:Q1_N3_1}.
Then, we obtain
\begin{equation}
\begin{cases}
 C^{(0)}=(\al_2 u_1+\al_3 u_4)q_1,\\
 C^{(1)}=-(\al_4 u_1+\al_1 u_4)q_2+(\al_2 u_1+\al_3 u_4)q_3,\\
 C^{(2)}=-(\al_4 u_1+\al_1 u_4)q_4.
\end{cases}
\end{equation}
From $C^{(0)}\equiv0$ and $C^{(2)}\equiv0$, we obtain
\begin{equation}
 \al_2 q_1=0,\quad
 \al_3 q_1=0,\quad
 \al_4 q_4=0,\quad
 \al_1 q_4=0.
\end{equation}
Since $(\al_1,\al_4),(\al_2,\al_3)\neq(0,0)$, we obtain $q_1=q_4=0$, which gives
\begin{equation}
 Q_2=u_5q_2+u_2q_3,\qquad
 Q_3=(\al_4 u_1+\al_1 u_4)q_2-(\al_2 u_1+\al_3 u_4)q_3.
\end{equation}
Therefore, we obtain the CAO octahedron with quad-equations \eqref{eqn:Q123_N3}.

We next consider the second case \eqref{eqn:Q1_N3_2}.
Then, we obtain
\begin{equation}
\begin{cases}
 C^{(0)}=u_1(\al_2 q_1-\al_3 u_4 q_2),\\
 C^{(1)}=\al_4 q_1-\al_1 u_4 q_2+\al_2 u_1 q_3-\al_3 u_1 u_4 q_4,\\
 C^{(2)}=\al_4 q_3-\al_1 u_4 q_4.
\end{cases}
\end{equation}
From $C^{(0)}\equiv0$ and $C^{(2)}\equiv0$, we obtain
\begin{equation}
 \al_2 q_1=0,\quad
 \al_3 q_2=0,\quad
 \al_4 q_3=0,\quad
 \al_1 q_4=0.
\end{equation}
We now have two possibilities.
\begin{description}
\item[(i)]
If $\al_2\neq0$, we obtain $q_1=0$ and $q_2,q_3\neq0$ which gives $\al_3=\al_4=0$, $\al_1\neq0$ and $q_4=0$.
Then, we obtain
\begin{equation}
 Q_1=\polyN3(u_4,u_2,u_5,u_1;\al_1,\al_2,0,0)
 =\polyN3(u_4,u_2,u_1,u_5;\al_1,\al_2,0,0),
\end{equation}
which is special case of the case \eqref{eqn:Q1_N3_1}.
\item[(ii)]
If $\al_2=0$ and $\al_3,\al_4\neq0$, we obtain $q_2=q_3=0$, $q_1,q_4\neq0$ and $\al_1=0$.
Then, we obtain
\begin{equation}
 Q_1=\polyN3(u_4,u_2,u_5,u_1;0,0,\al_3,\al_4)
 =u_1u_5\polyN3(u_4,u_2,{u_1}^{-1},{u_5}^{-1};0,0,\al_3,\al_4),
\end{equation}
which is special case of the case \eqref{eqn:Q1_N3_1}.
\end{description}

We finally consider the case \eqref{eqn:Q1_N3_3}.
Then, we obtain
\begin{equation}
\begin{cases}
 C^{(0)}=u_4(\al_3 q_1-\al_1 u_1 q_2),\\
 C^{(1)}=\al_2 q_1-\al_4 u_1 q_2+\al_3 u_4 q_3-\al_1 u_1 u_4 q_4,\\
 C^{(2)}=\al_2 q_3-\al_4 u_1 q_4.
\end{cases}
\end{equation}
From $C^{(0)}\equiv0$ and $C^{(2)}\equiv0$, we obtain
\begin{equation}
 \al_3 q_1=0,\quad
 \al_1 q_2=0,\quad
 \al_2 q_3=0,\quad
 \al_4 q_4=0.
\end{equation}
We now have two possibilities.
\begin{description}
\item[(i)]
If $\al_3\neq0$, we obtain $q_1=0$ and $q_2,q_3\neq0$ which gives $\al_1=\al_2=0$, $\al_4\neq0$ and $q_4=0$.
Then, we obtain
\begin{equation}
 Q_1=\polyN3(u_4,u_1,u_2,u_5;0,0,\al_3,\al_4)
 =\polyN3(u_4,u_2,u_1,u_5;0,0,\al_3,\al_4),
\end{equation}
which is special case of the case \eqref{eqn:Q1_N3_1}.
\item[(ii)]
If $\al_3=0$ and $\al_1,\al_2\neq0$, then we obtain $q_2=q_3=0$, $q_1,q_4\neq0$ and $\al_4=0$.
Then, we obtain
\begin{equation}
 Q_1=\polyN3(u_4,u_1,u_2,u_5;\al_1,\al_2,0,0)
 =u_1u_2\polyN3(u_4,{u_2}^{-1},{u_1}^{-1},u_5;\al_1,\al_2,0,0),
\end{equation}
which is special case of the case \eqref{eqn:Q1_N3_1}.
\end{description}
This completes the proof.
\end{proof}

\begin{lemma}\label{lemma:CAO_N2}
Let $Q_1$ be one of 
\begin{subequations}
\begin{align}
 &\polyN2(u_4,u_2,u_1,u_5;\al_1,\al_2,\al_3,\al_4),\label{eqn:Q1_N2_1}\\
 &\polyN2(u_4,u_2,u_5,u_1;\al_1,\al_2,\al_3,\al_4),\label{eqn:Q1_N2_2}\\
 &\polyN2(u_4,u_1,u_2,u_5;\al_1,\al_2,\al_3,\al_4).\label{eqn:Q1_N2_3}
\end{align}
\end{subequations}
Then, we obtain quad-equations \eqref{eqn:Q123_N2}, which have the CAO property.
\end{lemma}
\begin{proof}
We consider each case separately. 
We show below that the case \eqref{eqn:Q1_N2_1} gives the generic system of quad-equations \eqref{eqn:Q123_N2}, while the remaining two cases do not lead to the CAO property.
  
Consider the first case \eqref{eqn:Q1_N2_1}.
Then, we obtain
\begin{equation}
\begin{cases}
 C^{(0)}=(\al_2 u_1+\al_2 u_4+\al_4)q_1,\\
 C^{(1)}=-(\al_1 u_1+\al_1 u_4+\al_3)q_2+(\al_2 u_1+\al_2 u_4+\al_4)q_3,\\
 C^{(2)}=-(\al_1 u_1+\al_1 u_4+\al_3)q_4.
\end{cases}
\end{equation}
From $C^{(0)}\equiv0$ and $C^{(2)}\equiv0$, we obtain
\begin{equation}
 \al_2 q_1=0,\quad
 \al_4 q_1=0,\quad
 \al_1 q_4=0,\quad
 \al_3 q_4=0.
\end{equation}
Since $(\al_1,\al_3),(\al_2,\al_4)\neq(0,0)$, we obtain $q_1=q_4=0$, which gives
\begin{equation}
 Q_2=u_5q_2+u_2q_3,\qquad
 Q_3=(\al_1 u_1+\al_1 u_4+\al_3)q_2-(\al_2 u_1+\al_2 u_4+\al_4)q_3.
\end{equation}
Therefore, we obtain quad-equations \eqref{eqn:Q123_N2}, which satisfy the CAO property.

Now consider the second case \eqref{eqn:Q1_N2_2}.
Then, we obtain
\begin{equation}
\begin{cases}
 C^{(0)}=u_1(\al_2 q_1-\al_4 q_2-\al_2 q_2 u_4),\\
 C^{(1)}=\al_1 q_1-(\al_1 u_4+\al_3)q_2+\al_2 u_1 q_3-u_1(\al_2 u_4+\al_4)q_4,\\
 C^{(2)}=\al_1 q_3-\al_3 q_4-\al_1 q_4 u_4.
\end{cases}
\end{equation}
From $C^{(0)}\equiv0$ and $C^{(2)}\equiv0$, we obtain
\begin{equation}
 \al_2 q_1=\al_4 q_2,\quad
 \al_2 q_2=0,\quad
 \al_1 q_3=\al_3 q_4,\quad
 \al_1 q_4=0.
\end{equation}
We now have two possibilities.
\begin{description}
\item[(i)]
If $\al_2\neq0$, we obtain $q_2=0$ and $q_1,q_4\neq0$ which gives $\al_1=0$ and $\al_3\neq0$.
These conditions are inconsistent with $\al_2 q_1=\al_4 q_2$.
\item[(ii)]
If $\al_2=0$ and $\al_1,\al_4\neq0$, then from $\al_2 q_1=\al_4 q_2$, we obtain $q_2=0$ and $q_1,q_4\neq0$.
These conditions are inconsistent with $\al_1 q_4=0$.
\end{description}
Therefore, the CAO property does not hold in either sub-case.

We finally consider the case \eqref{eqn:Q1_N2_3}.
Then, we obtain
\begin{equation}
\begin{cases}
 C^{(0)}=(\al_2 u_4+\al_4)q_1-(\al_1 u_4+\al_3)u_1 q_2,\\
 C^{(1)}=\al_2 q_1-\al_1 u_1 q_2+(\al_2 u_4+\al_4)q_3-u_1(\al_1 u_4+\al_3)q_4,\\
 C^{(2)}=\al_2 q_3-\al_1 u_1 q_4.
\end{cases}
\end{equation}
From $C^{(2)}\equiv0$, we obtain
\begin{equation}
 \al_2 q_3=0,\quad
 \al_1 q_4=0.
\end{equation}
We now have two possibilities.
\begin{description}
\item[(i)]
If $\al_2\neq0$, we obtain $q_3=0$ and $q_1,q_4\neq0$ which gives $\al_1=0$ and $\al_3\neq0$.
Then, we obtain
\begin{equation}
 C^{(1)}=\al_2 q_1-(\al_1 q_2+\al_3 q_4)u_1\equiv0,
\end{equation}
which is inconsistent with $\al_2q_1\neq0$.
\item[(ii)]
If $\al_2=0$ and $\al_1,\al_4\neq0$, then we obtain $q_4=0$ and $q_2,q_3\neq0$.
Then, we obtain
\begin{equation}
 C^{(1)}=\al_4 q_3-\al_1 q_2 u_1\equiv0,
\end{equation}
which is inconsistent with $\al_4q_3\neq0$.
\end{description}
Therefore, the CAO property does not hold in either sub-case.
\end{proof}

\begin{lemma}\label{lemma:CAO_N1}
Let $Q_1$ be one of 
\begin{subequations}
\begin{align}
 &\polyN1(u_4,u_2,u_1,u_5;\al_1,\al_2,\al_3,\al_4),\label{eqn:Q1_N1_1}\\
 &\polyN1(u_4,u_2,u_5,u_1;\al_1,\al_2,\al_3,\al_4),\label{eqn:Q1_N1_2}\\
 &\polyN1(u_4,u_1,u_2,u_5;\al_1,\al_2,\al_3,\al_4).\label{eqn:Q1_N1_3}
\end{align}
\end{subequations}
Then, we obtain the CAO octahedron with quad-equations \eqref{eqn:Q123_N1}.
\end{lemma}
\begin{proof}
Consider the first case \eqref{eqn:Q1_N1_1}.
Then, we obtain
\begin{equation}
\begin{cases}
 C^{(0)}=\Big(\al_1(u_1+u_4)+\al_3\Big)q_1-\Big(\al_2(u_1+u_4)+\al_4\Big)q_2,\\
 C^{(1)}=-\Big(\al_1(u_1+u_4)+\al_3\Big)(q_2-q_3)-\Big(\al_2(u_1+u_4)+\al_4\Big)q_4,\\
 C^{(2)}=-\Big(\al_1(u_1+u_4)+\al_3\Big)q_4.
\end{cases}
\end{equation}
From $C^{(2)}\equiv0$, we obtain $q_4=0$.
Then, we obtain
\begin{equation}
 C^{(1)}=-\Big(\al_1(u_1+u_4)+\al_3\Big)(q_2-q_3)\equiv0,
\end{equation}
which gives $q_2=q_3$.
Therefore, we obtain
\begin{subequations}
\begin{align}
 Q_2=&q_1+(u_5+u_2)q_2,\\
 Q_3=&\Big(\al_1(u_1+u_4)+\al_3\Big)q_1-\Big(\al_2(u_1+u_4)+\al_4\Big)q_2,
\end{align}
\end{subequations}
which gives the CAO octahedron with quad-equations \eqref{eqn:Q123_N1}.

Note here that because of the following relation:
\begin{equation}
 \polyN1(x,y,w,z;0,\al_2,\al_3,\al_4)=\polyN1\left(x,y,\frac{\al_3}{\al_2}z,\frac{\al_2}{\al_3}w;0,\al_2,\al_3,\al_4\right),
\end{equation}
the cases \eqref{eqn:Q1_N1_2} and \eqref{eqn:Q1_N1_3} with $\al_1=0$ can be regarded as special cases of \eqref{eqn:Q1_N1_1}.
Therefore, we can assume $\al_1\neq0$, without loss of generality, for the cases \eqref{eqn:Q1_N1_2} and \eqref{eqn:Q1_N1_3}.

We now consider these cases. 
The case \eqref{eqn:Q1_N1_2} leads to
\begin{equation}
\begin{cases}
 C^{(0)}=(\al_1 u_1+\al_2)q_1-(\al_1 u_1 u_4+\al_2 u_4+\al_3 u_1+\al_4)q_2,\\
 C^{(1)}=\al_1 q_1-(\al_1 u_4+\al_3)q_2+(\al_1 u_1+\al_2)q_3\\
 \hspace{3em}-(\al_1 u_1 u_4+\al_2 u_4+\al_3 u_1+\al_4)q_4,\\
 C^{(2)}=\al_1 q_3-(\al_1 u_4+\al_3)q_4,
\end{cases}
\end{equation}
while from the case \eqref{eqn:Q1_N1_3} we obtain
\begin{equation}
\begin{cases}
 C^{(0)}=(\al_1 u_4+\al_3)q_1-( \al_1 u_1 u_4+\al_2 u_4+\al_3 u_1+\al_4)q_2,\\
 C^{(1)}=\al_1 q_1-(\al_1 u_1+\al_2)q_2+(\al_1 u_4+\al_3)q_3\\
 \hspace{3em}-(\al_1 u_1 u_4+\al_2 u_4+\al_3 u_1+\al_4)q_4,\\
 C^{(2)}=\al_1 q_3-(\al_1 u_1+\al_2)q_4.
\end{cases}
\end{equation}
In both cases, from $C^{(2)}\equiv0$ we obtain the condition \eqref{eqn:proof_q3_q4_0}, which is inadmissible.
Therefore, in each case the CAO property does not hold.
\end{proof}

\section{Proof of Lemma \ref{lemma:classification_rational_r}}
\label{section:proof_lemma_r}
Here, we give the proof of Lemma \ref{lemma:classification_rational_r}.
We show that the rational function $r(x,y;\{ B_i\})$ given by \eqref{eqn:def_r_qq} can be reduced to the rational functions $r_1(x, y; \{A_i\})$, $r_2(x, y; B)$ or $r_3(x, y; C_1,C_2)$, which are respectively given by \eqref{eqn:r_type1}, \eqref{eqn:r_type2} and \eqref{eqn:r_type3}.

Firstly, we show that by using a M\"obius transformation we can without loss of generality assume $B_5=0$.
Indeed, assuming $B_5\neq0$, then by the M\"obius transformation $x\to x'$, where
\begin{equation}
 x'=\dfrac{1}{x}-\dfrac{B_7}{B_5},
\end{equation}
we obtain
\begin{equation}
 r(x',y;B_1,B_2,B_3,B_4,B_5,B_6,B_7,B_8)
 =r(x,y;B_1',B_2',B_3',B_4',0,B_6',B_7',B_8'),
\end{equation}
where
\begin{align}
 &B_1'=B_3-\dfrac{B_1B_7}{B_5},\quad
 B_2'=B_4-\dfrac{B_2B_7}{B_5},\quad
 B_3'=B_1,\quad
 B_4'=B_2,\notag\\
 &B_6'=B_8-\dfrac{B_6B_7}{B_5},\quad
 B_7'=B_5,\quad
 B_8'=B_6.
\end{align}
Therefore, we can without loss of generality let $B_5=0$ for $r(x,y;B_1,\dots,B_8)$.
(In what follows, we omit this element of the argument for simplicity.)

In the following, dividing the conditions of the parameters $B_6$ and $B_7$ into four cases \eqref{eqn:typeIII_proof_c_1}--\eqref{eqn:typeIII_proof_c_4},
we prove Lemma \ref{lemma:classification_rational_r}.
\begin{description}
\item[(i) Case $B_6,B_7\neq0$]
By M\"obius transformation we can without loss of generality let $B_8=0$
and then again by M\"obius transformatio we can let $B_1=0$.
In the following, we divide the conditions of the parameters into two cases.
\begin{description}
\item[(i-1) Case $B_2B_7=B_3B_6$]
By the M\"obius transformation $(x,y)\to (x',y')$, where 
\begin{equation}
 x'=\frac{B_4}{B_6}x,\quad
 y'=\frac{B_4}{B_7}y,
\end{equation}
we obtain
\begin{equation}
 r\left(x',y',0,\dfrac{B_3B_6}{B_7},B_3,B_4,0,B_6,B_7,0\right)
 =r_2\left(x,y;\dfrac{B_3}{B_7}\right).
\end{equation}
\item[(i-2) Case $B_2B_7\neq B_3B_6$]
By the M\"obius transformation $(x,y)\to (x',y')$, where 
\begin{equation}
 x'=x-\frac{B_4 B_7}{B_2 B_7-B_3 B_6},\quad
 y'=y+\frac{B_4 B_6}{B_2 B_7-B_3 B_6},
\end{equation}
we obtain
\begin{equation}
 r(x',y',0,B_2,B_3,B_4,0,B_6,B_7,0)
 =r_1(x,y;B_2,B_3,B_6,B_7).
\end{equation}
\end{description}
\item[(ii) Case $B_6=B_7=0$]
We have
\begin{equation}\label{eqn:proof_class_r_ii}
 r(x,y;B_1,B_2,B_3,B_4,0,0,0,B_8)
 =\frac{B_1xy+B_2x+B_3y+B_4}{B_8}.
\end{equation}
In the following, we divide the conditions of the parameter $B_1$ into two cases.
\begin{description}
\item[(ii-1) Case $B_1\neq0$]
By the M\"obius transformation $(x,y)\to (x',y')$, where
\begin{equation}
 x'=\frac{1}{x}-\frac{B_3}{B_1},\quad
 y'=\frac{1}{y}-\frac{B_2}{B_1},
\end{equation}
we obtain
\begin{equation}
 r(x',y';B_1,B_2,B_3,B_4,0,0,0,B_8)
 =r_1(x,y;B_1B_4-B_2B_3,{B_1}^2,B_1B_8,0).
\end{equation}
\item[(ii-2) Case $B_1=0$]
Because of $\partial r/\partial x\neq0$ and $\partial r/\partial y\neq0$, 
we have $B_2,B_3\neq0$.
Therefore, by the M\"obius transformation $(x,y)\to (x',y')$, where
\begin{equation}
 x'=\frac{B_8x-B_4}{B_2},\quad
 y'=\frac{B_8y}{B_3},
\end{equation}
we obtain
\begin{equation}
 r(x',y';0,B_2,B_3,B_4,0,0,0,B_8)
 =r_3(x,y).
\end{equation}
\end{description}
\item[(iii) Case $B_6=0$ and $B_7\neq0$]
By the M\"obius transformation $y\to y'$, where
\begin{equation}
 y'=\dfrac{1}{y}-\dfrac{B_8}{B_7},
\end{equation}
we obtain
\begin{align}
 &r\left(x,y',B_1,B_2,B_3,B_4,0,0,B_7,B_8\right)\notag\\
 &=\dfrac{(B_2B_7-B_1B_8)xy+B_1B_7 x+(B_4B_7-B_3 B_8)y+B_3B_7}{{B_7}^2},
\end{align}
which is same form as Equation \eqref{eqn:proof_class_r_ii}.
Therefore, we can rewrite it in a similar manner to the case {\bf (ii)}.
\item[(iv) Case $B_6\neq0$ and $B_7=0$]
By the M\"obius transformation $x\to x'$, where
\begin{equation}
 x'=\dfrac{1}{x}-\dfrac{B_8}{B_6},
\end{equation}
we obtain
\begin{align}
 &r\left(x',y,B_1,B_2,B_3,B_4,0,B_6,0,B_8\right)\notag\\
 &=\dfrac{(B_3B_6-B_1B_8)xy+(B_4B_6-B_2 B_8)x+B_1B_6 y+B_2B_6}{{B_6}^2},
\end{align}
which is same form as Equation \eqref{eqn:proof_class_r_ii}.
Therefore, we can rewrite it in a similar manner to the case {\bf (ii)}.
\end{description}
Therefore, we have completed the proof.

\section{Proof of Lemma \ref{lemma:classification_Q123456789_1}}
\label{section:proof_lemma_CQ123456789_1}
In this section, we give the proof of Lemma \ref{lemma:classification_Q123456789_1}.

Let $\{Q_1,\dots,Q_9\}$ be the quad-equations on a cuboctahedron satisfying the properties \ref{cond:CO1}, \ref{cond:CO2} and \ref{cond:CO3} in Condition \ref{condition:CO}. 
Then, the following lemma holds.
\begin{lemma}\label{lemma:proof_F1234G1234}
The following properties hold:
\begin{description}
\item[(i)]
If $u_i$, $i=1,\dots,4$, and $v_1$ are initial values, then the following hold:
\begin{equation}\label{eqn:rel_F1234}
 \dfrac{\partial F_2}{\partial u_2}\dfrac{\partial F_1}{\partial u_5}
 +\dfrac{\partial F_3}{\partial u_2}\dfrac{\partial F_1}{\partial v_2}=0,
\end{equation}
where
\begin{align}
 &v_3=f_1(u_3,u_5,v_2)=:F_1,\quad
 u_5=f_2(u_4,u_2,u_1)=:F_2,\notag\\
 &v_2=f_3(v_1,u_2,u_4)=:F_3,
\end{align}
are solutions of $Q_3=0$, $Q_7=0$ and $Q_2=0$, respectively.
\item[(ii)]
If $u_i$, $i=1,4,5,6$, and $v_4$ are initial values, then the following hold:
\begin{equation}\label{eqn:rel_G1234}
 \dfrac{\partial G_2}{\partial u_6}\dfrac{\partial G_1}{\partial u_3}
 +\dfrac{\partial G_3}{\partial u_6}\dfrac{\partial G_1}{\partial v_3}=0,
\end{equation}
where
\begin{align}
 &v_2=g_1(u_3,u_5,v_3)=:G_1,\quad
 u_3=g_2(u_6,u_4,u_1)=:G_2,\notag\\
 &v_3=g_3(v_4,u_4,u_6)=:G_3,
\end{align}
are solutions of $Q_3=0$, $Q_9=0$ and $Q_6=0$, respectively.
\end{description}
\end{lemma}
\begin{proof}
Consider the case {\bf (i)}. (The argument is similar for the case {\bf (ii)}.) 
Using $Q_6=0$, $v_3$ can be expressed by the rational function in terms of $u_4$, $u_6$ and $v_4$.
Moreover, from $Q_9=0$, $u_6$ can be expressed by the rational function in terms of $u_1$, $u_3$ and $u_4$
and from the square equation $K_1(v_1,u_1,u_4,v_4)=0$, $v_4$ can be expressed by the rational function in terms of $u_1$, $u_4$ and $v_1$.
Therefore, we obtain $\partial v_3/\partial u_2=0$.
Hence, we obtain
\begin{equation}
 0=\dfrac{\partial v_3}{\partial u_2}
 =\dfrac{\partial F_1}{\partial u_2}
 =\dfrac{\partial u_5}{\partial u_2}\dfrac{\partial F_1}{\partial u_5}
 +\dfrac{\partial v_2}{\partial u_2}\dfrac{\partial F_1}{\partial v_2}
 =\dfrac{\partial F_2}{\partial u_2}\dfrac{\partial F_1}{\partial u_5}
 +\dfrac{\partial F_3}{\partial u_2}\dfrac{\partial F_1}{\partial v_2}.
\end{equation}
Therefore, we have completed the proof.
\end{proof}

\begin{remark}
We note that we can without loss of generality define $F_1$ as solving $Q_3=0$ by $r(v_3)$, where $r$ is a M\"obius transformation, instead of $v_3$.
It is also true for $F_4$, $G_1$ and $G_4$.
Hereinafter, we without saying define $F_1$, $F_4$, $G_1$ and $G_4$ with suitable M\"obius transformations.
For example, in \S\ref{subsection:Q123456789_1_case1}, 
$r_5(v_3)=F_1$ and $r_3(v_2)=G_1$ are solutions of $Q_3=0$,
and $r_6(v_5)=F_4$ and $r_2(v_6)=G_4$ are solutions of $Q_4=0$.
\end{remark}

In the following subsections, we divide the proof of Lemma \ref{lemma:classification_Q123456789_1} into four cases:
Case 1 is given by \eqref{eqn:P123_Mobius_1}, 
Case 2 is given by \eqref{eqn:P123_Mobius_2}, 
Case 3 is given by \eqref{eqn:P123_Mobius_3} and 
Case 4 is given by \eqref{eqn:P123_Mobius_4}.
In each case, we define some of $p_i$, $i=1,\dots,4$, satisfying
\begin{equation}
 p_1+p_2=0,\quad p_3+p_4=0,
\end{equation}
which respectively follow from Equations \eqref{eqn:rel_F1234} and \eqref{eqn:rel_G1234}.

As a result, we show that
Case 1 gives \eqref{eqn:P123_caseN3_1} and \eqref{eqn:P123_caseN3_3},
Case 4 gives \eqref{eqn:P123_caseN1}
and the others are inadmissible cases.

\subsection{Case 1: \eqref{eqn:P123_Mobius_1}}\label{subsection:Q123456789_1_case1}
Consider the quad-equations $\{Q_1,\dots,Q_9\}$ defined by the system of equations \eqref{eqns:Q123456789_CO3} with \eqref{eqn:P123_Mobius_1}.
Define $p_i$, $i=1,\dots,4$, as the following:
\begin{subequations}
\begin{align}
 &p_1=\frac{{\Gamma_1}^2}{A_{1234}^{(2)}r_2(u_3)}\dfrac{\partial F_2}{\partial u_2}\dfrac{\partial F_1}{\partial u_5},
 &&p_2=\frac{{\Gamma_1}^2}{A_{1234}^{(2)}r_2(u_3)}\dfrac{\partial F_3}{\partial u_2}\dfrac{\partial F_1}{\partial v_2},\\
 &p_3=\frac{{\Gamma_2}^2}{A_{1234}^{(2)}r_6(u_5)}\dfrac{\partial G_2}{\partial u_6}\dfrac{\partial G_1}{\partial u_3},
 &&p_4=\frac{{\Gamma_2}^2}{A_{1234}^{(2)}r_6(u_5)}\dfrac{\partial G_3}{\partial u_6}\dfrac{\partial G_1}{\partial v_3},
\end{align}
\end{subequations}
where $A_{1234}^{(2)}=A^{(2)}_1 A^{(2)}_2-A^{(2)}_3 A^{(2)}_4$ and
\begin{subequations}
\begin{align}
 \Gamma_1
 =&(d_{13} u_2+d_{14})
 (d_{23} u_2+d_{24})
 (d_{43} u_4+d_{44})
 (d_{53} u_4+d_{54})\notag\\
 &(d_{63}u_5+d_{64})
 (d_{33}v_2+d_{34})
 \Big(A^{(2)}_4 r_6(u_5)+A^{(2)}_2 r_3(v_2)\Big)\notag\\
 &\Big(d_{51} \Big(B^{(1)}_2 r_1(u_1)+B^{(1)}_3 r_4(u_4)\Big)+d_{53} r_2(u_2) \Big(B^{(1)}_4 r_1(u_1)+B^{(1)}_1 r_4(u_4)\Big)\Big)\notag\\
 &\Big(d_{41} \Big(A^{(1)}_3 r_5(u_4)+A^{(1)}_1 r_2(v_1)\Big)+d_{43} r_1(u_2) \Big(A^{(1)}_2 r_5(u_4)+A^{(1)}_4 r_2(v_1)\Big)\Big),\\
 \Gamma_2
 =&(d_{43} u_4+d_{44})
 (d_{33} u_4+d_{34})
 (d_{13} u_6+d_{14})
 (d_{63} u_6+d_{64})\notag\\
 &(d_{23} u_3+d_{24})
 (d_{53} v_3+d_{54})
 \Big(A^{(2)}_3 r_2(u_3)+A^{(2)}_2 r_5(v_3)\Big)\notag\\
 &\Big(d_{31} \Big(B^{(3)}_2 r_1(u_1)+B^{(3)}_4 r_4(u_4)\Big)+d_{33} r_6(u_6) \Big(B^{(3)}_3 r_1(u_1)+B^{(3)}_1 r_4(u_4)\Big)\Big)\notag\\
 &\Big(d_{41} \Big(A^{(3)}_4 r_3(u_4)+A^{(3)}_1 r_6(v_4)\Big)+d_{43} r_1(u_6) \Big(A^{(3)}_2 r_3(u_4)+A^{(3)}_3 r_6(v_4)\Big)\Big).
\end{align}
\end{subequations}
The relation $-p_1=p_2$ gives
\begin{align}\label{eqn:case1_p1p2}
 &(d_{22} d_{23}-d_{21} d_{24}) (d_{52} d_{53}-d_{51} d_{54}) (d_{62} d_{63}-d_{61} d_{64})(d_{14}+d_{13} u_2)^2\notag\\
 &\quad\Biggm((d_{32} d_{41}-d_{31} d_{42}) +(d_{32} d_{43}-d_{31} d_{44}) r_1(u_2) \dfrac{A^{(1)}_2 r_5(u_4)+A^{(1)}_4 r_2(v_1)}{A^{(1)}_3 r_5(u_4)+A^{(1)}_1 r_2(v_1)}\Biggm)\notag\\
 &\quad\Biggm((d_{34} d_{41}-d_{33} d_{42}) \dfrac{A^{(1)}_3 r_5(u_4)+A^{(1)}_1 r_2(v_1)}{A^{(1)}_2 r_5(u_4)+A^{(1)}_4 r_2(v_1)}+(d_{34} d_{43}-d_{33} d_{44}) r_1(u_2) \Biggm)\notag\\
 &=(d_{12} d_{13}-d_{11} d_{14}) (d_{32} d_{33}-d_{31} d_{34}) (d_{42} d_{43}-d_{41} d_{44})(d_{24}+d_{23} u_2)^2\notag\\
 &\quad\Biggm((d_{52} d_{61}-d_{51} d_{62}) +(d_{54} d_{61}-d_{53} d_{62}) r_2(u_2) \dfrac{B^{(1)}_1 r_4(u_4)+B^{(1)}_4 r_1(u_1)}{B^{(1)}_3 r_4(u_4)+B^{(1)}_2 r_1(u_1)}\Biggm)\notag\\
 &\quad\Biggm((d_{52} d_{63}-d_{51} d_{64}) \dfrac{B^{(1)}_3 r_4(u_4)+B^{(1)}_2 r_1(u_1)}{B^{(1)}_1 r_4(u_4)+B^{(1)}_4 r_1(u_1)}+(d_{54} d_{63}-d_{53} d_{64}) r_2(u_2) \Biggm).
\end{align}
Since the left-hand side of Equation \eqref{eqn:case1_p1p2} does not depend on $u_1$ and 
the right-hand side of Equation \eqref{eqn:case1_p1p2} does not depend on $v_1$,
we obtain 
\begin{equation}\label{eqn:case1_p1p2_1}
 d_{32} d_{41}-d_{31} d_{42}=d_{34} d_{43}-d_{33} d_{44}=0
 ~\,\text{or}~\,
 d_{32} d_{43}-d_{31} d_{44}=d_{34} d_{41}-d_{33} d_{42}=0,
\end{equation}
and
\begin{equation}\label{eqn:case1_p1p2_2}
 d_{52} d_{61}-d_{51} d_{62}=d_{54} d_{63}-d_{53} d_{64}=0
 ~\,\text{or}~\,
 d_{54} d_{61}-d_{53} d_{62}=d_{52} d_{63}-d_{51} d_{64}=0.
\end{equation}
In either case, from Equation \eqref{eqn:case1_p1p2} we obtain the condition that 
\begin{equation}
 \dfrac{(d_{11} u_2+d_{12})(d_{13} u_2+d_{14})}{(d_{21} u_2+d_{22})(d_{23} u_2+d_{24})}
\end{equation}
does not depend on $u_2$, which gives
\begin{equation}\label{eqn:case1_p1p2_3}
 d_{12} d_{21}-d_{11} d_{22}=d_{14} d_{23}-d_{13} d_{24}=0
 ~\,\text{or}~\,
 d_{12} d_{23}-d_{11} d_{24}=d_{14} d_{21}-d_{13} d_{22}=0.
\end{equation}
In a similar manner, from $-p_3=p_4$, we obtain the following conditions:
\begin{align}
 &d_{42} d_{51}-d_{41} d_{52}=d_{44} d_{53}-d_{43} d_{54}=0
 ~\,\text{or}~\,
 d_{44} d_{51}-d_{43} d_{52}=d_{42} d_{53}-d_{41} d_{54}=0,
 \label{eqn:case1_p3p4_1}\\
 &d_{22} d_{31}-d_{21} d_{32}=d_{24} d_{33}-d_{23} d_{34}=0
 ~\,\text{or}~\,
 d_{22} d_{33}-d_{21} d_{34}=d_{24} d_{31}-d_{23} d_{32}=0,
 \label{eqn:case1_p3p4_2}\\
\intertext{and}
 &d_{12} d_{61}-d_{11} d_{62}=d_{14} d_{63}-d_{13} d_{64}=0
 ~\,\text{or}~\,
 d_{12} d_{63}-d_{11} d_{64}=d_{14} d_{61}-d_{13} d_{62}=0.
 \label{eqn:case1_p3p4_3}
\end{align}
The conditions \eqref{eqn:case1_p1p2_1}, \eqref{eqn:case1_p1p2_2}, \eqref{eqn:case1_p1p2_3}--\eqref{eqn:case1_p3p4_3}
are equivalent to the conditions that each $r_i$, $i=2,\dots,6$, can be expressed by
\begin{equation}
 r_i(x)=\text{const.}\times r_1(x)
 \quad\text{or}\quad
 r_i(x)=\text{const.}\times \frac{1}{r_1(x)}.
\end{equation}
Using the M\"obius transformations of the variables $u_i$, $i=1,\dots,6$, and $v_j$, $j=1,\dots,6$, given by
\begin{equation}
 u_i\to r_1^{-1}(u_i)=-\dfrac{d_{12}-d_{14}u_i}{d_{11}-d_{13}u_i},\quad
 v_j\to r_1^{-1}(v_j)=-\dfrac{d_{12}-d_{14}v_j}{d_{11}-d_{13}v_j},
\end{equation}
and gauge transformations of the parameters $A_i^{(j)}$ and $B_i^{(j)}$,
we can without loss of generality let $r_1(x)=x$ and 
each $r_i$, $i=2,\dots,6$, be
\begin{equation}
 r_i(x)=x
 \quad\text{or}\quad
 r_i(x)=x^{-1}.
\end{equation}
Moreover, because of the following relations:
\begin{subequations}
\begin{align}
 &\polyN3\Big(x^{-1},y,z,w;\al_1,\al_2,\al_3,\al_4\Big)
 =\dfrac{z}{x}\, \polyN3\Big(x,y,z^{-1},w;\al_4,\al_3,\al_2,\al_1\Big),\\
 &\polyN3\Big(x,y^{-1},z,w;\al_1,\al_2,\al_3,\al_4\Big)
 =\dfrac{w}{y}\, \polyN3\Big(x,y,z,w^{-1};\al_3,\al_4,\al_1,\al_2\Big),
\end{align}
\end{subequations}
we can without loss of generality let \eqref{eqn:P123_Mobius_1} be
\begin{equation}
 \begin{cases}
 P_1(x,y,z,w)=\polyN3\Big(x, y, r_4(z), r_{25}(w);A_1,A_2,A_3,A_4\Big),\\
 P_2(x,y,z,w)=\polyN3\Big(x, y, r_{25}(z), r_{63}(w);A_1,A_2,A_3,A_4\Big),\\
 P_3(x,y,z,w)=\polyN3\Big(x, y, r_{63}(z), r_4(w);A_1,A_2,A_3,A_4\Big),
 \end{cases}
\end{equation}
where $r_{25}=r_2\circ r_5$ and $r_{63}=r_6\circ r_3$.
Furthermore, because of the symmetry of the cuboctahedron $G^{(co)}$, 
it is sufficient to consider \eqref{eqn:P123_caseN3_1}, \eqref{eqn:P123_caseN3_3} and
\begin{align}
 &\begin{cases}
 P_1=\polyN3\Big(x, y, z, w^{-1};A_1,A_2,A_3,A_4\Big),\\
 P_2=\polyN3\Big(x, y, z^{-1}, w;A_1,A_2,A_3,A_4\Big),\\
 P_3=\polyN3\Big(x, y, z, w;A_1,A_2,A_3,A_4\Big),
 \end{cases}\label{eqn:P123_caseN3_2}\\
 &\begin{cases}
 P_1=\polyN3\Big(x, y, z^{-1}, w^{-1};A_1,A_2,A_3,A_4\Big),\\
 P_2=\polyN3\Big(x, y, z^{-1}, w^{-1};A_1,A_2,A_3,A_4\Big),\\
 P_3=\polyN3\Big(x, y, z^{-1}, w^{-1};A_1,A_2,A_3,A_4\Big).
 \end{cases}\label{eqn:P123_caseN3_4}
\end{align}
We can easily verify that the cases \eqref{eqn:P123_caseN3_2} and \eqref{eqn:P123_caseN3_4} are inadequate.
Indeed, in the case \eqref{eqn:P123_caseN3_2}, we obtain
\begin{align}
 &\dfrac{\partial F_2}{\partial u_2}\dfrac{\partial F_1}{\partial u_5}
 +\dfrac{\partial F_3}{\partial u_2}\dfrac{\partial F_1}{\partial v_2}\notag\\
 &=\dfrac{2u_2(A^{(2)}_3A^{(2)}_4-A^{(2)}_1A^{(2)}_2)(B^{(1)}_4u_1+B^{(1)}_1u_4)(B^{(1)}_2u_1+B^{(1)}_3u_4)(A^{(1)}_3+A^{(1)}_1u_4v_1)(A^{(1)}_2+A^{(1)}_4u_4v_1)}
 {u_3\left(A^{(2)}_3{u_2}^2(A^{(1)}_2+A^{(1)}_4u_4v_1)(B^{(1)}_4u_1+B^{(1)}_1u_4)+A^{(2)}_1(A^{(1)}_3+A^{(1)}_1u_4v_1)(B^{(1)}_2u_1+B^{(1)}_3u_4)\right)^2}\notag\\
 &\neq0,
\end{align}
and in the case \eqref{eqn:P123_caseN3_4}, we obtain
\begin{align}
 &\dfrac{\partial F_2}{\partial u_2}\dfrac{\partial F_1}{\partial u_5}
 +\dfrac{\partial F_3}{\partial u_2}\dfrac{\partial F_1}{\partial v_2}\notag\\
 &=\dfrac{2u_2(A^{(2)}_3A^{(2)}_4-A^{(2)}_1A^{(2)}_2)(B^{(1)}_4+B^{(1)}_1u_1u_4)(B^{(1)}_2+B^{(1)}_3u_1u_4)(A^{(1)}_3+A^{(1)}_1u_4v_1)(A^{(1)}_2+A^{(1)}_4u_4v_1)}
 {u_3\left(A^{(2)}_3{u_2}^2(A^{(1)}_3+A^{(1)}_1u_4v_1)(B^{(1)}_4+B^{(1)}_1u_1u_4)+
A^{(2)}_1(A^{(1)}_2+A^{(1)}_4u_4v_1)(B^{(1)}_2+B^{(1)}_3u_1u_4)\right)^2}\notag\\
 &\neq0.
\end{align}
Therefore, from Case 1 we obtain \eqref{eqn:P123_caseN3_1} and \eqref{eqn:P123_caseN3_3}.
\subsection{Case 2: \eqref{eqn:P123_Mobius_2}}\label{subsection:Q123456789_1_case2}
Consider the quad-equations $\{Q_1,\dots,Q_9\}$ defined by the system of equations \eqref{eqns:Q123456789_CO3} with \eqref{eqn:P123_Mobius_2}.
Define $p_i$, $i=3,4$, as the following:
\begin{equation}
p_3=\frac{{\Gamma_2}^2}{A_{2314}^{(2)}r_6(u_5)}\dfrac{\partial G_2}{\partial u_6}\dfrac{\partial G_1}{\partial u_3},\quad
 p_4=\frac{{\Gamma_2}^2}{A_{2314}^{(2)}r_6(u_5)}\dfrac{\partial G_3}{\partial u_6}\dfrac{\partial G_1}{\partial v_3},
\end{equation}
where $A_{2314}^{(2)}=A^{(2)}_2 A^{(2)}_3-A^{(2)}_1 A^{(2)}_4$ and
\begin{align}
 \Gamma_2
 =&(d_{33} u_4+d_{34}) 
 (d_{43} u_4+d_{44}) 
 (d_{13} u_6+d_{14}) 
 (d_{63} u_6+d_{64})\notag\\
 &(d_{23} u_3+d_{24}) 
 (d_{53} v_3+d_{54}) 
 \Big(A^{(2)}_4+A^{(2)}_2 r_2(u_3)+A^{(2)}_2 r_5(v_3)\Big)\notag\\
 &\Big(d_{41} \Big(A^{(3)}_4 r_3(u_4)+A^{(3)}_1 r_6(v_4)\Big)+d_{43} r_1(u_6) \Big(A^{(3)}_2 r_3(u_4)+A^{(3)}_3 r_6(v_4)\Big)\Big)\notag\\
 &\Big(d_{31} \Big(B^{(3)}_2 r_1(u_1)+B^{(3)}_4 r_4(u_4)\Big)+d_{33} r_6(u_6) \Big(B^{(3)}_3 r_1(u_1)+B^{(3)}_1 r_4(u_4)\Big)\Big).
\end{align}
The relation $p_3=-p_4$ gives
\begin{align}\label{eqn:case2_p3p4}
 &(d_{13} u_6+d_{14})^2
 (d_{22} d_{23}-d_{21} d_{24}) (d_{32} d_{33}-d_{31} d_{34}) (d_{62} d_{63}-d_{61} d_{64})\notag\\
 &\quad\Biggm(d_{42} d_{53}-d_{41} d_{54}
 +(d_{44} d_{53}-d_{43} d_{54}) r_1(u_6)\dfrac{A^{(3)}_2 r_3(u_4)+A^{(3)}_3 r_6(v_4)}{A^{(3)}_4 r_3(u_4)+A^{(3)}_1 r_6(v_4)}\Biggm)\notag\\
 &\quad\Biggm((d_{42} d_{53}-d_{41} d_{54})\dfrac{A^{(3)}_4 r_3(u_4)+A^{(3)}_1 r_6(v_4)}{A^{(3)}_2 r_3(u_4)+A^{(3)}_3 r_6(v_4)}
 +(d_{44} d_{53}-d_{43} d_{54}) r_1(u_6)\Biggm)\notag\\ 
 &=-(d_{63} u_6+d_{64})^2
 (d_{12} d_{13}-d_{11} d_{14}) (d_{42} d_{43}-d_{41} d_{44}) (d_{52} d_{53}-d_{51} d_{54})\notag\\ 
 &\quad\Biggm(d_{24} d_{31}-d_{23} d_{32} 
 +(d_{24} d_{33}-d_{23} d_{34}) r_6(u_6) \dfrac{B^{(3)}_1 r_4(u_4)+B^{(3)}_3 r_1(u_1)}{B^{(3)}_4 r_4(u_4)+B^{(3)}_2 r_1(u_1)}\Biggm)\notag\\
 &\quad\Biggm((d_{24} d_{31}-d_{23} d_{32}) \dfrac{B^{(3)}_4 r_4(u_4)+B^{(3)}_2 r_1(u_1)}{B^{(3)}_1 r_4(u_4)+B^{(3)}_3 r_1(u_1)}
 +(d_{24} d_{33}-d_{23} d_{34}) r_6(u_6)\Biggm).
\end{align}
Since the left-hand side of Equation \eqref{eqn:case2_p3p4} does not depend on $u_1$,
we obtain 
\begin{equation}
 d_{24} d_{31}-d_{23} d_{32}=d_{24} d_{33}-d_{23} d_{34}=0,
\end{equation}
which gives the inappropriate condition for the M\"obius transformation $r_3$, that is, $d_{31}d_{34}=d_{32}d_{33}$.
Therefore, this case is inadequate.
\subsection{Case 3: \eqref{eqn:P123_Mobius_3}}\label{subsection:Q123456789_1_case3}
Consider the quad-equations $\{Q_1,\dots,Q_9\}$ defined by the system of equations \eqref{eqns:Q123456789_CO3} with \eqref{eqn:P123_Mobius_3}.
Define $p_i$, $i=1,2$, as the following:
\begin{equation}
 p_1=\frac{{\Gamma_1}^2}{A^{(2)}_{2314}r_2(u_3)}\dfrac{\partial F_2}{\partial u_2}\dfrac{\partial F_1}{\partial u_5},\quad
 p_2=\frac{{\Gamma_1}^2}{A^{(2)}_{2314}r_2(u_3)}\dfrac{\partial F_3}{\partial u_2}\dfrac{\partial F_1}{\partial v_2},
\end{equation}
where $A^{(2)}_{2314}=A^{(2)}_2 A^{(2)}_3-A^{(2)}_1 A^{(2)}_4$ and
\begin{align*}
 \Gamma_1
 =&(d_{13} u_2+d_{14}) 
 (d_{23} u_2+d_{24}) 
 (d_{43} u_4+d_{44}) 
 (d_{53} u_4+d_{54})\notag\\
 &(d_{63} u_5+d_{64}) 
 (d_{33} v_2+d_{34}) 
 \Big(A^{(2)}_4+A^{(2)}_2 \Big(r_3(v_2)+r_6(u_5)\Big)\Big)\notag\\
 &\biggm(\Big(A^{(1)}_2 d_{41}+A^{(1)}_4 d_{43}+A^{(1)}_2 d_{43} r_1(u_2)\Big) r_5(u_4)
 +\Big(A^{(1)}_1 d_{41}+A^{(1)}_3 d_{43}+A^{(1)}_1 d_{43} r_1(u_2)\Big) r_2(v_1) \biggm)\notag\\
 &\biggm(d_{53} \Big(B^{(1)}_3+B^{(1)}_1 r_1(u_1)+B^{(1)}_1 r_4(u_4)\Big) r_2(u_2)
 +d_{51} \Big(B^{(1)}_4+B^{(1)}_2 r_1(u_1)+B^{(1)}_2 r_4(u_4)\Big)\biggm).
\end{align*}
In the following, dividing the conditions of the parameters $A^{(1)}_1$ and $B^{(1)}_1$ into two cases, we show that Case 3 is inadmissible case.
\subsubsection{\rm\bf Case: $A^{(1)}_1,B^{(1)}_1\neq0$.}
The relation $-p_1=p_2$ gives
\begin{align}\label{eqn:case31_p1p2}
 &-(A^{(1)}_1)^{-2}(d_{13} u_2+d_{14})^2
 (d_{22} d_{23}-d_{21} d_{24})
 (d_{52} d_{53}-d_{51} d_{54})
 (d_{62} d_{63}-d_{61} d_{64})\notag\\
 &~\Biggm((d_{34} d_{43}-d_{33} d_{44})
 \Biggm(\dfrac{\Big(A^{(1)}_2 A^{(1)}_3-A^{(1)}_1 A^{(1)}_4\Big) r_5(u_4)}{A^{(1)}_2 r_5(u_4)+A^{(1)}_1 r_2(v_1)}
 -A^{(1)}_3-A^{(1)}_1 r_1(u_2)\Biggm)\notag\\
 &\quad-A^{(1)}_1 (d_{34} d_{41}-d_{33} d_{42})\Biggm)^2\notag\\
 &=(d_{23} u_2+d_{24})^2
 (d_{12} d_{13}-d_{11} d_{14})
 (d_{32} d_{33}-d_{31} d_{34})
 (d_{42} d_{43}-d_{41} d_{44})\notag\\
 &~\Biggm((d_{54} d_{63}-d_{53} d_{64}) r_2(u_2) 
 +(d_{52} d_{63}-d_{51} d_{64}) \dfrac{B^{(1)}_4+B^{(1)}_2 \Big(r_1(u_1)+r_4(u_4)\Big)}{B^{(1)}_3+B^{(1)}_1 \Big(r_1(u_1)+r_4(u_4)\Big)}\Biggm)\notag\\
 &~\Biggm((d_{54} d_{63}-d_{53} d_{64}) r_2(u_2) \dfrac{B^{(1)}_3+B^{(1)}_1 \Big(r_1(u_1)+r_4(u_4)\Big)}{B^{(1)}_4+B^{(1)}_2 \Big(r_1(u_1)+r_4(u_4)\Big)}
 +d_{52} d_{63}-d_{51} d_{64}\Biggm).
\end{align}
Since the left-hand side of Equation \eqref{eqn:case31_p1p2} does not depend on $u_1$,
we obtain 
\begin{equation}\label{eqn:case31_d56}
 d_{54} d_{63}-d_{53} d_{64}=d_{52} d_{63}-d_{51} d_{64}=0,
\end{equation}
which gives the inappropriate condition for the M\"obius transformation $r_5$, that is, $d_{51}d_{54}=d_{52}d_{53}$.
Therefore, this case is inadequate.
\subsubsection{\rm\bf Case: $A^{(1)}_1=B^{(1)}_1=0$.}
Note here that since $Q_1$ and $Q_7$ are quad-equations, we have
\begin{equation}
 A^{(1)}_2,A^{(1)}_3,B^{(1)}_2,B^{(1)}_3\neq0.
\end{equation}
The relation $-p_1=p_2$ gives
\begin{align}\label{eqn:case32_p1p2}
 &-B^{(1)}_3(d_{22} d_{23}-d_{21} d_{24}) (d_{52} d_{53}-d_{51} d_{54}) (d_{62} d_{63}-d_{61} d_{64})
 (d_{14}+d_{13} u_2)^2\notag\\
 &\Biggm(A^{(1)}_2 (d_{34} d_{41}-d_{33} d_{42}) r_5(u_4)\notag\\
 &\quad+(d_{34} d_{43}-d_{33} d_{44}) \Big(A^{(1)}_4 r_5(u_4)+A^{(1)}_3 r_2(v_1)+A^{(1)}_2 r_1(u_2) r_5(u_4)\Big)\Biggm)\notag\\
 &=(A^{(1)}_2)^2 (d_{12} d_{13}-d_{11} d_{14}) (d_{32} d_{33}-d_{31} d_{34}) (d_{42} d_{43}-d_{41} d_{44})
 (d_{24}+d_{23} u_2)^2 
 r_5(u_4)^2 \notag\\
 &\Biggm((d_{54} d_{63}-d_{53} d_{64})\dfrac{B^{(1)}_3 r_2(u_2)}{B^{(1)}_4+B^{(1)}_2 r_4(u_4)+B^{(1)}_2 r_1(u_1)}+d_{52} d_{63}-d_{51} d_{64}\Biggm)\notag\\
 &\Biggm(B^{(1)}_3 (d_{54} d_{63}-d_{53} d_{64}) r_2(u_2)\notag\\
 &\quad+(d_{52} d_{63}-d_{51} d_{64}) \Big(B^{(1)}_4+B^{(1)}_2 r_4(u_4)+B^{(1)}_2 r_1(u_1)\Big)\Biggm).
\end{align}
Since the left-hand side of Equation \eqref{eqn:case32_p1p2} does not depend on $u_1$,
we obtain the inappropriate condition \eqref{eqn:case31_d56}.
Therefore, this case is inadequate.
\subsection{Case 4: \eqref{eqn:P123_Mobius_4}}\label{subsection:Q123456789_1_case4}
Consider the quad-equations $\{Q_1,\dots,Q_9\}$ defined by the system of equations \eqref{eqns:Q123456789_CO3} with \eqref{eqn:P123_Mobius_4}.
Define $p_i$, $i=1,\dots,4$, as the following:
\begin{align*}
 p_1=\frac{{\Gamma_1}^2}{A^{(2)}_{1423}}\dfrac{\partial F_2}{\partial u_2}\dfrac{\partial F_1}{\partial u_5},~
 p_2=\frac{{\Gamma_1}^2}{A^{(2)}_{1423}}\dfrac{\partial F_3}{\partial u_2}\dfrac{\partial F_1}{\partial v_2},~
 p_3=\frac{{\Gamma_2}^2}{A^{(2)}_{1423}}\dfrac{\partial G_2}{\partial u_6}\dfrac{\partial G_1}{\partial u_3},~
 p_4=\frac{{\Gamma_2}^2}{A^{(2)}_{1423}}\dfrac{\partial G_3}{\partial u_6}\dfrac{\partial G_1}{\partial v_3},
\end{align*}
where $A^{(2)}_{1423}=A^{(2)}_1 A^{(2)}_4-A^{(2)}_2 A^{(2)}_3$ and
\begin{subequations}
\begin{align}
 \Gamma_1
 =&(d_{13} u_2+d_{14})
 (d_{23} u_2+d_{24}) 
 (d_{43} u_4+d_{44}) 
 (d_{53} u_4+d_{54})\notag\\
 &(d_{63} u_5+d_{64}) 
 (d_{33} v_2+d_{34}) 
 \Big(A^{(2)}_2+A^{(2)}_1 \Big(r_6(u_5)+r_3(v_2)\Big)\Big)\notag\\
 &\Biggm(A^{(1)}_2 d_{41}+A^{(1)}_4 d_{43}+A^{(1)}_2 d_{43} r_1(u_2)\notag\\
 &\quad+\Big(A^{(1)}_1 d_{41}+A^{(1)}_3 d_{43}+A^{(1)}_1 d_{43} r_1(u_2)\Big)\Big(r_2(v_1)+r_5(u_4)\Big)\Biggm)\notag\\
 &\Biggm(B^{(1)}_3 d_{51}+B^{(1)}_4 d_{53}+B^{(1)}_3 d_{53} r_2(u_2)\notag\\
 &\quad+\Big(B^{(1)}_1 d_{51}+B^{(1)}_2 d_{53}+B^{(1)}_1 d_{53} r_2(u_2)\Big)\Big(r_1(u_1)+r_4(u_4)\Big)\Biggm),\\
 \Gamma_2
 =&(d_{33} u_4+d_{34}) 
 (d_{43} u_4+d_{44}) 
 (d_{13} u_6+d_{14}) 
 (d_{63} u_6+d_{64})\notag\\
 &(d_{23} u_3+d_{24}) 
 (d_{53} v_3+d_{54}) 
 \Big(A^{(2)}_3+A^{(2)}_1 \Big(r_2(u_3)+r_5(v_3)\Big)\Big)\notag\\
 &\Biggm(A^{(3)}_3 d_{41}+A^{(3)}_4 d_{43}+A^{(3)}_3 d_{43} r_1(u_6)\notag\\
 &\quad+\Big(A^{(3)}_1 d_{41}+A^{(3)}_2 d_{43}+A^{(3)}_1 d_{43} r_1(u_6)\Big)\Big(r_3(u_4)+r_6(v_4)\Big)\Biggm)\notag\\
 &\Biggm(B^{(3)}_2 d_{31}+B^{(3)}_4 d_{33}+B^{(3)}_2 d_{33} r_6(u_6)\notag\\
 &\quad+\Big(B^{(3)}_1 d_{31}+B^{(3)}_3 d_{33}+B^{(3)}_1 d_{33} r_6(u_6)\Big)\Big(r_1(u_1)+r_4(u_4)\Big)\Biggm).
\end{align}
\end{subequations}
Note here that since $Q_1$ , $Q_5$, $Q_7$ and $Q_9$ are quad-equations, 
if $A^{(1)}_1=B^{(1)}_1=0$, then we obtain
\begin{equation}
 A^{(1)}_2,A^{(1)}_3,B^{(1)}_2,B^{(1)}_3\neq0.
\end{equation}
The relation $p_1=-p_2$ gives
\begin{align}\label{eqn:case4_p1p2_1}
 &(A^{(1)}_1)^{-2}(d_{13} u_2+d_{14})^2
 (d_{22} d_{23}-d_{21} d_{24}) 
 (d_{52} d_{53}-d_{51} d_{54}) 
 (d_{62} d_{63}-d_{61} d_{64})\notag\\
 &\Biggm((d_{34} d_{43}-d_{33} d_{44})
 \Biggm(\dfrac{A^{(1)}_2 A^{(1)}_3-A^{(1)}_1 A^{(1)}_4}{A^{(1)}_2+A^{(1)}_1 \Big(r_5(u_4)+r_2(v_1)\Big)}
 -A^{(1)}_3-A^{(1)}_1 r_1(u_2)\Biggm)\notag\\
 &\quad-A^{(1)}_1 (d_{34} d_{41}-d_{33} d_{42})\Biggm)^2\notag\\
 &=-(B^{(1)}_1)^{-2}(d_{23} u_2+d_{24})^2
 (d_{12} d_{13}-d_{11} d_{14})
 (d_{32} d_{33}-d_{31} d_{34})
 (d_{42} d_{43}-d_{41} d_{44})\notag\\
 &\Biggm((d_{54} d_{63}-d_{53} d_{64})
 \Biggm(\dfrac{B^{(1)}_2 B^{(1)}_3-B^{(1)}_1 B^{(1)}_4}{B^{(1)}_3+B^{(1)}_1 \Big(r_1(u_1)+r_4(u_4)\Big)}
 -B^{(1)}_2-B^{(1)}_1 r_2(u_2)\Biggm)\notag\\
 &\quad-B^{(1)}_1 (d_{52} d_{63}-d_{51} d_{64})\Biggm)^2,
\end{align}
if $A^{(1)}_1,B^{(1)}_1\neq0$, and
\begin{align}\label{eqn:case4_p1p2_2}
 &(B^{(1)}_3)^2 (d_{13} u_2+d_{14})^2
 (d_{22} d_{23}-d_{21} d_{24}) (d_{52} d_{53}-d_{51} d_{54}) (d_{62} d_{63}-d_{61} d_{64})\notag\\
 &\Biggm(A^{(1)}_2 (d_{34} d_{41}-d_{33} d_{42})\notag\\
 &\quad+(d_{34} d_{43}-d_{33} d_{44}) \Big(A^{(1)}_4+A^{(1)}_3 r_2(v_1)+A^{(1)}_2 r_1(u_2)+A^{(1)}_3 r_5(u_4)\Big)\Biggm)^2\notag\\
 &=-(A^{(1)}_2)^2 (d_{23} u_2+d_{24})^2
 (d_{12} d_{13}-d_{11} d_{14}) (d_{32} d_{33}-d_{31} d_{34}) (d_{42} d_{43}-d_{41} d_{44})\notag\\
 &\Biggm(B^{(1)}_3 (d_{52} d_{63}-d_{51} d_{64})\notag\\
 &\quad+(d_{54} d_{63}-d_{53} d_{64}) \Big(B^{(1)}_4+B^{(1)}_2 r_1(u_1)+B^{(1)}_3 r_2(u_2)+B^{(1)}_2 r_4(u_4)\Big)\Biggm)^2,
\end{align}
if $A^{(1)}_1=B^{(1)}_1=0$.
Since the left-hand sides of Equations \eqref{eqn:case4_p1p2_1} and \eqref{eqn:case4_p1p2_2} do not depend on $u_1$ and 
the right-hand sides of Equations \eqref{eqn:case4_p1p2_1} and \eqref{eqn:case4_p1p2_2} do not depend on $v_1$,
we obtain 
\begin{equation}\label{eqn:case4_condition_d_1}
 d_{34} d_{43}-d_{33} d_{44}=d_{54} d_{63}-d_{53} d_{64}=0.
\end{equation}
Moreover, using the condition \eqref{eqn:case4_condition_d_1}, from Equations \eqref{eqn:case4_p1p2_1} and \eqref{eqn:case4_p1p2_2} we obtain the condition that 
\begin{equation}
 \dfrac{d_{13} u_2+d_{14}}{d_{23} u_2+d_{24}}
\end{equation}
does not depend on $u_2$, which gives
\begin{equation}\label{eqn:case4_condition_d_2}
 d_{13} d_{24}-d_{14} d_{23}=0.
\end{equation}
In a similar manner, from $-p_3=p_4$, we obtain the following conditions:
\begin{equation}\label{eqn:case4_condition_d_3}
 d_{44} d_{53}-d_{43} d_{54}=d_{24} d_{33}-d_{23} d_{34}=d_{13} d_{64}-d_{14} d_{63}=0.
\end{equation}

The conditions \eqref{eqn:case4_condition_d_1}, \eqref{eqn:case4_condition_d_2} and \eqref{eqn:case4_condition_d_3}
are equivalent to the conditions that each $r_i$, $i=2,\dots,6$, can be expressed by
\begin{equation}\label{eqn:case4_condition_d_35}
 r_i(x)=\frac{\text{const.}}{d_{13}x+d_{14}}\,(d_{i1}x+d_{i2}).
\end{equation}
This leads that we can without loss of generality let 
\begin{equation}
 r_i(x)=d_{i1}x+d_{i2},\quad i=1,\dots,6.
\end{equation}
Indeed, 
if $d_{13}=0$, then it is obvious from \eqref{eqn:case4_condition_d_35} and
if $d_{13}\neq0$, then the M\"obius transformations of the variables $u_i$, $i=1,\dots,6$, and $v_j$, $j=1,\dots,6$, given by
\begin{equation}
 u_i\to \frac{1}{u_i}-\dfrac{d_{14}}{d_{13}},\quad
 v_j\to \frac{1}{v_i}-\dfrac{d_{14}}{d_{13}},
\end{equation}
causes to change from \eqref{eqn:case4_condition_d_35} to
\begin{equation}
 r_i(x)=\text{const.}\times \Big((d_{i2}d_{13}-d_{i1}d_{14})x+d_{i1}d_{13}\Big).
\end{equation}
Therefore, for the system \eqref{eqn:P123_Mobius_4} we can without loss of generality let 
\begin{equation}\label{eqn:case4_condition_d_4}
 d_{13}=\cdots=d_{63}=0,\quad d_{14}=\cdots=d_{64}=1.
\end{equation}

Moreover, under the condition \eqref{eqn:case4_condition_d_4} 
the system of equations \eqref{eqn:P123_Mobius_4} can be rewritten as the following:
\begin{equation}
 \begin{cases}
 P_1(x,y,z,w)=\polyN1\Biggm(x, y, \dfrac{d_{11}}{d_{41}}z, \dfrac{d_{51}}{d_{21}}w;A_{11},A_{12},A_{13},A_{14}\Biggm),\\
 P_2(x,y,z,w)=\polyN1\Biggm(x, y, \dfrac{d_{51}}{d_{21}}z, \dfrac{d_{31}}{d_{61}}w;A_{21},A_{22},A_{23},A_{24}\Biggm),\\
 P_3(x,y,z,w)=\polyN1\Biggm(x, y, \dfrac{d_{31}}{d_{61}}z, \dfrac{d_{11}}{d_{41}}w;A_{31},A_{32},A_{33},A_{34}\Biggm),
 \end{cases}
\end{equation}
where
\begin{subequations}
\begin{align}
 &A_{11}=A_1 d_{21} d_{41},\quad
 A_{12}=d_{41} \Big(A_2+A_1 (d_{22}+d_{52})\Big),\\
 &A_{13}=d_{21} \Big(A_3+A_1 (d_{12}+d_{42})\Big),\\
 &A_{14}=A_4+A_2 (d_{12}+d_{42})+\Big(A_3+A_1 (d_{12}+d_{42})\Big) (d_{22}+d_{52}),\\
 &A_{21}=A_1 d_{61} d_{21},\quad
 A_{22}= d_{21} \Big(A_2+A_1 (d_{62}+d_{32})\Big),\\
 &A_{23}=d_{61} \Big(A_3+A_1 (d_{52}+d_{22})\Big),\\
 &A_{24}=A_4+A_2 (d_{52}+d_{22})+\Big(A_3+A_1 (d_{52}+d_{22})\Big) (d_{62}+d_{32}),\\
 &A_{31}=A_1 d_{41} d_{61},\quad
 A_{32}= d_{61} \Big(A_2+A_1 (d_{42}+d_{12})\Big),\\
 &A_{33}=d_{41} \Big(A_3+A_1 (d_{32}+d_{62})\Big),\\
 &A_{34}=A_4+A_2 (d_{32}+d_{62})+\Big(A_3+A_1 (d_{32}+d_{62})\Big) (d_{42}+d_{12}).
\end{align}
\end{subequations}
Therefore, for the system \eqref{eqn:P123_Mobius_4} with \eqref{eqn:case4_condition_d_4} we can without loss of generality let 
\begin{equation}\label{eqn:case4_condition_d_5}
 d_{12}=\cdots=d_{62}=0,\quad d_{21}=d_{41}=d_{61}=1.
\end{equation}

Under the conditions \eqref{eqn:case4_condition_d_4} and \eqref{eqn:case4_condition_d_5},
$p_i$, $i=1,2,3,4$, are given as the following:
\begin{subequations}
\begin{align}
 &p_1=-d_{51} 
 \Big(B^{(1)}_3+B^{(1)}_1 (d_{11} u_1+u_4)\Big)^2 
 \Big(A^{(1)}_2+A^{(1)}_1 (d_{51} u_4+v_1)\Big)^2,\\
 &p_2=-d_{11} d_{31} {d_{51}}^2 
 \Big(B^{(1)}_3+B^{(1)}_1 (d_{11} u_1+u_4)\Big)^2 
 \Big(A^{(1)}_2+A^{(1)}_1 (d_{51} u_4+v_1)\Big)^2,\\
 &p_3=-d_{31} 
 \Big(B^{(3)}_2+B^{(3)}_1 (d_{11} u_1+u_4)\Big)^2
 \Big(A^{(3)}_3+A^{(3)}_1 (d_{31} u_4+v_4)\Big)^2,\\
 &p_4=-d_{11} {d_{31}}^2 d_{51} 
 \Big(B^{(3)}_2+B^{(3)}_1 (d_{11} u_1+u_4)\Big)^2 
 \Big(A^{(3)}_3+A^{(3)}_1 (d_{31} u_4+v_4)\Big)^2.
\end{align}
\end{subequations}
Then, $p_1+p_2=p_3+p_4=0$ gives the condition $d_{11} d_{31} d_{51}=-1$.
Therefore, we obtain \eqref{eqn:P123_caseN1}.

\section{Conditions in Lemmas \ref{lemma:classification_CACO_CO1CO2CO3_1}--\ref{lemma:classification_CACO_CO1CO2CO3_3}}
\label{section:list_theorem_conditions}
This appendix collects and presents all the explicit conditions needed for the statements of Lemmas \ref{lemma:classification_CACO_CO1CO2CO3_1}--\ref{lemma:classification_CACO_CO1CO2CO3_3}. 
(Note that any terms appearing in the denominators of the equations in this appendix are assumed to be non-zero.)

\subsection{For Lemma \ref{lemma:classification_CACO_CO1CO2CO3_1}}
Lemma \ref{lemma:classification_CACO_CO1CO2CO3_1} describes system of quad-equations of Type I, which has two sub-cases called Type I-1 and Type I-2. 
The first case is given by\\
{\rm(Type I-1):}
\begin{equation}\tag{I-1}\label{eqn:typeI_cond1}
\begin{cases}
 (A^{(1)}_3A^{(2)}_1B^{(1)}_1+A^{(1)}_2A^{(2)}_3B^{(1)}_3)A^{(3)}_4B^{(3)}_1
 =-(A^{(1)}_3A^{(2)}_4B^{(1)}_1+A^{(1)}_2A^{(2)}_2B^{(1)}_3)A^{(3)}_2B^{(3)}_4,\\
 \Big((A^{(1)}_2A^{(2)}_3B^{(1)}_2+A^{(1)}_3A^{(2)}_1B^{(1)}_4)A^{(3)}_4
 -(A^{(1)}_1A^{(2)}_1B^{(1)}_1+A^{(1)}_4A^{(2)}_3B^{(1)}_3)A^{(3)}_1\Big)B^{(3)}_1\\
 \hspace{0.7em}+(A^{(1)}_3A^{(2)}_4B^{(1)}_1+A^{(1)}_2A^{(2)}_2B^{(1)}_3)A^{(3)}_2B^{(3)}_2
 +(A^{(1)}_3A^{(2)}_1B^{(1)}_1+A^{(1)}_2A^{(2)}_3B^{(1)}_3)A^{(3)}_4B^{(3)}_3\\
 \hspace{0.7em}+\Big((A^{(1)}_2A^{(2)}_2B^{(1)}_2+A^{(1)}_3A^{(2)}_4B^{(1)}_4)A^{(3)}_2-(A^{(1)}_1A^{(2)}_4B^{(1)}_1+A^{(1)}_4A^{(2)}_2B^{(1)}_3)A^{(3)}_3\Big)B^{(3)}_4=0,
 \end{cases}
 \end{equation}
and the images of these equations under the symmetry $1\leftrightarrow 2$ and $3\leftrightarrow 4$ applied to the subscripts.
The second case is given by\\
{\rm(Type I-2):}
\begin{equation}\tag{I-2}\label{eqn:typeI_cond2}
\begin{cases}
 A^{(1)}_3=B^{(1)}_3=A^{(2)}_3=B^{(2)}_3=A^{(3)}_3=B^{(3)}_3=0,\quad
 B^{(2)}_1=B^{(1)}_4 B^{(3)}_1,\\
 B^{(2)}_2=-\dfrac{A^{(1)}_1 A^{(2)}_4 B^{(1)}_4 B^{(3)}_4}{A^{(1)}_4 A^{(2)}_2},\
 B^{(2)}_4=-\dfrac{A^{(1)}_1 A^{(2)}_4 B^{(1)}_4 B^{(3)}_1}{A^{(1)}_4 A^{(2)}_2},\quad
 B^{(3)}_2=\dfrac{B^{(1)}_4 B^{(3)}_4}{B^{(1)}_1},\\
 A^{(1)}_2=-\dfrac{A^{(1)}_4 A^{(3)}_1 B^{(1)}_1}{A^{(3)}_4 B^{(1)}_4},\quad
 B^{(1)}_2=-\dfrac{A^{(1)}_1 A^{(2)}_4 B^{(1)}_4}{A^{(1)}_4 A^{(2)}_2},\quad
 A^{(3)}_2=-\dfrac{A^{(2)}_1 A^{(3)}_4 B^{(3)}_1}{A^{(2)}_4 B^{(3)}_4}.
\end{cases}
\end{equation}
\subsection{For Lemma \ref{lemma:classification_CACO_CO1CO2CO3_2}}~\\
There are two sub-cases in Type II, which correspond to the choice of sign in the equations below.\\
{\rm(Type II-1):}
\begin{equation}\tag{II-1}\label{eqn:typeII_cond1}
\begin{split}
 &C^{(32)}=C^{(12)},\quad
 C^{(33)}=C^{(15)},\quad
 C^{(34)}=C^{(14)},\quad 
 C^{(35)}=C^{(11)},\quad
 C^{(36)}=C^{(16)},\\ 
 &C^{(41)}=C^{(21)},\quad 
 C^{(42)}=C^{(26)},\quad
 C^{(43)}=C^{(23)},\quad
 C^{(45)}=C^{(25)},\quad
 C^{(46)}=C^{(24)}.
\end{split}
\end{equation}
{\rm(Type II-2):}
\begin{equation}\tag{II-2}\label{eqn:typeII_cond2}
\begin{split}
 &C^{(32)}=-C^{(12)},\ 
 C^{(33)}=-C^{(15)},\
 C^{(34)}=-C^{(14)},\ 
 C^{(35)}=-C^{(11)},\
 C^{(36)}=-C^{(16)},\\ 
 &C^{(41)}=-C^{(21)},\
 C^{(42)}=-C^{(26)},\ 
 C^{(43)}=-C^{(23)},\
 C^{(45)}=-C^{(25)},\ 
 C^{(46)}=-C^{(24)}.
\end{split}
\end{equation}
Here, the parameters $C^{(11)}$, \dots, $C^{(16)}$, \dots, $C^{(41)}$, \dots, $C^{(46)}$ are given by 
\begin{subequations}
\begin{align}
 C^{(11)}
 =&A^{(1)}_2B^{(1)}_3(A^{(3)}_2B^{(3)}_4A^{(2)}_1+A^{(3)}_4B^{(3)}_1A^{(2)}_4)
 +A^{(1)}_3B^{(1)}_1(A^{(3)}_4B^{(3)}_1A^{(2)}_2+A^{(3)}_2B^{(3)}_4A^{(2)}_3),\label{eqn:typeII_defC11}\\
 C^{(12)}
 =&A^{(1)}_4B^{(1)}_3(A^{(3)}_2B^{(3)}_4A^{(2)}_1+A^{(3)}_4B^{(3)}_1A^{(2)}_4)
 +A^{(1)}_1B^{(1)}_1(A^{(3)}_4B^{(3)}_1A^{(2)}_2+A^{(3)}_2B^{(3)}_4A^{(2)}_3),\\
 C^{(13)}
 =&A^{(1)}_2B^{(1)}_2(A^{(3)}_2B^{(3)}_2A^{(2)}_1+A^{(3)}_4B^{(3)}_3A^{(2)}_4)
 +A^{(1)}_3B^{(1)}_4(A^{(3)}_4B^{(3)}_3A^{(2)}_2+A^{(3)}_2B^{(3)}_2A^{(2)}_3),\\
 C^{(14)}
 =&A^{(1)}_4B^{(1)}_2(A^{(3)}_2B^{(3)}_2A^{(2)}_1+A^{(3)}_4B^{(3)}_3A^{(2)}_4)
 +A^{(1)}_1B^{(1)}_4(A^{(3)}_4B^{(3)}_3A^{(2)}_2+A^{(3)}_2B^{(3)}_2A^{(2)}_3),\\
 C^{(15)}
 =&A^{(1)}_2(A^{(2)}_1A^{(3)}_2\hat{B}_1+A^{(2)}_4A^{(3)}_4\hat{B}_4)
 +A^{(1)}_3(A^{(2)}_2A^{(3)}_4\hat{B}_2+A^{(2)}_3A^{(3)}_2\hat{B}_3),\\
 C^{(16)}
 =&A^{(1)}_4(A^{(2)}_1A^{(3)}_2\hat{B}_1+A^{(2)}_4A^{(3)}_4\hat{B}_4)
 +A^{(1)}_1(A^{(2)}_2A^{(3)}_4\hat{B}_2+A^{(2)}_3A^{(3)}_2\hat{B}_3),
\end{align}
\end{subequations}
\vspace{-1.5em}
\begin{subequations}
\begin{align}
 C^{(21)}
 =&A^{(1)}_2B^{(1)}_3(A^{(3)}_3B^{(3)}_4A^{(2)}_1+A^{(3)}_1B^{(3)}_1A^{(2)}_4)
 +A^{(1)}_3B^{(1)}_1(A^{(3)}_1B^{(3)}_1A^{(2)}_2+A^{(3)}_3B^{(3)}_4A^{(2)}_3),\\
 C^{(22)}
 =&A^{(1)}_4B^{(1)}_3(A^{(3)}_3B^{(3)}_4A^{(2)}_1+A^{(3)}_1B^{(3)}_1A^{(2)}_4)
 +A^{(1)}_1B^{(1)}_1(A^{(3)}_1B^{(3)}_1A^{(2)}_2+A^{(3)}_3B^{(3)}_4A^{(2)}_3),\\
 C^{(23)}
 =&A^{(1)}_2B^{(1)}_2(A^{(3)}_3B^{(3)}_2A^{(2)}_1+A^{(3)}_1B^{(3)}_3A^{(2)}_4)
 +A^{(1)}_3B^{(1)}_4(A^{(3)}_1B^{(3)}_3A^{(2)}_2+A^{(3)}_3B^{(3)}_2A^{(2)}_3),\\
 C^{(24)}
 =&A^{(1)}_4B^{(1)}_2(A^{(3)}_3B^{(3)}_2A^{(2)}_1+A^{(3)}_1B^{(3)}_3A^{(2)}_4)
 +A^{(1)}_1B^{(1)}_4(A^{(3)}_1B^{(3)}_3A^{(2)}_2+A^{(3)}_3B^{(3)}_2A^{(2)}_3),\\
 C^{(25)}
 =&A^{(1)}_2(A^{(2)}_1A^{(3)}_3\hat{B}_1+A^{(2)}_4A^{(3)}_1\hat{B}_4)
 +A^{(1)}_3(A^{(2)}_2A^{(3)}_1\hat{B}_2+A^{(2)}_3A^{(3)}_3\hat{B}_3),\\
 C^{(26)}
 =&A^{(1)}_4(A^{(2)}_1A^{(3)}_3\hat{B}_1+A^{(2)}_4A^{(3)}_1\hat{B}_4)
 +A^{(1)}_1(A^{(2)}_2A^{(3)}_1\hat{B}_2+A^{(2)}_3A^{(3)}_3\hat{B}_3),
\end{align}
\end{subequations}
\vspace{-1.5em}
\begin{subequations}
\begin{align}
 C^{(31)}
 =&A^{(1)}_3B^{(1)}_3(A^{(3)}_4B^{(3)}_4A^{(2)}_1+A^{(3)}_2B^{(3)}_1A^{(2)}_4)
 +A^{(1)}_2B^{(1)}_1(A^{(3)}_2B^{(3)}_1A^{(2)}_2+A^{(3)}_4B^{(3)}_4A^{(2)}_3),\\
 C^{(32)}
 =&A^{(1)}_3B^{(1)}_3(A^{(3)}_1B^{(3)}_4A^{(2)}_1+A^{(3)}_3B^{(3)}_1A^{(2)}_4)
 +A^{(1)}_2B^{(1)}_1(A^{(3)}_3B^{(3)}_1A^{(2)}_2+A^{(3)}_1B^{(3)}_4A^{(2)}_3),\\
 C^{(33)}
 =&A^{(1)}_3B^{(1)}_2(A^{(3)}_4B^{(3)}_2A^{(2)}_1+A^{(3)}_2B^{(3)}_3A^{(2)}_4)
 +A^{(1)}_2B^{(1)}_4(A^{(3)}_2B^{(3)}_3A^{(2)}_2+A^{(3)}_4B^{(3)}_2A^{(2)}_3),\\
 C^{(34)}
 =&A^{(1)}_3B^{(1)}_2(A^{(3)}_1B^{(3)}_2A^{(2)}_1+A^{(3)}_3B^{(3)}_3A^{(2)}_4)
 +A^{(1)}_2B^{(1)}_4(A^{(3)}_3B^{(3)}_3A^{(2)}_2+A^{(3)}_1B^{(3)}_2A^{(2)}_3),\\
 C^{(35)}
 =&A^{(1)}_3(A^{(2)}_1A^{(3)}_4\hat{B}_1+A^{(2)}_4A^{(3)}_2\hat{B}_4)
 +A^{(1)}_2(A^{(2)}_2A^{(3)}_2\hat{B}_2+A^{(2)}_3A^{(3)}_4\hat{B}_3),\\
 C^{(36)}
 =&A^{(1)}_3(A^{(2)}_1A^{(3)}_1\hat{B}_1+A^{(2)}_4A^{(3)}_3\hat{B}_4)
 +A^{(1)}_2(A^{(2)}_2A^{(3)}_3\hat{B}_2+A^{(2)}_3A^{(3)}_1\hat{B}_3),
\end{align}
\end{subequations}
\vspace{-1.5em}
\begin{subequations}
\begin{align}
 C^{(41)}
 =&A^{(1)}_1B^{(1)}_3(A^{(3)}_4B^{(3)}_4A^{(2)}_1+A^{(3)}_2B^{(3)}_1A^{(2)}_4)
 +A^{(1)}_4B^{(1)}_1(A^{(3)}_2B^{(3)}_1A^{(2)}_2+A^{(3)}_4B^{(3)}_4A^{(2)}_3),\\
 C^{(42)}
 =&A^{(1)}_1B^{(1)}_3(A^{(3)}_1B^{(3)}_4A^{(2)}_1+A^{(3)}_3B^{(3)}_1A^{(2)}_4)
 +A^{(1)}_4B^{(1)}_1(A^{(3)}_3B^{(3)}_1A^{(2)}_2+A^{(3)}_1B^{(3)}_4A^{(2)}_3),\\
 C^{(43)}
 =&A^{(1)}_1B^{(1)}_2(A^{(3)}_4B^{(3)}_2A^{(2)}_1+A^{(3)}_2B^{(3)}_3A^{(2)}_4)
 +A^{(1)}_4B^{(1)}_4(A^{(3)}_2B^{(3)}_3A^{(2)}_2+A^{(3)}_4B^{(3)}_2A^{(2)}_3),\\
 C^{(44)}
 =&A^{(1)}_1B^{(1)}_2(A^{(3)}_1B^{(3)}_2A^{(2)}_1+A^{(3)}_3B^{(3)}_3A^{(2)}_4)
 +A^{(1)}_4B^{(1)}_4(A^{(3)}_3B^{(3)}_3A^{(2)}_2+A^{(3)}_1B^{(3)}_2A^{(2)}_3),\\
 C^{(45)}
 =&A^{(1)}_1(A^{(2)}_1A^{(3)}_4\hat{B}_1+A^{(2)}_4A^{(3)}_2\hat{B}_4)
 +A^{(1)}_4(A^{(2)}_2A^{(3)}_2\hat{B}_2+A^{(2)}_3A^{(3)}_4\hat{B}_3),\\
 C^{(46)}
 =&A^{(1)}_1(A^{(2)}_1A^{(3)}_1\hat{B}_1+A^{(2)}_4A^{(3)}_3\hat{B}_4)
 +A^{(1)}_4(A^{(2)}_2A^{(3)}_3\hat{B}_2+A^{(2)}_3A^{(3)}_1\hat{B}_3),\label{eqn:typeII_defC46}
\end{align}
\end{subequations}
where
\begin{align}
 &\hat{B}_1=B^{(1)}_3B^{(3)}_2+B^{(1)}_2B^{(3)}_4,\quad
 \hat{B}_2=B^{(1)}_4B^{(3)}_1+B^{(1)}_1B^{(3)}_3,\quad
 \hat{B}_3=B^{(1)}_1B^{(3)}_2+B^{(1)}_4B^{(3)}_4,\notag\\
 &\hat{B}_4=B^{(1)}_2B^{(3)}_1+B^{(1)}_3B^{(3)}_3.
\end{align}
\subsection{For Lemma \ref{lemma:classification_CACO_CO1CO2CO3_3}}
\label{subsection:proof_cond_para_N1}~\\
Type III contains 32 subcases, which are listed below.
~\\
{\rm(Type III-1-1):}\quad
the condition
\begin{equation}\tag{III-1}\label{eqn:typeIII_cond1}
 A^{(1)}_1=B^{(1)}_1=A^{(2)}_1=B^{(2)}_1=A^{(3)}_1=B^{(3)}_1=0,
\end{equation}
and
\begin{equation}\tag{III-1-1}\label{eqn:typeIII_cond1_1}
\begin{cases}
 A^{(3)}_3=\dfrac{A^{(1)}_2A^{(2)}_2A^{(3)}_2}{A^{(1)}_3A^{(2)}_3},\quad
 B^{(3)}_3=-\dfrac{A^{(2)}_3B^{(3)}_2d_3(A^{(1)}_2B^{(1)}_2+A^{(1)}_3B^{(1)}_3{d_2}^2d_3)}{A^{(1)}_2A^{(2)}_2B^{(1)}_3d_2},\\
 B^{(3)}_4=\dfrac{A^{(2)}_3B^{(3)}_2d_3\Big(A^{(1)}_3A^{(3)}_4(1-d_2)+A^{(1)}_4A^{(3)}_2(1-d_3)\Big)}{2A^{(1)}_2A^{(2)}_2A^{(3)}_2}\\
 \hspace{3.em}+\dfrac{B^{(3)}_2\Big(B^{(1)}_3A^{(2)}_4(1+d_2d_3)-2A^{(2)}_3B^{(1)}_4d_3\Big)}{2A^{(2)}_2B^{(1)}_3d_2}.
\end{cases}
\end{equation}
{\rm(Type III-1-2):}\quad
the condition \eqref{eqn:typeIII_cond1} and
\begin{equation}\tag{III-1-2}\label{eqn:typeIII_cond1_2}
 A^{(3)}_3=-\dfrac{A^{(1)}_2A^{(2)}_2A^{(3)}_2}{A^{(1)}_3A^{(2)}_3},\quad
 A^{(3)}_4=\dfrac{A^{(3)}_2\Big(A^{(1)}_4A^{(2)}_3d_2d_3(1+d_3)+A^{(1)}_2A^{(2)}_4(1-d_2d_3)\Big)}{A^{(1)}_3A^{(2)}_3d_2d_3(1+d_2)}.
\end{equation}
{\rm(Type III-1-3):}\quad
the condition \eqref{eqn:typeIII_cond1} and
\begin{equation}\tag{III-1-3}\label{eqn:typeIII_cond1_3}
 A^{(3)}_3=-\dfrac{A^{(1)}_2A^{(2)}_2A^{(3)}_2}{A^{(1)}_3A^{(2)}_3},\quad
 A^{(2)}_4=\dfrac{A^{(1)}_4A^{(2)}_3d_3}{A^{(1)}_2},\quad
 d_2=-1.
\end{equation}
{\rm(Type III-1-4):}\quad
the condition \eqref{eqn:typeIII_cond1} and
\begin{equation}\tag{III-1-4}\label{eqn:typeIII_cond1_4}
 A^{(3)}_3=-\dfrac{A^{(1)}_2A^{(2)}_2A^{(3)}_2}{A^{(1)}_3A^{(2)}_3},\quad
 d_2=d_3=-1.
\end{equation}
{\rm(Type III-2-1):}\quad
the condition
\begin{equation}\tag{III-2}\label{eqn:typeIII_cond2}
 A^{(1)}_1=B^{(1)}_1=0,\quad A^{(2)}_1,B^{(2)}_1,A^{(3)}_1,B^{(3)}_1\neq0,
\end{equation}
and
\begin{equation}\tag{III-2-1}\label{eqn:typeIII_cond2_1}
\begin{cases}
 B^{(1)}_2=\dfrac{A^{(1)}_3B^{(1)}_3}{A^{(1)}_2},\quad
 B^{(3)}_3=-\dfrac{A^{(2)}_3B^{(3)}_1}{A^{(2)}_1}-\dfrac{A^{(3)}_2B^{(3)}_1}{A^{(3)}_1},\\
 A^{(3)}_4=\dfrac{A^{(1)}_2A^{(3)}_1(A^{(2)}_2A^{(2)}_3-A^{(2)}_1A^{(2)}_4)}{A^{(1)}_3{A^{(2)}_1}^2}+\dfrac{A^{(3)}_2A^{(3)}_3}{A^{(3)}_1},\\
 B^{(1)}_4=-\dfrac{(A^{(1)}_4A^{(3)}_1+A^{(1)}_3A^{(3)}_3)B^{(1)}_3}{A^{(1)}_2A^{(3)}_1}-\dfrac{A^{(2)}_2B^{(1)}_3}{A^{(2)}_1},\quad
 d_2=d_3=-1.
\end{cases}
\end{equation}
{\rm(Type III-2-2):}\quad
the condition \eqref{eqn:typeIII_cond2} and
\begin{equation}\tag{III-2-2}\label{eqn:typeIII_cond2_2}
\begin{cases}
 B^{(1)}_2=\dfrac{A^{(1)}_3B^{(1)}_3}{A^{(1)}_2},\quad
 B^{(3)}_3=-\dfrac{(A^{(2)}_3A^{(3)}_1+A^{(2)}_1A^{(3)}_2)B^{(3)}_1}{A^{(2)}_1A^{(3)}_1},\\
 B^{(1)}_4=-\dfrac{(A^{(1)}_4A^{(3)}_1+A^{(1)}_3A^{(3)}_3)B^{(1)}_3}{A^{(1)}_2A^{(3)}_1}-\dfrac{A^{(2)}_2B^{(1)}_3}{A^{(2)}_1},\\
 A^{(3)}_4=-\dfrac{A^{(1)}_2A^{(3)}_1(A^{(2)}_2A^{(2)}_3-A^{(2)}_1A^{(2)}_4)}{A^{(1)}_3{A^{(2)}_1}^2}+\dfrac{A^{(3)}_2A^{(3)}_3}{A^{(3)}_1},\\
 B^{(3)}_4=-\dfrac{(A^{(2)}_3A^{(3)}_1+A^{(2)}_1A^{(3)}_2)B^{(3)}_2}{A^{(2)}_1A^{(3)}_1},\quad
 d_2=d_3=-1.
\end{cases}
\end{equation}
{\rm(Type III-2-3):}\quad
the condition \eqref{eqn:typeIII_cond2} and
\begin{equation}\tag{III-2-3}\label{eqn:typeIII_cond2_3}
\begin{cases}
 B^{(1)}_2=\dfrac{A^{(1)}_3B^{(1)}_3}{A^{(1)}_2},\quad
 A^{(2)}_4=-\dfrac{A^{(2)}_3B^{(1)}_4(1-d_2)}{B^{(1)}_3{d_2}^2}-\dfrac{A^{(1)}_3A^{(2)}_1A^{(3)}_4}{A^{(1)}_2A^{(3)}_1},\\
 A^{(3)}_2=\dfrac{A^{(2)}_3A^{(3)}_1}{A^{(2)}_1},\quad
 A^{(3)}_3=-\dfrac{A^{(1)}_2A^{(3)}_1\Big(A^{(2)}_1B^{(1)}_4(1-d_2)+A^{(2)}_2B^{(1)}_3{d_2}^2\Big)}{A^{(1)}_3A^{(2)}_1B^{(1)}_3{d_2}^2},\\
 B^{(3)}_3=-\dfrac{A^{(2)}_3B^{(3)}_1(1-d_2)}{A^{(2)}_1{d_2}^2},\quad
 B^{(3)}_4=-\dfrac{A^{(2)}_3B^{(3)}_2(1-d_2)}{A^{(2)}_1{d_2}^2},\quad
 d_3=-\dfrac{1}{{d_2}^2}.
\end{cases}
\end{equation}
{\rm(Type III-2-4):}\quad
the condition \eqref{eqn:typeIII_cond2} and
\begin{equation}\tag{III-2-4}\label{eqn:typeIII_cond2_4}
\begin{cases}
A^{(3)}_3=-\dfrac{A^{(3)}_1\Big(A^{(2)}_1B^{(1)}_4(1-d_2)+A^{(2)}_2B^{(1)}_3{d_2}^2\Big)}{A^{(2)}_1B^{(1)}_2{d_2}^2},\quad
 A^{(1)}_3=\dfrac{A^{(1)}_2B^{(1)}_2}{B^{(1)}_3},\\
 A^{(1)}_4=\dfrac{A^{(1)}_2B^{(1)}_4}{B^{(1)}_3},\quad
 A^{(2)}_4=-\dfrac{A^{(2)}_3A^{(3)}_1B^{(1)}_4(1-d_2)+A^{(2)}_1A^{(3)}_4B^{(1)}_2{d_2}^2}{A^{(3)}_1B^{(1)}_3{d_2}^2},\\
 A^{(3)}_2=\dfrac{A^{(2)}_3A^{(3)}_1}{A^{(2)}_1},\quad
 B^{(3)}_3=-\dfrac{A^{(2)}_3B^{(3)}_1(1-d_2)}{A^{(2)}_1{d_2}^2},\quad
 d_3=-\dfrac{1}{{d_2}^2}.
\end{cases}
\end{equation}
{\rm(Type III-2-5):}\quad
the condition \eqref{eqn:typeIII_cond2} and
\begin{equation}\tag{III-2-5}\label{eqn:typeIII_cond2_5}
\begin{cases}
 B^{(1)}_2=-\dfrac{A^{(1)}_3B^{(1)}_3}{A^{(1)}_2},\quad
 A^{(3)}_2=\dfrac{A^{(2)}_3A^{(3)}_1(1-d_2)}{A^{(2)}_1(1+d_2)},\\
 A^{(2)}_4=\dfrac{2A^{(1)}_4A^{(2)}_3(1-d_2)}{A^{(1)}_2d_2(1+d_2)^2}+\dfrac{2A^{(2)}_2A^{(2)}_3(1+{d_2}^2)}{A^{(2)}_1(1+d_2)^2}-\dfrac{A^{(1)}_3A^{(2)}_1A^{(3)}_4}{A^{(1)}_2A^{(3)}_1}\\
 \hspace{2.5em}-\dfrac{A^{(2)}_3B^{(1)}_4(1-d_2)^2}{B^{(1)}_3{d_2}^2(1+d_2)},\\
 A^{(3)}_3=-\dfrac{A^{(1)}_2A^{(3)}_1(1-d_2)}{A^{(1)}_3}\left(\dfrac{B^{(1)}_4}{B^{(1)}_3{d_2}^2}-\dfrac{A^{(2)}_2}{A^{(2)}_1(1+d_2)}\right)+\dfrac{2A^{(1)}_4A^{(3)}_1}{A^{(1)}_3d_2(1+d_2)},\\
 B^{(3)}_3=\dfrac{A^{(2)}_3B^{(3)}_1(1+{d_2}^2)}{A^{(2)}_1{d_2}^2(1+d_2)},\quad
 B^{(3)}_4=\dfrac{A^{(2)}_3B^{(3)}_2(1+{d_2}^2)}{A^{(2)}_1{d_2}^2(1+d_2)},\quad
 d_3=\dfrac{1}{{d_2}^2}.
\end{cases}
\end{equation}
{\rm(Type III-2-6):}\quad
the condition \eqref{eqn:typeIII_cond2} and
\begin{equation}\tag{III-2-6}\label{eqn:typeIII_cond2_6}
\begin{cases}
 B^{(1)}_2=-\dfrac{A^{(1)}_3B^{(1)}_3}{A^{(1)}_2},\quad
 A^{(1)}_4=\dfrac{(1-d_2)A^{(1)}_2A^{(3)}_1B^{(1)}_4+2A^{(1)}_3A^{(3)}_3B^{(1)}_3{d_2}^2}{A^{(3)}_1B^{(1)}_3(1+d_2)},\\
 A^{(2)}_2=\dfrac{A^{(2)}_1\Big(A^{(1)}_3A^{(3)}_3B^{(1)}_3(1-d_2){d_2}^2+A^{(1)}_2A^{(3)}_1B^{(1)}_4(1+{d_2}^2)\Big)}{A^{(1)}_2A^{(3)}_1B^{(1)}_3{d_2}^2(1+d_2)},\\
 A^{(2)}_4=\dfrac{2A^{(1)}_3A^{(2)}_3A^{(3)}_3(1-d_2)}{A^{(1)}_2A^{(3)}_1(1+d_2)}+\dfrac{A^{(2)}_3B^{(1)}_4(1+{d_2}^2)}{B^{(1)}_3{d_2}^2(1+d_2)}-\dfrac{A^{(1)}_3A^{(2)}_1A^{(3)}_4}{A^{(1)}_2A^{(3)}_1},\\
 A^{(3)}_2=\dfrac{A^{(2)}_3A^{(3)}_1(1-d_2)}{A^{(2)}_1(1+d_2)},\quad
 B^{(3)}_3=\dfrac{A^{(2)}_3B^{(3)}_1(1+{d_2}^2)}{A^{(2)}_1{d_2}^2(1+d_2)},\quad
 d_3=\dfrac{1}{{d_2}^2}.
\end{cases}
\end{equation}
{\rm(Type III-2-7):}\quad
the condition \eqref{eqn:typeIII_cond2} and
\begin{equation}\tag{III-2-7}\label{eqn:typeIII_cond2_7}
\begin{cases}
 B^{(1)}_2=\dfrac{A^{(1)}_3B^{(1)}_3}{A^{(1)}_2},\quad
 A^{(1)}_4=-\dfrac{A^{(1)}_2A^{(2)}_2}{A^{(2)}_1}-\dfrac{A^{(1)}_3A^{(3)}_3}{A^{(3)}_1}-\dfrac{A^{(1)}_2B^{(1)}_4}{B^{(1)}_3},\\
 A^{(2)}_4=\dfrac{A^{(2)}_2A^{(2)}_3}{A^{(2)}_1}+\dfrac{A^{(1)}_3A^{(2)}_1(A^{(3)}_2A^{(3)}_3-A^{(3)}_1A^{(3)}_4)}{A^{(1)}_2{A^{(3)}_1}^2},\quad
 d_2=d_3=-1.
\end{cases}
\end{equation}
{\rm(Type III-2-8):}\quad
the condition \eqref{eqn:typeIII_cond2} and
\begin{equation}\tag{III-2-8}\label{eqn:typeIII_cond2_8}
\begin{cases}
 B^{(1)}_2=\dfrac{A^{(1)}_3B^{(1)}_3}{A^{(1)}_2},\
 A^{(1)}_4=\dfrac{A^{(1)}_2B^{(1)}_4}{B^{(1)}_3},\
 A^{(2)}_2=-\dfrac{A^{(2)}_1B^{(1)}_4(1-d_2)}{B^{(1)}_3{d_2}^2}-\dfrac{A^{(1)}_3A^{(2)}_1A^{(3)}_3}{A^{(1)}_2A^{(3)}_1},\\
 A^{(2)}_4=-\dfrac{(1-d_2)A^{(2)}_3B^{(1)}_4}{B^{(1)}_3{d_2}^2}-\dfrac{A^{(1)}_3A^{(2)}_1A^{(3)}_4}{A^{(1)}_2A^{(3)}_1},\quad
 A^{(3)}_2=\dfrac{A^{(2)}_3A^{(3)}_1}{A^{(2)}_1},\quad
 d_3=-\dfrac{1}{{d_2}^2}.
\end{cases}
\end{equation}
{\rm(Type III-2-9):}\quad
the condition \eqref{eqn:typeIII_cond2} and
\begin{equation}\tag{III-2-9}\label{eqn:typeIII_cond2_9}
\begin{cases}
 B^{(1)}_2=-\dfrac{A^{(1)}_3B^{(1)}_3}{A^{(1)}_2},\quad
 A^{(1)}_4=\dfrac{(1-d_2)A^{(1)}_2A^{(3)}_1B^{(1)}_4+2A^{(1)}_3A^{(3)}_3B^{(1)}_3{d_2}^2}{A^{(3)}_1B^{(1)}_3(1+d_2)},\\
 A^{(2)}_2=\dfrac{A^{(2)}_1B^{(1)}_4(1+{d_2}^2)}{B^{(1)}_3{d_2}^2(1+d_2)}+\dfrac{A^{(1)}_3A^{(2)}_1A^{(3)}_3(1-d_2)}{A^{(1)}_2A^{(3)}_1(1+d_2)},\quad
 A^{(3)}_2=\dfrac{A^{(2)}_3A^{(3)}_1(1-d_2)}{A^{(2)}_1(1+d_2)},\\
 A^{(2)}_4=\dfrac{2A^{(1)}_3A^{(2)}_3A^{(3)}_3(1-d_2)}{A^{(1)}_2A^{(3)}_1(1+d_2)}+\dfrac{A^{(2)}_3B^{(1)}_4(1+{d_2}^2)}{B^{(1)}_3{d_2}^2(1+d_2)}-\dfrac{A^{(1)}_3A^{(2)}_1A^{(3)}_4}{A^{(1)}_2A^{(3)}_1},\
 d_3=\dfrac{1}{{d_2}^2}.
\end{cases}
\end{equation}
{\rm(Type III-2-10):}\quad
the condition \eqref{eqn:typeIII_cond2} and
\begin{equation}\tag{III-2-10}\label{eqn:typeIII_cond2_10}
\begin{cases}
 A^{(2)}_3=-\dfrac{A^{(2)}_1(A^{(3)}_2B^{(3)}_1+A^{(3)}_1B^{(3)}_3)}{A^{(3)}_1B^{(3)}_1},\quad
 d_2=d_3=-1,\\
 A^{(2)}_4=\dfrac{A^{(1)}_3A^{(2)}_1(A^{(3)}_2A^{(3)}_3-A^{(3)}_1A^{(3)}_4)}{A^{(1)}_2{A^{(3)}_1}^2}-\dfrac{A^{(2)}_2(A^{(3)}_2B^{(3)}_1+A^{(3)}_1B^{(3)}_3)}{A^{(3)}_1B^{(3)}_1}\\
 \hspace{2.5em}+\dfrac{A^{(2)}_1(A^{(1)}_2B^{(1)}_2-A^{(1)}_3B^{(1)}_3)(B^{(3)}_2B^{(3)}_3-B^{(3)}_1B^{(3)}_4)}{A^{(1)}_2B^{(1)}_3{B^{(3)}_1}^2},\\
 B^{(1)}_4=\dfrac{B^{(3)}_2(A^{(1)}_2B^{(1)}_2-A^{(1)}_3B^{(1)}_3)}{A^{(1)}_2B^{(3)}_1}-\dfrac{B^{(1)}_3(A^{(1)}_4A^{(3)}_1+A^{(1)}_3A^{(3)}_3)}{A^{(1)}_2A^{(3)}_1}-\dfrac{A^{(2)}_2B^{(1)}_3}{A^{(2)}_1}.
\end{cases}
\end{equation}
{\rm(Type III-2-11):}\quad
the condition \eqref{eqn:typeIII_cond2} and
\begin{equation}\tag{III-2-11}\label{eqn:typeIII_cond2_11}
\begin{cases}
 A^{(2)}_3=B^{(2)}_3=0,\quad
 A^{(1)}_4=\dfrac{A^{(1)}_2A^{(2)}_2}{A^{(2)}_1},\quad
 d_2=-1,\quad d_3=1,\\
 A^{(2)}_4=\dfrac{A^{(1)}_3A^{(2)}_1(A^{(3)}_1B^{(3)}_4-A^{(3)}_4B^{(3)}_1)}{A^{(1)}_2A^{(3)}_1B^{(3)}_1}+\dfrac{A^{(2)}_1(B^{(1)}_2B^{(3)}_4-B^{(1)}_4B^{(3)}_3)}{B^{(1)}_3B^{(3)}_1},\\
 A^{(3)}_2=\dfrac{A^{(3)}_1B^{(3)}_3}{B^{(3)}_1},\quad
 B^{(3)}_2=\dfrac{(A^{(1)}_3A^{(3)}_3B^{(1)}_3+A^{(1)}_2A^{(3)}_1B^{(1)}_4)B^{(3)}_1}{A^{(3)}_1(A^{(1)}_2B^{(1)}_2+A^{(1)}_3B^{(1)}_3)}.
\end{cases}
\end{equation}
{\rm(Type III-2-12):}\quad
the condition \eqref{eqn:typeIII_cond2} and
\begin{equation}\tag{III-2-12}\label{eqn:typeIII_cond2_12}
\begin{cases}
 A^{(3)}_2=B^{(3)}_2=0,\quad
 A^{(1)}_4=\dfrac{A^{(1)}_3A^{(3)}_3}{A^{(3)}_1},\quad
 A^{(2)}_3=\dfrac{A^{(2)}_1B^{(3)}_3}{B^{(3)}_1},\\
 A^{(3)}_4=\dfrac{A^{(1)}_2A^{(3)}_1(A^{(2)}_2B^{(3)}_3-A^{(2)}_4B^{(3)}_1)}{A^{(1)}_3A^{(2)}_1B^{(3)}_1}-\dfrac{A^{(3)}_1B^{(3)}_4(A^{(1)}_2B^{(1)}_2+A^{(1)}_3B^{(1)}_3)}{A^{(1)}_3B^{(1)}_3B^{(3)}_1},\\
 B^{(1)}_4=\dfrac{A^{(2)}_2B^{(1)}_3}{A^{(2)}_1},\quad
 d_2=d_3=1.
\end{cases}
\end{equation}
{\rm(Type III-3-1):}\quad
the condition
\begin{equation}\tag{III-3}\label{eqn:typeIII_cond3}
 A^{(i)}_1=B^{(i)}_1=1,\quad i=1,2,3,
\end{equation}
and
\begin{equation}\tag{III-3-1}\label{eqn:typeIII_cond3_1}
\begin{cases}
a^{(1)}_4=b^{(1)}_4=-\dfrac{(A^{(1)}_3+A^{(2)}_2+B^{(1)}_2)(A^{(1)}_2+A^{(3)}_3+B^{(3)}_2)}{1+A^{(1)}_2+A^{(3)}_3+B^{(3)}_2},\quad
B^{(1)}_3=1+B^{(3)}_2\\
a^{(2)}_4=-(A^{(1)}_3+A^{(2)}_2+B^{(1)}_2)(A^{(2)}_3+A^{(3)}_2+B^{(3)}_3),\quad
d_2=d_3=-1\\
a^{(3)}_4=b^{(3)}_4=\dfrac{(1+A^{(1)}_2+A^{(3)}_3+B^{(3)}_2)(A^{(2)}_3+A^{(3)}_2+B^{(3)}_3)}{A^{(1)}_2+A^{(3)}_3+B^{(3)}_2}.
\end{cases}
\end{equation}
Here, parameters $a^{(i)}_4$ and $b^{(i)}_4$, $i=1,2,3$, are given by
\begin{equation}\label{eqn:typeIII_def_ab}
 a^{(i)}_4=A^{(i)}_4-A^{(i)}_2A^{(i)}_3,\quad
 b^{(i)}_4=B^{(i)}_4-B^{(i)}_2B^{(i)}_3.
\end{equation}
Note that under the condition \eqref{eqn:typeIII_cond3} and setting \eqref{eqn:typeIII_def_ab},
the condition \eqref{eqn:caseN1_CAO_B2} can be rewritten as
\begin{equation}\label{eqn:caseN1_CAO_B2_2}
 B^{(1)}_3=1+B^{(3)}_2,\quad
 B^{(2)}_2=B^{(3)}_3-b^{(3)}_4,\quad
 B^{(2)}_3=B^{(1)}_2+b^{(1)}_4,\quad
 b^{(2)}_4=b^{(1)}_4 b^{(3)}_4.
\end{equation}
{\rm(Type III-3-2):}\quad
the condition \eqref{eqn:typeIII_cond3} and the following condition with \eqref{eqn:typeIII_def_ab}{\rm :}
\begin{equation}\tag{III-3-2}\label{eqn:typeIII_cond3_2}
\begin{cases}
 a^{(1)}_4=\dfrac{(A^{(2)}_2-B^{(1)}_2)(1+d_2)}{(1-d_2){d_2}^2},\quad
 a^{(2)}_4=-\dfrac{(A^{(2)}_2-B^{(1)}_2)(A^{(2)}_3-B^{(3)}_3)(1+d_2)^2}{(1-d_2)^2},\\
 a^{(3)}_4=-\dfrac{(A^{(2)}_3-B^{(3)}_3)(1+d_2)}{(1-d_2){d_2}^2},\quad
 b^{(1)}_4=-\dfrac{(A^{(2)}_2-B^{(1)}_2)(1+d_2)}{(1-d_2)d_2},\quad
 d_3=d_2,\\
 b^{(3)}_4=\dfrac{(A^{(2)}_3-B^{(3)}_3)(1+d_2)}{(1-d_2)d_2},\quad
 A^{(1)}_3=-\dfrac{B^{(1)}_2(1+{d_2}^2)-2A^{(2)}_2d_2}{(1-d_2){d_2}^2},\\
 B^{(1)}_3=1+B^{(3)}_2,\quad
 A^{(3)}_2=-\dfrac{B^{(3)}_3(1+{d_2}^2)-2A^{(2)}_3d_2}{(1-d_2){d_2}^2},\quad
 1-{d_2}^2+2{d_2}^4=0,\\
 A^{(3)}_3=\dfrac{1-2d_2(1+B^{(3)}_2d_2)}{2(1-d_2)},\quad
 A^{(1)}_2=-\dfrac{1+2{d_2}^2B^{(3)}_2}{2(1-d_2)}+d_2.
\end{cases}
\end{equation}
{\rm(Type III-3-3):}\quad
the condition \eqref{eqn:typeIII_cond3} and the following condition with \eqref{eqn:typeIII_def_ab}{\rm :}
\begin{equation}\tag{III-3-3}\label{eqn:typeIII_cond3_3}
\begin{cases}
 a^{(1)}_4=\dfrac{A^{(2)}_2}{{d_2}^2}-\dfrac{B^{(1)}_2(1-d_2)}{{d_2}^2(1+d_2)},\quad
 a^{(3)}_4=\dfrac{B^{(3)}_3(1+d_2)}{{d_2}^2(1-d_2)}-\dfrac{A^{(2)}_3}{{d_2}^2},\\
 a^{(2)}_4=\dfrac{2A^{(2)}_2B^{(3)}_3}{1-d_2}+\dfrac{2A^{(2)}_3B^{(1)}_2}{1+d_2}-(A^{(2)}_2+B^{(1)}_2)(A^{(2)}_3+B^{(3)}_3),\\
 b^{(1)}_4=\dfrac{A^{(2)}_2}{d_2}-\dfrac{B^{(1)}_2(1-d_2)}{d_2(1+d_2)},\quad
 b^{(3)}_4=\dfrac{A^{(2)}_3}{d_2}-\dfrac{B^{(3)}_3(1+d_2)}{d_2(1-d_2)},\\
 A^{(1)}_3=\dfrac{B^{(1)}_2(1+{d_2}^2)}{{d_2}^2(1+d_2)},\quad
 B^{(1)}_3=\dfrac{1}{d_2}-\dfrac{1}{2{d_2}^2}+\dfrac{A^{(1)}_2(1+{d_2}^2)}{{d_2}^2(1+d_2)},\\
 A^{(3)}_2=\dfrac{B^{(3)}_3(1+{d_2}^2)}{{d_2}^2(1-d_2)},\quad
 A^{(3)}_3=-\dfrac{(1-d_2)(1-A^{(1)}_2+d_2)}{1+d_2},\\
 B^{(3)}_2=-1-\dfrac{1}{2{d_2}^2}+\dfrac{1}{d_2}+\dfrac{A^{(1)}_2(1+{d_2}^2)}{{d_2}^2(1+d_2)},\\
 d_3=-d_2,\quad
 1-{d_2}^2+2{d_2}^4=0.
\end{cases}
\end{equation}
{\rm(Type III-3-4):}\quad
the condition \eqref{eqn:typeIII_cond3} and the following condition with \eqref{eqn:typeIII_def_ab}{\rm :}
\begin{equation}\tag{III-3-4}\label{eqn:typeIII_cond3_4}
\begin{cases}
 a^{(1)}_4=\dfrac{(1+{d_2}^2{d_3}^2)\Big(A^{(2)}_2(1+d_2)-B^{(1)}_2(1+d_3)\Big)}{{d_2}^2{d_3}^2(1-d_3)^2(1+d_3)},\\
 a^{(2)}_4=\dfrac{\Big(B^{(3)}_3(1+d_2)-A^{(2)}_3(1+d_3)\Big)\Big(A^{(2)}_2(1+d_2)-B^{(1)}_2(1+d_3)\Big)}{(1-d_2)(1-d_3)},\\
 a^{(3)}_4=\dfrac{(1+{d_2}^2{d_3}^2)\Big(B^{(3)}_3(1+d_2)-A^{(2)}_3(1+d_3)\Big)}{{d_2}^2{d_3}^2(1-d_2)^2)(1+d_2)},\\
 b^{(1)}_4=-\dfrac{(1+{d_2}^2{d_3}^2)\Big(A^{(2)}_2(1+d_2)-B^{(1)}_2(1+d_3)\Big)}{{d_2}^2d_3(1-d_3)^2(1+d_3)},\\
 b^{(3)}_4=-\dfrac{(1+{d_2}^2{d_3}^2)\Big(B^{(3)}_3(1+d_2)-A^{(2)}_3(1+d_3)\Big)}{d_2{d_3}^2(1-d_2)^2(1+d_2)},\\
 A^{(1)}_3=\dfrac{A^{(2)}_2(d_2+d_3)-B^{(1)}_2(1+{d_3}^2)}{d_2d_3(1-d_3)},\
 A^{(3)}_3=\dfrac{1+{d_2}^2{d_3}^2+A^{(1)}_2d_2d_3(1+d_3)}{d_2d_3(1+d_2)},\\
 B^{(1)}_3=-\dfrac{1+{d_2}^2{d_3}^2}{{d_2}^2{d_3}^2(1+d_2)}-\dfrac{1-{d_2}^2+2{d_2}^2{d_3}^2}{{d_2}^2({d_2}^2-{d_3}^2)}-\dfrac{A^{(1)}_2(1-d_2d_3)}{d_2d_3(1+d_2)},\\
 A^{(3)}_2=-\dfrac{B^{(3)}_3(1+{d_2}^2)-A^{(2)}_3(d_2+d_3)}{(1-d_2)d_2d_3},\
 1+{d_2}^2{d_3}^2+{d_2}^4{d_3}^2+{d_2}^2{d_3}^4=0,\\
 B^{(3)}_2=-\dfrac{1-{d_2}^3{d_3}^2}{{d_2}^2{d_3}^2(1+d_2)}-\dfrac{1-{d_2}^2+2{d_2}^4}{{d_2}^2({d_2}^2-{d_3}^2)}-\dfrac{A^{(1)}_2(1-d_2d_3)}{d_2d_3(1+d_2)}.
\end{cases}
\end{equation}
{\rm(Type III-3-5):}\quad
the condition \eqref{eqn:typeIII_cond3} and the following condition with \eqref{eqn:typeIII_def_ab}{\rm :}
\begin{equation}\tag{III-3-5}\label{eqn:typeIII_cond3_5}
\begin{cases}
 a^{(1)}_4=\dfrac{(2+\ii\sqrt{2})A^{(1)}_3}{3}-2A^{(2)}_2,\quad
 a^{(2)}_4=\dfrac{\Big((2+\ii\sqrt{2})A^{(1)}_3-6A^{(2)}_2\Big)a^{(3)}_4}{12},\\
 b^{(1)}_4=\dfrac{(1-\ii\sqrt{2})A^{(1)}_3}{3}+\ii\sqrt{2}A^{(2)}_2,\quad
 b^{(3)}_4=-\dfrac{\ii a^{(3)}_4}{\sqrt{2}},\quad
 B^{(1)}_2=\dfrac{(2+\ii\sqrt{2})A^{(1)}_3}{6},\\
 B^{(1)}_3=1-\ii\sqrt{2}+(2-\ii\sqrt{2})A^{(1)}_2,\\
 A^{(3)}_2=\dfrac{(2-\ii\sqrt{2})(2A^{(2)}_3-a^{(3)}_4)}{2},\quad
 A^{(3)}_3=\dfrac{2-\ii\sqrt{2}}{2}+A^{(1)}_2,\\
 B^{(3)}_2=-\ii\sqrt{2}+(2-\ii\sqrt{2})A^{(1)}_2,\quad
 B^{(3)}_3=A^{(2)}_3-\dfrac{a^{(3)}_4}{2},\quad
 d_2=d_3=\dfrac{\ii}{\sqrt{2}},
\end{cases}
\end{equation}
where $\ii=\sqrt{-1}$.\\
{\rm(Type III-3-6):}\quad
the condition \eqref{eqn:typeIII_cond3} and the following condition with \eqref{eqn:typeIII_def_ab}{\rm :}
\begin{equation}\tag{III-3-6}\label{eqn:typeIII_cond3_6}
\begin{cases}
 a^{(1)}_4=\dfrac{(2-\ii\sqrt{2})A^{(1)}_3}{3}-2A^{(2)}_2,\quad
 a^{(2)}_4=\dfrac{\Big((2-\ii\sqrt{2})A^{(1)}_3-6A^{(2)}_2\Big)a^{(3)}_4}{12},\\
 b^{(1)}_4=\dfrac{(1+\ii\sqrt{2})A^{(1)}_3}{3}-\ii\sqrt{2}A^{(2)}_2,\quad
 b^{(3)}_4=\dfrac{\ii a^{(3)}_4}{\sqrt{2}},\\
 B^{(1)}_2=\dfrac{(2-\ii\sqrt{2})A^{(1)}_3}{6},\quad
 B^{(1)}_3=1+\ii\sqrt{2}+(2+\ii\sqrt{2})A^{(1)}_2,\\
 A^{(3)}_2=\dfrac{(2+\ii\sqrt{2})(2A^{(2)}_3-a^{(3)}_4)}{2},\quad
 A^{(3)}_3=\dfrac{2+\ii\sqrt{2}}{2}+A^{(1)}_2,\\
 B^{(3)}_2=\ii\sqrt{2}+(2+\ii\sqrt{2})A^{(1)}_2,\quad
 B^{(3)}_3=A^{(2)}_3-\dfrac{a^{(3)}_4}{2},\quad
 d_2=d_3=-\dfrac{\ii}{\sqrt{2}}.
\end{cases}
\end{equation}
{\rm(Type III-3-7):}\quad
the condition \eqref{eqn:typeIII_cond3} and the following condition with \eqref{eqn:typeIII_def_ab}{\rm :}
\begin{equation}\tag{III-3-7}\label{eqn:typeIII_cond3_7}
\begin{cases}
 a^{(1)}_4=-\dfrac{(2-\ii\sqrt{2})A^{(1)}_3+2A^{(2)}_2}{3},\quad
 a^{(2)}_4=-\dfrac{\Big((2-\ii\sqrt{2})A^{(1)}_3+2A^{(2)}_2\Big)a^{(3)}_4}{12},\\
 A^{(3)}_2=-\dfrac{(2-\ii\sqrt{2})(2A^{(2)}_3-3a^{(3)}_4)}{6},\quad
 A^{(3)}_3=\dfrac{(2\ii+\sqrt{2})\Big(3\ii-(2\ii+\sqrt{2})A^{(1)}_2\Big)}{6},\\
 B^{(1)}_2=\dfrac{4\ii A^{(2)}_2-(\sqrt{2}-\ii)A^{(1)}_3}{3\sqrt{2}},\quad
 B^{(1)}_3=1-\ii\sqrt{2}-\dfrac{(2-\ii\sqrt{2})A^{(1)}_2}{3},\\
 b^{(1)}_4=-\dfrac{(1+\ii\sqrt{2})A^{(1)}_3+\ii\sqrt{2} A^{(2)}_2}{3},\quad
 B^{(3)}_2=-\ii\sqrt{2}-\dfrac{(2-\ii\sqrt{2})A^{(1)}_2}{3},\\
 B^{(3)}_3=\dfrac{(1-2\ii\sqrt{2})A^{(2)}_3}{3}-\dfrac{a^{(3)}_4}{2},\quad
 b^{(3)}_4=-\dfrac{\ii a^{(3)}_4}{\sqrt{2}},\quad
 d_2=\dfrac{\ii}{\sqrt{2}},\quad
 d_3=-\dfrac{\ii}{\sqrt{2}}.
\end{cases}
\end{equation}
{\rm(Type III-3-8):}\quad
the condition \eqref{eqn:typeIII_cond3} and the following condition with \eqref{eqn:typeIII_def_ab}{\rm :}
\begin{equation}\tag{III-3-8}\label{eqn:typeIII_cond3_8}
\begin{cases}
 a^{(1)}_4=-\dfrac{(2+\ii\sqrt{2})A^{(1)}_3+2A^{(2)}_2}{3},\quad
 a^{(2)}_4=-\dfrac{\Big((2+\ii\sqrt{2})A^{(1)}_3+2A^{(2)}_2\Big)a^{(3)}_4}{12},\\
 b^{(1)}_4=-\dfrac{(1-\ii\sqrt{2})A^{(1)}_3-\sqrt{2}\ii A^{(2)}_2}{3},\quad
 b^{(3)}_4=\dfrac{\ii a^{(3)}_4}{\sqrt{2}},\quad
 d_2=-d_3=-\dfrac{\ii}{\sqrt{2}},\\
 B^{(1)}_2=-\dfrac{(\sqrt{2}+\ii)A^{(1)}_3+4\ii A^{(2)}_2}{3\sqrt{2}},\quad
 B^{(1)}_3=1+\ii\sqrt{2}-\dfrac{(2+\ii\sqrt{2})A^{(1)}_2}{3},\\
 A^{(3)}_2=-\dfrac{(2+\ii\sqrt{2})(2A^{(2)}_3-3a^{(3)}_4)}{6},\quad
 B^{(3)}_2=\ii\sqrt{2}-\dfrac{(2+\ii\sqrt{2})A^{(1)}_2}{3},\\
 A^{(3)}_3=-\dfrac{(2+\ii\sqrt{2})(3-(2+\ii\sqrt{2})A^{(1)}_2}{6},\quad
 B^{(3)}_3=\dfrac{(1+2\sqrt{2}\,\ii)A^{(2)}_3}{3}-\dfrac{a^{(3)}_4}{2}.
\end{cases}
\end{equation}
{\rm(Type III-3-9):}\quad
the condition \eqref{eqn:typeIII_cond3} and the following condition with \eqref{eqn:typeIII_def_ab}{\rm :}
\begin{equation}\tag{III-3-9}\label{eqn:typeIII_cond3_9}
\begin{cases}
 a^{(1)}_4=-\dfrac{(1-{d_2}^2)\Big(A^{(2)}_2(1-d_2d_3)+A^{(1)}_3d_2d_3(1+d_3)\Big)}{(1-{d_3}^2)(1-{d_2}^2{d_3}^2)},\\
 a^{(2)}_4=-\dfrac{\Big(A^{(2)}_3(1-d_2d_3)+A^{(3)}_2d_2d_3(1+d_2)\Big)\Big(A^{(2)}_2(1-d_2d_3)+A^{(1)}_3d_2d_3(1+d_3)\Big)}{(1-{d_2}^2)(1-{d_3}^2)},\\
 a^{(3)}_4=\dfrac{(1-{d_3}^2)\Big(A^{(2)}_3(1-d_2d_3)+A^{(3)}_2d_2(1+d_2)d_3\Big)}{(1-{d_2}^2)(1-{d_2}^2{d_3}^2)},\\
 b^{(1)}_4=\dfrac{d_3(1-{d_2}^2)\Big(A^{(2)}_2(1-d_2d_3)+A^{(1)}_3d_2d_3(1+d_3)\Big)}{(1-{d_3}^2)(1-{d_2}^2{d_3}^2)},\\
 b^{(3)}_4=-\dfrac{d_2(1-{d_3}^2)\Big(A^{(2)}_3(1-d_2d_3)+A^{(3)}_2d_2(1+d_2)d_3\Big)}{(1-{d_2}^2)(1-{d_2}^2{d_3}^2)},\\
 B^{(1)}_2=-\dfrac{A^{(1)}_3d_2d_3(1+d_3)-A^{(2)}_2(d_2-d_3)}{1-{d_3}^2},\\
 B^{(1)}_3=-\dfrac{(1-{d_2}^2)(1-{d_3}^2)}{({d_2}^2-{d_3}^2)(1-{d_2}^2{d_3}^2)}-\dfrac{d_2(1-{d_3}^2}{1-{d_2}^2{d_3}^2}+\dfrac{2(1-{d_3}^2}{{d_2}^2-{d_3}^2}-\dfrac{A^{(1)}_2(1-d_2d_3)}{d_2d_3(1+d_2)},\\
 B^{(3)}_2=\dfrac{(2+d_2)(1-d_2)}{{d_2}^2-{d_3}^2}+\dfrac{(1-d_2)\Big(d_2{d_3}^2(1-{d_3}^2)-(1-{d_2}^2)\Big)}{({d_2}^2-{d_3}^2)(1-{d_2}^2{d_3}^2)}-\dfrac{A^{(1)}_2(1-d_2d_3)}{d_2d_3(1+d_2)},\\
 A^{(3)}_3=-\dfrac{d_2d_3(1-d_2)(1-{d_3}^2)}{1-{d_2}^2{d_3}^2}+\dfrac{(1+d_3)A^{(1)}_2}{1+d_2},\\
 B^{(3)}_3=-\dfrac{A^{(2)}_3(d_2-d_3)+A^{(3)}_2d_2d_3(1+d_2)}{1-{d_2}^2},\quad
 1-3{d_2}^2{d_3}^2+{d_2}^4{d_3}^2+{d_2}^2{d_3}^4=0.
\end{cases}
\end{equation}
{\rm(Type III-3-10):}\quad
the condition \eqref{eqn:typeIII_cond3} and the following condition with \eqref{eqn:typeIII_def_ab}{\rm :}
\begin{equation}\tag{III-3-10}\label{eqn:typeIII_cond3_10}
\begin{cases}
 a^{(2)}_4=-\dfrac{(A^{(1)}_3+A^{(2)}_2+B^{(1)}_2)^2b^{(3)}_4}{A^{(1)}_3-2a^{(1)}_4+A^{(2)}_2+B^{(1)}_2},\quad
 a^{(3)}_4=-\dfrac{a^{(1)}_4b^{(3)}_4}{A^{(1)}_3-2a^{(1)}_4+A^{(2)}_2+B^{(1)}_2},\\
 b^{(1)}_4=-A^{(1)}_3+2a^{(1)}_4-A^{(2)}_2-B^{(1)}_2,\quad
 B^{(1)}_3=\dfrac{1-2A^{(1)}_2-2A^{(3)}_3}{2},\\
 B^{(3)}_2=-\dfrac{1+2A^{(1)}_2+2A^{(3)}_3}{2},\\
 B^{(3)}_3=-A^{(2)}_3-A^{(3)}_2+\dfrac{(A^{(1)}_3+A^{(2)}_2+B^{(1)}_2)b^{(3)}_4}{A^{(1)}_3-2a^{(1)}_4+A^{(2)}_2+B^{(1)}_2},\quad
 d_2=d_3=-1.
\end{cases}
\end{equation}
{\rm(Type III-3-11):}\quad
the condition \eqref{eqn:typeIII_cond3} and the following condition with \eqref{eqn:typeIII_def_ab}{\rm :}
\begin{equation}\tag{III-3-11}\label{eqn:typeIII_cond3_11}
\begin{cases}
 a^{(2)}_4=\dfrac{a^{(1)}_4a^{(3)}_4(A^{(1)}_3+A^{(2)}_2+B^{(1)}_2)^2}{(A^{(1)}_3-a^{(1)}_4+A^{(2)}_2+B^{(1)}_2+b^{(1)}_4)^2},\\
 b^{(3)}_4=a^{(3)}_4-\dfrac{a^{(1)}_4a^{(3)}_4(a^{(1)}_4-b^{(1)}_4)}{(A^{(1)}_3-a^{(1)}_4+A^{(2)}_2+B^{(1)}_2+b^{(1)}_4)^2},\\
 B^{(1)}_3=1-A^{(1)}_2-A^{(3)}_3-\dfrac{a^{(1)}_4}{A^{(1)}_3+A^{(2)}_2+B^{(1)}_2+b^{(1)}_4},\\
 B^{(3)}_2=-A^{(1)}_2-A^{(3)}_3-\dfrac{a^{(1)}_4}{A^{(1)}_3+A^{(2)}_2+B^{(1)}_2+b^{(1)}_4},\\
 B^{(3)}_3=-A^{(2)}_3-A^{(3)}_2-\dfrac{a^{(1)}_4a^{(3)}_4(A^{(1)}_3+A^{(2)}_2+B^{(1)}_2)}{(A^{(1)}_3-a^{(1)}_4+A^{(2)}_2+B^{(1)}_2+b^{(1)}_4)^2},\quad
 d_2=d_3=-1.
\end{cases}
\end{equation}
{\rm(Type III-3-12):}\quad
the condition \eqref{eqn:typeIII_cond3} and the following condition with \eqref{eqn:typeIII_def_ab}{\rm :}
\begin{equation}\tag{III-3-12}\label{eqn:typeIII_cond3_12}
\begin{cases}
 a^{(1)}_4=\dfrac{{d_2}^2(1-d_3)(1-{d_3}^2)\Big(A^{(2)}_2(1+d_2)-B^{(1)}_2(1+d_3)\Big)}{(1-{d_2}^2{d_3}^2)^2},\\
 a^{(2)}_4=\dfrac{b^{(3)}_4(1-{d_2}^2{d_3}^2)\Big(A^{(2)}_2(1+d_2)-B^{(1)}_2(1+d_3)\Big)}{d_2(1-d_3)(1+d_3)^2},\\
 a^{(3)}_4=\dfrac{b^{(3)}_4(1-{d_2}^2)^2{d_3}^2}{d_2(1-{d_3}^2)(1-{d_2}^2{d_3}^2)},\\
 b^{(1)}_4=\dfrac{d_3(1-{d_2}^2)\Big(A^{(2)}_2(1+d_2)-B^{(1)}_2(1+d_3)\Big)}{(1+d_3)(1-{d_2}^2{d_3}^2)},\\
 A^{(1)}_3=-\dfrac{B^{(1)}_2(1-{d_3}^2)-A^{(2)}_2(d_2-d_3)}{d_2d_3(1+d_3)},\\
 B^{(1)}_3=\dfrac{(1-d_2)(1+d_2{d_3}^2)}{1-{d_2}^2{d_3}^2}-\dfrac{A^{(1)}_2(1-d_2d_3)}{d_2d_3(1+d_2)},\\
 A^{(3)}_2=-\dfrac{b^{(3)}_4(1-d_2)(1-{d_2}^2{d_3}^2)}{{d_2}^2d_3(1-{d_3}^2)}-\dfrac{A^{(2)}_3(1-d_2d_3)}{d_2d_3(1+d_2)},\\
 B^{(3)}_2=-\dfrac{d_2(1-{d_3}^2)}{1-{d_2}^2{d_3}^2}-\dfrac{A^{(1)}_2(1-d_2d_3)}{d_2d_3(1+d_2)},\\
 A^{(3)}_3=-\dfrac{d_2d_3(1-d_2)(1-{d_3}^2)}{1-{d_2}^2{d_3}^2}+\dfrac{A^{(1)}_2(1+d_3)}{1+d_2},\\
 B^{(3)}_3=\dfrac{A^{(2)}_3(1+d_3)}{1+d_2}+\dfrac{b^{(3)}_4(1-{d_2}^2{d_3}^2)}{d_2(1-{d_3}^2)}.
\end{cases}
\end{equation}
{\rm(Type III-3-13):}\quad
the condition \eqref{eqn:typeIII_cond3} and the following condition with \eqref{eqn:typeIII_def_ab}{\rm :}
\begin{equation}\tag{III-3-13}\label{eqn:typeIII_cond3_13}
\begin{cases}
 a^{(1)}_4=-\dfrac{(A^{(2)}_2-B^{(1)}_2){d_2}^2(1+{d_2}^2)}{(1-d_2)^2(1+{d_2}^4)},\quad
 a^{(2)}_4=-\dfrac{(A^{(2)}_2-B^{(1)}_2)b^{(3)}_4(1+{d_2}^4)}{(1-d_2)^2d_2},\\
 a^{(3)}_4=\dfrac{b^{(3)}_4d_2(1+{d_2}^2)}{(1-{d_2}^2)^2},\quad
 b^{(1)}_4=-\dfrac{(A^{(2)}_2-B^{(1)}_2)d_2(1+d_2)^2}{1+{d_2}^4},\\
 A^{(1)}_3=-\dfrac{B^{(1)}_2(1+{d_2}^2)-2A^{(2)}_2d_2}{(1-d_2){d_2}^2},\quad
 B^{(3)}_3=A^{(2)}_3-\dfrac{b^{(3)}_4(1+{d_2}^4)}{d_2(1+d_2)^2},\\
 B^{(1)}_3=\dfrac{1}{(1+d_2)^2}+\dfrac{d_2}{1+d_2}-\dfrac{A^{(1)}_2(1-d_2)}{{d_2}^2},\\
 A^{(3)}_2=\dfrac{b^{(3)}_4(1-{d_2}^8)}{{d_2}^3(1+d_2)(1-{d_2}^2)^2}-\dfrac{A^{(2)}_3d_2(1-{d_2})}{{d_2}^3},\quad
 d_3=d_2,\\
 B^{(3)}_2=-\dfrac{d_2}{(1+d_2)^2}-\dfrac{A^{(1)}_2(1-d_2)}{{d_2}^2},\quad
 A^{(3)}_3=A^{(1)}_2-\dfrac{{d_2}^2(1+{d_2}^2)}{(1-d_2)(1+d_2)^2}.
\end{cases}
\end{equation}
{\rm(Type III-3-14):}\quad
the condition \eqref{eqn:typeIII_cond3} and the following condition with \eqref{eqn:typeIII_def_ab}{\rm :}
\begin{equation}\tag{III-3-14}\label{eqn:typeIII_cond3_14}
\begin{cases}
 a^{(1)}_4=\dfrac{{d_2}^2(1+{d_2}^2)\Big(B^{(1)}_2(1-d_2)-A^{(2)}_2(1+d_2)\Big)}{(1-d_2)(1+d_2)^2(1+{d_2}^4)},\\
 a^{(2)}_4=\dfrac{b^{(3)}_4(1+{d_2}^4)\Big(B^{(1)}_2(1-d_2)-A^{(2)}_2(1+d_2)\Big)}{d_2(1-d_2)(1+d_2)^2},\\
 a^{(3)}_4=\dfrac{b^{(3)}_4d_2(1+{d_2}^2)}{(1-{d_2}^2)^2},\quad
 b^{(1)}_4=-\dfrac{d_2(1-d_2)\Big(B^{(1)}_2(1-d_2)-A^{(2)}_2(1+d_2)\Big)}{1+{d_2}^4},\\
 A^{(1)}_3=\dfrac{B^{(1)}_2(1+{d_2}^2)}{{d_2}^2(1+d_2)},\quad
 B^{(1)}_3=\dfrac{{d_2}^2(1-{d_2}^3)+A^{(1)}_2(1-{d_2}^4)}{{d_2}^2(1+d_2)^2(1-d_2)},\\
 A^{(3)}_2=\dfrac{(1+{d_2}^2)\Big(A^{(2)}_3d_2(1-{d_2}^2)-b^{(3)}_4(1+{d_2}^4)\Big)}{{d_2}^3(1-d_2)(1+d_2)^2},\\
 B^{(3)}_2=\dfrac{A^{(1)}_2(1+{d_2}^2)}{{d_2}^2(1+d_2)}-\dfrac{d_2}{(1+d_2)^2},\
 B^{(3)}_3=\dfrac{A^{(2)}_3d_2(1-{d_2}^2)-b^{(3)}_4(1+{d_2}^4)}{d_2(1+d_2)^2},\\
 A^{(3)}_3=\dfrac{{d_2}^2(1+{d_2}^2)}{(1-d_2)(1+d_2)^2}+\dfrac{A^{(1)}_2(1-d_2)}{1+d_2},\quad
 d_3=-d_2.
\end{cases}
\end{equation}
{\rm(Type III-3-15):}\quad
the condition \eqref{eqn:typeIII_cond3} and the following condition with \eqref{eqn:typeIII_def_ab}{\rm :}
\begin{equation}\tag{III-3-15}\label{eqn:typeIII_cond3_15}
\begin{cases}
 a^{(1)}_4=-\dfrac{{d_2}^2\Big(A^{(2)}_2(1+d_2)-B^{(1)}_2(1+d_3)\Big)}{(1-d_3)(1+{d_2}^2{d_3}^2)},\\
 a^{(2)}_4=-\dfrac{b^{(3)}_4(1+{d_2}^2{d_3}^2)\Big(A^{(2)}_2(1+d_2)-B^{(1)}_2(1+d_3)\Big)}{d_2(1-d_3)^2(1+d_3)},\\
 a^{(3)}_4=\dfrac{b^{(3)}_4{d_3}^2}{d_2(1-{d_3}^2)},\quad
 b^{(1)}_4=-\dfrac{d_3(1-{d_2}^2)\Big(A^{(2)}_2(1+d_2)-B^{(1)}_2(1+d_3)\Big)}{(1-d_3)(1+{d_2}^2{d_3}^2)},\\
 A^{(1)}_3=\dfrac{A^{(2)}_2(d_2+d_3)-B^{(1)}_2(1+{d_3}^2)}{d_2d_3(1-d_3)},\\
 B^{(1)}_3=\dfrac{1}{(1+d_2)(1+{d_2}^2{d_3}^2)}+\dfrac{d_2{d_3}^2(d_2+{d_3}^2-2d_2{d_3}^2)}{({d_2}^2-{d_3}^2)(1+
{d_2}^2{d_3}^2)}\\
 \hspace{2.5em}+\dfrac{{d_2}^4(1-{d_3}^2)}{(1+d_2)({d_2}^2-{d_3}^2)(1+{d_2}^2{d_3}^2)}-\dfrac{A^{(1)}_2(1-d_2d_3)}{d_2d_3(1+d_2)},\\
 A^{(3)}_2=-\dfrac{A^{(2)}_3(1-d_2d_3)}{d_2d_3(1+d_2)}+\dfrac{b^{(3)}_4(1+{d_2}^2)(1+{d_2}^2{d_3}^2)}{{d_2}^2d_3(1+d_2)(1-{d_3}^2)},\\
 B^{(3)}_2=-\dfrac{(d_2(1-d_2)}{(1+d_2)(1+{d_2}^2{d_3}^2)}+\dfrac{d_2{d_3}^2(1-d_2)(2d_2+{d_2}^2+{d_3}^2)}{({d_2}^2-{d_3}^2)(1+{d_2}^2{d_3}^2)}\\
 \hspace{2.5em}-\dfrac{A^{(1)}_2(1-d_2d_3)}{d_2d_3(1+d_2)},\\
 A^{(3)}_3=\dfrac{A^{(1)}_2(1+d_3)}{1+d_2}-\dfrac{d_2d_3(1+{d_2}^2+{d_3}^2-{d_2}^2{d_3}^2)}{(1+d_2)(1+{d_2}^2{d_3}^2)},\\
 B^{(3)}_3=\dfrac{A^{(2)}_3(1+d_3)}{1+d_2}-\dfrac{b^{(3)}_4(1-d_2)(1+{d_2}^2{d_3}^2)}{d_2(1+d_2)(1-{d_3}^2)},\quad
 {d_2}^2+{d_3}^2=-2.
\end{cases}
\end{equation}
{\rm(Type III-3-16):}\quad
the condition \eqref{eqn:typeIII_cond3} and the following condition with \eqref{eqn:typeIII_def_ab}{\rm :}
\begin{equation}\tag{III-3-16}\label{eqn:typeIII_cond3_16}
\begin{cases}
 a^{(1)}_4=-\dfrac{{d_2}^2(1+{d_3}^2)^2\Big(A^{(2)}_2(1+d_2)-B^{(1)}_2(1+d_3)\Big)}{(1-d_3)(1-{d_2}^4{d_3}^4)},\\
 a^{(2)}_4=-\dfrac{b^{(3)}_4(1+{d_2}^2{d_3}^2)\Big(A^{(2)}_2(1+d_2)-B^{(1)}_2(1+d_3)\Big)}{d_2(1-d_3)^2(1+d_3)},\\
 a^{(3)}_4=\dfrac{b^{(3)}_4{d_3}^2(1+{d_2}^2)^2}{d_2(1-{d_3}^2)(1-{d_2}^2{d_3}^2)},\\
 b^{(1)}_4=-\dfrac{d_3(1-{d_2}^2)\Big(A^{(2)}_2(1+d_2)-B^{(1)}_2(1+d_3)\Big)}{(1-d_3)(1+{d_2}^2{d_3}^2)},\\
 A^{(1)}_3=\dfrac{A^{(2)}_2(d_2+d_3)-B^{(1)}_2(1+{d_3}^2)}{d_2d_3(1-d_3)},\\
 B^{(1)}_3=-\dfrac{A^{(1)}_2(1-d_2d_3)}{d_2d_3(1+d_2)}+\dfrac{(1+{d_2}^2)(1-d_2{d_3}^2)}{(1+d_2)(1-{d_2}^2{d_3}^2)},\\
 A^{(3)}_2=-\dfrac{A^{(2)}_3(1-d_2d_3)}{d_2d_3(1+d_2)}+\dfrac{b^{(3)}_4(1+{d_2}^2)(1+{d_2}^2{d_3}^2)}{{d_2}^2d_3(1+d_2)(1-{d_3}^2)},\\
 B^{(3)}_2=-\dfrac{A^{(1)}_2(1-d_2d_3)}{d_2d_3(1+d_2)}-\dfrac{d_2(1-d_2)(1+{d_3}^2)}{(1+d_2)(1-{d_2}^2{d_3}^2)},\\
 A^{(3)}_3=-\dfrac{d_2d_3(1+{d_2}^2)(1+{d_3}^2)}{(1+d_2)(1-{d_2}^2{d_3}^2)}+\dfrac{A^{(1)}_2(1+d_3)}{1+d_2},\\
 B^{(3)}_3=\dfrac{A^{(2)}_3(1+d_3)}{1+d_2}-\dfrac{b^{(3)}_4(1-d_2)(1+{d_2}^2{d_3}^2)}{d_2(1+d_2)(1-{d_3}^2)}.
\end{cases}
\end{equation}
\section{Proofs of Lemmas \ref{lemma:classification_CACO_CO1CO2CO3_1}--\ref{lemma:classification_CACO_CO1CO2CO3_3}}
\label{section:proof_lemma_CACO_CO1CO2CO3}
In this section, we give the proofs of Lemmas \ref{lemma:classification_CACO_CO1CO2CO3_1}--\ref{lemma:classification_CACO_CO1CO2CO3_3}.

Let $\{Q_1,\dots,Q_9\}$ be a system of quad-equations satisfying the properties \ref{cond:CO1}, \ref{cond:CO2} and \ref{cond:CO3} in Condition \ref{condition:CO}. 
Then, from the property \ref{cond:CO3} and Lemma \ref{lemma:classification_Q123456789_1}, 
the quad-equations $\{Q_1,\dots,Q_9\}$ are given by \eqref{eqn:Q1_P1_CO3}--\eqref{eqn:Q9_P3_CO3},
where the set of polynomials $\{P_1,P_2,P_3\}$ is given by one of \eqref{eqn:P123_caseN3_1}--\eqref{eqn:P123_caseN1}.
Note that without saying separately
the conditions of parameters are assumed so as to $Q_i$, $i=1,\dots,9$, be quad equations.
Because of the CAO property of the octahedron with quad-equations $\{Q_7,Q_8,Q_9\}$, 
when $\{P_1,P_2,P_3\}$ is given by \eqref{eqn:P123_caseN3_1} or \eqref{eqn:P123_caseN3_3}, we obtain the condition \eqref{eqn:caseN3_CAO_B2}
and when $\{P_1,P_2,P_3\}$ is given by \eqref{eqn:P123_caseN1}, we obtain the condition \eqref{eqn:caseN1_CAO_B2}.

Let $u_i$, $i=1,2,3,4$, and $v_1$ be initial values. 
Then, there are two ways to obtain the value of $v_4=v_4(u_1,u_2,u_3,u_4,v_1)$ (see Remark \ref{remark:CACO_2D}).
In the following, we denote the value of $v_4$ given by the clockwise manner by $v_4^{(r)}=v_4^{(r)}(u_1,u_2,u_3,u_4,v_1)$
and the value of $v_4$ given by the anti-clockwise manner by $v_4^{(l)}=v_4^{(l)}(u_1,u_2,u_3,u_4,v_1)$.
In a similar manner with $v_4^{(r)}$ and $v_4^{(l)}$,
we also define 
$v_5^{(r)}=v_5^{(r)}(u_2,u_3,u_4,u_5,v_2)$,
$v_5^{(l)}=v_5^{(l)}(u_2,u_3,u_4,u_5,v_2)$,
$v_6^{(r)}=v_6^{(r)}(u_3,u_4,u_5,u_6,v_3)$ and
$v_6^{(l)}=v_6^{(l)}(u_3,u_4,u_5,u_6,v_3)$.
The CACO property is none other than 
\begin{equation}
 v_4^{(r)}=v_4^{(l)},\quad
 v_5^{(r)}=v_5^{(l)},\quad
 \text{or}\quad
 v_6^{(r)}=v_6^{(l)}.
\end{equation}

\subsection{Proof of Lemma \ref{lemma:classification_CACO_CO1CO2CO3_1}}
\label{subsection:caseN3_1}
In this section, we consider quad-equations \eqref{eqn:Q1_P1_CO3}--\eqref{eqn:Q9_P3_CO3} with \eqref{eqn:P123_caseN3_1}.
Note that in this setting, the transformations $s_{12}$, $s_{23}$, $s_{13}$ and $\iota$ defined in \eqref{eqn:Gco} keep invariant the system of quad-equations $\{Q_1,\dots,Q_9\}$ under the appropriate transformations of parameters (see Remark \ref{remark:symmetry_GCO}), e.g.
\begin{equation}
\begin{split}
 s_{12}:~&
 A^{(1)}_3\leftrightarrow A^{(1)}_4,\quad
 A^{(2)}_1\leftrightarrow A^{(3)}_1,\quad
 A^{(2)}_2\leftrightarrow A^{(3)}_2,\quad
 A^{(2)}_3\leftrightarrow A^{(3)}_4,\quad
 A^{(2)}_4\leftrightarrow A^{(3)}_3,\\
 &B^{(1)}_3\leftrightarrow B^{(1)}_4,\quad
 B^{(2)}_1\leftrightarrow B^{(3)}_1,\quad
 B^{(2)}_2\leftrightarrow B^{(3)}_2,\quad
 B^{(2)}_3\leftrightarrow B^{(3)}_4,\quad
 B^{(2)}_4\leftrightarrow B^{(3)}_3.
\end{split}
\end{equation}

In this setting, $v_4^{(r)}$ and $v_4^{(l)}$ are given by
\begin{equation}\label{eqn:case1_v4rl}
 v_4^{(r)}=-\dfrac{u_4(u_4h_1+h_2v_1)}{u_4h_3+h_4v_1},\quad
 v_4^{(l)}=-\dfrac{u_1(u_1h_4+h_2v_1)}{u_1h_3+h_1v_1},
\end{equation}
where
\begin{subequations}
\begin{align}
 &h_1=A^{(1)}_2(B^{(1)}_2u_1+B^{(1)}_3u_4)
 \left(A^{(2)}_3A^{(3)}_4(B^{(3)}_3u_1+B^{(3)}_1u_4)+A^{(2)}_2A^{(3)}_2(B^{(3)}_2u_1+B^{(3)}_4u_4)\right)\notag\\
 &\quad+A^{(1)}_3(B^{(1)}_4u_1+B^{(1)}_1u_4)
 \left(A^{(2)}_1A^{(3)}_4(B^{(3)}_3u_1+B^{(3)}_1u_4)+A^{(2)}_4A^{(3)}_2(B^{(3)}_2u_1+B^{(3)}_4u_4)\right),\\
 &h_2=A^{(1)}_4(B^{(1)}_2u_1+B^{(1)}_3u_4)
 \left(A^{(2)}_3A^{(3)}_4(B^{(3)}_3u_1+B^{(3)}_1u_4)+A^{(2)}_2A^{(3)}_2(B^{(3)}_2u_1+B^{(3)}_4u_4)\right)\notag\\
 &\quad+A^{(1)}_1(B^{(1)}_4u_1+B^{(1)}_1u_4)
 \left(A^{(2)}_1A^{(3)}_4(B^{(3)}_3u_1+B^{(3)}_1u_4)+A^{(2)}_4A^{(3)}_2(B^{(3)}_2u_1+B^{(3)}_4u_4)\right),\\
 &h_3=A^{(1)}_2(B^{(1)}_2u_1+B^{(1)}_3u_4)
 \left(A^{(2)}_3A^{(3)}_1(B^{(3)}_3u_1+B^{(3)}_1u_4)+A^{(2)}_2A^{(3)}_3(B^{(3)}_2u_1+B^{(3)}_4u_4)\right)\notag\\
 &\quad+A^{(1)}_3(B^{(1)}_4u_1+B^{(1)}_1u_4)
 \left(A^{(2)}_1A^{(3)}_1(B^{(3)}_3u_1+B^{(3)}_1u_4)+A^{(2)}_4A^{(3)}_3(B^{(3)}_2u_1+B^{(3)}_4u_4)\right),\\
 &h_4=A^{(1)}_4(B^{(1)}_2u_1+B^{(1)}_3u_4)
 \left(A^{(2)}_3A^{(3)}_1(B^{(3)}_3u_1+B^{(3)}_1u_4)+A^{(2)}_2A^{(3)}_3(B^{(3)}_2u_1+B^{(3)}_4u_4)\right)\notag\\
 &\quad+A^{(1)}_1(B^{(1)}_4u_1+B^{(1)}_1u_4)
 \left(A^{(2)}_1A^{(3)}_1(B^{(3)}_3u_1+B^{(3)}_1u_4)+A^{(2)}_4A^{(3)}_3(B^{(3)}_2u_1+B^{(3)}_4u_4)\right).
\end{align}
\end{subequations}
Note that $h_i$, $i=1,\dots,4$, do not depend on $v_1$.
The condition $v_4^{(r)}=v_4^{(l)}$ gives
\begin{equation}
 (h_4u_1-h_1u_4)\left(h_3u_1u_4+(h_4u_1+h_1u_4)v_1+h_2{v_1}^2\right)=0.
\end{equation}
Then, we obtain the following two cases:
\begin{subequations}
\begin{align}
 &h_4u_1-h_1u_4=0;\label{eqn:case1_h_cond1}\\
 &h_2=h_3=0,\quad h_4u_1+h_1u_4=0.\label{eqn:case1_h_cond2}
\end{align}
\end{subequations}
Under the condition \eqref{eqn:case1_h_cond1} we obtain the following square equation $K_1$:
\begin{equation}
 (h_3 u_1+h_1 v_1)v_4+u_1(h_1 u_4+h_2 v_1)=0,
\end{equation}
and under the condition \eqref{eqn:case1_h_cond2} we obtain the following square equation $K_1$:
\begin{equation}
 v_1v_4-u_1u_4=0.
\end{equation}
Because of the symmetry $\iota\circ s_{12}\circ s_{23}$ whose action on the variables $u_i$ and $v_j$ is given by
\begin{equation}
 \iota\circ s_{12}\circ s_{23}:u_i\mapsto u_{i+1},\quad v_j\mapsto v_{j+1},
\end{equation}
where $i,j\in\bbZ/(6\bbZ)$,
it is obvious that in a similar manner as the case $v_4^{(r)}=v_4^{(r)}$, 
under the necessary and sufficient condition of $v_5^{(r)}=v_5^{(r)}$ the square equation $K_2$ exists,
and under the necessary and sufficient condition of $v_6^{(r)}=v_6^{(r)}$ the square equation $K_3$ exists.
Moreover, since the necessary and sufficient condition of the CACO property is given by \eqref{eqn:case1_h_cond1} or \eqref{eqn:case1_h_cond2}, under the condition \eqref{eqn:case1_h_cond1} or \eqref{eqn:case1_h_cond2} the square property holds.

Next, we consider the conditions of the parameters $A^{(i)}_j$ and $B^{(i)}_j$ for the conditions \eqref{eqn:case1_h_cond1} and \eqref{eqn:case1_h_cond2}.
We can easily verify that the condition \eqref{eqn:case1_h_cond1} is equivalent to the condition \eqref{eqn:typeI_cond1}.
Therefore, we here consider the condition \eqref{eqn:case1_h_cond2}.
The condition \eqref{eqn:case1_h_cond2} is equivalent to
\begin{equation}\label{eqn:proof_case1_2_C}
 C^{(i)}=0,\quad i=1,\dots,10,
\end{equation}
where
\begin{align*}
 C^{(1)}=&B^{(1)}_1A^{(1)}_1(B^{(3)}_1A^{(3)}_4A^{(2)}_1+A^{(3)}_2B^{(3)}_4A^{(2)}_4)
 +B^{(1)}_3A^{(1)}_4(A^{(3)}_2B^{(3)}_4A^{(2)}_2+B^{(3)}_1A^{(3)}_4A^{(2)}_3),\\
 C^{(2)}=&A^{(1)}_1B^{(1)}_4(B^{(3)}_3A^{(3)}_4A^{(2)}_1+B^{(3)}_2A^{(3)}_2A^{(2)}_4)
 +B^{(1)}_2A^{(1)}_4(B^{(3)}_2A^{(3)}_2A^{(2)}_2+B^{(3)}_3A^{(3)}_4A^{(2)}_3),\\
 C^{(3)}=&B^{(1)}_1A^{(1)}_3(B^{(3)}_1A^{(3)}_1A^{(2)}_1+A^{(3)}_3B^{(3)}_4A^{(2)}_4)
 +A^{(1)}_2B^{(1)}_3(A^{(3)}_3B^{(3)}_4A^{(2)}_2+B^{(3)}_1A^{(3)}_1A^{(2)}_3),\\
 C^{(4)}=&A^{(1)}_3B^{(1)}_4(A^{(3)}_1B^{(3)}_3A^{(2)}_1+B^{(3)}_2A^{(3)}_3A^{(2)}_4)
 +B^{(1)}_2A^{(1)}_2(B^{(3)}_2A^{(3)}_3A^{(2)}_2+A^{(3)}_1B^{(3)}_3A^{(2)}_3),\\
 C^{(5)}=&B^{(1)}_1A^{(1)}_3(B^{(3)}_1A^{(3)}_4A^{(2)}_1+A^{(3)}_2B^{(3)}_4A^{(2)}_4)
 +A^{(1)}_2B^{(1)}_3(A^{(3)}_2B^{(3)}_4A^{(2)}_2+B^{(3)}_1A^{(3)}_4A^{(2)}_3),\\
 C^{(6)}=&A^{(1)}_1B^{(1)}_4(A^{(3)}_1B^{(3)}_3A^{(2)}_1+B^{(3)}_2A^{(3)}_3A^{(2)}_4)
 +B^{(1)}_2A^{(1)}_4(B^{(3)}_2A^{(3)}_3A^{(2)}_2+A^{(3)}_1B^{(3)}_3A^{(2)}_3),\\
 C^{(7)}=&(B^{(1)}_4B^{(3)}_1+B^{(1)}_1B^{(3)}_3)A^{(1)}_1A^{(3)}_4A^{(2)}_1
 +(B^{(1)}_3B^{(3)}_2+B^{(1)}_2B^{(3)}_4)A^{(1)}_4A^{(3)}_2A^{(2)}_2\notag\\
 &+(B^{(1)}_2B^{(3)}_1+B^{(1)}_3B^{(3)}_3)A^{(1)}_4A^{(3)}_4A^{(2)}_3
 +(B^{(1)}_1B^{(3)}_2+B^{(1)}_4B^{(3)}_4)A^{(1)}_1A^{(3)}_2A^{(2)}_4,\\
 C^{(8)}=&(B^{(1)}_4B^{(3)}_1+B^{(1)}_1B^{(3)}_3)A^{(1)}_3A^{(3)}_1A^{(2)}_1
 +(B^{(1)}_3B^{(3)}_2+B^{(1)}_2B^{(3)}_4)A^{(1)}_2A^{(3)}_3A^{(2)}_2\notag\\
 &+(B^{(1)}_2B^{(3)}_1+B^{(1)}_3B^{(3)}_3)A^{(1)}_2A^{(3)}_1A^{(2)}_3
 +(B^{(1)}_1B^{(3)}_2+B^{(1)}_4B^{(3)}_4)A^{(1)}_3A^{(3)}_3A^{(2)}_4,\\
 C^{(9)}=&(B^{(1)}_1A^{(1)}_1B^{(3)}_1A^{(3)}_1+A^{(1)}_3B^{(1)}_4B^{(3)}_1A^{(3)}_4+B^{(1)}_1A^{(1)}_3B^{(3)}_3A^{(3)}_4)A^{(2)}_1\notag\\
 &+(A^{(1)}_2B^{(1)}_3B^{(3)}_2A^{(3)}_2+B^{(1)}_2A^{(1)}_2A^{(3)}_2B^{(3)}_4+B^{(1)}_3A^{(1)}_4A^{(3)}_3B^{(3)}_4)A^{(2)}_2\notag\\
 &+(B^{(1)}_3A^{(1)}_4B^{(3)}_1A^{(3)}_1+B^{(1)}_2A^{(1)}_2B^{(3)}_1A^{(3)}_4+A^{(1)}_2B^{(1)}_3B^{(3)}_3A^{(3)}_4)A^{(2)}_3\notag\\
 &+(B^{(1)}_1A^{(1)}_3B^{(3)}_2A^{(3)}_2+A^{(1)}_3B^{(1)}_4A^{(3)}_2B^{(3)}_4+B^{(1)}_1A^{(1)}_1A^{(3)}_3B^{(3)}_4)A^{(2)}_4,\\
 C^{(10)}=&(A^{(1)}_1B^{(1)}_4B^{(3)}_1A^{(3)}_1+B^{(1)}_1A^{(1)}_1A^{(3)}_1B^{(3)}_3+A^{(1)}_3B^{(1)}_4B^{(3)}_3A^{(3)}_4)A^{(2)}_1\notag\\
 &+(B^{(1)}_2A^{(1)}_2B^{(3)}_2A^{(3)}_2+B^{(1)}_3A^{(1)}_4B^{(3)}_2A^{(3)}_3+B^{(1)}_2A^{(1)}_4A^{(3)}_3B^{(3)}_4)A^{(2)}_2\notag\\
 &+(B^{(1)}_2A^{(1)}_4B^{(3)}_1A^{(3)}_1+B^{(1)}_3A^{(1)}_4A^{(3)}_1B^{(3)}_3+B^{(1)}_2A^{(1)}_2B^{(3)}_3A^{(3)}_4)A^{(2)}_3\notag\\
 &+(A^{(1)}_3B^{(1)}_4B^{(3)}_2A^{(3)}_2+B^{(1)}_1A^{(1)}_1B^{(3)}_2A^{(3)}_3+A^{(1)}_1B^{(1)}_4A^{(3)}_3B^{(3)}_4)A^{(2)}_4.
\end{align*}

For a set of four equations $x_j=0$, $i=1,\dots,4$, which have the following forms:
\begin{equation}
 x_i=c_{i1}A^{(2)}_1+c_{i2}A^{(2)}_2+c_{i3}A^{(2)}_3+c_{i4}A^{(2)}_4,
\end{equation}
where $c_{ij}$ do not depend of $A^{(2)}_i$, $i=1,\dots,4$,
we introduce a bracket $[x_1,x_2,x_3,x_4]_{A^{(2)}}$ as the following:
\begin{equation}
 [x_1,x_2,x_3,x_4]_{A^{(2)}}
 =\mathrm{det}(c_{ij})_{i,j=1}^4.
\end{equation}
Note that since $(A^{(2)}_1,A^{(2)}_2,A^{(2)}_3,A^{(2)}_4)\neq(0,0,0,0)$, the following relation holds:
\begin{equation}
 [x_1,x_2,x_3,x_4]_{A^{(2)}}=0.
\end{equation}
The equations
\begin{subequations}
\begin{align}
 &[C^{(1)},C^{(3)},C^{(5)},C^{(9)}]_{A^{(2)}}=0,
 &&[C^{(1)},C^{(5)},C^{(7)},C^{(9)}]_{A^{(2)}}=0,\\
 &[C^{(2)},C^{(4)},C^{(6)},C^{(10)}]_{A^{(2)}}=0,
 &&[C^{(2)},C^{(6)},C^{(7)},C^{(10)}]_{A^{(2)}}=0,\\
 &[C^{(3)},C^{(5)},C^{(8)},C^{(9)}]_{A^{(2)}}=0,
 &&[C^{(4)},C^{(6)},C^{(8)},C^{(10)}]_{A^{(2)}}=0,
\end{align}
\end{subequations}
give the conditions
\begin{subequations}
\begin{align}
 &B^{(1)}_1B^{(1)}_3B^{(3)}_1B^{(3)}_4=0,\quad
 A^{(3)}_2A^{(3)}_4B^{(1)}_1B^{(1)}_3=0,\quad
 B^{(1)}_2B^{(1)}_4B^{(3)}_2B^{(3)}_3=0,\\
 &A^{(1)}_1A^{(1)}_4B^{(3)}_2B^{(3)}_3=0,\quad
 A^{(1)}_2A^{(1)}_3B^{(3)}_1B^{(3)}_4=0,\quad
 A^{(3)}_1A^{(3)}_3B^{(1)}_2B^{(1)}_4=0,
\end{align}
\end{subequations}
respectively.
Therefore, we obtain the following six cases:
\begin{subequations}
\begin{align}
 &A^{(1)}_1=B^{(1)}_1=A^{(1)}_2=B^{(1)}_2=0;\label{eqn:theo_proof_case1_21}\\
 &A^{(1)}_3=B^{(1)}_3=A^{(1)}_4=B^{(1)}_4=0;\label{eqn:theo_proof_case1_22}\\
 &A^{(3)}_1=B^{(3)}_1=A^{(3)}_2=B^{(3)}_2=0;\label{eqn:theo_proof_case1_23}\\
 &A^{(3)}_3=B^{(3)}_3=A^{(3)}_4=B^{(3)}_4=0;\label{eqn:theo_proof_case1_24}\\
 &A^{(1)}_3=B^{(1)}_3=A^{(3)}_3=B^{(3)}_3=0;\label{eqn:theo_proof_case1_25}\\
 &A^{(1)}_4=B^{(1)}_4=A^{(3)}_4=B^{(3)}_4=0.\label{eqn:theo_proof_case1_26}
\end{align}
\end{subequations}
Considering the condition \eqref{eqn:theo_proof_case1_21} and 
considering the condition \eqref{eqn:theo_proof_case1_23}
mean the same under the symmetry $s_{23}$ whose action on the parameters is given by
\begin{equation}
\begin{split}
 s_{23}:~&
 A^{(1)}_1\leftrightarrow A^{(3)}_1,\quad
 A^{(1)}_2\leftrightarrow A^{(3)}_2,\quad
 A^{(1)}_3\leftrightarrow A^{(3)}_4,\quad
 A^{(1)}_4\leftrightarrow A^{(3)}_3,\quad
 A^{(2)}_3\leftrightarrow A^{(2)}_4,\\
 &B^{(1)}_1\leftrightarrow B^{(3)}_1,\quad
 B^{(1)}_2\leftrightarrow B^{(3)}_2,\quad
 B^{(1)}_3\leftrightarrow B^{(3)}_4,\quad
 B^{(1)}_4\leftrightarrow B^{(3)}_3,\quad
 B^{(2)}_3\leftrightarrow B^{(2)}_4.
\end{split}
\end{equation}
Similarly, considering the condition \eqref{eqn:theo_proof_case1_22} and 
considering the condition \eqref{eqn:theo_proof_case1_24}
mean the same under the symmetry $s_{23}$,
and 
considering the condition \eqref{eqn:theo_proof_case1_25} and 
considering the condition \eqref{eqn:theo_proof_case1_26}
mean the same under the symmetry $\iota$ whose action on the parameters is given by
\begin{equation}
\begin{split}
\iota:~&
 A^{(1)}_1\leftrightarrow A^{(1)}_2,\quad
 A^{(1)}_3\leftrightarrow A^{(1)}_4,\quad
 A^{(2)}_1\leftrightarrow A^{(2)}_2,\quad
 A^{(2)}_3\leftrightarrow A^{(2)}_4,\quad
 A^{(3)}_1\leftrightarrow A^{(3)}_2,\\
 &A^{(3)}_3\leftrightarrow A^{(3)}_4,\quad
 B^{(1)}_1\leftrightarrow B^{(1)}_2,\quad
 B^{(1)}_3\leftrightarrow B^{(1)}_4,\quad
 B^{(2)}_1\leftrightarrow B^{(2)}_2,\quad
 B^{(2)}_3\leftrightarrow B^{(2)}_4,\\\
 &B^{(3)}_1\leftrightarrow B^{(3)}_2,\quad
 B^{(3)}_3\leftrightarrow B^{(3)}_4.
\end{split}
\end{equation}
Therefore, it is sufficient to consider the conditions \eqref{eqn:theo_proof_case1_21}, \eqref{eqn:theo_proof_case1_22} and \eqref{eqn:theo_proof_case1_25}.
Below, we show that the condition \eqref{eqn:theo_proof_case1_25} gives the condition \eqref{eqn:typeI_cond2} and the others are inadmissible cases.

Firstly, we assume \eqref{eqn:theo_proof_case1_21}.
Then, the equations
\begin{equation}
 [C^{(1)},C^{(7)},C^{(9)},C^{(10)}]_{A^{(2)}}=0,\quad
 [C^{(4)},C^{(8)},C^{(9)},C^{(10)}]_{A^{(2)}}=0,
\end{equation}
give the conditions
\begin{equation}
 A^{(1)}_3 B^{(1)}_3 A^{(1)}_4 B^{(1)}_4 A^{(3)}_2 A^{(3)}_4=0,\quad
 A^{(1)}_3 B^{(1)}_3 A^{(1)}_4 B^{(1)}_4 A^{(3)}_1 A^{(3)}_3=0,
\end{equation}
respectively.
Hence, we obtain the following two cases:
\begin{subequations}
\begin{align}
 &A^{(3)}_1=B^{(3)}_1=A^{(3)}_2=B^{(3)}_2=0;\\
 &A^{(3)}_3=B^{(3)}_3=A^{(3)}_4=B^{(3)}_4=0.
\end{align}
\end{subequations}
If $A^{(3)}_1=B^{(3)}_1=A^{(3)}_2=B^{(3)}_2=0$, then we obtain
\begin{equation}
 C^{(7)}=A^{(1)}_4 A^{(2)}_3 A^{(3)}_4 B^{(1)}_3 B^{(3)}_3,\quad
 C^{(9)}=A^{(1)}_4 A^{(2)}_2 A^{(3)}_3 B^{(1)}_3 B^{(3)}_4,
\end{equation}
and if $A^{(3)}_3=B^{(3)}_3=A^{(3)}_4=B^{(3)}_4=0$, then we obtain
\begin{equation}
 C^{(7)}=A^{(1)}_4 A^{(2)}_2 A^{(3)}_2 B^{(1)}_3 B^{(3)}_2,\quad
 C^{(9)}=A^{(1)}_4 A^{(2)}_3 A^{(3)}_1 B^{(1)}_3 B^{(3)}_1.
\end{equation}
In the both cases $C^{(7)}$ and $C^{(9)}$ cannot be simultaneously zero.
Therefore, the case \eqref{eqn:theo_proof_case1_21} is inadequate.

Secondly, we assume \eqref{eqn:theo_proof_case1_22}.
From the equations
\begin{equation}
 [C^{(1)},C^{(7)},C^{(9)},C^{(10)}]_{A^{(2)}}=0,\quad
 [C^{(4)},C^{(8)},C^{(9)},C^{(10)}]_{A^{(2)}}=0,
\end{equation}
we obtain the conditions
\begin{equation}\label{eqn:no_case1_22}
 A^{(1)}_1B^{(1)}_1A^{(1)}_2B^{(1)}_2A^{(3)}_2A^{(3)}_4=0,\quad
 A^{(1)}_1B^{(1)}_1A^{(1)}_2B^{(1)}_2A^{(3)}_1A^{(3)}_3=0.
\end{equation}
In a similar manner as the case \eqref{eqn:theo_proof_case1_21}, we can show that 
under the condition \eqref{eqn:no_case1_22}, $C^{(7)}$ and $C^{(9)}$ cannot be simultaneously zero.
Therefore, the cases \eqref{eqn:theo_proof_case1_22} is inadequate.

Lastly, we consider the case \eqref{eqn:theo_proof_case1_25}.
From $C^{(8)}=0$, we obtain the following condition:
\begin{equation}\label{eqn:theo_proof_case1_25_1}
 A^{(2)}_3=B^{(2)}_3=0.
\end{equation}
Then, from the condition \eqref{eqn:caseN3_CAO_B2}, we have
\begin{equation}\label{eqn:theo_proof_case1_25_2}
 B^{(2)}_1=B^{(1)}_4 B^{(3)}_1,\quad
 B^{(2)}_2=B^{(1)}_2 B^{(3)}_4,\quad
 B^{(3)}_2=\dfrac{B^{(1)}_4 B^{(3)}_4}{B^{(1)}_1},\quad
 B^{(2)}_4=B^{(1)}_2 B^{(3)}_1.
\end{equation}
From the equations
\begin{subequations}
\begin{align}
 &C^{(1)}=A^{(1)}_1 B^{(1)}_1 (A^{(2)}_1 A^{(3)}_4 B^{(3)}_1 + A^{(2)}_4 A^{(3)}_2 B^{(3)}_4)=0,\\
 &C^{(2)}=A^{(3)}_2 B^{(3)}_2 (A^{(1)}_4 A^{(2)}_2 B^{(1)}_2 + A^{(1)}_1 A^{(2)}_4 B^{(1)}_4)=0,\\
 &C^{(9)}=A^{(1)}_1 A^{(2)}_1 A^{(3)}_1 B^{(1)}_1 B^{(3)}_1 + A^{(1)}_2 A^{(2)}_2 A^{(3)}_2 B^{(1)}_2 B^{(3)}_4=0,
\end{align}
\end{subequations}
we obtain the conditiuons
\begin{equation}\label{eqn:theo_proof_case1_25_3}
 A^{(1)}_2=-\dfrac{A^{(1)}_4 A^{(3)}_1 B^{(1)}_1}{A^{(3)}_4 B^{(1)}_4},\quad 
 B^{(1)}_2=-\dfrac{A^{(1)}_1 A^{(2)}_4 B^{(1)}_4}{A^{(1)}_4 A^{(2)}_2},\quad
 A^{(3)}_2=-\dfrac{A^{(2)}_1 A^{(3)}_4 B^{(3)}_1}{A^{(2)}_4 B^{(3)}_4}.
\end{equation}
We can easily verify that under the conditions \eqref{eqn:theo_proof_case1_25}, \eqref{eqn:theo_proof_case1_25_1}, \eqref{eqn:theo_proof_case1_25_2} and \eqref{eqn:theo_proof_case1_25_3}, 
the condition \eqref{eqn:proof_case1_2_C} holds,
and these conditions with \eqref{eqn:caseN3_CAO_B2} are equivalent to the condition \eqref{eqn:typeI_cond2}.
Therefore, we have completed the proof.
\subsection{Proof of Lemma \ref{lemma:classification_CACO_CO1CO2CO3_2}}
\label{subsection:caseN3_3}
In this section, we consider quad-equations \eqref{eqn:Q1_P1_CO3}--\eqref{eqn:Q9_P3_CO3} with \eqref{eqn:P123_caseN3_3}.
Note that in this setting, the transformations $s_{23}$ and $\iota$ defined in \eqref{eqn:Gco} keep invariant the system of quad-equations $\{Q_1,\dots,Q_9\}$ under the appropriate transformations of parameters (see Remark \ref{remark:symmetry_GCO}).

We obtain
\begin{equation}
 v_4^{(r)}=\dfrac{F_1}{F_2},\quad
 v_4^{(l)}=\dfrac{G_1}{G_2},
\end{equation}
where
\begin{align*}
 F_1=&-(C^{(11)}+C^{(12)}u_4v_1){u_4}^2
 -(C^{(13)}+C^{(14)}u_4v_1){u_1}^2
 -(C^{(15)}+C^{(16)}u_4v_1)u_1u_4,\\
 F_2=&u_4\Big((C^{(21)}+C^{(22)}u_4v_1){u_4}^2+(C^{(23)}+C^{(24)}u_4v_1){u_1}^2+(C^{(25)}+C^{(26)}u_4v_1)u_1u_4\Big),\\
 G_1=&-(C^{(31)}+C^{(32)}u_1v_1){u_4}^2-(C^{(33)}+C^{(34)}u_1v_1){u_1}^2-(C^{(35)}+C^{(36)}u_1v_1)u_1u_4,\\
 G_2=&u_1\Big((C^{(41)}+C^{(42)}u_1v_1){u_4}^2+(C^{(43)}+C^{(44)}u_1v_1){u_1}^2+(C^{(45)}+C^{(46)}u_1v_1)u_1u_4\Big),
\end{align*}
where the parameters $C^{(11)}$, \dots, $C^{(16)}$, \dots, $C^{(41)}$, \dots, $C^{(46)}$ are given by \eqref{eqn:typeII_defC11}--\eqref{eqn:typeII_defC46}.

\begin{remark}\label{remark:case3_pole_zero_v4}
Since the symmetry $\iota\circ s_{23}$ acts on the variables $u_i$ and $v_j$, $i,j=1,\dots,6$, as the following:
\begin{equation}
 \iota\circ s_{23}:~u_1\leftrightarrow u_4,\quad
 u_2\leftrightarrow \frac{1}{u_3},\quad
 u_5\leftrightarrow \frac{1}{u_6},\quad
 v_2\leftrightarrow \frac{1}{v_6},\quad
 v_3\leftrightarrow \frac{1}{v_5},
\end{equation}
if there exists a condition of parameters such that 
the CACO property holds and $v_4$ has a pole at $u_1=0$, then
there exists a condition of parameters such that 
the CACO property holds and $v_4$ has a pole at $u_4=0$,
and their orders are same.
The same is true in the zeros of $v_4$ at $u_1=0$ and $u_4=0$.
\end{remark}

We show that $u_1=0$ and $u_4=0$ are not a pole nor a zero of $v_4$, and $v_1=0$ is not a pole of $v_4$ by the following four lemmas.

\begin{lemma}\label{lemma:case3_u1u4_pole1}
There exist no condition of parameters such that 
the CACO property holds and $v_4$ has a simple pole at $u_1=0$ or $u_4=0$.
\end{lemma}
\begin{proof}
Because of Remark \ref{remark:case3_pole_zero_v4}, it is sufficient to prove that 
there exist no condition of parameters such that the CACO property holds and $v_4$ has a simple pole at $u_1=0$.

Assume that $u_1=0$ is a simple pole of $v_4$.
Since the following hold:
\begin{equation*}
 F_1=-{u_4}^2(C^{(11)}+u_4v_1C^{(12)})+O(u_1),\quad
 F_2={u_4}^3(C^{(21)}+u_4v_1C^{(22)})+O(u_1),
\end{equation*}
as $u_1\to0$, we obtain 
\begin{equation}
 C^{(21)}=C^{(22)}=0.
\end{equation} 
Then, the following relations hold:
\begin{subequations}
\begin{align}
 0&=(A^{(1)}_1A^{(2)}_2B^{(1)}_1+A^{(1)}_4A^{(2)}_4B^{(1)}_3)C^{(21)}-(A^{(1)}_3A^{(2)}_2B^{(1)}_1+A^{(1)}_2A^{(2)}_4B^{(1)}_3)C^{(22)}\notag\\
 &=(A^{(1)}_1A^{(1)}_2-A^{(1)}_3A^{(1)}_4)(A^{(2)}_1A^{(2)}_2-A^{(2)}_3A^{(2)}_4)B^{(1)}_1B^{(1)}_3A^{(3)}_3B^{(3)}_4,\\
 0&=(A^{(1)}_1A^{(2)}_3B^{(1)}_1+A^{(1)}_4A^{(2)}_1B^{(1)}_3)C^{(21)}-(A^{(1)}_3A^{(2)}_3B^{(1)}_1+A^{(1)}_2A^{(2)}_1B^{(1)}_3)C^{(22)}\notag\\
 &=-(A^{(1)}_1A^{(1)}_2-A^{(1)}_3A^{(1)}_4)(A^{(2)}_1A^{(2)}_2-A^{(2)}_3A^{(2)}_4)B^{(1)}_1B^{(1)}_3A^{(3)}_1B^{(3)}_1.
\end{align}
\end{subequations}
Therefore, we obtain the following two cases:
\begin{subequations}
\begin{align}
 &A^{(1)}_1=B^{(1)}_1=0;\label{eqn:caseN3_3_spole_con_1}\\
 &A^{(1)}_3=B^{(1)}_3=0.\label{eqn:caseN3_3_spole_con_2}
\end{align}
\end{subequations}
Below, we show that the both cases are inadmissible.

Firstly, we consider the case \eqref{eqn:caseN3_3_spole_con_1}. 
Then, we obtain
\begin{equation}
 C^{(41)}=C^{(42)}=0.
\end{equation}
Moreover, since
\begin{equation*}
 G_1= -{u_4}^2C^{(31)}+O(u_1),\quad
 G_2={u_1}^2u_4C^{(45)}+O({u_1}^3),
\end{equation*}
as $u_1\to0$ and $v_4^{(l)}$ has a simple pole at $u_1=0$,
we obtain
\begin{equation}
 C^{(31)}=0.
\end{equation}

We here prove $C^{(26)}\neq0$.
Assume $C^{(26)}=0$. 
Then, from $C^{(22)}=0$ and $C^{(26)}=0$, we obtain
\begin{align*}
 &0=(B^{(1)}_2B^{(3)}_1+B^{(1)}_3B^{(3)}_3)C^{(22)}-B^{(1)}_3B^{(3)}_1C^{(26)}
 =-(B^{(3)}_1B^{(3)}_2-B^{(3)}_3B^{(3)}_4)A^{(1)}_4A^{(2)}_1A^{(3)}_3{B^{(1)}_3}^2,\\
 &0=(B^{(1)}_3B^{(3)}_2+B^{(1)}_2B^{(3)}_4)C^{(22)}-B^{(1)}_3B^{(3)}_4C^{(26)}
 =(B^{(3)}_1B^{(3)}_2-B^{(3)}_3B^{(3)}_4)A^{(1)}_4A^{(2)}_4A^{(3)}_1{B^{(1)}_3}^2,
\end{align*}
which give the following two cases:
\begin{subequations}
\begin{align}
 &A^{(2)}_1=B^{(2)}_1=A^{(3)}_1=B^{(3)}_1=0;\label{eqn:proof_case3_11}\\
 &A^{(2)}_4=B^{(2)}_4=A^{(3)}_3=B^{(3)}_3=0.\label{eqn:proof_case3_12}
\end{align}
\end{subequations}
In the case \eqref{eqn:proof_case3_11}, we obtain
\begin{equation}
 C^{(11)}=C^{(12)}=C^{(24)}=C^{(32)}=C^{(46)}=0,
\end{equation}
which gives
\begin{equation}
v_4^{(r)}=-\dfrac{u_1C^{(13)}+u_1u_4v_1C^{(14)}+u_4C^{(15)}+{u_4}^2v_1C^{(16)}}{u_4(u_1C^{(23)}+u_4C^{(25)})}.
\end{equation}
Because of 
\begin{equation}
 C^{(25)}=A^{(1)}_3B^{(1)}_4A^{(2)}_3A^{(3)}_3B^{(3)}_4\neq0,
\end{equation}
the case \eqref{eqn:proof_case3_11} is inconsistent with the assumption that $u_1=0$ is a simple pole of $v_4$.
Moreover, in the case \eqref{eqn:proof_case3_12}, 
from $C^{(31)}=0$ we obtain
\begin{equation}
 A^{(3)}_4=B^{(3)}_4=0,
\end{equation}
which gives
\begin{equation}
 C^{(11)}=C^{(12)}=C^{(23)}=C^{(24)}=C^{(32)}=C^{(33)}=C^{(43)}=C^{(46)}=0.
\end{equation}
Then, we have
\begin{equation}
v_4^{(r)}=-\dfrac{u_1C^{(13)}+u_1u_4v_1C^{(14)}+u_4C^{(15)}+{u_4}^2v_1C^{(16)}}{{u_4}^2C^{(25)}},
\end{equation}
which is inconsistent with the assumption that $u_1=0$ is a simple pole of $v_4$.
Therefore, $C^{(26)}\neq0$ holds.

The following hold:
\begin{align*}
 &F_1 G_2=-{u_1}^2{u_4}^3C^{(45)}(C^{(11)}+u_4v_1C^{(12)})+O({u_1}^3),\\
 &F_2 G_1=-{u_1}^2{u_4}^3(C^{(25)}+u_4v_1C^{(26)})(C^{(35)}+u_4v_1C^{(32)})+O({u_1}^3).
\end{align*}
as $u_1\to 0$.
Since $C^{(26)}\neq0$, from $F_1 G_2=F_2 G_1$ we obtain
\begin{equation}
 C^{(32)}=0.
\end{equation}
Then, from 
\begin{subequations}
\begin{align}
 0&=A^{(3)}_1C^{(31)}-A^{(3)}_4C^{(32)}
 =(A^{(3)}_1A^{(3)}_2-A^{(3)}_3A^{(3)}_4)A^{(1)}_3B^{(1)}_3A^{(2)}_4B^{(3)}_1,\\
 0&=A^{(3)}_2C^{(32)}-A^{(3)}_3C^{(31)}
 =(A^{(3)}_1A^{(3)}_2-A^{(3)}_3A^{(3)}_4)A^{(1)}_3B^{(1)}_3A^{(2)}_1B^{(3)}_4,
\end{align}
\end{subequations}
we obtain 
\begin{equation}
 A^{(2)}_1=B^{(2)}_1=A^{(3)}_1=B^{(3)}_1=0
\quad\text{or}\quad
 A^{(2)}_4=B^{(2)}_4=A^{(3)}_4=B^{(3)}_4=0.
\end{equation}
Because of $C^{(26)}\neq0$, the condition $A^{(2)}_1=B^{(2)}_1=A^{(3)}_1=B^{(3)}_1=0$ is ineligible 
and then the condition $A^{(2)}_4=B^{(2)}_4=A^{(3)}_4=B^{(3)}_4=0$ holds.
Then, we obtain $C^{(11)}=C^{(12)}=0$.
Therefore, we obtain 
\begin{equation}
 v_4^{(r)}
=-\dfrac{u_1C^{(13)}+u_1u_4v_1C^{(14)}+u_4C^{(15)}+{u_4}^2v_1C^{(16)}}
{u_4(u_1C^{(23)}+u_1u_4v_1C^{(24)}+u_4C^{(25)}+{u_4}^2v_1C^{(26)})},
\end{equation}
which is inconsistent with the assumption that $u_1=0$ is a simple pole of $v_4$ because of $C^{(26)}\neq0$.
Therefore, the case \eqref{eqn:caseN3_3_spole_con_1} is inadequate.

We next consider the case \eqref{eqn:caseN3_3_spole_con_2}.
Then, we obtain
\begin{equation}
 C^{(11)}=0.
\end{equation}
Since the right-hand sides of the following equations cannot be simultaneously zero:
\begin{align*}
 &A^{(1)}_2B^{(1)}_2A^{(2)}_4C^{(22)}-A^{(1)}_1B^{(1)}_1A^{(2)}_2C^{(25)}
 =-(A^{(2)}_1A^{(2)}_2-A^{(2)}_3A^{(2)}_4)A^{(1)}_1B^{(1)}_1A^{(1)}_2B^{(1)}_2A^{(3)}_3B^{(3)}_4,\\
 &A^{(1)}_2B^{(1)}_2A^{(2)}_1C^{(22)}-A^{(1)}_1B^{(1)}_1A^{(2)}_3C^{(25)}
 =(A^{(2)}_1A^{(2)}_2-A^{(2)}_3A^{(2)}_4)A^{(1)}_1B^{(1)}_1A^{(1)}_2B^{(1)}_2A^{(3)}_1B^{(3)}_1,
\end{align*}
and $C^{(22)}=0$, we obtain
\begin{equation}
 C^{(25)}\neq0.
\end{equation}
Since the following relation holds:
\begin{equation*}
 F_1G_2-F_2G_1
=\left({u_4}^4C^{(25)}C^{(31)}+{u_4}^5v_1(C^{(26)}C^{(31)}-C^{(12)}C^{(41)})\right)u_1+O({u_1}^2),
\end{equation*}
as $u_1\to0$, we obtain 
\begin{equation}
 C^{(31)}=0.
\end{equation}
Moreover, since 
\begin{equation*}
 G_1=-u_1u_4(u_4v_1C^{(32)}+C^{(35)})+O({u_1}^2),\quad
 G_2=u_1{u_4}^2C^{(41)}+O({u_1}^2),
\end{equation*}
as $u_1\to0$ and $u_1=0$ is a simple pole of $v_4^{(l)}$, we obtain
\begin{equation}
 C^{(41)}=0.
\end{equation}
Then, from
\begin{align*}
 F_1G_2-F_2G_1
 =&\Big({u_4}^3C^{(25)}C^{(35)}+{u_4}^4v_1(C^{(25)}C^{(32)}+C^{(26)}C^{(35)}-C^{(12)}C^{(45)})\notag\\
 &+{u_4}^5{v_1}^2(C^{(26)}C^{(32)}-C^{(12)}C^{(42)})\Big){u_1}^2+O({u_1}^3),
\end{align*}
as $u_1\to0$, we obtain
\begin{equation}
 C^{(35)}=0.
\end{equation}
From 
\begin{align*}
 &0=(B^{(1)}_1B^{(3)}_2+B^{(1)}_4B^{(3)}_4)C^{(31)}-B^{(1)}_1B^{(3)}_4C^{(35)}
 =(B^{(3)}_1B^{(3)}_2-B^{(3)}_3B^{(3)}_4)A^{(1)}_2A^{(2)}_2A^{(3)}_2{B^{(1)}_1}^2,\\
 &0=(B^{(1)}_1B^{(3)}_3+B^{(1)}_4B^{(3)}_1)C^{(31)}-B^{(1)}_1B^{(3)}_1C^{(35)}
 =(B^{(3)}_3B^{(3)}_4-B^{(3)}_1B^{(3)}_2)A^{(1)}_2A^{(2)}_3A^{(3)}_4{B^{(1)}_1}^2,
\end{align*}
we obtain 
\begin{equation}
 A^{(2)}_2=B^{(2)}_2=A^{(3)}_4=B^{(3)}_4=0
\quad\text{or}\quad
 A^{(2)}_3=B^{(2)}_3=A^{(3)}_2=B^{(3)}_2=0.
\end{equation}
Assume 
\begin{equation}
 A^{(2)}_2=B^{(2)}_2=A^{(3)}_4=B^{(3)}_4=0.
\end{equation}
Then, we have
\begin{equation}
 C^{(12)}=C^{(15)}=C^{(32)}=C^{(33)}=C^{(35)}=C^{(42)}=0,
\end{equation}
and
\begin{equation}
 C^{(45)}=A^{(1)}_1A^{(2)}_4A^{(3)}_2B^{(1)}_2B^{(3)}_1\neq0.
\end{equation}
We obtain
\begin{equation}
v_4^{(l)}
=-\dfrac{v_1(u_1C^{(34)}+u_4C^{(36)})}{u_1(C^{(43)}+u_1v_1C^{(44)}+u_4v_1C^{(46)})+u_4C^{(45)}},
\end{equation}
which is inconsistent with the assumption that $u_1=0$ is a simple pole of $v_4$ because of $C^{(45)}\neq0$.
Therefore, we obtain
\begin{equation}
 A^{(2)}_3=B^{(2)}_3=A^{(3)}_2=B^{(3)}_2=0.
\end{equation}
From
\begin{equation}
 C^{(22)}=A^{(1)}_1A^{(2)}_2A^{(3)}_1B^{(1)}_1B^{(3)}_1,
\end{equation}
we obtain
\begin{equation}
 A^{(3)}_1=B^{(3)}_1=0.
\end{equation}
Then, we obtain
\begin{equation}
 C^{(12)}=C^{(15)}=C^{(23)}=C^{(24)}=C^{(32)}=C^{(33)}=C^{(35)}=C^{(42)}=C^{(43)}=0.
\end{equation}
We obtain
\begin{equation}
 v_4^{(r)}
 =-\dfrac{u_1C^{(13)}+u_1u_4v_1C^{(14)}+{u_4}^2v_1C^{(16)}}{{u_4}^2(C^{(25)}+u_4v_1C^{(26)})},
\end{equation}
which is inconsistent with the assumption that $u_1=0$ is a simple pole of $v_4$.
Therefore, the case \eqref{eqn:caseN3_3_spole_con_2} is also inadequate.
Hence, we have completed the proof.
\end{proof}

\begin{lemma}\label{lemma:case3_u1u4_pole2}
There exist no condition of parameters such that 
the CACO property holds and $v_4$ has a pole of order 2 at $u_1=0$ or $u_4=0$.
\end{lemma}
\begin{proof}
Because of Remark \ref{remark:case3_pole_zero_v4}, it is sufficient to prove that 
there exist no condition of parameters such that the CACO property holds and $v_4$ has a pole of order 2 at $u_1=0$.

Assume that $u_1=0$ is a pole of order 2 of $v_4$.
Since 
\begin{align*}
 &F_1=-{u_4}^2(C^{(11)}+u_4v_1C^{(12)})+O(u_1),\\
 &F_2={u_4}^3(C^{(21)}+u_4v_1C^{(22)})+u_1{u_4}^2(C^{(25)}+u_4v_1C^{(26)})+O({u_1}^2),\\
 &G_1=-{u_4}^2C^{(31)}+O(u_1),\\
 &G_2=u_1{u_4}^2C^{(41)}+O({u_1}^2),
\end{align*}
as $u_1\to 0$, we obtain 
\begin{equation}
 C^{(21)}=C^{(22)}=C^{(25)}=C^{(26)}=C^{(41)}=0.
\end{equation}
Because of the following relations:
\begin{subequations}
\begin{align}
 0&=(A^{(1)}_1A^{(2)}_2B^{(1)}_1+A^{(1)}_4A^{(2)}_4B^{(1)}_3)C^{(21)}-(A^{(1)}_3A^{(2)}_2B^{(1)}_1+A^{(1)}_2A^{(2)}_4B^{(1)}_3)C^{(22)}\notag\\
 &=(A^{(1)}_1A^{(1)}_2-A^{(1)}_3A^{(1)}_4)(A^{(2)}_1A^{(2)}_2-A^{(2)}_3A^{(2)}_4)B^{(1)}_1B^{(1)}_3A^{(3)}_3B^{(3)}_4,\\
 0&=(A^{(1)}_1A^{(2)}_3B^{(1)}_1+A^{(1)}_4A^{(2)}_1B^{(1)}_3)C^{(21)}-(A^{(1)}_3A^{(2)}_3B^{(1)}_1+A^{(1)}_2A^{(2)}_1B^{(1)}_3)C^{(22)}\notag\\
 &=-(A^{(1)}_1A^{(1)}_2-A^{(1)}_3A^{(1)}_4)(A^{(2)}_1A^{(2)}_2-A^{(2)}_3A^{(2)}_4)B^{(1)}_1B^{(1)}_3A^{(3)}_1B^{(3)}_1,
\end{align}
\end{subequations}
we obtain 
\begin{equation}
 A^{(1)}_1=B^{(1)}_1=0
\quad\text{or}\quad
 A^{(1)}_3=B^{(1)}_3=0.
\end{equation}
However, both cases are inadequate.
Indeed, if $A^{(1)}_1=B^{(1)}_1=0$, then we obtain
\begin{subequations}
\begin{align}
 &A^{(1)}_3B^{(1)}_4A^{(2)}_2C^{(22)}-B^{(1)}_3A^{(2)}_4(A^{(1)}_4C^{(25)}-A^{(1)}_2C^{(26)})\notag\\
 &\quad=(A^{(2)}_1A^{(2)}_2-A^{(2)}_3A^{(2)}_4)A^{(1)}_3B^{(1)}_3A^{(1)}_4B^{(1)}_4A^{(3)}_3B^{(3)}_4,\\
 &A^{(1)}_3B^{(1)}_4A^{(2)}_3C^{(22)}-B^{(1)}_3A^{(2)}_1(A^{(1)}_4C^{(25)}-A^{(1)}_2C^{(26)})\notag\\
 &\quad=-(A^{(2)}_1A^{(2)}_2-A^{(2)}_3A^{(2)}_4)A^{(1)}_3B^{(1)}_3A^{(1)}_4B^{(1)}_4A^{(3)}_1B^{(3)}_1,
\end{align}
\end{subequations}
whose left-hand sides are zero but right-hand sides cannot be simultaneously zero, and
if $A^{(1)}_3=B^{(1)}_3=0$, then we obtain
\begin{align*}
 &A^{(1)}_2B^{(1)}_2A^{(2)}_4C^{(22)}-A^{(1)}_1B^{(1)}_1A^{(2)}_2C^{(25)}
 =-(A^{(2)}_1A^{(2)}_2-A^{(2)}_3A^{(2)}_4)A^{(1)}_1B^{(1)}_1A^{(1)}_2B^{(1)}_2A^{(3)}_3B^{(3)}_4,\\
 &A^{(1)}_2B^{(1)}_2A^{(2)}_1C^{(22)}-A^{(1)}_1B^{(1)}_1A^{(2)}_3C^{(25)}
 =(A^{(2)}_1A^{(2)}_2-A^{(2)}_3A^{(2)}_4)A^{(1)}_1B^{(1)}_1A^{(1)}_2B^{(1)}_2A^{(3)}_1B^{(3)}_1,
\end{align*}
whose left-hand sides are zero but right-hand sides cannot be simultaneously zero.
Therefore, we have completed the proof.
\end{proof}

\begin{lemma}\label{lemma:case3_u1u4_zero}
There exist no condition of parameters such that 
the CACO property holds and $v_4$ has a zero at $u_1=0$ or $u_4=0$.
\end{lemma}
\begin{proof}
Because of Remark \ref{remark:case3_pole_zero_v4}, it is sufficient to prove that 
there exist no condition of parameters such that the CACO property holds and $v_4$ has a zero at $u_1=0$.

Assume that $u_1=0$ is a zero of $v_4$.
Since 
\begin{align*}
 &F_1=-{u_4}^2(C^{(11)}+u_4v_1C^{(12)})+O(u_1),
 &&F_2={u_4}^3(C^{(21)}+u_4v_1C^{(22)})+O(u_1),\\
 &G_1=-{u_4}^2C^{(31)}-u_1u_4(u_4v_1C^{(32)}+C^{(35)})+O({u_1}^2),
 &&G_2=u_1{u_4}^2C^{(41)}+O({u_1}^2),
\end{align*}
as $u_1\to 0$, we obtain 
\begin{equation}
 C^{(11)}=C^{(12)}=C^{(31)}=C^{(32)}=C^{(35)}=0.
\end{equation}
Moreover, since
\begin{equation*}
 F_1=-{u_1}^2C^{(13)}+O(u_4),\quad
 F_2={u_1}^2u_4C^{(23)}+O({u_4}^2),
\end{equation*}
as $u_4\to0$ and $u_4=0$ is not a pole of $v_4$, we obtain
\begin{equation}
 C^{(13)}=0.
\end{equation}
Because of the following relations:
\begin{subequations}
\begin{align}
 0&=(A^{(1)}_1A^{(2)}_2B^{(1)}_1+A^{(1)}_4A^{(2)}_4B^{(1)}_3)C^{(11)}-(A^{(1)}_3A^{(2)}_2B^{(1)}_1+A^{(1)}_2A^{(2)}_4B^{(1)}_3)C^{(12)}\notag\\
 &=(A^{(1)}_1A^{(1)}_2-A^{(1)}_3A^{(1)}_4)(A^{(2)}_1A^{(2)}_2-A^{(2)}_3A^{(2)}_4)B^{(1)}_1B^{(1)}_3A^{(3)}_2B^{(3)}_4,\\
 0&=(A^{(1)}_1A^{(2)}_3B^{(1)}_1+A^{(1)}_4A^{(2)}_1B^{(1)}_3)C^{(11)}-(A^{(1)}_3A^{(2)}_3B^{(1)}_1+A^{(1)}_2A^{(2)}_1B^{(1)}_3)C^{(12)}\notag\\
 &=-(A^{(1)}_1A^{(1)}_2-A^{(1)}_3A^{(1)}_4)(A^{(2)}_1A^{(2)}_2-A^{(2)}_3A^{(2)}_4)B^{(1)}_1B^{(1)}_3B^{(3)}_1A^{(3)}_4,\\
 0&=(A^{(2)}_4A^{(3)}_3B^{(3)}_1+A^{(2)}_1A^{(3)}_1B^{(3)}_4)C^{(31)}-(A^{(2)}_4A^{(3)}_2B^{(3)}_1+A^{(2)}_1A^{(3)}_4B^{(3)}_4)C^{(32)}\notag\\
 &=(A^{(2)}_1A^{(2)}_2-A^{(2)}_3A^{(2)}_4)(A^{(3)}_1A^{(3)}_2-A^{(3)}_3A^{(3)}_4)B^{(1)}_1A^{(1)}_2B^{(3)}_1B^{(3)}_4,\\
 0&=(A^{(2)}_2A^{(3)}_3B^{(3)}_1+A^{(2)}_3A^{(3)}_1B^{(3)}_4)C^{(31)}-(A^{(2)}_2A^{(3)}_2B^{(3)}_1+A^{(2)}_3A^{(3)}_4B^{(3)}_4)C^{(32)}\notag\\
 &=-(A^{(2)}_1A^{(2)}_2-A^{(2)}_3A^{(2)}_4)(A^{(3)}_1A^{(3)}_2-A^{(3)}_3A^{(3)}_4)A^{(1)}_3B^{(1)}_3B^{(3)}_1B^{(3)}_4,
\end{align}
\end{subequations}
we obtain the following four cases:
\begin{subequations}
\begin{align}
 &A^{(1)}_1=B^{(1)}_1=A^{(3)}_1=B^{(3)}_1=0;\label{eqn:N3_3_cond_u1zero_1}\\
 &A^{(1)}_1=B^{(1)}_1=A^{(3)}_4=B^{(3)}_4=0;\label{eqn:N3_3_cond_u1zero_2}\\
 &A^{(1)}_3=B^{(1)}_3=A^{(3)}_1=B^{(3)}_1=0;\label{eqn:N3_3_cond_u1zero_3}\\
 &A^{(1)}_3=B^{(1)}_3=A^{(3)}_4=B^{(3)}_4=0.\label{eqn:N3_3_cond_u1zero_4}
\end{align}
\end{subequations}
Below, we show that all cases are inadmissible.

Firstly, we consider the case \eqref{eqn:N3_3_cond_u1zero_1}. 
From 
\begin{equation}
 0=C^{(31)}=A^{(1)}_3B^{(1)}_3A^{(2)}_1A^{(3)}_4B^{(3)}_4,
\end{equation}
we obtain $A^{(2)}_1=B^{(2)}_1=0$.
Then, we obtain
\begin{equation}
 C^{(21)}=C^{(22)}=0,\quad
 C^{(16)}=B^{(1)}_3A^{(1)}_4A^{(2)}_4A^{(3)}_4B^{(3)}_3\neq0.
\end{equation}
Since $C^{(16)}\neq0$ and 
\begin{equation*}
 F_1=-u_1u_4(C^{(15)}+u_4v_1C^{(16)})+O({u_1}^2),\quad
 F_2=u_1{u_4}^2(C^{(25)}+u_4v_1C^{(26)})+O({u_1}^2),
\end{equation*}
as $u_1\to0$, $u_1=0$ cannot be a zero of $v_4$. 
Therefore, the case \eqref{eqn:N3_3_cond_u1zero_1} is inadequate.

Secondly, we consider the case \eqref{eqn:N3_3_cond_u1zero_2}.
From 
\begin{equation}
 0=C^{(31)}=A^{(1)}_3B^{(1)}_3A^{(2)}_4B^{(3)}_1A^{(3)}_2,
\end{equation}
we obtain $A^{(2)}_4=B^{(2)}_4=0$.
Then, we obtain
\begin{equation}
 C^{(21)}=C^{(22)}=0,\quad
 C^{(16)}=B^{(1)}_3A^{(1)}_4A^{(2)}_1A^{(3)}_2B^{(3)}_2\neq0.
\end{equation}
Since $C^{(16)}\neq0$ and
\begin{equation*}
 F_1=-u_1u_4(C^{(15)}+u_4v_1C^{(16)})+O({u_1}^2),\quad
 F_2=u_1{u_4}^2(C^{(25)}+u_4v_1C^{(26)})+O({u_1}^2),
\end{equation*}
as $u_1\to0$, $u_1=0$ cannot be a zero of $v_4$. 
Therefore, the case \eqref{eqn:N3_3_cond_u1zero_2} is inadequate.

Thirdly, We consider the case \eqref{eqn:N3_3_cond_u1zero_3}.
From 
\begin{equation}
 0=C^{(31)}=B^{(1)}_1A^{(1)}_2A^{(2)}_3A^{(3)}_4B^{(3)}_4,
\end{equation}
we obtain $A^{(2)}_3=B^{(2)}_3=0$.
Moreover, from
\begin{equation}
 0=C^{(35)}=B^{(1)}_1A^{(1)}_2A^{(2)}_2A^{(3)}_2B^{(3)}_3,
\end{equation}
we obtain $A^{(3)}_2=B^{(3)}_2=0$.
Then, we obtain
\begin{equation}
 C^{(21)}=C^{(22)}=0,\quad
 C^{(16)}=A^{(1)}_1B^{(1)}_1A^{(2)}_2B^{(3)}_3A^{(3)}_4\neq0.
\end{equation}
Since $C^{(16)}\neq0$ and
\begin{equation*}
 F_1=-u_1u_4(C^{(15)}+u_4v_1C^{(16)})+O({u_1}^2),\quad
 F_2=u_1{u_4}^2(C^{(25)}+u_4v_1C^{(26)})+O({u_1}^2),
\end{equation*}
as $u_1\to0$, $u_1=0$ cannot be a zero of $v_4$. 
Therefore, the case \eqref{eqn:N3_3_cond_u1zero_3} is inadequate.

Lastly, we consider the case \eqref{eqn:N3_3_cond_u1zero_4}.
From 
\begin{equation}
 0=C^{(31)}=B^{(1)}_1A^{(1)}_2A^{(2)}_2B^{(3)}_1A^{(3)}_2,
\end{equation}
we obtain $A^{(2)}_2=B^{(2)}_2=0$.
Then, we obtain
\begin{equation}
 C^{(21)}=C^{(22)}=0,\quad
 C^{(16)}=A^{(1)}_1B^{(1)}_1A^{(2)}_3A^{(3)}_2B^{(3)}_2\neq0.
\end{equation}
Since $C^{(16)}\neq0$ and 
\begin{equation*}
 F_1=-u_1u_4(C^{(15)}+u_4v_1C^{(16)})+O({u_1}^2),\quad
 F_2=u_1{u_4}^2(C^{(25)}+u_4v_1C^{(26)})+O({u_1}^2),
\end{equation*}
as $u_1\to0$,
$u_1=0$ cannot be a zero of $v_4$. 
Therefore, the case \eqref{eqn:N3_3_cond_u1zero_4} is inadequate.
Hence, we have completed the proof.
\end{proof}

\begin{lemma}\label{lemma:case3_v1_pole}
There exist no condition of parameters such that 
the CACO property holds and $v_4$ has a pole at $v_1=0$.
\end{lemma}
\begin{proof}
Assume that $v_1=0$ is a pole of $v_4$.
Since 
\begin{align*}
 &G_1=-{u_4}^2C^{(31)}-u_1u_4(u_4v_1C^{(32)}+C^{(35)})+O({u_1}^2)&&(u_1\to0),\\
 &G_2=u_1u_4^2C^{(41)}+O({u_1}^2)&&(u_1\to0),\\
 &F_1=-{u_1}^2C^{(13)}-u_1u_4(u_1v_1C^{(14)}+C^{(15)})+O({u_4}^2)&&(u_4\to0),\\
 &F_2={u_1}^2u_4C^{(23)}+O({u_4}^2)&&(u_4\to0),\\
 &F_1=-({u_4}^2C^{(11)}+{u_1}^2C^{(13)}+u_1u_4C^{(15)})+O(v_1)&&(v_1\to0),\\
 &F_2=u_4({u_4}^2C^{(21)}+{u_1}^2C^{(23)}+u_1u_4C^{(25)})+O(v_1)&&(v_1\to0),\\
 &G_1=-({u_4}^2C^{(31)}+{u_1}^2C^{(33)}+u_1u_4C^{(35)})+O(v_1)&&(v_1\to0),\\
 &G_2=u_1({u_4}^2C^{(41)}+{u_1}^2C^{(43)}+u_1u_4C^{(45)})+O(v_1)&&(v_1\to0),
\end{align*}
$u_1=0$ and $u_4=0$ are not a pole nor zero of $v_4$
and $v_1=0$ is a pole of $v_4$, we obtain
\begin{align}
 C^{(13)}&=C^{(14)}=C^{(15)}=C^{(21)}=C^{(23)}=C^{(25)}=C^{(31)}=C^{(32)}=C^{(35)}\notag\\
 &=C^{(41)}=C^{(43)}=C^{(45)}=0.
\end{align}
Because of the following relations:
\begin{subequations}
\begin{align}
 0&=(A^{(1)}_4A^{(2)}_4B^{(1)}_2+A^{(1)}_1A^{(2)}_2B^{(1)}_4)C^{(13)}-(A^{(1)}_2A^{(2)}_4B^{(1)}_2+A^{(1)}_3A^{(2)}_2B^{(1)}_4)C^{(14)}\notag\\
 &=(A^{(1)}_1A^{(1)}_2-A^{(1)}_3A^{(1)}_4)(A^{(2)}_1A^{(2)}_2-A^{(2)}_3A^{(2)}_4)B^{(1)}_2B^{(1)}_4A^{(3)}_2B^{(3)}_2,\\
 0&=(A^{(1)}_4A^{(2)}_1B^{(1)}_2+A^{(1)}_1A^{(2)}_3B^{(1)}_4)C^{(13)}-(A^{(1)}_2A^{(2)}_1B^{(1)}_2+A^{(1)}_3A^{(2)}_3B^{(1)}_4)C^{(14)}\notag\\
 &=-(A^{(1)}_1A^{(1)}_2-A^{(1)}_3A^{(1)}_4)(A^{(2)}_1A^{(2)}_2-A^{(2)}_3A^{(2)}_4)B^{(1)}_2B^{(1)}_4B^{(3)}_3A^{(3)}_4,\\
 0&=(A^{(2)}_4A^{(3)}_3B^{(3)}_1+A^{(2)}_1A^{(3)}_1B^{(3)}_4)C^{(31)}-(A^{(2)}_4A^{(3)}_2B^{(3)}_1+A^{(2)}_1A^{(3)}_4B^{(3)}_4)C^{(32)}\notag\\
 &=(A^{(2)}_1A^{(2)}_2-A^{(2)}_3A^{(2)}_4)(A^{(3)}_1A^{(3)}_2-A^{(3)}_3A^{(3)}_4)B^{(1)}_1A^{(1)}_2B^{(3)}_1B^{(3)}_4,\\
 0&=(A^{(2)}_2A^{(3)}_3B^{(3)}_1+A^{(2)}_3A^{(3)}_1B^{(3)}_4)C^{(31)}-(A^{(2)}_2A^{(3)}_2B^{(3)}_1+A^{(2)}_3A^{(3)}_4B^{(3)}_4)C^{(32)}\notag\\
 &=-(A^{(2)}_1A^{(2)}_2-A^{(2)}_3A^{(2)}_4)(A^{(3)}_1A^{(3)}_2-A^{(3)}_3A^{(3)}_4)A^{(1)}_3B^{(1)}_3B^{(3)}_1B^{(3)}_4,\\
 0&=(A^{(2)}_3A^{(3)}_3B^{(3)}_2+A^{(2)}_2A^{(3)}_1B^{(3)}_3)C^{(13)}-(A^{(2)}_3A^{(3)}_2B^{(3)}_2+A^{(2)}_2A^{(3)}_4B^{(3)}_3)C^{(23)}\notag\\
 &=(A^{(2)}_1A^{(2)}_2-A^{(2)}_3A^{(2)}_4)(A^{(3)}_1A^{(3)}_2-A^{(3)}_3A^{(3)}_4)A^{(1)}_2B^{(1)}_2B^{(3)}_2B^{(3)}_3,\\
 0&=(A^{(2)}_1A^{(3)}_3B^{(3)}_2+A^{(2)}_4A^{(3)}_1B^{(3)}_3)C^{(13)}-(A^{(2)}_1A^{(3)}_2B^{(3)}_2+A^{(2)}_4A^{(3)}_4B^{(3)}_3)C^{(23)}\notag\\
 &=-(A^{(2)}_1A^{(2)}_2-A^{(2)}_3A^{(2)}_4)(A^{(3)}_1A^{(3)}_2-A^{(3)}_3A^{(3)}_4)A^{(1)}_3B^{(1)}_4B^{(3)}_2B^{(3)}_3,\\
 0&=(A^{(1)}_4A^{(2)}_3B^{(1)}_1+A^{(1)}_1A^{(2)}_1B^{(1)}_3)C^{(31)}-(A^{(1)}_2A^{(2)}_3B^{(1)}_1+A^{(1)}_3A^{(2)}_1B^{(1)}_3)C^{(41)}\notag\\
 &=(A^{(1)}_1A^{(1)}_2-A^{(1)}_3A^{(1)}_4)(A^{(2)}_1A^{(2)}_2-A^{(2)}_3A^{(2)}_4)B^{(1)}_1B^{(1)}_3B^{(3)}_1A^{(3)}_2,\\
 0&=(A^{(1)}_4A^{(2)}_2B^{(1)}_1+A^{(1)}_1A^{(2)}_4B^{(1)}_3)C^{(31)}-(A^{(1)}_2A^{(2)}_2B^{(1)}_1+A^{(1)}_3A^{(2)}_4B^{(1)}_3)C^{(41)}\notag\\
 &=-(A^{(1)}_1A^{(1)}_2-A^{(1)}_3A^{(1)}_4)(A^{(2)}_1A^{(2)}_2-A^{(2)}_3A^{(2)}_4)B^{(1)}_1B^{(1)}_3A^{(3)}_4B^{(3)}_4,
\end{align}
\end{subequations}
we obtain the following four cases:
\begin{subequations}
\begin{align}
 &A^{(1)}_1=B^{(1)}_1=A^{(1)}_2=B^{(1)}_2=A^{(3)}_1=B^{(3)}_1=A^{(3)}_2=B^{(3)}_2=0;\label{eqn:case3_v1_pole_1}\\
 &A^{(1)}_1=B^{(1)}_1=A^{(1)}_2=B^{(1)}_2=A^{(3)}_3=B^{(3)}_3=A^{(3)}_4=B^{(3)}_4=0;\label{eqn:case3_v1_pole_2}\\
 &A^{(1)}_3=B^{(1)}_3=A^{(1)}_4=B^{(1)}_4=A^{(3)}_1=B^{(3)}_1=A^{(3)}_2=B^{(3)}_2=0;\label{eqn:case3_v1_pole_3}\\
 &A^{(1)}_3=B^{(1)}_3=A^{(1)}_4=B^{(1)}_4=A^{(3)}_3=B^{(3)}_3=A^{(3)}_4=B^{(3)}_4=0.\label{eqn:case3_v1_pole_4}
\end{align}
\end{subequations}
However, all cases are inadmissible.
Indeed, in the case \eqref{eqn:case3_v1_pole_1} we obtain
\begin{equation}
 C^{(13)}=A^{(1)}_3B^{(1)}_4A^{(2)}_2B^{(3)}_3A^{(3)}_4,\quad
 C^{(45)}=A^{(1)}_4B^{(1)}_4A^{(2)}_3A^{(3)}_4B^{(3)}_4,
\end{equation}
which cannot be simultaneously zero,
in the case \eqref{eqn:case3_v1_pole_2} we obtain
\begin{equation}
 C^{(13)}=A^{(1)}_3B^{(1)}_4A^{(2)}_3A^{(3)}_2B^{(3)}_2,\quad
 C^{(45)}=A^{(1)}_4B^{(1)}_4A^{(2)}_2B^{(3)}_1A^{(3)}_2,
\end{equation}
which cannot be simultaneously zero,
in the case \eqref{eqn:case3_v1_pole_3} we obtain
\begin{equation}
 C^{(13)}=A^{(1)}_2B^{(1)}_2A^{(2)}_4A^{(3)}_4B^{(3)}_3,\quad
 C^{(45)}=A^{(1)}_1B^{(1)}_2A^{(2)}_1A^{(3)}_4B^{(3)}_4,
\end{equation}
which cannot be simultaneously zero, and 
in the case \eqref{eqn:case3_v1_pole_4} we obtain
\begin{equation}
 C^{(13)}=A^{(1)}_2B^{(1)}_2A^{(2)}_1A^{(3)}_2B^{(3)}_2,\quad
 C^{(45)}=A^{(1)}_1B^{(1)}_2A^{(2)}_4A^{(3)}_2B^{(3)}_1,
\end{equation}
which cannot be simultaneously zero,
Therefore, we have completed the proof.
\end{proof}

Note here that from Lemmas \ref{lemma:case3_u1u4_pole1}, \ref{lemma:case3_u1u4_pole2}, \ref{lemma:case3_u1u4_zero} and \ref{lemma:case3_v1_pole}, we have the results that $u_1=0$ and $u_4=0$ are not a pole nor a zero of $v_4$, and $v_1=0$ is not a pole of $v_4$. 

Since 
\begin{align*}
 &G_1=-{u_4}^2C^{(31)}+O(u_1),
 &&G_2=u_1u_4^2C^{(41)}+O({u_1}^2)\qquad(u_1\to0),\\
 &F_1=-{u_1}^2C^{(13)}+O(u_4),
 &&F_2={u_1}^2u_4C^{(23)}+O({u_4}^2)\qquad (u_4\to0),
\end{align*}
and $u_1=0$ and $u_4=0$ are not a pole nor a zero of $v_4$, we obtain
\begin{equation}
 C^{(13)}=C^{(31)}=0.
\end{equation}
Since $v_1=0$ is not a pole of $v_4$, $v_4^{(r)}$ and $v_4^{(l)}$ can be expressed as
\begin{equation}
 v_4^{(r)}=-\dfrac{F_{11}+F_{12}v_1}{1+F_{21}v_1},\quad
 v_4^{(l)}=-\dfrac{G_{11}+G_{12}v_1}{1+G_{21}v_1},
\end{equation}
where
\begin{align*}
 &F_{11}=\dfrac{u_4C^{(11)}+u_1C^{(15)}}{{u_4}^2C^{(21)}+u_1(u_1C^{(23)}+u_4C^{(25)})},
 &&F_{12}=\dfrac{{u_4}^2C^{(12)}+u_1(u_1C^{(14)}+u_4C^{(16)})}{{u_4}^2C^{(21)}+u_1(u_1C^{(23)}+u_4C^{(25)})},\\
 &F_{21}=\dfrac{u_4\Big({u_4}^2C^{(22)}+u_1(u_1C^{(24)}+u_4C^{(26)})\Big)}{{u_4}^2C^{(21)}+u_1(u_1C^{(23)}+u_4C^{(25)})},
 &&G_{11}=\dfrac{u_1C^{(33)}+u_4C^{(35)}}{{u_4}^2C^{(41)}+u_1(u_1C^{(43)}+u_4C^{(45)})},\\
 &G_{12}=\dfrac{{u_4}^2C^{(32)}+u_1(u_1C^{(34)}+u_4C^{(36)})}{{u_4}^2C^{(41)}+u_1(u_1C^{(43)}+u_4C^{(45)})},
 &&G_{21}=\dfrac{u_1\Big({u_4}^2C^{(42)}+u_1(u_1C^{(44)}+u_4C^{(46)})\Big)}{{u_4}^2C^{(41)}+u_1(u_1C^{(43)}+u_4C^{(45)})}.
\end{align*}
Because of $v_4^{(r)}=v_4^{(l)}$, we obtain $F_{21}=G_{21}$, which is equivalent to
\begin{align}
 &{u_4}^5C^{(22)}C^{(41)}-{u_1}^5C^{(23)}C^{(44)}+u_1{u_4}^4(C^{(26)}C^{(41)}-C^{(21)}C^{(42)}+C^{(22)}C^{(45)})\notag\\
 &\quad+{u_1}^2{u_4}^3(C^{(24)}C^{(41)}-C^{(25)}C^{(42)}+C^{(22)}C^{(43)}+C^{(26)}C^{(45)}-C^{(21)}C^{(46)})\notag\\
 &\quad+{u_1}^4u_4(C^{(24)}C^{(43)}-C^{(25)}C^{(44)}-C^{(23)}C^{(46)})\notag\\
 &\quad-{u_1}^3{u_4}^2(C^{(23)}C^{(42)}-C^{(26)}C^{(43)}+C^{(21)}C^{(44)}-C^{(24)}C^{(45)}+C^{(25)}C^{(46)})=0.
\end{align}
Therefore, we obtain
\begin{equation}
 C^{(23)}C^{(44)}=0,\quad
 C^{(22)}C^{(41)}=0.
\end{equation}
Moreover, the coefficients of $u_1{u_4}^6{v_1}^2$ and ${u_1}^6u_4{v_1}^2$ in $F_1G_2-F_2G_1=0$ give
\begin{equation}
 C^{(22)}C^{(32)}=0,\quad
 C^{(14)}C^{(44)}=0,
\end{equation}
respectively.
Therefore, we obtain the following four cases:
\begin{subequations}
\begin{align}
 &C^{(13)}=C^{(14)}=C^{(23)}=C^{(31)}=C^{(32)}=C^{(41)}=0;\label{eqn:case3_condp1}\\
 &C^{(13)}=C^{(31)}=C^{(32)}=C^{(41)}=C^{(44)}=0;\label{eqn:case3_condp2}\\
 &C^{(13)}=C^{(14)}=C^{(22)}=C^{(23)}=C^{(31)}=0;\label{eqn:case3_condp3}\\
 &C^{(13)}=C^{(22)}=C^{(31)}=C^{(44)}=0.\label{eqn:case3_condp4_1}
\end{align}
\end{subequations}
The following three lemmas show that conditions \eqref{eqn:case3_condp1}--\eqref{eqn:case3_condp3} are inadequate.

\begin{lemma}\label{lemma:case3_condp1}
The condition \eqref{eqn:case3_condp1} is inadequate.
\end{lemma}
\begin{proof}
Since 
\begin{subequations}
\begin{align}
 0&=(A^{(1)}_4A^{(2)}_4B^{(1)}_2+A^{(1)}_1A^{(2)}_2B^{(1)}_4)C^{(13)}-(A^{(1)}_2A^{(2)}_4B^{(1)}_2+A^{(1)}_3A^{(2)}_2B^{(1)}_4)C^{(14)}\notag\\
 &=(A^{(1)}_1A^{(1)}_2-A^{(1)}_3A^{(1)}_4)(A^{(2)}_1A^{(2)}_2-A^{(2)}_3A^{(2)}_4)B^{(1)}_2B^{(1)}_4A^{(3)}_2B^{(3)}_2,\\
 0&=(A^{(1)}_4A^{(2)}_1B^{(1)}_2+A^{(1)}_1A^{(2)}_3B^{(1)}_4)C^{(13)}-(A^{(1)}_2A^{(2)}_1B^{(1)}_2+A^{(1)}_3A^{(2)}_3B^{(1)}_4)C^{(14)}\notag\\
 &=-(A^{(1)}_1A^{(1)}_2-A^{(1)}_3A^{(1)}_4)(A^{(2)}_1A^{(2)}_2-A^{(2)}_3A^{(2)}_4)B^{(1)}_2B^{(1)}_4B^{(3)}_3A^{(3)}_4,\\
 0&=(A^{(2)}_4A^{(3)}_3B^{(3)}_1+A^{(2)}_1A^{(3)}_1B^{(3)}_4)C^{(31)}-(A^{(2)}_4A^{(3)}_2B^{(3)}_1+A^{(2)}_1A^{(3)}_4B^{(3)}_4)C^{(32)}\notag\\
 &=(A^{(2)}_1A^{(2)}_2-A^{(2)}_3A^{(2)}_4)(A^{(3)}_1A^{(3)}_2-A^{(3)}_3A^{(3)}_4)B^{(1)}_1A^{(1)}_2B^{(3)}_1B^{(3)}_4,\\
 0&=(A^{(2)}_2A^{(3)}_3B^{(3)}_1+A^{(2)}_3A^{(3)}_1B^{(3)}_4)C^{(31)}-(A^{(2)}_2A^{(3)}_2B^{(3)}_1+A^{(2)}_3A^{(3)}_4B^{(3)}_4)C^{(32)}\notag\\
 &=-(A^{(2)}_1A^{(2)}_2-A^{(2)}_3A^{(2)}_4)(A^{(3)}_1A^{(3)}_2-A^{(3)}_3A^{(3)}_4)A^{(1)}_3B^{(1)}_3B^{(3)}_1B^{(3)}_4,\\
 0&=(A^{(2)}_3A^{(3)}_3B^{(3)}_2+A^{(2)}_2A^{(3)}_1B^{(3)}_3)C^{(13)}-(A^{(2)}_3A^{(3)}_2B^{(3)}_2+A^{(2)}_2A^{(3)}_4B^{(3)}_3)C^{(23)}\notag\\
 &=(A^{(2)}_1A^{(2)}_2-A^{(2)}_3A^{(2)}_4)(A^{(3)}_1A^{(3)}_2-A^{(3)}_3A^{(3)}_4)A^{(1)}_2B^{(1)}_2B^{(3)}_2B^{(3)}_3,\\
 0&=(A^{(2)}_1A^{(3)}_3B^{(3)}_2+A^{(2)}_4A^{(3)}_1B^{(3)}_3)C^{(13)}-(A^{(2)}_1A^{(3)}_2B^{(3)}_2+A^{(2)}_4A^{(3)}_4B^{(3)}_3)C^{(23)}\notag\\
 &=-(A^{(2)}_1A^{(2)}_2-A^{(2)}_3A^{(2)}_4)(A^{(3)}_1A^{(3)}_2-A^{(3)}_3A^{(3)}_4)A^{(1)}_3B^{(1)}_4B^{(3)}_2B^{(3)}_3,\\
 0&=(A^{(1)}_4A^{(2)}_3B^{(1)}_1+A^{(1)}_1A^{(2)}_1B^{(1)}_3)C^{(31)}-(A^{(1)}_2A^{(2)}_3B^{(1)}_1+A^{(1)}_3A^{(2)}_1B^{(1)}_3)C^{(41)}\notag\\
 &=(A^{(1)}_1A^{(1)}_2-A^{(1)}_3A^{(1)}_4)(A^{(2)}_1A^{(2)}_2-A^{(2)}_3A^{(2)}_4)B^{(1)}_1B^{(1)}_3B^{(3)}_1A^{(3)}_2,\\
 0&=(A^{(1)}_4A^{(2)}_2B^{(1)}_1+A^{(1)}_1A^{(2)}_4B^{(1)}_3)C^{(31)}-(A^{(1)}_2A^{(2)}_2B^{(1)}_1+A^{(1)}_3A^{(2)}_4B^{(1)}_3)C^{(41)}\notag\\
 &=-(A^{(1)}_1A^{(1)}_2-A^{(1)}_3A^{(1)}_4)(A^{(2)}_1A^{(2)}_2-A^{(2)}_3A^{(2)}_4)B^{(1)}_1B^{(1)}_3A^{(3)}_4B^{(3)}_4,
\end{align}
we obtain the following four cases:
\begin{align}
 &A^{(1)}_1=B^{(1)}_1=A^{(1)}_2=B^{(1)}_2=A^{(3)}_1=B^{(3)}_1=A^{(3)}_2=B^{(3)}_2=0;\label{eqn:case3_condp1_1}\\
 &A^{(1)}_1=B^{(1)}_1=A^{(1)}_2=B^{(1)}_2=A^{(3)}_3=B^{(3)}_3=A^{(3)}_4=B^{(3)}_4=0;\label{eqn:case3_condp1_2}\\
 &A^{(1)}_3=B^{(1)}_3=A^{(1)}_4=B^{(1)}_4=A^{(3)}_1=B^{(3)}_1=A^{(3)}_2=B^{(3)}_2=0;\label{eqn:case3_condp1_3}\\
 &A^{(1)}_3=B^{(1)}_3=A^{(1)}_4=B^{(1)}_4=A^{(3)}_3=B^{(3)}_3=A^{(3)}_4=B^{(3)}_4=0.\label{eqn:case3_condp1_4}
\end{align}
\end{subequations}

However, all cases are inadmissible.
Indeed, if \eqref{eqn:case3_condp1_1}, then from
\begin{equation}
 0=C^{(13)}=A^{(1)}_3B^{(1)}_4A^{(2)}_2B^{(3)}_3A^{(3)}_4,\quad
 0=C^{(31)}=A^{(1)}_3B^{(1)}_3A^{(2)}_1A^{(3)}_4B^{(3)}_4,
\end{equation}
we obtain $A^{(2)}_1=B^{(2)}_1=A^{(2)}_2=B^{(2)}_2=0$.
Then, we obtain 
\begin{equation}
 v_4^{(r)}=-\dfrac{A^{(1)}_4A^{(2)}_4A^{(3)}_4B^{(1)}_3B^{(3)}_3}{A^{(1)}_3A^{(2)}_3A^{(3)}_3B^{(1)}_4B^{(3)}_4}v_1,
\end{equation}
which is inconsistent with the square property.
Similarly, if \eqref{eqn:case3_condp1_2}, then from
\begin{equation}
 0=C^{(13)}=A^{(1)}_3B^{(1)}_4A^{(2)}_3A^{(3)}_2B^{(3)}_2,\quad
 0=C^{(31)}=A^{(1)}_3B^{(1)}_3A^{(2)}_4A^{(3)}_2B^{(3)}_1,
\end{equation}
we obtain $A^{(2)}_3=B^{(2)}_3=A^{(2)}_4=B^{(2)}_4=0$,
which gives the following inadequate result:
\begin{equation}
 v_4^{(r)}=-\dfrac{A^{(1)}_4A^{(2)}_1A^{(3)}_2B^{(1)}_3B^{(3)}_2}{A^{(1)}_3A^{(2)}_2A^{(3)}_1B^{(1)}_4B^{(3)}_1}v_1,
\end{equation}
if \eqref{eqn:case3_condp1_3}, then from
\begin{equation}
 0=C^{(13)}=A^{(1)}_2B^{(1)}_2A^{(2)}_4A^{(3)}_4B^{(3)}_3,\quad
 0=C^{(31)}=B^{(1)}_1A^{(1)}_2A^{(2)}_3A^{(3)}_4B^{(3)}_4,
\end{equation}
we obtain $A^{(2)}_3=B^{(2)}_3=A^{(2)}_4=B^{(2)}_4=0$,
which gives the following inadequate result:
\begin{equation}
 v_4^{(r)}=-\dfrac{A^{(1)}_1A^{(2)}_2A^{(3)}_4B^{(1)}_1B^{(3)}_3}{A^{(1)}_2A^{(2)}_1A^{(3)}_3B^{(1)}_2B^{(3)}_4}v_1,
\end{equation}
and if \eqref{eqn:case3_condp1_4}, then from
\begin{equation}
 0=C^{(13)}=A^{(1)}_2B^{(1)}_2A^{(2)}_1A^{(3)}_2B^{(3)}_2,\quad
 0=C^{(31)}=B^{(1)}_1A^{(1)}_2A^{(2)}_2A^{(3)}_2B^{(3)}_1,
\end{equation}
we obtain $A^{(2)}_1=B^{(2)}_1=A^{(2)}_2=B^{(2)}_2=0$, 
which gives the following inadequate result:
\begin{equation}
 v_4^{(r)}=-\dfrac{A^{(1)}_1A^{(2)}_3A^{(3)}_2B^{(1)}_1B^{(3)}_2}{A^{(1)}_2A^{(2)}_4A^{(3)}_1B^{(1)}_2B^{(3)}_1}v_1,
\end{equation}
Therefore, we have completed the proof.
\end{proof}

\begin{lemma}\label{lemma:case3_condp2}
The condition \eqref{eqn:case3_condp2} is inadequate.
\end{lemma}
\begin{proof}
Since 
\begin{subequations}
\begin{align}
 0&=(A^{(2)}_4A^{(3)}_3B^{(3)}_1+A^{(2)}_1A^{(3)}_1B^{(3)}_4)C^{(31)}-(A^{(2)}_4A^{(3)}_2B^{(3)}_1+A^{(2)}_1A^{(3)}_4B^{(3)}_4)C^{(32)}\notag\\
 &=(A^{(2)}_1A^{(2)}_2-A^{(2)}_3A^{(2)}_4)(A^{(3)}_1A^{(3)}_2-A^{(3)}_3A^{(3)}_4)B^{(1)}_1A^{(1)}_2B^{(3)}_1B^{(3)}_4,\\
 0&=(A^{(2)}_2A^{(3)}_3B^{(3)}_1+A^{(2)}_3A^{(3)}_1B^{(3)}_4)C^{(31)}-(A^{(2)}_2A^{(3)}_2B^{(3)}_1+A^{(2)}_3A^{(3)}_4B^{(3)}_4)C^{(32)}\notag\\
 &=-(A^{(2)}_1A^{(2)}_2-A^{(2)}_3A^{(2)}_4)(A^{(3)}_1A^{(3)}_2-A^{(3)}_3A^{(3)}_4)A^{(1)}_3B^{(1)}_3B^{(3)}_1B^{(3)}_4,\\
 0&=(A^{(1)}_4A^{(2)}_3B^{(1)}_1+A^{(1)}_1A^{(2)}_1B^{(1)}_3)C^{(31)}-(A^{(1)}_2A^{(2)}_3B^{(1)}_1+A^{(1)}_3A^{(2)}_1B^{(1)}_3)C^{(41)}\notag\\
 &=(A^{(1)}_1A^{(1)}_2-A^{(1)}_3A^{(1)}_4)(A^{(2)}_1A^{(2)}_2-A^{(2)}_3A^{(2)}_4)B^{(1)}_1B^{(1)}_3B^{(3)}_1A^{(3)}_2,\\
 0&=(A^{(1)}_4A^{(2)}_2B^{(1)}_1+A^{(1)}_1A^{(2)}_4B^{(1)}_3)C^{(31)}-(A^{(1)}_2A^{(2)}_2B^{(1)}_1+A^{(1)}_3A^{(2)}_4B^{(1)}_3)C^{(41)}\notag\\
 &=-(A^{(1)}_1A^{(1)}_2-A^{(1)}_3A^{(1)}_4)(A^{(2)}_1A^{(2)}_2-A^{(2)}_3A^{(2)}_4)B^{(1)}_1B^{(1)}_3A^{(3)}_4B^{(3)}_4,
\end{align}
\end{subequations}
we obtain the following four cases:
\begin{subequations}
\begin{align}
 &A^{(1)}_1=B^{(1)}_1=A^{(3)}_1=B^{(3)}_1=0;\label{eqn:case3_condp2_1}\\
 &A^{(1)}_1=B^{(1)}_1=A^{(3)}_4=B^{(3)}_4=0;\label{eqn:case3_condp2_2}\\
 &A^{(1)}_3=B^{(1)}_3=A^{(3)}_1=B^{(3)}_1=0;\label{eqn:case3_condp2_3}\\
 &A^{(1)}_3=B^{(1)}_3=A^{(3)}_4=B^{(3)}_4= 0.\label{eqn:case3_condp2_4}
\end{align}
\end{subequations}
However, all cases are inadmissible.
Indeed, if \eqref{eqn:case3_condp2_1}, then from
\begin{equation}
 0=C^{(31)}=A^{(1)}_3B^{(1)}_3A^{(2)}_1A^{(3)}_4B^{(3)}_4,\quad
 0=C^{(44)}=A^{(1)}_4B^{(1)}_4A^{(2)}_2A^{(3)}_3B^{(3)}_3,
\end{equation}
we obtain $A^{(2)}_1=B^{(2)}_1=A^{(2)}_2=B^{(2)}_2=0$.
Then, with the condition \eqref{eqn:caseN3_CAO_B2}, we obtain 
\begin{equation}
 v_4^{(r)}=-\dfrac{A^{(1)}_2A^{(2)}_4A^{(3)}_4B^{(1)}_3B^{(3)}_3+A^{(1)}_3A^{(2)}_3A^{(3)}_2B^{(1)}_4B^{(3)}_4+A^{(1)}_4A^{(2)}_4A^{(3)}_4B^{(1)}_3B^{(3)}_3u_4v_1}{A^{(1)}_3A^{(2)}_3A^{(3)}_3B^{(1)}_4B^{(3)}_4u_4},
\end{equation}
which is inconsistent with the square property.
Similarly, if \eqref{eqn:case3_condp2_2}, then from
\begin{equation}
 C^{(31)}=0,\quad
 C^{(35)}=0,\quad
 C^{(13)}=0,\quad
 C^{(44)}=0,
\end{equation}
we respectively obtain
\begin{equation}
 A^{(2)}_4=B^{(2)}_4=0,\quad
 A^{(1)}_2=B^{(1)}_2=0,\quad
 A^{(2)}_3=B^{(2)}_3=0,\quad
 A^{(3)}_3=B^{(3)}_3=0,
\end{equation}
and then with the condition \eqref{eqn:caseN3_CAO_B2} we obtain the following inadequate result:
\begin{equation}
 v_4^{(r)}=-\dfrac{A^{(1)}_4A^{(2)}_1A^{(3)}_2B^{(1)}_3B^{(3)}_2}{A^{(1)}_3A^{(2)}_2A^{(3)}_1B^{(1)}_4B^{(3)}_1}v_1.
\end{equation}
If \eqref{eqn:case3_condp2_3}, then from
\begin{equation}
 C^{(31)}=0,\quad
 C^{(35)}=0,\quad
 C^{(13)}=0,
\end{equation}
we respectively obtain
\begin{equation}
 A^{(2)}_3=B^{(2)}_3=0,\quad
 A^{(3)}_2=B^{(3)}_2=0,\quad
 A^{(2)}_4=B^{(2)}_4=0,
\end{equation}
and then with the condition \eqref{eqn:caseN3_CAO_B2} we obtain the following inadequate result:
\begin{equation}
 v_4^{(r)}=-\dfrac{A^{(1)}_1A^{(2)}_2A^{(3)}_4B^{(1)}_1B^{(3)}_3}{A^{(1)}_2A^{(2)}_1A^{(3)}_3B^{(1)}_2B^{(3)}_4}v_1.
\end{equation}
Finally, if \eqref{eqn:case3_condp2_4}, then from
\begin{equation}
 0=C^{(13)}=A^{(1)}_2B^{(1)}_2A^{(2)}_1A^{(3)}_2B^{(3)}_2,\quad
 0=C^{(31)}=B^{(1)}_1A^{(1)}_2A^{(2)}_2B^{(3)}_1A^{(3)}_2,
\end{equation}
we obtain $A^{(2)}_1=B^{(2)}_1=A^{(2)}_2=B^{(2)}_2=0$,
and then with the condition \eqref{eqn:caseN3_CAO_B2} we obtain the following inadequate result:
\begin{equation}
 v_4^{(r)}=-\dfrac{A^{(1)}_1A^{(2)}_3A^{(3)}_2B^{(1)}_1B^{(3)}_2v_1}{A^{(1)}_2A^{(2)}_4A^{(3)}_1B^{(1)}_2B^{(3)}_1+(A^{(1)}_4A^{(2)}_4A^{(3)}_1B^{(1)}_2B^{(3)}_1+A^{(1)}_1A^{(2)}_3A^{(3)}_3B^{(1)}_1B^{(3)}_2)u_4v_1}.
\end{equation}
Therefore, we have completed the proof.
\end{proof}

\begin{lemma}\label{lemma:case3_condp3}
The condition \eqref{eqn:case3_condp3} is inadequate.
\end{lemma}
\begin{proof}
Since
\begin{equation*}
 F_1=-u_1u_4C^{(15)}+O({u_4}^2),\quad
 F_2=u_1{u_4}^2(u_1v_1C^{(24)}+C^{(25)})+O({u_4}^3),
\end{equation*}
as $u_4\to0$ and $u_4=0$ is not a pole nor zero of $v_4$, we obtain
\begin{equation}
 C^{(15)}=0.
\end{equation}
Since 
\begin{subequations}
\begin{align}
 0&=(A^{(1)}_4A^{(2)}_4B^{(1)}_2+A^{(1)}_1A^{(2)}_2B^{(1)}_4)C^{(13)}-(A^{(1)}_2A^{(2)}_4B^{(1)}_2+A^{(1)}_3A^{(2)}_2B^{(1)}_4)C^{(14)}\notag\\
 &=(A^{(1)}_1A^{(1)}_2-A^{(1)}_3A^{(1)}_4)(A^{(2)}_1A^{(2)}_2-A^{(2)}_3A^{(2)}_4)B^{(1)}_2B^{(1)}_4A^{(3)}_2B^{(3)}_2,\\
 0&=(A^{(1)}_4A^{(2)}_1B^{(1)}_2+A^{(1)}_1A^{(2)}_3B^{(1)}_4)C^{(13)}-(A^{(1)}_2A^{(2)}_1B^{(1)}_2+A^{(1)}_3A^{(2)}_3B^{(1)}_4)C^{(14)}\notag\\
 &=-(A^{(1)}_1A^{(1)}_2-A^{(1)}_3A^{(1)}_4)(A^{(2)}_1A^{(2)}_2-A^{(2)}_3A^{(2)}_4)B^{(1)}_2B^{(1)}_4B^{(3)}_3A^{(3)}_4,\\
 0&=(A^{(2)}_3A^{(3)}_3B^{(3)}_2+A^{(2)}_2A^{(3)}_1B^{(3)}_3)C^{(13)}-(A^{(2)}_3A^{(3)}_2B^{(3)}_2+A^{(2)}_2A^{(3)}_4B^{(3)}_3)C^{(23)}\notag\\
 &=(A^{(2)}_1A^{(2)}_2-A^{(2)}_3A^{(2)}_4)(A^{(3)}_1A^{(3)}_2-A^{(3)}_3A^{(3)}_4)A^{(1)}_2B^{(1)}_2B^{(3)}_2B^{(3)}_3,\\
 0&=(A^{(2)}_1A^{(3)}_3B^{(3)}_2+A^{(2)}_4A^{(3)}_1B^{(3)}_3)C^{(13)}-(A^{(2)}_1A^{(3)}_2B^{(3)}_2+A^{(2)}_4A^{(3)}_4B^{(3)}_3)C^{(23)}\notag\\
 &=-(A^{(2)}_1A^{(2)}_2-A^{(2)}_3A^{(2)}_4)(A^{(3)}_1A^{(3)}_2-A^{(3)}_3A^{(3)}_4)A^{(1)}_3B^{(1)}_4B^{(3)}_2B^{(3)}_3,
\end{align}
\end{subequations}
we obtain the following four cases:
\begin{subequations}
\begin{align}
 &A^{(1)}_2=B^{(1)}_2=A^{(3)}_2=B^{(3)}_2=0;\label{eqn:case3_condp3_1}\\
 &A^{(1)}_2=B^{(1)}_2=A^{(3)}_3=B^{(3)}_3=0;\label{eqn:case3_condp3_2}\\
 &A^{(1)}_4=B^{(1)}_4=A^{(3)}_2=B^{(3)}_2=0;\label{eqn:case3_condp3_3}\\
 &A^{(1)}_4=B^{(1)}_4=A^{(3)}_3=B^{(3)}_3=0.\label{eqn:case3_condp3_4}
\end{align}
\end{subequations}
However, all cases are inadmissible.
Indeed, if \eqref{eqn:case3_condp3_1}, then from
\begin{equation}
 0=C^{(13)}=A^{(1)}_3B^{(1)}_4A^{(2)}_2A^{(3)}_4B^{(3)}_3,\quad
 0=C^{(31)}=A^{(1)}_3B^{(1)}_3A^{(2)}_1A^{(3)}_4B^{(3)}_4,
\end{equation}
we obtain $A^{(2)}_1=B^{(2)}_1=A^{(2)}_2=B^{(2)}_2=0$.
Then, with the condition \eqref{eqn:caseN3_CAO_B2}, we obtain 
\begin{equation}
 v_4^{(r)}=-\dfrac{A^{(1)}_4A^{(2)}_4A^{(3)}_4B^{(1)}_3B^{(3)}_3v_1}{A^{(1)}_3A^{(2)}_3A^{(3)}_3B^{(1)}_4B^{(3)}_4+(A^{(1)}_4A^{(2)}_4A^{(3)}_1B^{(1)}_3B^{(3)}_3+A^{(1)}_1A^{(2)}_3A^{(3)}_3B^{(1)}_4B^{(3)}_4)u_4v_1},
\end{equation}
which is inconsistent with the square property.
Similarly, if \eqref{eqn:case3_condp3_2}, then from
\begin{equation}
 C^{(13)}=0,\quad
 C^{(15)}=0,\quad
 C^{(31)}=0,\quad
 C^{(22)}=0,
\end{equation}
we respectively obtain
\begin{equation}
 A^{(2)}_3=B^{(2)}_3=0,\quad
 A^{(3)}_4=B^{(3)}_4=0,\quad
 A^{(2)}_4=B^{(2)}_4=0,\quad
 A^{(1)}_1=B^{(1)}_1=0,
\end{equation}
and then with the condition \eqref{eqn:caseN3_CAO_B2} we obtain the following inadequate result:
\begin{equation}
 v_4^{(r)}=-\dfrac{A^{(1)}_4A^{(2)}_1A^{(3)}_2B^{(1)}_3B^{(3)}_2}{A^{(1)}_3A^{(2)}_2A^{(3)}_1B^{(1)}_4B^{(3)}_1}v_1.
\end{equation}
If \eqref{eqn:case3_condp3_3}, then from
\begin{equation}
 C^{(13)}=0,\quad
 C^{(15)}=0,\quad
 C^{(31)}=0,\quad
 C^{(22)}=0,\quad
\end{equation}
we respectively obtain
\begin{equation}
 A^{(2)}_4=B^{(2)}_4=0,\quad
 A^{(1)}_3=B^{(1)}_3=0,\quad
 A^{(2)}_3=B^{(2)}_3=0,\quad
 A^{(3)}_1=B^{(3)}_1=0,
\end{equation}
and then with the condition \eqref{eqn:caseN3_CAO_B2} we obtain the following inadequate result:
\begin{equation}
 v_4^{(r)}=-\dfrac{A^{(1)}_1A^{(2)}_2A^{(3)}_4B^{(1)}_1B^{(3)}_3}{A^{(1)}_2A^{(2)}_1A^{(3)}_3B^{(1)}_2B^{(3)}_4}v_1.
\end{equation}
Finally, if \eqref{eqn:case3_condp3_4}, then from
\begin{equation}
 0=C^{(13)}=A^{(1)}_2B^{(1)}_2A^{(2)}_1A^{(3)}_2B^{(3)}_2,\quad
 0=C^{(22)}=A^{(1)}_1B^{(1)}_1A^{(2)}_2A^{(3)}_1B^{(3)}_1,
\end{equation}
we obtain $A^{(2)}_1=B^{(2)}_1=A^{(2)}_2=B^{(2)}_2=0$,
and then with the condition \eqref{eqn:caseN3_CAO_B2} we obtain the following inadequate result:
\begin{equation}
 v_4^{(r)}=-\dfrac{A^{(1)}_2A^{(2)}_4A^{(3)}_4B^{(1)}_2B^{(3)}_1+A^{(1)}_3A^{(2)}_3A^{(3)}_2B^{(1)}_1B^{(3)}_2+A^{(1)}_1A^{(2)}_3A^{(3)}_2B^{(1)}_1B^{(3)}_2u_4v_1}{A^{(1)}_2A^{(2)}_4A^{(3)}_1B^{(1)}_2B^{(3)}_1u_4}.
\end{equation}
Therefore, we have completed the proof.
\end{proof}

From Lemmas \ref{lemma:case3_condp1}--\ref{lemma:case3_condp3}, the condition \eqref{eqn:case3_condp4_1} holds.
This condition leads the conditions \eqref{eqn:typeII_C_nonzero} and \eqref{eqn:typeII_AB_nonzero} as shown in the following lemma.

\begin{lemma}\label{lemma:case3_condp4_Cneq0}
The conditions \eqref{eqn:typeII_C_nonzero} and \eqref{eqn:typeII_AB_nonzero} hold.
\end{lemma}
\begin{proof}
We shall prove the lemma by dividing the proof into the following five steps:
\begin{center}
\begin{description}
\item[\underline{Step 1}]
$C^{(23)},C^{(41)}\neq0$;
\item[\underline{Step 2}]
$C^{(14)}\neq0$;
\item[\underline{Step 3}]
$C^{(32)}\neq0$;
\item[\underline{Step 4}]
$A^{(i)}_1,A^{(i)}_2,A^{(i)}_3,A^{(i)}_4,
 B^{(i)}_1,B^{(i)}_2,B^{(i)}_3,B^{(i)}_4\neq0,\quad (i=1,3);$
\item[\underline{Step 5}]
$C^{(12)},C^{(21)},C^{(34)},C^{(43)}\neq0$.
\end{description}
\end{center}

\underline{\bf Step 1.}\quad
The following hold:
\begin{align*}
 &v_4^{(r)}=-\dfrac{u_1v_1C^{(14)}+C^{(15)}}{u_1C^{(23)}}+O(u_4)&& (u_4\to0),\\
 &v_4^{(l)}=-\dfrac{u_4v_1C^{(32)}+C^{(35)}}{u_4C^{(41)}}+O(u_1)&& (u_1\to0).
\end{align*}
Since $u_1=0$ and $u_4=0$ are not a pole nor a zero of $v_4$,
if $C^{(23)}=0$, then we obtain $C^{(14)}=C^{(15)}=0$, which satisfy inadequate condition \ref{lemma:case3_condp3},
and if $C^{(41)}=0$, then we obtain $C^{(32)}=C^{(35)}=0$, which satisfy inadequate condition \ref{lemma:case3_condp2}.
Therefore, we obtain
\begin{equation}
 C^{(23)},C^{(41)}\neq0.
\end{equation}

\underline{\bf Step 2.}\quad
Assume $C^{(14)}=0$.
Then, from
\begin{subequations}
\begin{align}
 0&=(A^{(1)}_4A^{(2)}_4B^{(1)}_2+A^{(1)}_1A^{(2)}_2B^{(1)}_4)C^{(13)}-(A^{(1)}_2A^{(2)}_4B^{(1)}_2+A^{(1)}_3A^{(2)}_2B^{(1)}_4)C^{(14)}\notag\\
 &=(A^{(1)}_1A^{(1)}_2-A^{(1)}_3A^{(1)}_4)(A^{(2)}_1A^{(2)}_2-A^{(2)}_3A^{(2)}_4)B^{(1)}_2B^{(1)}_4A^{(3)}_2B^{(3)}_2,\\
 0&=(A^{(1)}_4A^{(2)}_1B^{(1)}_2+A^{(1)}_1A^{(2)}_3B^{(1)}_4)C^{(13)}-(A^{(1)}_2A^{(2)}_1B^{(1)}_2+A^{(1)}_3A^{(2)}_3B^{(1)}_4)C^{(14)}\notag\\
 &=-(A^{(1)}_1A^{(1)}_2-A^{(1)}_3A^{(1)}_4)(A^{(2)}_1A^{(2)}_2-A^{(2)}_3A^{(2)}_4)B^{(1)}_2B^{(1)}_4B^{(3)}_3A^{(3)}_4,
\end{align} 
\end{subequations}
we obtain
\begin{equation}
 A^{(1)}_2=B^{(1)}_2=0
\quad\text{or}\quad
 A^{(1)}_4=B^{(1)}_4=0.
\end{equation}
Below, we show that both cases are inadmissible.

Firstly, we consider the case $A^{(1)}_2=B^{(1)}_2=0$.
From
\begin{equation*}
 G_1=-u_1u_4(C^{(35)}+u_1v_1C^{(36)})+O({u_4}^2),\quad
 G_2={u_1}^3C^{(43)}+O(u_4),
\end{equation*}
as $u_4\to0$ and $u_4=0$ is not a zero of $v_4$, we obtain
\begin{equation}
 C^{(43)}=0.
\end{equation}
Since
\begin{subequations}
\begin{align}
 0&=(A^{(2)}_1A^{(3)}_1B^{(3)}_2+A^{(2)}_4A^{(3)}_3B^{(3)}_3)C^{(43)}-(A^{(2)}_1A^{(3)}_4B^{(3)}_2+A^{(2)}_4A^{(3)}_2B^{(3)}_3)C^{(44)}\notag\\
 &=(A^{(2)}_1A^{(2)}_2-A^{(2)}_3A^{(2)}_4)(A^{(3)}_1A^{(3)}_2-A^{(3)}_3A^{(3)}_4)A^{(1)}_4B^{(1)}_4B^{(3)}_2B^{(3)}_3,\\
 0&\neq C^{(23)}=A^{(1)}_3B^{(1)}_4(A^{(2)}_3A^{(3)}_3B^{(3)}_2+A^{(2)}_2A^{(3)}_1B^{(3)}_3),
\end{align}
\end{subequations}
hold, we obtain
\begin{equation}
 A^{(3)}_2=B^{(3)}_2=0,
\end{equation}
which gives
\begin{equation}
 C^{(35)}=0.
\end{equation}
Since $C^{(41)}\neq0$ and
\begin{equation*}
 G_1=-u_1u_4v_1(u_4C^{(32)}+u_1C^{(36)})+O({v_1}^2),\quad
 G_2=u_1u_4(u_4C^{(41)}+u_1C^{(45)})+O(v_1),
\end{equation*}
as $v_1\to0$, $v_1=0$ is a zero of $v_4$.
Then, from
\begin{equation*}
 F_1=-u_4(u_4C^{(11)}+u_1C^{(15)})+O(v_1),\quad
 F_2=u_4({u_4}^2C^{(21)}+{u_1}^2C^{(23)}+u_1u_4C^{(25)})+O(v_1),
\end{equation*}
as $v_1\to0$, we obtain
\begin{equation}
 C^{(11)}=C^{(15)}=0.
\end{equation}
Then, since
$C^{(23)}\neq0$ and
\begin{equation*}
 F_1=-u_1{u_4}^2v_1C^{(16)}+O({u_4}^3),\quad
 F_2={u_1}^2u_4C^{(23)}+O({u_4}^2),
\end{equation*}
as $u_4\to0$, $u_4=0$ must be a zero of $v_4$.
Therefore, the case $A^{(1)}_2=B^{(1)}_2=0$ is inadequate.

Next, we consider the case 
\begin{equation}
 A^{(1)}_4=B^{(1)}_4=0.
\end{equation} 
Then, 
\begin{equation}
 C^{(24)}=0
\end{equation}
holds.
Since
\begin{equation*}
 v_4^{(r)}=-\dfrac{C^{(15)}}{u_1C^{(23)}}+O(u_4),\quad
 v_4^{(l)}=-\dfrac{C^{(33)}+u_1v_1C^{(34)}}{u_1C^{(43)}}+O(u_4),
\end{equation*}
as $u_4\to0$ and $C^{(23)}\neq0$, we obtain
\begin{equation}
 C^{(34)}=0.
\end{equation}
Assume
\begin{equation}
 C^{(46)}\neq0.
\end{equation}
Then, from
\begin{align*}
 &F_1=-u_1u_4(C^{(15)}+u_4v_1C^{(16)})+O({u_1}^0),\quad
 F_2={u_1}^2u_4C^{(23)}+O(u_1),\\
 &G_1=-{u_1}^2(C^{(33)}+u_4v_1C^{(36)})+O(u_1),\quad
 G_2={u_1}^3(C^{(43)}+u_4v_1C^{(46)})+O({u_1}^2),
\end{align*}
as $u_1\to\infty$, and $C^{(23)},C^{(46)}\neq0$, we obtain
\begin{equation}
 C^{(16)}=C^{(33)}=C^{(43)}=0.
\end{equation}
From
\begin{subequations}
\begin{align}
 0&=(A^{(2)}_3A^{(3)}_1B^{(3)}_2+A^{(2)}_2A^{(3)}_3B^{(3)}_3)C^{(43)}-(A^{(2)}_3A^{(3)}_4B^{(3)}_2+A^{(2)}_2A^{(3)}_2B^{(3)}_3)C^{(44)}\notag\\
 &=-(A^{(2)}_1A^{(2)}_2-A^{(2)}_3A^{(2)}_4)(A^{(3)}_1A^{(3)}_2-A^{(3)}_3A^{(3)}_4)A^{(1)}_1B^{(1)}_2B^{(3)}_2B^{(3)}_3,\\
 0&\neq C^{(23)}=A^{(1)}_2B^{(1)}_2(A^{(2)}_1A^{(3)}_3B^{(3)}_2+A^{(2)}_4A^{(3)}_1B^{(3)}_3),
\end{align}
\end{subequations}
we obtain
\begin{equation}
 A^{(3)}_2=B^{(3)}_2=0.
\end{equation}
Since 
\begin{equation}
 C^{(13)}=A^{(1)}_2B^{(1)}_2A^{(2)}_4B^{(3)}_3A^{(3)}_4,\quad
 C^{(16)}=A^{(1)}_1B^{(1)}_1A^{(2)}_2B^{(3)}_3A^{(3)}_4,
\end{equation}
cannot be simultaneously zero, the assumption $C^{(46)}\neq0$ is inadequate.
Therefore, we obtain
\begin{equation}
 C^{(46)}=0.
\end{equation}
Since
\begin{align*}
 &0=B^{(1)}_2B^{(3)}_3C^{(46)}-(B^{(1)}_3B^{(3)}_3+B^{(1)}_2B^{(3)}_1)C^{(44)}
 =-A^{(1)}_1{B^{(1)}_2}^2A^{(2)}_1A^{(3)}_1(B^{(3)}_1B^{(3)}_2-B^{(3)}_3B^{(3)}_4),\\
 &0=B^{(1)}_2B^{(3)}_2C^{(46)}-(B^{(1)}_3B^{(3)}_2+B^{(1)}_2B^{(3)}_4)C^{(44)}
 =A^{(1)}_1{B^{(1)}_2}^2A^{(2)}_4A^{(3)}_3(B^{(3)}_1B^{(3)}_2-B^{(3)}_3B^{(3)}_4),\\
 &0\neq C^{(23)}=A^{(1)}_2B^{(1)}_2(A^{(2)}_1A^{(3)}_3B^{(3)}_2+A^{(2)}_4A^{(3)}_1B^{(3)}_3),
\end{align*}
we obtain
\begin{equation}
 A^{(2)}_4=B^{(2)}_4=A^{(3)}_1=B^{(3)}_1=0.
\end{equation}
Moreover, from
\begin{equation}
 0=C^{(13)}=A^{(1)}_2B^{(1)}_2A^{(2)}_1A^{(3)}_2B^{(3)}_2,\quad
 0=C^{(22)}=A^{(1)}_1B^{(1)}_1A^{(2)}_3A^{(3)}_3B^{(3)}_4,
\end{equation}
we obtain
\begin{equation}
 A^{(2)}_3=B^{(2)}_3=A^{(3)}_2=B^{(3)}_2=0,
\end{equation}
which is inconsistent with the condition $C^{(23)}\neq0$.
Therefore, the condition $A^{(1)}_4=B^{(1)}_4=0$ is also inadequate.

\underline{\bf Step 3.}\quad
Assume
\begin{equation}
 C^{(32)}=0.
\end{equation}
Then, from
\begin{subequations}
\begin{align}
 0&=(A^{(2)}_4A^{(3)}_3B^{(3)}_1+A^{(2)}_1A^{(3)}_1B^{(3)}_4)C^{(31)}-(A^{(2)}_4A^{(3)}_2B^{(3)}_1+A^{(2)}_1A^{(3)}_4B^{(3)}_4)C^{(32)}\notag\\
 &=(A^{(2)}_1A^{(2)}_2-A^{(2)}_3A^{(2)}_4)(A^{(3)}_1A^{(3)}_2-A^{(3)}_3A^{(3)}_4)B^{(1)}_1A^{(1)}_2B^{(3)}_1B^{(3)}_4,\\
 0&=(A^{(2)}_2A^{(3)}_3B^{(3)}_1+A^{(2)}_3A^{(3)}_1B^{(3)}_4)C^{(31)}-(A^{(2)}_2A^{(3)}_2B^{(3)}_1+A^{(2)}_3A^{(3)}_4B^{(3)}_4)C^{(32)}\notag\\
 &=-(A^{(2)}_1A^{(2)}_2-A^{(2)}_3A^{(2)}_4)(A^{(3)}_1A^{(3)}_2-A^{(3)}_3A^{(3)}_4)A^{(1)}_3B^{(1)}_3B^{(3)}_1B^{(3)}_4,
\end{align}
\end{subequations}
we obtain
\begin{equation}
 A^{(3)}_1=B^{(3)}_1=0
\quad\text{or}\quad
 A^{(3)}_4=B^{(3)}_4=0.
\end{equation}
Below, we show that both cases are inadmissible.

We first consider the case $A^{(3)}_1=B^{(3)}_1=0$.
Then, we obtain
\begin{equation}
 C^{(42)}=0.
\end{equation}
Form
\begin{equation*}
 v_4^{(r)}=-\dfrac{C^{(11)}+u_4v_1C^{(12)}}{u_4C^{(21)}}+O(u_1),\quad
 v_4^{(l)}=-\dfrac{C^{(35)}}{u_4C^{(41)}}+O(u_1),
\end{equation*}
as $u_1\to0$ and $C^{(41)}\neq0$, we obtain
\begin{equation}
 C^{(12)}=0,\quad C^{(35)}\neq0.
\end{equation}
Assume 
\begin{equation}
 C^{(26)}\neq0.
\end{equation}
Then, from
\begin{equation*}
 v_4^{(r)}=-\dfrac{C^{(11)}+u_1v_1C^{(16)}}{u_4(C^{(21)}+u_1v_1C^{(26)})}+O({u_4}^{-2}),\quad
 v_4^{(l)}=-\dfrac{C^{(35)}+u_1v_1C^{(36)}}{u_4 C^{(41)})}+O({u_4}^{-2}),
\end{equation*}
as $u_4\to\infty$ and $C^{(41)}\neq0$, we obtain
\begin{equation}
 C^{(11)}=C^{(21)}=C^{(36)}=0.
\end{equation}
Moreover, from
\begin{equation*}
 v_4^{(r)}=-\dfrac{C^{(15)}+u_4v_1C^{(16)}}{u_4(C^{(25)}+u_4v_1C^{(26)})}+O(u_1),\quad
 v_4^{(l)}=-\dfrac{C^{(35)}}{u_4C^{(41)}}+O(u_1),
\end{equation*}
as $u_1\to0$ and $C^{(26)},C^{(41)}\neq0$, we obtain
\begin{equation}
 C^{(15)}=C^{(25)}=0.
\end{equation}
From
\begin{equation*}
 v_4^{(r)}=O(v_1),\quad
 v_4^{(l)}=-\dfrac{u_1C^{(33)}+u_4C^{(35)}}{{u_4}^2C^{(41)}+{u_1}^2C^{(43)}+u_1u_4C^{(45)}}+O(v_1),
\end{equation*}
as $v_1\to0$ and $C^{(35)}\neq0$, the CACO property $v_4^{(r)}\neq v_4^{(l)}$ does not hold.
Therefore, the assumption $C^{(26)}\neq0$ is inadequate, and then we obtain
\begin{equation}
 C^{(26)}=0.
\end{equation}
From
\begin{align*}
 0&=(B^{(1)}_2B^{(3)}_4+B^{(1)}_3B^{(3)}_2)C^{(22)}-B^{(1)}_3B^{(3)}_4C^{(26)}
 =A^{(1)}_1A^{(2)}_3A^{(3)}_3{B^{(3)}_4}^2(B^{(1)}_1B^{(1)}_2-B^{(1)}_3B^{(1)}_4),\\
 0&=(B^{(1)}_4B^{(3)}_4+B^{(1)}_1B^{(3)}_2)C^{(22)}-B^{(1)}_1B^{(3)}_4C^{(26)}
 =-A^{(1)}_4A^{(2)}_1A^{(3)}_3{B^{(3)}_4}^2(B^{(1)}_1B^{(1)}_2-B^{(1)}_3B^{(1)}_4),\\
 0&\neq C^{(41)}=A^{(3)}_4B^{(3)}_4(A^{(1)}_4A^{(2)}_3B^{(1)}_1+A^{(1)}_1A^{(2)}_1B^{(1)}_3),
\end{align*}
we obtain
\begin{equation}
 A^{(1)}_4=B^{(1)}_4=A^{(2)}_3=B^{(2)}_3=0.
\end{equation}
From
\begin{equation}
 0=C^{(31)}=A^{(1)}_3B^{(1)}_3A^{(2)}_1A^{(3)}_4B^{(3)}_4,\quad
 0=C^{(44)}=A^{(1)}_1B^{(1)}_2A^{(2)}_4A^{(3)}_3B^{(3)}_3,
\end{equation}
we obtain
\begin{equation}
 A^{(1)}_3=B^{(1)}_3=A^{(2)}_4=B^{(2)}_4=0,
\end{equation}
which is inconsistent with the condition $C^{(41)}\neq0$.
Therefore, the case $A^{(3)}_1=B^{(3)}_1=0$ is inadequate.

Next, we consider the case $A^{(3)}_4=B^{(3)}_4=0$.
Then, we obtain
\begin{equation}
 C^{(11)}=C^{(12)}=0.
\end{equation}
From
\begin{equation*}
 v_4^{(r)}=-\dfrac{u_1(C^{(15)}+u_4v_1C^{(16)})}{{u_4}^2C^{(21)}}+O({u_1}^2),
\end{equation*}
as $u_1\to0$ and $u_1=0$ is not a zero of $v_4$, we obtain
\begin{equation}
 C^{(21)}=0.
\end{equation}
From
\begin{subequations}
\begin{align}
 0&=(A^{(1)}_1A^{(2)}_3B^{(1)}_1+A^{(1)}_4A^{(2)}_1B^{(1)}_3)C^{(21)}-(A^{(1)}_3A^{(2)}_3B^{(1)}_1+A^{(1)}_2A^{(2)}_1B^{(1)}_3)C^{(22)}\notag\\
 &=-(A^{(1)}_1A^{(1)}_2-A^{(1)}_3A^{(1)}_4)(A^{(2)}_1A^{(2)}_2-A^{(2)}_3A^{(2)}_4)B^{(1)}_1B^{(1)}_3A^{(3)}_1B^{(3)}_1,\\
 0&\neq C^{(41)}=A^{(3)}_2B^{(3)}_1(A^{(1)}_4A^{(2)}_2B^{(1)}_1+A^{(1)}_1A^{(2)}_4B^{(1)}_3),
\end{align}
\end{subequations}
we obtain
\begin{equation}
 A^{(1)}_3=B^{(1)}_3=0,
\end{equation}
which gives
\begin{equation}
 C^{(15)}=0.
\end{equation}
From
\begin{equation*}
 v_4^{(r)}=-\dfrac{v_1(u_1C^{(14)}+u_4C^{(16)})}{u_1C^{(23)}+u_4C^{(25)}}+O({v_1}^2),\quad
 v_4^{(l)}=-\dfrac{u_1C^{(33)}+u_4C^{(35)}}{{u_4}^2C^{(41)}+{u_1}^2C^{(43)}+u_1u_4C^{(45)}}+O(v_1),
\end{equation*}
as $v_1\to0$ and $C^{(23)}\neq0$, we obtain
\begin{equation}
 C^{(33)}=C^{(35)}=0.
\end{equation}
Then, we obtain
\begin{equation*}
 v_4^{(l)}=-\dfrac{u_1v_1C^{(36)}}{u_4C^{(41)}}+O({u_1}^2),
\end{equation*}
as $u_1\to0$.
Since $C^{(41)}\neq0$, this gives the inadequate property that $u_1=0$ is a zero of $v_4$.
Therefore, the condition $A^{(3)}_4=B^{(3)}_4=0$ is also inadequate.
Hence, we obtain $C^{(32)}\neq0$.

\underline{\bf Step 4.}\quad
From
\begin{equation}
 C^{(41)}\neq0,\quad
 C^{(14)}\neq0,\quad
 C^{(32)}\neq0,\quad
 C^{(23)}\neq0,
\end{equation}
we respectively obtain
\begin{equation}
 A^{(1)}_1,B^{(1)}_1\neq0,\quad
 A^{(1)}_4,B^{(1)}_4\neq0,\quad
 A^{(3)}_1,B^{(3)}_1\neq0,\quad
 A^{(3)}_3,B^{(3)}_3\neq0.
\end{equation}

Assume $A^{(1)}_2=B^{(1)}_2=0$.
Then, we obtain $A^{(1)}_3C^{(14)}=A^{(1)}_1C^{(13)}=0$.
Since $C^{(14)}\neq0$, we obtain $A^{(1)}_3=B^{(1)}_3=0$,
which is consistent with $C^{(23)}\neq0$.
Therefore, we obtain
\begin{equation}
 A^{(1)}_2,B^{(1)}_2\neq0.
\end{equation}

Assume $A^{(1)}_3=B^{(1)}_3=0$.
Then, we obtain $A^{(1)}_2C^{(41)}=A^{(1)}_4C^{(31)}=0$,
which is inconsistent with $A^{(1)}_2,C^{(41)}\neq0$.
Therefore, we obtain
\begin{equation}
 A^{(1)}_3,B^{(1)}_3\neq0.
\end{equation}

Assume $A^{(3)}_2=B^{(3)}_2=0$.
Then, we obtain $A^{(3)}_4C^{(23)}=A^{(3)}_1C^{(13)}=0$.
Since $C^{(23)}\neq0$, we obtain $A^{(3)}_4=B^{(3)}_4=0$,
which is consistent with $C^{(41)}\neq0$.
Therefore, we obtain
\begin{equation}
 A^{(3)}_4,B^{(3)}_4\neq0.
\end{equation}

Assume $A^{(3)}_4=B^{(3)}_4=0$.
Then, we obtain $A^{(3)}_2C^{(32)}=A^{(3)}_3C^{(31)}=0$,
which is inconsistent with $A^{(3)}_2,C^{(32)}\neq0$.
Therefore, we obtain
\begin{equation}
 A^{(3)}_4,B^{(3)}_4\neq0.
\end{equation}
Therefore, we obtain
\begin{equation}
 A^{(i)}_1,A^{(i)}_2,A^{(i)}_3,A^{(i)}_4,
 B^{(i)}_1,B^{(i)}_2,B^{(i)}_3,B^{(i)}_4\neq0,\quad (i=1,3).
\end{equation}



\underline{\bf Step 5.}\quad
From
\begin{subequations}
\begin{align}
 &(A^{(2)}_2A^{(3)}_1B^{(3)}_1+A^{(2)}_3A^{(3)}_3B^{(3)}_4)C^{(12)}-(A^{(2)}_2A^{(3)}_4B^{(3)}_1+A^{(2)}_3A^{(3)}_2B^{(3)}_4)C^{(22)}\notag\\
 &\quad=(A^{(2)}_1A^{(2)}_2-A^{(2)}_3A^{(2)}_4)(A^{(3)}_1A^{(3)}_2-A^{(3)}_3A^{(3)}_4)B^{(1)}_3A^{(1)}_4B^{(3)}_1B^{(3)}_4\neq0,\\
 &(A^{(1)}_1A^{(2)}_2B^{(1)}_1+A^{(1)}_4A^{(2)}_4B^{(1)}_3)C^{(21)}-(A^{(1)}_3A^{(2)}_2B^{(1)}_1+A^{(1)}_2A^{(2)}_4B^{(1)}_3)C^{(22)}\notag\\
 &\quad=(A^{(1)}_1A^{(1)}_2-A^{(1)}_3A^{(1)}_4)(A^{(2)}_1A^{(2)}_2-A^{(2)}_3A^{(2)}_4)B^{(1)}_1B^{(1)}_3A^{(3)}_3B^{(3)}_4\neq0,\\
 &(A^{(1)}_1A^{(2)}_1B^{(1)}_2+A^{(1)}_4A^{(2)}_3B^{(1)}_4)C^{(34)}-(A^{(1)}_3A^{(2)}_1B^{(1)}_2+A^{(1)}_2A^{(2)}_3B^{(1)}_4)C^{(44)}\notag\\
 &\quad=(A^{(1)}_1A^{(1)}_2-A^{(1)}_3A^{(1)}_4)(A^{(2)}_1A^{(2)}_2-A^{(2)}_3A^{(2)}_4)B^{(1)}_2B^{(1)}_4A^{(3)}_3B^{(3)}_3\neq0,\\
 &(A^{(2)}_1A^{(3)}_1B^{(3)}_2+A^{(2)}_4A^{(3)}_3B^{(3)}_3)C^{(43)}-(A^{(2)}_1A^{(3)}_4B^{(3)}_2+A^{(2)}_4A^{(3)}_2B^{(3)}_3)C^{(44)}\notag\\
 &\quad=(A^{(2)}_1A^{(2)}_2-A^{(2)}_3A^{(2)}_4)(A^{(3)}_1A^{(3)}_2-A^{(3)}_3A^{(3)}_4)A^{(1)}_4B^{(1)}_4B^{(3)}_2B^{(3)}_3\neq0,
\end{align}
\end{subequations}
and the condition \eqref{eqn:case3_condp4_1}, we obtain
\begin{equation}
 C^{(12)},C^{(21)},C^{(34)},C^{(43)}\neq0.
\end{equation}

Therefore, we have completed the proof. 
\end{proof}

\begin{lemma}
The following conditions hold:
\begin{equation}\label{eqn:case3_condp4_2}
\begin{split}
 &C^{(11)}=\dfrac{C^{(23)}C^{(35)}}{C^{(43)}},\quad
 C^{(12)}=\dfrac{C^{(23)}C^{(32)}}{C^{(43)}},\quad
 C^{(14)}=\dfrac{C^{(23)}C^{(34)}}{C^{(43)}},\\
 &C^{(15)}=\dfrac{C^{(23)}C^{(33)}}{C^{(43)}},\quad
 C^{(16)}=\dfrac{C^{(23)}C^{(36)}}{C^{(43)}},\quad
 C^{(21)}=\dfrac{C^{(23)}C^{(41)}}{C^{(43)}},\\
 &C^{(24)}=\dfrac{C^{(23)}C^{(46)}}{C^{(43)}},\quad
 C^{(25)}=\dfrac{C^{(23)}C^{(45)}}{C^{(43)}},\quad
 C^{(26)}=\dfrac{C^{(23)}C^{(42)}}{C^{(43)}}.
\end{split}
\end{equation}
Moreover, we can easily verify that under the conditions \eqref{eqn:case3_condp4_1} and \eqref{eqn:case3_condp4_2}
the CACO property holds.
\end{lemma}
\begin{proof}
From 
\begin{align*}
 &v_4^{(r)}=-\dfrac{C^{(11)}}{u_4C^{(21)}}-\dfrac{v_1C^{(12)}}{C^{(21)}}+O(u_1),\quad
 v_4^{(l)}=-\dfrac{C^{(35)}}{u_4C^{(41)}}-\dfrac{v_1C^{(32)}}{C^{(41)}}+O(u_1)&&(u_1\to0),\\
 &v_4^{(r)}=-\dfrac{C^{(15)}}{u_1C^{(23)})}-\dfrac{v_1C^{(14)}}{C^{(23)}}+O(u_4),\quad
 v_4^{(l)}=-\dfrac{C^{(33)}}{u_1C^{(43)}}-\dfrac{v_1C^{(34)}}{C^{(43)}}+O(u_4)&&(u_4\to0),
\end{align*}
we obtain
\begin{equation}\label{eqn:case3_condp4_proof_1}
 C^{(11)}=\dfrac{C^{(21)}C^{(35)}}{C^{(41)}},\
 C^{(12)}=\dfrac{C^{(21)}C^{(32)}}{C^{(41)}},\
 C^{(14)}=\dfrac{C^{(23)}C^{(34)}}{C^{(43)}},\
 C^{(15)}=\dfrac{C^{(23)}C^{(33)}}{C^{(43)}}.
\end{equation}
Moreover, from
\begin{align*}
 &v_4^{(r)}=-\dfrac{v_1C^{(23)}C^{(34)}}{(C^{(23)}+u_4v_1C^{(24)})C^{(43)}}+O({u_1}^{-1}),\
 v_4^{(l)}=-\dfrac{v_1C^{(34)}}{C^{(43)}+u_4v_1C^{(46)}}+O({u_1}^{-1}),\ (u_1\to\infty),\\
 &v_4^{(r)}=-\dfrac{v_1C^{(21)}C^{(32)}}{(C^{(21)}+u_1v_1C^{(26)})C^{(41)}}+O({u_4}^{-1}),\
 v_4^{(l)}=-\dfrac{v_1C^{(32)}}{C^{(41)}+u_1v_1C^{(42)}}+O({u_4}^{-1}),\ (u_4\to\infty),
\end{align*}
we obtain
\begin{equation}\label{eqn:case3_condp4_proof_2}
 C^{(24)}=\dfrac{C^{(23)}C^{(46)}}{C^{(43)}},\quad
 C^{(26)}=\dfrac{C^{(21)}C^{(42)}}{C^{(41)}}.
\end{equation}

If $v_1=0$ is not a zero of $v_4$, then, from 
\begin{subequations}\label{eqn:case3_v4rl_v1_0}
\begin{align}
 &v_4^{(r)}=-\cfrac{u_1\cfrac{C^{(23)}C^{(33)}}{C^{(21)}C^{(43)}}+u_4\cfrac{C^{(35)}}{C^{(41)}}}{{u_4}^2+{u_1}^2\cfrac{C^{(23)}}{C^{(21)}}+u_1u_4\cfrac{C^{(25)}}{C^{(21)}}}+O(v_1),\\
 &v_4^{(l)}=-\cfrac{u_1\cfrac{C^{(33)}}{C^{(41)}}+u_4\cfrac{C^{(35)}}{C^{(41)}}}{{u_4}^2+{u_1}^2\cfrac{C^{(43)}}{C^{(41)}}+u_1u_4\cfrac{C^{(45)}}{C^{(41)}}}+O(v_1),
 \end{align}
\end{subequations}
at $v_1\to0$, we obtain
\begin{equation}\label{eqn:case3_condp4_proof_3}
 C^{(21)}=\dfrac{C^{(23)}C^{(41)}}{C^{(43)}},\quad
 C^{(25)}=\dfrac{C^{(23)}C^{(45)}}{C^{(43)}}.
\end{equation}
Moreover, from the condition of the CACO property:
\begin{equation*}
 0=v_4^{(r)}-v_4^{(l)}
 =\dfrac{u_1u_4v_1(C^{(23)}C^{(36)}-C^{(16)}C^{(43)})}{C^{(23)}({u_1}^2C^{(43)}+u_1u_4C^{(45)}+{u_4}^2C^{(41)}+u_1u_4(u_1C^{(46)}+u_4C^{(42)})v_1)},
\end{equation*}
we obtain
\begin{equation}\label{eqn:case3_condp4_proof_4}
 C^{(16)}=\dfrac{C^{(23)}C^{(36)}}{C^{(43)}}.
\end{equation}

If $v_1=0$ is a zero of $v_4$, then from \eqref{eqn:case3_v4rl_v1_0} we obtain
\begin{equation}\label{eqn:case3_condp4_proof_5}
 C^{(33)}=C^{(35)}=0.
\end{equation}
From
\begin{align*}
 &v_4^{(r)}=-v_1\cfrac{{u_1}^2\cfrac{C^{(34)}}{C^{(43)}}+u_1u_4\cfrac{C^{(16)}}{C^{(23)}}+{u_4}^2\cfrac{C^{(21)}C^{(32)}}{C^{(23)}C^{(41)}}}{{u_1}^2+u_1u_4\cfrac{C^{(25)}}{C^{(23)}}+{u_4}^2\cfrac{C^{(21)}}{C^{(23)}}}+O({v_1}^2),\\
 &v_4^{(l)}=-v_1\cfrac{{u_1}^2\cfrac{C^{(34)}}{C^{(43)}}+u_1u_4\cfrac{C^{(36)}}{C^{(43)}}+{u_4}^2\cfrac{C^{(32)}}{C^{(43)}}}{{u_1}^2+u_1u_4\cfrac{C^{(45)}}{C^{(43)}}+{u_4}^2\cfrac{C^{(41)}}{C^{(43)}}}+O({v_1}^2),
\end{align*}
as $v_1\to0$, we obtain
\begin{equation}\label{eqn:case3_condp4_proof_6}
 C^{(16)}=\dfrac{C^{(23)}C^{(36)}}{C^{(43)}},\quad
 C^{(21)}=\dfrac{C^{(23)}C^{(41)}}{C^{(43)}},\quad
 C^{(25)}=\dfrac{C^{(23)}C^{(45)}}{C^{(43)}}.
\end{equation}
We can easily verify that under the conditions 
\eqref{eqn:case3_condp4_1},
\eqref{eqn:case3_condp4_proof_1}, 
\eqref{eqn:case3_condp4_proof_2},
\eqref{eqn:case3_condp4_proof_5} and \eqref{eqn:case3_condp4_proof_6} the CACO property holds, 
and the conditions \eqref{eqn:case3_condp4_proof_1}, \eqref{eqn:case3_condp4_proof_2},
\eqref{eqn:case3_condp4_proof_5} and \eqref{eqn:case3_condp4_proof_6} are special case of the condition \eqref{eqn:case3_condp4_2}.

Therefore, we have completed the proof.
\end{proof}

\begin{lemma}
The condition \eqref{eqn:case3_condp4_2} can be divided into the conditions \eqref{eqn:typeII_cond1} and \eqref{eqn:typeII_cond2}.
\end{lemma}
\begin{proof}
Since
\begin{subequations}
\begin{align}
%
 &(A^{(2)}_3A^{(3)}_3B^{(3)}_2+A^{(2)}_2A^{(3)}_1B^{(3)}_3)C^{(13)}-(A^{(2)}_3A^{(3)}_2B^{(3)}_2+A^{(2)}_2A^{(3)}_4B^{(3)}_3)C^{(23)}\notag\\
 &\quad=(A^{(2)}_1A^{(2)}_2-A^{(2)}_3A^{(2)}_4)(A^{(3)}_1A^{(3)}_2-A^{(3)}_3A^{(3)}_4)A^{(1)}_2B^{(1)}_2B^{(3)}_2B^{(3)}_3\neq0,\\
 &(A^{(2)}_1A^{(3)}_3B^{(3)}_2+A^{(2)}_4A^{(3)}_1B^{(3)}_3)C^{(13)}-(A^{(2)}_1A^{(3)}_2B^{(3)}_2+A^{(2)}_4A^{(3)}_4B^{(3)}_3)C^{(23)}\notag\\
 &\quad=-(A^{(2)}_1A^{(2)}_2-A^{(2)}_3A^{(2)}_4)(A^{(3)}_1A^{(3)}_2-A^{(3)}_3A^{(3)}_4)A^{(1)}_3B^{(1)}_4B^{(3)}_2B^{(3)}_3\neq0,\\
 &(A^{(2)}_4A^{(3)}_3B^{(3)}_1+A^{(2)}_1A^{(3)}_1B^{(3)}_4)C^{(31)}-(A^{(2)}_4A^{(3)}_2B^{(3)}_1+A^{(2)}_1A^{(3)}_4B^{(3)}_4)C^{(32)}\notag\\
 &\quad=(A^{(2)}_1A^{(2)}_2-A^{(2)}_3A^{(2)}_4)(A^{(3)}_1A^{(3)}_2-A^{(3)}_3A^{(3)}_4)B^{(1)}_1A^{(1)}_2B^{(3)}_1B^{(3)}_4\neq0,\\
 &(A^{(2)}_2A^{(3)}_3B^{(3)}_1+A^{(2)}_3A^{(3)}_1B^{(3)}_4)C^{(31)}-(A^{(2)}_2A^{(3)}_2B^{(3)}_1+A^{(2)}_3A^{(3)}_4B^{(3)}_4)C^{(32)}\notag\\
 &\quad=-(A^{(2)}_1A^{(2)}_2-A^{(2)}_3A^{(2)}_4)(A^{(3)}_1A^{(3)}_2-A^{(3)}_3A^{(3)}_4)A^{(1)}_3B^{(1)}_3B^{(3)}_1B^{(3)}_4\neq0,\\
%
 &(A^{(2)}_2A^{(3)}_1B^{(3)}_1+A^{(2)}_3A^{(3)}_3B^{(3)}_4)C^{(12)}-(A^{(2)}_2A^{(3)}_4B^{(3)}_1+A^{(2)}_3A^{(3)}_2B^{(3)}_4)C^{(22)}\notag\\
 &\quad=(A^{(2)}_1A^{(2)}_2-A^{(2)}_3A^{(2)}_4)(A^{(3)}_1A^{(3)}_2-A^{(3)}_3A^{(3)}_4)B^{(1)}_3A^{(1)}_4B^{(3)}_1B^{(3)}_4\neq0,\\
 &(A^{(2)}_4A^{(3)}_1B^{(3)}_1+A^{(2)}_1A^{(3)}_3B^{(3)}_4)C^{(12)}-(A^{(2)}_4A^{(3)}_4B^{(3)}_1+A^{(2)}_1A^{(3)}_2B^{(3)}_4)C^{(22)}\notag\\
 &\quad=-(A^{(2)}_1A^{(2)}_2-A^{(2)}_3A^{(2)}_4)(A^{(3)}_1A^{(3)}_2-A^{(3)}_3A^{(3)}_4)A^{(1)}_1B^{(1)}_1B^{(3)}_1B^{(3)}_4\neq0,
\end{align}
\end{subequations}
and \eqref{eqn:case3_condp4_1}, the following relations hold:
\begin{subequations}
\begin{align}
 &A^{(2)}_3A^{(3)}_2B^{(3)}_2+A^{(2)}_2A^{(3)}_4B^{(3)}_3\neq0,\quad
 A^{(2)}_1A^{(3)}_2B^{(3)}_2+A^{(2)}_4A^{(3)}_4B^{(3)}_3\neq0,\\
 &A^{(2)}_4A^{(3)}_2B^{(3)}_1+A^{(2)}_1A^{(3)}_4B^{(3)}_4\neq0,\quad
 A^{(2)}_2A^{(3)}_2B^{(3)}_1+A^{(2)}_3A^{(3)}_4B^{(3)}_4\neq0,\\
 &A^{(2)}_2A^{(3)}_1B^{(3)}_1+A^{(2)}_3A^{(3)}_3B^{(3)}_4\neq0,\quad
 A^{(2)}_4A^{(3)}_1B^{(3)}_1+A^{(2)}_1A^{(3)}_3B^{(3)}_4\neq0.
\end{align}
\end{subequations}
From $C^{(13)}=C^{(22)}=C^{(31)}=0$ we obtain
\begin{subequations}
\begin{align}
 &A^{(1)}_1=-\dfrac{A^{(1)}_4B^{(1)}_3(A^{(2)}_4A^{(3)}_1B^{(3)}_1+A^{(2)}_1A^{(3)}_3B^{(3)}_4)}{B^{(1)}_1(A^{(2)}_2A^{(3)}_1B^{(3)}_1+A^{(2)}_3A^{(3)}_3B^{(3)}_4)},\\
 &A^{(1)}_2=-\dfrac{A^{(1)}_3B^{(1)}_3(A^{(2)}_4A^{(3)}_2B^{(3)}_1+A^{(2)}_1A^{(3)}_4B^{(3)}_4)}{B^{(1)}_1(A^{(2)}_2A^{(3)}_2B^{(3)}_1+A^{(2)}_3A^{(3)}_4B^{(3)}_4)},\\
 &B^{(1)}_4=\dfrac{B^{(1)}_2B^{(1)}_3(A^{(2)}_1A^{(3)}_2B^{(3)}_2+A^{(2)}_4A^{(3)}_4B^{(3)}_3)(A^{(2)}_4A^{(3)}_2B^{(3)}_1+A^{(2)}_1A^{(3)}_4B^{(3)}_4)}{B^{(1)}_1(A^{(2)}_3A^{(3)}_2B^{(3)}_2+A^{(2)}_2A^{(3)}_4B^{(3)}_3)(A^{(2)}_2A^{(3)}_2B^{(3)}_1+A^{(2)}_3A^{(3)}_4B^{(3)}_4)}.
\end{align}
\end{subequations}
Then, 
\begin{align}
 0=&C^{(44)}\notag\\
 =&\dfrac{B^{(1)}_2B^{(1)}_3A^{(1)}_4(A^{(2)}_1A^{(2)}_2-A^{(2)}_3A^{(2)}_4)(A^{(3)}_1A^{(3)}_4-A^{(3)}_2A^{(3)}_3)}{B^{(1)}_1(A^{(2)}_3A^{(3)}_2B^{(3)}_2+A^{(2)}_2A^{(3)}_4B^{(3)}_3)(A^{(2)}_2A^{(3)}_1B^{(3)}_1+A^{(2)}_3A^{(3)}_3B^{(3)}_4)}\notag\\
 &\dfrac{(A^{(3)}_1A^{(3)}_2B^{(3)}_1B^{(3)}_2-A^{(3)}_3A^{(3)}_4B^{(3)}_3B^{(3)}_4)(A^{(2)}_1A^{(2)}_3B^{(3)}_2B^{(3)}_4-A^{(2)}_2A^{(2)}_4B^{(3)}_1B^{(3)}_3)}{A^{(2)}_2A^{(3)}_2B^{(3)}_1+A^{(2)}_3A^{(3)}_4B^{(3)}_4},
\end{align}
gives
\begin{align}
 &(A^{(3)}_1A^{(3)}_4-A^{(3)}_2A^{(3)}_3)
 (A^{(3)}_1A^{(3)}_2B^{(3)}_1B^{(3)}_2-A^{(3)}_3A^{(3)}_4B^{(3)}_3B^{(3)}_4)\notag\\
 &\quad (A^{(2)}_1A^{(2)}_3B^{(3)}_2B^{(3)}_4-A^{(2)}_2A^{(2)}_4B^{(3)}_1B^{(3)}_3)
 =0.
\end{align}
Therefore, we obtain the following five cases:
\begin{subequations}
\begin{align}
 &A^{(3)}_1A^{(3)}_4-A^{(3)}_2A^{(3)}_3=0;\label{eqn:case3_condp4_11}\\
 &A^{(3)}_1A^{(3)}_2B^{(3)}_1B^{(3)}_2-A^{(3)}_3A^{(3)}_4B^{(3)}_3B^{(3)}_4=0;\label{eqn:case3_condp4_12}\\
 &A^{(2)}_1A^{(2)}_3B^{(3)}_2B^{(3)}_4-A^{(2)}_2A^{(2)}_4B^{(3)}_1B^{(3)}_3=0,\quad 
 A^{(2)}_1,\dots,A^{(2)}_4,B^{(2)}_1,\dots,B^{(2)}_4\neq0;\label{eqn:case3_condp4_13}\\
 &A^{(2)}_1=B^{(2)}_1=A^{(2)}_2=B^{(2)}_2=0;\label{eqn:case3_condp4_14}\\
 &A^{(2)}_3=B^{(2)}_3=A^{(2)}_4=B^{(2)}_4=0.\label{eqn:case3_condp4_15}
\end{align}
\end{subequations}

Firstly, we consider the case \eqref{eqn:case3_condp4_11}.
From 
\begin{equation}
 0=C^{(12)}-\dfrac{C^{(23)}C^{(32)}}{C^{(43)}}
=-\dfrac{C^{(12)}({A^{(1)}_3}^2{A^{(3)}_1}^2-{A^{(1)}_4}^2{A^{(3)}_2}^2)}{{A^{(1)}_4}^2{A^{(3)}_2}^2},
\end{equation}
we obtain
\begin{equation}
 A^{(1)}_4A^{(3)}_2-A^{(1)}_3A^{(3)}_1=0
\quad\text{or}\quad
 A^{(1)}_4A^{(3)}_2+A^{(1)}_3A^{(3)}_1=0.
\end{equation}
When $A^{(1)}_4A^{(3)}_2-A^{(1)}_3A^{(3)}_1=0$, then we obtain $C^{(43)}=C^{(23)}$,
and when $A^{(1)}_4A^{(3)}_2+A^{(1)}_3A^{(3)}_1=0$, then we obtain $C^{(43)}=-C^{(23)}$.
Therefore, the condition \eqref{eqn:case3_condp4_2} can be rewritten as 
\eqref{eqn:typeII_cond1} or \eqref{eqn:typeII_cond2}.

Next, we consider the case \eqref{eqn:case3_condp4_12}.
Under the condition \eqref{eqn:case3_condp4_12} the following relation holds:
\begin{align}
 0&=\dfrac{A^{(1)}_4A^{(3)}_3A^{(3)}_4}{B^{(1)}_1B^{(3)}_1B^{(3)}_3}
 \left(C^{(21)}-\dfrac{C^{(23)}C^{(41)}}{C^{(43)}}\right)
 -\dfrac{A^{(1)}_4A^{(3)}_1A^{(3)}_2}{B^{(1)}_1{B^{(3)}_3}^2}\left(C^{(25)}-\dfrac{C^{(23)}C^{(45)}}{C^{(43)}}\right)\notag\\
 &=\dfrac{A^{(2)}_2A^{(2)}_4(A^{(1)}_1A^{(1)}_2-A^{(1)}_3A^{(1)}_4)(A^{(3)}_1A^{(3)}_2-A^{(3)}_3A^{(3)}_4)(A^{(3)}_1A^{(3)}_4-A^{(3)}_2A^{(3)}_3)}{A^{(2)}_1A^{(3)}_2B^{(3)}_2+A^{(2)}_4A^{(3)}_4B^{(3)}_3},
\end{align}
which gives the conditions \eqref{eqn:case3_condp4_11}, $A^{(2)}_2=B^{(2)}_2=0$ or $A^{(2)}_4=B^{(2)}_4$.
If $A^{(2)}_2=B^{(2)}_2=0$, then we obtain
\begin{align}
 0&=\left(\dfrac{A^{(1)}_4A^{(2)}_4A^{(3)}_2B^{(1)}_2}{A^{(2)}_3A^{(3)}_3B^{(1)}_1B^{(1)}_3B^{(3)}_2}\right)^2\left(C^{(14)}-\dfrac{C^{(23)}}{C^{(34)}C^{(43)}}\right)\notag\\
 &=\dfrac{C^{(14)}({C^{(43)}}^2-{C^{(23)}}^2)}{(A^{(2)}_4A^{(3)}_1B^{(1)}_2B^{(3)}_1+A^{(2)}_1A^{(3)}_3B^{(1)}_3B^{(3)}_2)^2(B^{(1)}_2B^{(3)}_1-B^{(1)}_3B^{(3)}_3)^2},
\end{align}
and if $A^{(2)}_2=B^{(2)}_2=0$, then we obtain
\begin{equation*}
 0=C^{(14)}-\dfrac{C^{(23)}}{C^{(34)}C^{(43)}}
 =\dfrac{A^{(1)}_3B^{(1)}_1(A^{(2)}_2A^{(3)}_1B^{(3)}_1+A^{(2)}_3A^{(3)}_3B^{(3)}_4)C^{(21)}({C^{(43)}}^2-{C^{(23)}}^2)}{A^{(1)}_4A^{(2)}_1A^{(3)}_2B^{(1)}_2B^{(3)}_2{C^{(32)}}^2}.
\end{equation*}
In the both cases, we obtain the condition ${C^{(43)}}^2={C^{(23)}}^2$,
which allows the condition \eqref{eqn:case3_condp4_2} to be rewritten as 
\eqref{eqn:typeII_cond1} or \eqref{eqn:typeII_cond2}.

Thirdly, we consider the case \eqref{eqn:case3_condp4_13}.
Under the condition \eqref{eqn:case3_condp4_13} we obtain the following relation:
\begin{align}
 0&=\dfrac{B^{(1)}_2B^{(3)}_1}{A^{(1)}_3A^{(2)}_1A^{(2)}_3B^{(3)}_2}\left(C^{(25)}-\dfrac{C^{(23)}C^{(45)}}{C^{(43)}}\right)\notag\\
 &=-\dfrac{(B^{(1)}_1B^{(1)}_2-B^{(1)}_3B^{(1)}_4)(B^{(3)}_1B^{(3)}_2-B^{(3)}_3B^{(3)}_4)(A^{(3)}_1A^{(3)}_4-A^{(3)}_2A^{(3)}_3)}{A^{(2)}_1A^{(3)}_2B^{(3)}_2+A^{(2)}_4A^{(3)}_4B^{(3)}_3},
\end{align}
which gives the condition \eqref{eqn:case3_condp4_11}.

Fourthly, we consider the case \eqref{eqn:case3_condp4_14}.
Under the condition \eqref{eqn:case3_condp4_14} the following holds:
\begin{equation*}
 0=C^{(12)}-\dfrac{C^{(23)}C^{(32)}}{C^{(43)}}
 =\dfrac{A^{(2)}_4B^{(1)}_3B^{(3)}_1(A^{(3)}_1A^{(3)}_2-A^{(3)}_3A^{(3)}_4)({A^{(1)}_3}^2{A^{(3)}_3}^2-{A^{(1)}_4}^2{A^{(3)}_4}^2)}{A^{(1)}_4A^{(3)}_3{A^{(3)}_4}^2}.
\end{equation*}
Therefore, we obtain ${A^{(1)}_3}^2{A^{(3)}_3}^2-{A^{(1)}_4}^2{A^{(3)}_4}^2=0$,
and then we obtain ${C^{(43)}}^2={C^{(23)}}^2$,
which allows the condition \eqref{eqn:case3_condp4_2} to be rewritten as 
\eqref{eqn:typeII_cond1} or \eqref{eqn:typeII_cond2}.

Lastly, we consider the case \eqref{eqn:case3_condp4_15}.
Under the condition \eqref{eqn:case3_condp4_15} the following holds:
\begin{equation*}
 0=C^{(12)}-\dfrac{C^{(23)}C^{(32)}}{C^{(43)}}
 =-\dfrac{A^{(2)}_1B^{(1)}_3B^{(3)}_4(A^{(3)}_1A^{(3)}_2-A^{(3)}_3A^{(3)}_4)({A^{(1)}_3}^2{A^{(3)}_1}^2-{A^{(1)}_4}^2{A^{(3)}_2}^2)}{A^{(1)}_4A^{(3)}_1{A^{(3)}_2}^2}.
\end{equation*}
Therefore, we obtain
\begin{equation}
 {A^{(1)}_3}^2{A^{(3)}_1}^2-{A^{(1)}_4}^2{A^{(3)}_2}^2=0,
\end{equation}
and under the condition above we obtain ${C^{(43)}}^2={C^{(23)}}^2$,
which allows the condition \eqref{eqn:case3_condp4_2} to be rewritten as 
\eqref{eqn:typeII_cond1} or \eqref{eqn:typeII_cond2}.
Therefore, we have completed the proof.
\end{proof}

We can easily verify that under the conditions \eqref{eqn:typeII_cond0}--\eqref{eqn:typeII_AB_nonzero} (which exactly equals to \eqref{eqn:case3_condp4_1}) and \eqref{eqn:typeII_cond1},
the CACO property holds and the following square equation $K_1$ holds:
\begin{align}
 &v_4\Big(C^{(21)}{u_4}^2+C^{(23)}{u_1}^2+C^{(24)}{u_1}^2u_4v_1+C^{(25)}u_1u_4+C^{(26)}u_1{u_4}^2v_1\Big)\notag\\
 &\quad+C^{(11)}u_4+C^{(12)}{u_4}^2v_1+C^{(14)}{u_1}^2v_1+C^{(15)}u_1+C^{(16)}u_1u_4v_1=0,
\end{align}
and same is true for the conditions \eqref{eqn:typeII_cond0}--\eqref{eqn:typeII_AB_nonzero} and \eqref{eqn:typeII_cond2}.
It is difficult to explicitly show the square equations $K_2$ and $K_3$ by using the conditions \eqref{eqn:typeII_cond0}--\eqref{eqn:typeII_AB_nonzero} and \eqref{eqn:typeII_cond1} (or, the conditions \eqref{eqn:typeII_cond0}--\eqref{eqn:typeII_AB_nonzero} and \eqref{eqn:typeII_cond2}).
Therefore, we just show the existences of the square equations $K_2$ and $K_3$ 
under the conditions \eqref{eqn:typeII_cond0}--\eqref{eqn:typeII_AB_nonzero} and \eqref{eqn:typeII_cond1}
and under the conditions \eqref{eqn:typeII_cond0}--\eqref{eqn:typeII_AB_nonzero} and \eqref{eqn:typeII_cond2} as follows.

We can easily verify that for the general parameters $A^{(i)}_j$ and $B^{(i)}_j$, 
$v_5^{(r)}=v_5^{(r)}(u_2,u_3,u_4,u_5,v_2)$ and $v_6^{(r)}=v_6^{(r)}(u_3,u_4,u_5,u_6,v_3)$ satisfy
\begin{equation}
 \dfrac{\partial v_5^{(r)}}{u_3}=\dfrac{\partial v_5^{(r)}}{u_4}=0,\quad
 \dfrac{\partial v_6^{(r)}}{u_4}=\dfrac{\partial v_6^{(r)}}{u_5}=0,
\end{equation}
respectively.
Moreover, under the following special condition:
\begin{equation}
\begin{cases}
 A^{(2)}_1=B^{(2)}_1=A^{(2)}_2=B^{(2)}_2=0,\quad
 A^{(3)}_2=\dfrac{A^{(1)}_2A^{(3)}_1}{A^{(1)}_1},\quad
 A^{(3)}_3=\dfrac{A^{(1)}_4A^{(3)}_1}{A^{(1)}_1},\\
 A^{(3)}_4=\dfrac{A^{(1)}_3A^{(3)}_1}{A^{(1)}_1},\quad
 B^{(2)}_3=\dfrac{A^{(2)}_4B^{(3)}_1(B^{(1)}_1B^{(1)}_2-B^{(1)}_3B^{(1)}_4)}{A^{(2)}_3B^{(1)}_1},\\
 B^{(2)}_4=\dfrac{B^{(3)}_1(B^{(1)}_1B^{(1)}_2-B^{(1)}_3B^{(1)}_4)}{B^{(1)}_1},\quad
 B^{(3)}_2=-\dfrac{A^{(2)}_4B^{(1)}_2B^{(3)}_1}{A^{(2)}_3B^{(1)}_1},\quad
 B^{(3)}_3=\dfrac{B^{(1)}_4B^{(3)}_1}{B^{(1)}_1},\\
 B^{(3)}_4=-\dfrac{A^{(2)}_4B^{(1)}_3B^{(3)}_1}{A^{(2)}_3B^{(1)}_1},
\end{cases}
\end{equation}
which satisfies the conditions \eqref{eqn:caseN3_CAO_B2}, \eqref{eqn:typeII_cond0}--\eqref{eqn:typeII_AB_nonzero} and \eqref{eqn:typeII_cond1}, 
and under the following special condition:
\begin{equation}
\begin{cases}
 A^{(2)}_1=B^{(2)}_1=A^{(2)}_2=B^{(2)}_2=0,\quad
 A^{(3)}_2=-\dfrac{A^{(1)}_2A^{(3)}_1}{A^{(1)}_1},\quad
 A^{(3)}_3=\dfrac{A^{(1)}_4A^{(3)}_1}{A^{(1)}_1},\\
 A^{(3)}_4=-\dfrac{A^{(1)}_3A^{(3)}_1}{A^{(1)}_1},\quad
 B^{(2)}_3=\dfrac{A^{(2)}_4B^{(3)}_1(B^{(1)}_1B^{(1)}_2-B^{(1)}_3B^{(1)}_4)}{A^{(2)}_3B^{(1)}_1},\\
 B^{(2)}_4=\dfrac{B^{(3)}_1(B^{(1)}_1B^{(1)}_2-B^{(1)}_3B^{(1)}_4)}{B^{(1)}_1},\quad
 B^{(3)}_2=-\dfrac{A^{(2)}_4B^{(1)}_2B^{(3)}_1}{A^{(2)}_3B^{(1)}_1},\quad
 B^{(3)}_3=\dfrac{B^{(1)}_4B^{(3)}_1}{B^{(1)}_1},\\
 B^{(3)}_4=-\dfrac{A^{(2)}_4B^{(1)}_3B^{(3)}_1}{A^{(2)}_3B^{(1)}_1},
\end{cases}
\end{equation}
which satisfies the conditions \eqref{eqn:caseN3_CAO_B2}, \eqref{eqn:typeII_cond0}--\eqref{eqn:typeII_AB_nonzero} and \eqref{eqn:typeII_cond2}, 
we can easily verify
\begin{equation}\label{eqn:case3_v5v6_nonzero}
 \dfrac{\partial v_5}{v_2},\dfrac{\partial v_5}{u_2},\dfrac{\partial v_5}{u_5}\neq 0,\quad
 \dfrac{\partial v_6}{v_3},\dfrac{\partial v_6}{u_3},\dfrac{\partial v_6}{u_6}\neq 0.
\end{equation}
Therefore, under the conditions \eqref{eqn:typeII_cond0}--\eqref{eqn:typeII_AB_nonzero} and \eqref{eqn:typeII_cond1}
and under the conditions \eqref{eqn:typeII_cond0}--\eqref{eqn:typeII_AB_nonzero} and \eqref{eqn:typeII_cond2},
the square property holds.
Hence, we have complete the proof. 

\subsection{Proof of Lemma \ref{lemma:classification_CACO_CO1CO2CO3_3}}
\label{subsection:caseN1}
In this section, we consider quad-equations \eqref{eqn:Q1_P1_CO3}--\eqref{eqn:Q9_P3_CO3} with \eqref{eqn:P123_caseN1}.
Note that in this setting, the transformations $s_{12}$, $s_{23}$, $s_{13}$ and $\iota$ defined in \eqref{eqn:Gco} keep invariant the system of quad-equations $\{Q_1,\dots,Q_9\}$ under the appropriate transformations of parameters (see Remark \ref{remark:symmetry_GCO}).

In this setting, $v_4^{(r)}$ and $v_4^{(l)}$ are given by
\begin{equation}
 v_4^{(r)}=\dfrac{F_1}{F_2},\quad
 v_4^{(l)}=\dfrac{G_1}{G_2}.
\end{equation}
Here, 
\begin{subequations}
\begin{align}
 F_1=&-{d_2}^2{d_3}^2{U_1}^2{U_4}^2\Big(U_4(d_2+U_4v_1)C^{(15)}+{U_4}^2C^{(16)}+d_3(d_2+U_4v_1)C^{(17)}\notag\\
 &+d_3U_4C^{(18)}\Big)-d_2d_3U_1U_4(d_2d_3U_1-U_4)\Big(d_3(d_2+U_4v_1)C^{(21)}+d_3U_4C^{(22)}\notag\\
 &+U_4(d_2+U_4v_1)C^{(23)}+{U_4}^2C^{(24)}\Big)-(d_2d_3U_1-U_4)^2\Big(d_3(d_2+U_4v_1)C^{(11)}\notag\\
 &+d_3U_4C^{(12)}+U_4(d_2+U_4v_1)C^{(13)}+{U_4}^2C^{(14)}\Big),\\
 F_2=&{d_2}^2{d_3}^2{U_1}^2{U_4}^3\Big((d_2+U_4v_1)C^{(17)}+U_4C^{(18)}\Big)\notag\\
 &+d_2d_3U_1{U_4}^2(d_2d_3U_1-U_4)\Big((d_2+U_4v_1)C^{(21)}+U_4C^{(22)}\Big)\notag\\
 &+U_4(d_2d_3U_1-U_4)^2\Big((d_2+U_4v_1)C^{(11)}+U_4C^{(12)}\Big),\\
 G_1=&-{d_2}^2{d_3}^2{U_1}^2{U_4}^2\Big(U_1(1+d_3U_1v_1)C^{(35)}+{U_1}^2C^{(36)}+(1+d_3U_1v_1)C^{(37)}+U_1C^{(38)}\Big)\notag\\
 &-d_2d_3U_1U_4(d_2d_3U_1-U_4)\Big((1+d_3U_1v_1)C^{(41)}+U_1C^{(42)}+U_1(1+d_3U_1v_1)C^{(43)}\notag\\
 &+{U_1}^2C^{(44)}\Big)
 -(d_2d_3U_1-U_4)^2\Big((1+d_3U_1v_1)C^{(31)}+U_1C^{(32)}\notag\\
 &+U_1(1+d_3U_1v_1)C^{(33)}+{U_1}^2C^{(34)}\Big),\\
 G_2=&{d_2}^3{d_3}^2{U_1}^3{U_4}^2\Big((1+d_3U_1v_1)C^{(37)}+U_1C^{(38)}\Big)\notag\\
 &+{d_2}^2d_3{U_1}^2U_4(d_2d_3U_1-U_4)\Big((1+d_3U_1v_1)C^{(41)}+U_1C^{(42)}\Big)\notag\\
 &+d_2U_1(d_2d_3U_1-U_4)^2\Big((1+d_3U_1v_1)C^{(31)}+U_1C^{(32)}\Big),
\end{align}
\end{subequations}
where
\begin{equation}
 U_1=\dfrac{1}{u_1},\quad
 U_4=\dfrac{1}{u_4},
\end{equation}
and
\begin{subequations}
\begin{align}
 C^{(11)}&=A^{(1)}_1A^{(3)}_1(A^{(2)}_1B^{(1)}_2B^{(3)}_3-A^{(2)}_2B^{(1)}_1B^{(3)}_3d_2-A^{(2)}_3B^{(1)}_2B^{(3)}_1d_3)\notag\\
 &+d_2d_3\Big(A^{(1)}_1A^{(3)}_2B^{(3)}_1(A^{(2)}_1B^{(1)}_2-A^{(2)}_2B^{(1)}_1d_2)+A^{(3)}_1B^{(1)}_1B^{(3)}_1(A^{(1)}_1A^{(2)}_4-A^{(1)}_3A^{(2)}_3d_3)\notag\\
 &+A^{(1)}_3A^{(2)}_1B^{(1)}_1(A^{(3)}_1B^{(3)}_3+A^{(3)}_2B^{(3)}_1d_2d_3)\Big),\\
 C^{(12)}&=A^{(1)}_2A^{(3)}_1(A^{(2)}_1B^{(1)}_2B^{(3)}_3-A^{(2)}_2B^{(1)}_1B^{(3)}_3d_2-A^{(2)}_3B^{(1)}_2B^{(3)}_1d_3)\notag\\
 &+d_2d_3\Big(A^{(1)}_2A^{(3)}_2B^{(3)}_1(A^{(2)}_1B^{(1)}_2-A^{(2)}_2B^{(1)}_1d_2)+A^{(3)}_1B^{(1)}_1B^{(3)}_1(A^{(1)}_2A^{(2)}_4-A^{(1)}_4A^{(2)}_3d_3)\notag\\
 &+A^{(1)}_4A^{(2)}_1B^{(1)}_1(A^{(3)}_1B^{(3)}_3+A^{(3)}_2B^{(3)}_1d_2d_3)\Big),\\
 C^{(13)}&=A^{(1)}_1A^{(3)}_3(A^{(2)}_1B^{(1)}_2B^{(3)}_3-A^{(2)}_2B^{(1)}_1B^{(3)}_3d_2-A^{(2)}_3B^{(1)}_2B^{(3)}_1d_3)\notag\\
 &+d_2d_3\Big(A^{(1)}_1A^{(3)}_4B^{(3)}_1(A^{(2)}_1B^{(1)}_2-A^{(2)}_2B^{(1)}_1d_2)+A^{(3)}_3B^{(1)}_1B^{(3)}_1(A^{(1)}_1A^{(2)}_4-A^{(1)}_3A^{(2)}_3d_3)\notag\\
 &+A^{(1)}_3A^{(2)}_1B^{(1)}_1(A^{(3)}_3B^{(3)}_3+A^{(3)}_4B^{(3)}_1d_2d_3)\Big),\\
 C^{(14)}&=A^{(1)}_2A^{(3)}_3(A^{(2)}_1B^{(1)}_2B^{(3)}_3-A^{(2)}_2B^{(1)}_1B^{(3)}_3d_2-A^{(2)}_3B^{(1)}_2B^{(3)}_1d_3)\notag\\
 &+d_2d_3\Big(A^{(1)}_2A^{(3)}_4B^{(3)}_1(A^{(2)}_1B^{(1)}_2-A^{(2)}_2B^{(1)}_1d_2)+A^{(3)}_3B^{(1)}_1B^{(3)}_1(A^{(1)}_2A^{(2)}_4-A^{(1)}_4A^{(2)}_3d_3)\notag\\
 &+A^{(1)}_4A^{(2)}_1B^{(1)}_1(A^{(3)}_3B^{(3)}_3+A^{(3)}_4B^{(3)}_1d_2d_3)\Big),\\
 C^{(15)}&=A^{(1)}_1A^{(3)}_3(A^{(2)}_1B^{(1)}_4B^{(3)}_4-A^{(2)}_2B^{(1)}_3B^{(3)}_4d_2-A^{(2)}_3B^{(1)}_4B^{(3)}_2d_3)\notag\\
 &+d_2d_3\Big(A^{(1)}_1A^{(3)}_4B^{(3)}_2(A^{(2)}_1B^{(1)}_4-A^{(2)}_2B^{(1)}_3d_2)+A^{(3)}_3B^{(1)}_3B^{(3)}_2(A^{(1)}_1A^{(2)}_4-A^{(1)}_3A^{(2)}_3d_3)\notag\\
 &+A^{(1)}_3A^{(2)}_1B^{(1)}_3(A^{(3)}_3B^{(3)}_4+A^{(3)}_4B^{(3)}_2d_2d_3)\Big),\\
 C^{(16)}&=A^{(1)}_2A^{(3)}_3(A^{(2)}_1B^{(1)}_4B^{(3)}_4-A^{(2)}_2B^{(1)}_3B^{(3)}_4d_2-A^{(2)}_3B^{(1)}_4B^{(3)}_2d_3)\notag\\
 &+d_2d_3\Big(A^{(1)}_2A^{(3)}_4B^{(3)}_2(A^{(2)}_1B^{(1)}_4-A^{(2)}_2B^{(1)}_3d_2)+A^{(3)}_3B^{(1)}_3B^{(3)}_2(A^{(1)}_2A^{(2)}_4-A^{(1)}_4A^{(2)}_3d_3)\notag\\
 &+A^{(1)}_4A^{(2)}_1B^{(1)}_3(A^{(3)}_3B^{(3)}_4+A^{(3)}_4B^{(3)}_2d_2d_3)\Big),\\
 C^{(17)}&=A^{(1)}_1A^{(3)}_1(A^{(2)}_1B^{(1)}_4B^{(3)}_4-A^{(2)}_2B^{(1)}_3B^{(3)}_4d_2-A^{(2)}_3B^{(1)}_4B^{(3)}_2d_3)\notag\\
 &+d_2d_3\Big(A^{(1)}_1A^{(3)}_2B^{(3)}_2(A^{(2)}_1B^{(1)}_4-A^{(2)}_2B^{(1)}_3d_2)+A^{(3)}_1B^{(1)}_3B^{(3)}_2(A^{(1)}_1A^{(2)}_4-A^{(1)}_3A^{(2)}_3d_3)\notag\\
 &+A^{(1)}_3A^{(2)}_1B^{(1)}_3(A^{(3)}_1B^{(3)}_4+A^{(3)}_2B^{(3)}_2d_2d_3)\Big),\\
 C^{(18)}&=A^{(1)}_2A^{(3)}_1(A^{(2)}_1B^{(1)}_4B^{(3)}_4-A^{(2)}_2B^{(1)}_3B^{(3)}_4d_2-A^{(2)}_3B^{(1)}_4B^{(3)}_2d_3)\notag\\
 &+d_2d_3\Big(A^{(1)}_2A^{(3)}_2B^{(3)}_2(A^{(2)}_1B^{(1)}_4-A^{(2)}_2B^{(1)}_3d_2)+A^{(3)}_1B^{(1)}_3B^{(3)}_2(A^{(1)}_2A^{(2)}_4-A^{(1)}_4A^{(2)}_3d_3)\notag\\
 &+A^{(1)}_4A^{(2)}_1B^{(1)}_3(A^{(3)}_1B^{(3)}_4+A^{(3)}_2B^{(3)}_2d_2d_3)\Big),\\
 C^{(21)}&=A^{(1)}_1A^{(2)}_1A^{(3)}_1(B^{(1)}_4B^{(3)}_3+B^{(1)}_2B^{(3)}_4)-A^{(3)}_1d_2(B^{(1)}_3B^{(3)}_3+B^{(1)}_1B^{(3)}_4)(A^{(1)}_1A^{(2)}_2-A^{(1)}_3A^{(2)}_1d_3)\notag\\
 &-A^{(1)}_1d_3(B^{(1)}_4B^{(3)}_1+B^{(1)}_2B^{(3)}_2)(A^{(2)}_3A^{(3)}_1-A^{(2)}_1A^{(3)}_2d_2)\notag\\
 &+d_2d_3(B^{(1)}_3B^{(3)}_1+B^{(1)}_1B^{(3)}_2)\Big(A^{(1)}_1(A^{(2)}_4A^{(3)}_1-A^{(2)}_2A^{(3)}_2d_2)\notag\\
 &-A^{(1)}_3d_3(A^{(2)}_3A^{(3)}_1-A^{(2)}_1A^{(3)}_2d_2)\Big),\\
 C^{(22)}&=A^{(1)}_2A^{(2)}_1A^{(3)}_1(B^{(1)}_4B^{(3)}_3+B^{(1)}_2B^{(3)}_4)-A^{(3)}_1d_2(B^{(1)}_3B^{(3)}_3+B^{(1)}_1B^{(3)}_4)(A^{(1)}_2A^{(2)}_2-A^{(1)}_4A^{(2)}_1d_3)\notag\\
 &-A^{(1)}_2d_3(B^{(1)}_4B^{(3)}_1+B^{(1)}_2B^{(3)}_2)(A^{(2)}_3A^{(3)}_1-A^{(2)}_1A^{(3)}_2d_2)\notag\\
 &+d_2d_3(B^{(1)}_3B^{(3)}_1+B^{(1)}_1B^{(3)}_2)\Big(A^{(1)}_2(A^{(2)}_4A^{(3)}_1-A^{(2)}_2A^{(3)}_2d_2)\notag\\
 &-A^{(1)}_4d_3(A^{(2)}_3A^{(3)}_1-A^{(2)}_1A^{(3)}_2d_2)\Big),\\
 C^{(23)}&=A^{(1)}_1A^{(2)}_1A^{(3)}_3(B^{(1)}_4B^{(3)}_3+B^{(1)}_2B^{(3)}_4)-A^{(3)}_3d_2(B^{(1)}_3B^{(3)}_3+B^{(1)}_1B^{(3)}_4)(A^{(1)}_1A^{(2)}_2-A^{(1)}_3A^{(2)}_1d_3)\notag\\
 &-A^{(1)}_1d_3(B^{(1)}_4B^{(3)}_1+B^{(1)}_2B^{(3)}_2)(A^{(2)}_3A^{(3)}_3-A^{(2)}_1A^{(3)}_4d_2)\notag\\
 &+d_2d_3(B^{(1)}_3B^{(3)}_1+B^{(1)}_1B^{(3)}_2)\Big(A^{(1)}_1(A^{(2)}_4A^{(3)}_3-A^{(2)}_2A^{(3)}_4d_2)\notag\\
 &-A^{(1)}_3d_3(A^{(2)}_3A^{(3)}_3-A^{(2)}_1A^{(3)}_4d_2)\Big),\\
 C^{(24)}&=A^{(1)}_2A^{(2)}_1A^{(3)}_3(B^{(1)}_4B^{(3)}_3+B^{(1)}_2B^{(3)}_4)-A^{(3)}_3d_2(B^{(1)}_3B^{(3)}_3+B^{(1)}_1B^{(3)}_4)(A^{(1)}_2A^{(2)}_2-A^{(1)}_4A^{(2)}_1d_3)\notag\\
 &-A^{(1)}_2d_3(B^{(1)}_4B^{(3)}_1+B^{(1)}_2B^{(3)}_2)(A^{(2)}_3A^{(3)}_3-A^{(2)}_1A^{(3)}_4d_2)\notag\\
 &+d_2d_3(B^{(1)}_3B^{(3)}_1+B^{(1)}_1B^{(3)}_2)\Big(A^{(1)}_2(A^{(2)}_4A^{(3)}_3-A^{(2)}_2A^{(3)}_4d_2)\notag\\
 &-A^{(1)}_4d_3(A^{(2)}_3A^{(3)}_3-A^{(2)}_1A^{(3)}_4d_2)\Big),\\
 C^{(31)}&=A^{(1)}_1A^{(3)}_1(A^{(2)}_4B^{(1)}_1B^{(3)}_1-A^{(2)}_2B^{(1)}_1B^{(3)}_3d_2-A^{(2)}_3B^{(1)}_2B^{(3)}_1d_3)\notag\\
 &+d_2d_3\Big(A^{(2)}_1A^{(3)}_1B^{(3)}_3(A^{(1)}_1B^{(1)}_2-A^{(1)}_3B^{(1)}_1d_2)+A^{(1)}_1A^{(3)}_2B^{(3)}_1(A^{(2)}_2B^{(1)}_1-A^{(2)}_1B^{(1)}_2d_3)\notag\\
 &+A^{(1)}_3B^{(1)}_1B^{(3)}_1(A^{(2)}_3A^{(3)}_1+A^{(2)}_1A^{(3)}_2d_2d_3)\Big),\\
 C^{(32)}&=A^{(1)}_1A^{(3)}_3(A^{(2)}_4B^{(1)}_1B^{(3)}_1-A^{(2)}_2B^{(1)}_1B^{(3)}_3d_2-A^{(2)}_3B^{(1)}_2B^{(3)}_1d_3)\notag\\
 &+d_2d_3\Big(A^{(2)}_1A^{(3)}_3B^{(3)}_3(A^{(1)}_1B^{(1)}_2-A^{(1)}_3B^{(1)}_1d_2)+A^{(1)}_1A^{(3)}_4B^{(3)}_1(A^{(2)}_2B^{(1)}_1-A^{(2)}_1B^{(1)}_2d_3)\notag\\
 &+A^{(1)}_3B^{(1)}_1B^{(3)}_1(A^{(2)}_3A^{(3)}_3+A^{(2)}_1A^{(3)}_4d_2d_3)\Big),\\
 C^{(33)}&=A^{(1)}_2A^{(3)}_1(A^{(2)}_4B^{(1)}_1B^{(3)}_1-A^{(2)}_2B^{(1)}_1B^{(3)}_3d_2-A^{(2)}_3B^{(1)}_2B^{(3)}_1d_3)\notag\\
 &+d_2d_3\Big(A^{(2)}_1A^{(3)}_1B^{(3)}_3(A^{(1)}_2B^{(1)}_2-A^{(1)}_4B^{(1)}_1d_2)+A^{(1)}_2A^{(3)}_2B^{(3)}_1(A^{(2)}_2B^{(1)}_1-A^{(2)}_1B^{(1)}_2d_3)\notag\\
 &+A^{(1)}_4B^{(1)}_1B^{(3)}_1(A^{(2)}_3A^{(3)}_1+A^{(2)}_1A^{(3)}_2d_2d_3)\Big),\\
 C^{(34)}&=A^{(1)}_2A^{(3)}_3(A^{(2)}_4B^{(1)}_1B^{(3)}_1-A^{(2)}_2B^{(1)}_1B^{(3)}_3d_2-A^{(2)}_3B^{(1)}_2B^{(3)}_1d_3)\notag\\
 &+d_2d_3\Big(A^{(2)}_1A^{(3)}_3B^{(3)}_3(A^{(1)}_2B^{(1)}_2-A^{(1)}_4B^{(1)}_1d_2)+A^{(1)}_2A^{(3)}_4B^{(3)}_1(A^{(2)}_2B^{(1)}_1-A^{(2)}_1B^{(1)}_2d_3)\notag\\
 &+A^{(1)}_4B^{(1)}_1B^{(3)}_1(A^{(2)}_3A^{(3)}_3+A^{(2)}_1A^{(3)}_4d_2d_3)\Big),\\
 C^{(35)}&=A^{(1)}_2A^{(3)}_1(A^{(2)}_4B^{(1)}_3B^{(3)}_2-A^{(2)}_2B^{(1)}_3B^{(3)}_4d_2-A^{(2)}_3B^{(1)}_4B^{(3)}_2d_3)\notag\\
 &+d_2d_3\Big(A^{(2)}_1A^{(3)}_1B^{(3)}_4(A^{(1)}_2B^{(1)}_4-A^{(1)}_4B^{(1)}_3d_2)+A^{(1)}_2A^{(3)}_2B^{(3)}_2(A^{(2)}_2B^{(1)}_3-A^{(2)}_1B^{(1)}_4d_3)\notag\\
 &+A^{(1)}_4B^{(1)}_3B^{(3)}_2(A^{(2)}_3A^{(3)}_1+A^{(2)}_1A^{(3)}_2d_2d_3)\Big),\\
 C^{(36)}&=A^{(1)}_2A^{(3)}_3(A^{(2)}_4B^{(1)}_3B^{(3)}_2-A^{(2)}_2B^{(1)}_3B^{(3)}_4d_2-A^{(2)}_3B^{(1)}_4B^{(3)}_2d_3)\notag\\
 &+d_2d_3\Big(A^{(2)}_1A^{(3)}_3B^{(3)}_4(A^{(1)}_2B^{(1)}_4-A^{(1)}_4B^{(1)}_3d_2)+A^{(1)}_2A^{(3)}_4B^{(3)}_2(A^{(2)}_2B^{(1)}_3-A^{(2)}_1B^{(1)}_4d_3)\notag\\
 &+A^{(1)}_4B^{(1)}_3B^{(3)}_2(A^{(2)}_3A^{(3)}_3+A^{(2)}_1A^{(3)}_4d_2d_3)\Big),\\
 C^{(37)}&=A^{(1)}_1A^{(3)}_1(A^{(2)}_4B^{(1)}_3B^{(3)}_2-A^{(2)}_2B^{(1)}_3B^{(3)}_4d_2-A^{(2)}_3B^{(1)}_4B^{(3)}_2d_3)\notag\\
 &+d_2d_3\Big(A^{(2)}_1A^{(3)}_1B^{(3)}_4(A^{(1)}_1B^{(1)}_4-A^{(1)}_3B^{(1)}_3d_2)+A^{(1)}_1A^{(3)}_2B^{(3)}_2(A^{(2)}_2B^{(1)}_3-A^{(2)}_1B^{(1)}_4d_3)\notag\\
 &+A^{(1)}_3B^{(1)}_3B^{(3)}_2(A^{(2)}_3A^{(3)}_1+A^{(2)}_1A^{(3)}_2d_2d_3)\Big),\\
 C^{(38)}&=A^{(1)}_1A^{(3)}_3(A^{(2)}_4B^{(1)}_3B^{(3)}_2-A^{(2)}_2B^{(1)}_3B^{(3)}_4d_2-A^{(2)}_3B^{(1)}_4B^{(3)}_2d_3)\notag\\
 &+d_2d_3\Big(A^{(2)}_1A^{(3)}_3B^{(3)}_4(A^{(1)}_1B^{(1)}_4-A^{(1)}_3B^{(1)}_3d_2)+A^{(1)}_1A^{(3)}_4B^{(3)}_2(A^{(2)}_2B^{(1)}_3-A^{(2)}_1B^{(1)}_4d_3)\notag\\
 &+A^{(1)}_3B^{(1)}_3B^{(3)}_2(A^{(2)}_3A^{(3)}_3+A^{(2)}_1A^{(3)}_4d_2d_3)\Big),\\
 C^{(41)}&=A^{(1)}_1A^{(2)}_1A^{(3)}_1d_2d_3(B^{(1)}_4B^{(3)}_3+B^{(1)}_2B^{(3)}_4)\notag\\
 &-A^{(3)}_1d_2(B^{(1)}_3B^{(3)}_3+B^{(1)}_1B^{(3)}_4)(A^{(1)}_1A^{(2)}_2+A^{(1)}_3A^{(2)}_1d_2d_3)\notag\\
 &-A^{(1)}_1d_3(B^{(1)}_4B^{(3)}_1+B^{(1)}_2B^{(3)}_2)(A^{(2)}_3A^{(3)}_1+A^{(2)}_1A^{(3)}_2d_2d_3)\notag\\
 &+(B^{(1)}_3B^{(3)}_1+B^{(1)}_1B^{(3)}_2)\Big(A^{(3)}_2d_2d_3(A^{(1)}_1A^{(2)}_2+A^{(1)}_3A^{(2)}_1d_2d_3)\notag\\
 &+A^{(3)}_1(A^{(1)}_1A^{(2)}_4+A^{(1)}_3A^{(2)}_3d_2d_3)\Big),\\
 C^{(42)}&=A^{(1)}_1A^{(2)}_1A^{(3)}_3d_2d_3(B^{(1)}_4B^{(3)}_3+B^{(1)}_2B^{(3)}_4)\notag\\
 &-A^{(3)}_3d_2(B^{(1)}_3B^{(3)}_3+B^{(1)}_1B^{(3)}_4)(A^{(1)}_1A^{(2)}_2+A^{(1)}_3A^{(2)}_1d_2d_3)\notag\\
 &-A^{(1)}_1d_3(B^{(1)}_4B^{(3)}_1+B^{(1)}_2B^{(3)}_2)(A^{(2)}_3A^{(3)}_3+A^{(2)}_1A^{(3)}_4d_2d_3)\notag\\
 &+(B^{(1)}_3B^{(3)}_1+B^{(1)}_1B^{(3)}_2)\Big(A^{(3)}_4d_2d_3(A^{(1)}_1A^{(2)}_2+A^{(1)}_3A^{(2)}_1d_2d_3)\notag\\
 &+A^{(3)}_3(A^{(1)}_1A^{(2)}_4+A^{(1)}_3A^{(2)}_3d_2d_3)\Big),\\
 C^{(43)}&=A^{(1)}_2A^{(2)}_1A^{(3)}_1d_2d_3(B^{(1)}_4B^{(3)}_3+B^{(1)}_2B^{(3)}_4)\notag\\
 &-A^{(3)}_1d_2(B^{(1)}_3B^{(3)}_3+B^{(1)}_1B^{(3)}_4)(A^{(1)}_2A^{(2)}_2+A^{(1)}_4A^{(2)}_1d_2d_3)\notag\\
 &-A^{(1)}_2d_3(B^{(1)}_4B^{(3)}_1+B^{(1)}_2B^{(3)}_2)(A^{(2)}_3A^{(3)}_1+A^{(2)}_1A^{(3)}_2d_2d_3)\notag\\
 &+(B^{(1)}_3B^{(3)}_1+B^{(1)}_1B^{(3)}_2)\Big(A^{(3)}_2d_2d_3(A^{(1)}_2A^{(2)}_2+A^{(1)}_4A^{(2)}_1d_2d_3)\notag\\
 &+A^{(3)}_1(A^{(1)}_2A^{(2)}_4+A^{(1)}_4A^{(2)}_3d_2d_3)\Big),\\
 C^{(44)}&=A^{(1)}_2A^{(2)}_1A^{(3)}_3d_2d_3(B^{(1)}_4B^{(3)}_3+B^{(1)}_2B^{(3)}_4)\notag\\
 &-A^{(3)}_3d_2(B^{(1)}_3B^{(3)}_3+B^{(1)}_1B^{(3)}_4)(A^{(1)}_2A^{(2)}_2+A^{(1)}_4A^{(2)}_1d_2d_3)\notag\\
 &-A^{(1)}_2d_3(B^{(1)}_4B^{(3)}_1+B^{(1)}_2B^{(3)}_2)(A^{(2)}_3A^{(3)}_3+A^{(2)}_1A^{(3)}_4d_2d_3)\notag\\
 &+(B^{(1)}_3B^{(3)}_1+B^{(1)}_1B^{(3)}_2)\Big(A^{(3)}_4d_2d_3(A^{(1)}_2A^{(2)}_2+A^{(1)}_4A^{(2)}_1d_2d_3)\notag\\
 &+A^{(3)}_3(A^{(1)}_2A^{(2)}_4+A^{(1)}_4A^{(2)}_3d_2d_3)\Big).
\end{align}
\end{subequations}

We here consider by dividing into the following five cases:
\begin{align}
 &A^{(1)}_1=B^{(1)}_1=A^{(2)}_1=B^{(2)}_1=A^{(3)}_1=B^{(3)}_1=0;\tag{i}\label{eqn:typeIII_proof_c_1}\\
 &A^{(1)}_1=B^{(1)}_1=0,\quad A^{(2)}_1,B^{(2)}_1,A^{(3)}_1,B^{(3)}_1\neq0;\tag{ii}\label{eqn:typeIII_proof_c_2}\\
 &A^{(1)}_1,B^{(1)}_1,A^{(2)}_1,B^{(2)}_1,A^{(3)}_1,B^{(3)}_1\neq0;\tag{iii}\label{eqn:typeIII_proof_c_3}\\
 &A^{(2)}_1=B^{(2)}_1=0,\quad A^{(1)}_1,B^{(1)}_1,A^{(3)}_1,B^{(3)}_1\neq0;\tag{iv}\label{eqn:typeIII_proof_c_4}\\
 &A^{(3)}_1=B^{(3)}_1=0,\quad A^{(1)}_1,B^{(1)}_1,A^{(2)}_1,B^{(2)}_1\neq0.\tag{v}\label{eqn:typeIII_proof_c_5}
\end{align}
%
%
%
%
Note that because of the condition \eqref{eqn:caseN1_CAO_B2} the following cases are inadequate:
\begin{equation}
 A^{(i)}_1=B^{(i)}_1=A^{(i+1)}_1=B^{(i+1)}_1=0,\quad A^{(i+2)}_1,B^{(i+2)}_1\neq 0,
\end{equation}
where $i\in\bbZ/(3\bbZ)$.
Moreover, because of the symmetries $s_{12}$ and $s_{23}$,
considering the case \eqref{eqn:typeIII_proof_c_2}, 
considering the case \eqref{eqn:typeIII_proof_c_4} and 
considering the case \eqref{eqn:typeIII_proof_c_5} mean the same.
Note that the actions of $s_{12}$ and $s_{23}$ on the parameters are given by the following:
\begin{subequations}
\begin{align}
\begin{split}
 s_{12}:~&
 A^{(1)}_2\leftrightarrow A^{(1)}_3,\quad
 A^{(2)}_1\leftrightarrow A^{(3)}_1,\quad
 A^{(2)}_2\leftrightarrow A^{(3)}_3,\quad
 A^{(2)}_3\leftrightarrow A^{(3)}_2,\quad
 A^{(2)}_4\leftrightarrow A^{(3)}_4,\\
 &B^{(1)}_2\leftrightarrow B^{(1)}_3,\quad
 B^{(2)}_1\leftrightarrow B^{(3)}_1,\quad
 B^{(2)}_2\leftrightarrow B^{(3)}_3,\quad
 B^{(2)}_3\leftrightarrow B^{(3)}_2,\quad
 B^{(2)}_4\leftrightarrow B^{(3)}_4,\\
 &d_1\leftrightarrow d_2,
\end{split}\\
\begin{split}
 s_{23}:~&
 A^{(1)}_1\leftrightarrow A^{(3)}_1,\quad
 A^{(1)}_2\leftrightarrow A^{(3)}_3,\quad
 A^{(1)}_3\leftrightarrow A^{(3)}_2,\quad
 A^{(1)}_4\leftrightarrow A^{(3)}_4,\quad
 A^{(2)}_2\leftrightarrow A^{(2)}_3,\\
 &B^{(1)}_1\leftrightarrow B^{(3)}_1,\quad
 B^{(1)}_2\leftrightarrow B^{(3)}_3,\quad
 B^{(1)}_3\leftrightarrow B^{(3)}_2,\quad
 B^{(1)}_4\leftrightarrow B^{(3)}_4,\quad
 B^{(2)}_2\leftrightarrow B^{(2)}_3,\\
 &d_2\leftrightarrow d_3.
\end{split}
\end{align}
\end{subequations}
Therefore, it is sufficient to consider the three cases \eqref{eqn:typeIII_proof_c_1}--\eqref{eqn:typeIII_proof_c_3}.

\subsubsection{Case \eqref{eqn:typeIII_proof_c_1}.}
\label{subsubsection:caseN1_case1}
In this case, we obtain
\begin{align}
 C^{(11)}&=C^{(12)}=C^{(13)}=C^{(14)}=C^{(17)}=C^{(21)}=C^{(22)}=C^{(23)}\notag\\
 &=C^{(31)}=C^{(32)}=C^{(33)}=C^{(34)}=C^{(37)}=C^{(41)}=C^{(42)}=C^{(43)}=0,
\end{align}
and
\begin{equation}
 C^{(18)}=-A^{(1)}_2A^{(2)}_2A^{(3)}_2B^{(1)}_3B^{(3)}_2{d_2}^2d_3\neq0.
\end{equation}
From
\begin{align}
 0=&F_1G_2-F_2G_1\notag\\
 =&{d_2}^3{d_3}^3(d_2d_3C^{(18)}C^{(35)}+d_2d_3C^{(18)}C^{(38)}+d_2C^{(24)}C^{(38)}-C^{(18)}C^{(44)}){U_1}^5{U_4}^6\notag\\
&-{d_2}^4{d_3}^4({d_2}^2C^{(15)}C^{(38)}+d_2d_3C^{(18)}C^{(38)}+d_2C^{(24)}C^{(38)}-C^{(18)}C^{(44)}){U_1}^6{U_4}^5\notag\\
&+{d_2}^4{d_3}^4(C^{(18)}C^{(36)}-d_2C^{(16)}C^{(38)}){U_1}^6{U_4}^6\notag\\
&+{d_2}^4{d_3}^4(d_3C^{(18)}C^{(35)}-d_2C^{(15)}C^{(38)}){U_1}^6{U_4}^6v_1,
\end{align}
we obtain
\begin{align}
 &C^{(35)}=\dfrac{d_2C^{(15)}C^{(38)}}{d_3C^{(18)}},\quad
 C^{(44)}=\dfrac{d_2(d_2C^{(15)}+d_3C^{(18)}+C^{(24)})C^{(38)}}{C^{(18)}},\notag\\
 &C^{(36)}=\dfrac{d_2C^{(16)}C^{(38)}}{C^{(18)}}.
\end{align}
Moreover, from
\begin{align}
 0&=d_3C^{(18)}C^{(35)}-d_2C^{(15)}C^{(38)}\notag\\
 &=-{B^{(1)}_3}^2{B^{(3)}_2}^2{d_2}^3{d_3}^3
 (A^{(1)}_2A^{(2)}_2A^{(3)}_2-A^{(1)}_3A^{(2)}_3A^{(3)}_3)
 (A^{(1)}_2A^{(2)}_2A^{(3)}_2+A^{(1)}_3A^{(2)}_3A^{(3)}_3),
\end{align}
we obtain
\begin{equation}
A^{(1)}_2A^{(2)}_2A^{(3)}_2-A^{(1)}_3A^{(2)}_3A^{(3)}_3=0\quad
\text{or}\quad
A^{(1)}_2A^{(2)}_2A^{(3)}_2+A^{(1)}_3A^{(2)}_3A^{(3)}_3=0.
\end{equation}

In the case $A^{(1)}_2A^{(2)}_2A^{(3)}_2-A^{(1)}_3A^{(2)}_3A^{(3)}_3=0$, we obtain the condition \eqref{eqn:typeIII_cond1_1},
and under this condition we can verify by computation that the CACO property and square property hold. 

In the case $A^{(1)}_2A^{(2)}_2A^{(3)}_2+A^{(1)}_3A^{(2)}_3A^{(3)}_3=0$, 
if $d_2\neq-1$, then we obtain
\begin{equation}
 A^{(3)}_3=-\dfrac{A^{(1)}_2A^{(2)}_2A^{(3)}_2}{A^{(1)}_3A^{(2)}_3},\quad
 A^{(3)}_4=\dfrac{A^{(3)}_2\Big(A^{(1)}_4A^{(2)}_3d_2d_3(1+d_3)+A^{(1)}_2A^{(2)}_4(1-d_2d_3)\Big)}{A^{(1)}_3A^{(2)}_3d_2d_3(1+d_2)},
\end{equation}
if $d_2=-1$ and $d_3\neq-1$, then we obtain
\begin{equation}
 A^{(3)}_3=-\dfrac{A^{(1)}_2A^{(2)}_2A^{(3)}_2}{A^{(1)}_3A^{(2)}_3},\quad
 A^{(2)}_4=\dfrac{A^{(1)}_4A^{(2)}_3d_3}{A^{(1)}_2},
\end{equation}
and if $d_2=d_3=-1$, then we obtain
\begin{equation}
A^{(3)}_3=-\dfrac{A^{(1)}_2A^{(2)}_2A^{(3)}_2}{A^{(1)}_3A^{(2)}_3}.
\end{equation}
Therefore, we obtain the conditions \eqref{eqn:typeIII_cond1_2}--\eqref{eqn:typeIII_cond1_4}.
We can verify by computation that under these conditions the CACO property and square property hold.

\subsubsection{Case \eqref{eqn:typeIII_proof_c_2}.}
\label{subsubsection:caseN1_case2}
In this case, we obtain
\begin{subequations}
\begin{align}
 &C^{(11)}=C^{(13)}=C^{(31)}=C^{(32)}=0,\quad
 C^{(21)}=\dfrac{A^{(1)}_3B^{(1)}_3d_2d_3C^{(12)}}{A^{(1)}_2B^{(1)}_2},\\
 &C^{(23)}=\dfrac{A^{(1)}_3B^{(1)}_3d_2d_3C^{(14)}}{A^{(1)}_2B^{(1)}_2},\quad
 C^{(41)}=-\dfrac{A^{(1)}_3B^{(1)}_3d_2C^{(3)}_3}{A^{(1)}_2B^{(1)}_2},\\
 &C^{(42)}=-\dfrac{A^{(1)}_3B^{(1)}_3d_2C^{(34)}}{A^{(1)}_2B^{(1)}_2}.
\end{align}
\end{subequations}
From
\begin{equation}
 A^{(3)}_3C^{(33)}-A^{(3)}_1C^{(34)}
 =A^{(1)}_2A^{(2)}_1B^{(1)}_2B^{(3)}_1d_2{d_3}^2(A^{(3)}_1A^{(3)}_4-A^{(3)}_2A^{(3)}_3)\neq0,
\end{equation}
we obtain
\begin{equation}
 (C^{(33)},C^{(34)})\neq(0,0).
\end{equation}
Because of the following coefficients of $F_1G_2-F_2G_1=0$:
\begin{align*}
 &{U_1}^1{U_4}^6{v_1}^0:
 &&\dfrac{(A^{(1)}_2B^{(1)}_2+A^{(1)}_3B^{(1)}_3{d_2}^2d_3)C^{(12)}C^{(33)}}{A^{(1)}_2B^{(1)}_2}=0,\\
 &{U_1}^6{U_4}^2{v_1}^0:
 &&\dfrac{{d_2}^4{d_3}^4(A^{(1)}_2B^{(1)}_2+A^{(1)}_3B^{(1)}_3{d_2}^2d_3)^2C^{(12)}C^{(34)}}{{A^{(1)}_2}^2{B^{(1)}_2}^2}=0,
\end{align*}
we consider by dividing into the following four cases:
\begin{align}
 &A^{(1)}_2B^{(1)}_2+A^{(1)}_3B^{(1)}_3{d_2}^2d_3=0,\quad C^{(12)}=0;\tag{ii:1}\label{eqn:typeIII_proof_c_21}\\
 &A^{(1)}_2B^{(1)}_2+A^{(1)}_3B^{(1)}_3{d_2}^2d_3=0,\quad C^{(12)}\neq0;\tag{ii:2}\label{eqn:typeIII_proof_c_22}\\
 &A^{(1)}_2B^{(1)}_2+A^{(1)}_3B^{(1)}_3{d_2}^2d_3\neq0,\quad C^{(12)}=C^{(33)}=0;\tag{ii:3}\label{eqn:typeIII_proof_c_23}\\
 &A^{(1)}_2B^{(1)}_2+A^{(1)}_3B^{(1)}_3{d_2}^2d_3\neq0,\quad C^{(12)}=0,\quad C^{(33)}\neq0.\tag{ii:4}\label{eqn:typeIII_proof_c_24}
\end{align}
%
%
%

\subsubsection*{\quad Case \eqref{eqn:typeIII_proof_c_21}.}
From
\begin{equation}
 A^{(3)}_3C^{(12)}-A^{(3)}_1C^{(14)}
 =A^{(1)}_2A^{(2)}_1B^{(1)}_2B^{(3)}_1d_2d_3(A^{(3)}_2A^{(3)}_3-A^{(3)}_1A^{(3)}_4)\neq0,
\end{equation}
we obtain
\begin{equation}
 C^{(14)}\neq0.
\end{equation}
Since
\begin{equation*}
 v_4^{(r)}=\dfrac{C^{(14)}}{d_2d_3C^{(22)}U_1}+O({U_1}^0),
\end{equation*}
as $U_1\to0$, $U_1=0$ is a pole of $v_4$.
Hence, from
\begin{equation*}
 v_4^{(l)}=\dfrac{d_3(d_2d_3U_4C^{(37)}-C^{(33)}-U_4C^{(43)})}{U_4C^{(33)}}+O(U_1),
\end{equation*}
as $U_1\to0$, we obtain
\begin{equation}
 C^{(33)}=0.
\end{equation}
Because of the following relation:
\begin{equation}
 0=d_3C^{(12)}+C^{(33)}
 =A^{(1)}_3A^{(3)}_1B^{(1)}_3{d_2}^2{d_3}^2\Big((1+d_3)A^{(2)}_3B^{(3)}_1-(1+d_2)A^{(2)}_1B^{(3)}_3\Big),
\end{equation}
we consider by dividing into the following two cases:
\begin{align}
 &d_2=d_3=-1;\tag{ii:1.1}\label{eqn:typeIII_proof_c_211}\\
 &d_2,d_3\neq-1.\tag{ii:1.2}\label{eqn:typeIII_proof_c_212}
\end{align}
%
\subsubsection*{\qquad Case \eqref{eqn:typeIII_proof_c_211}.}
In this case, from 
\begin{equation}
 0=C^{(12)}=A^{(1)}_3B^{(1)}_3(A^{(2)}_3A^{(3)}_1B^{(3)}_1+A^{(2)}_1A^{(3)}_2B^{(3)}_1+A^{(2)}_1A^{(3)}_1B^{(3)}_3),
\end{equation}
we obtain
\begin{equation}
 B^{(3)}_3=-\dfrac{(A^{(2)}_3A^{(3)}_1+A^{(2)}_1A^{(3)}_2)B^{(3)}_1}{A^{(2)}_1A^{(3)}_1},
\end{equation}
and then the following hold:
\begin{align}
 &C^{(14)}=C^{(23)}=C^{(34)}=C^{(42)},\quad
 C^{(35)}=C^{(18)},\quad
 C^{(36)}=C^{(16)},\notag\\
 &C^{(37)}=C^{(17)},\quad
 C^{(38)}=C^{(15)},\quad
 C^{(43)}=C^{(22)},\quad
 C^{(44)}=C^{(24)}.
\end{align}
Then, we obtain
\begin{align}
 0=&F_1G_2-F_2G_1\notag\\
 =&-U_1^3U_4^4\Big((C^{(14)}-C^{(17)}+C^{(22)})(U_1-U_4)+(C^{(15)}+C^{(18)})U_1U_4\Big)\notag\\
 &~\Big((C^{(14)}+C^{(17)}-C^{(22)})U_1+(C^{(15)}-C^{(24)}){U_1}^2\notag\\
 &\quad-(C^{(14)}+C^{(17)}-C^{(22)}){U_1}^2v_1+(C^{(22)}-C^{(14)})U_4\notag\\
 &\quad+(C^{(24)}-C^{(18)})U_1U_4+(C^{(14)}-C^{(17)}-C^{(22)})U_1U_4v_1\notag\\
 &\quad-C^{(16)}U_1^2U_4+(C^{(18)}-C^{(15)}){U_1}^2U_4v_1+C^{(17)}{U_1}^2U_4{v_1}^2\Big).
\end{align}
Therefore, we obtain the following two cases.
The first case is given by
\begin{equation}
 C^{(14)}-C^{(17)}+C^{(22)}=0,\quad
 C^{(15)}+C^{(18)}=0,
\end{equation}
which gives the condition \eqref{eqn:typeIII_cond2_1},
and the second case is given by
\begin{equation}
 C^{(16)}=C^{(17)}=0,\quad
 C^{(14)}=C^{(22)},\quad
 C^{(15)}=C^{(18)}=C^{(24)},
\end{equation}
which gives the condition \eqref{eqn:typeIII_cond2_2}.
We can verify by computation that under the conditions \eqref{eqn:typeIII_cond2_1} and \eqref{eqn:typeIII_cond2_2}
the CACO property and square property hold.
\subsubsection*{\qquad Case \eqref{eqn:typeIII_proof_c_212}.}
From $C^{(12)}=C^{(33)}=0$, we obtain
\begin{equation}
B^{(3)}_3=\dfrac{A^{(2)}_3B^{(3)}_1(1+d_3)}{A^{(2)}_1(1+d_2)},\quad
A^{(3)}_2=-\dfrac{A^{(2)}_3A^{(3)}_1(1-d_2d_3)}{A^{(2)}_1d_2d_3(1+d_2)},
\end{equation}
and then the following hold:
\begin{align}
 &-d_2C^{(23)}=-\dfrac{C^{(34)}}{d_3}=-d_2C^{(42)}=C^{(14)},\quad
 C^{(37)}=-d_2C^{(17)}.
\end{align}
Then, from $F_1G_2=F_2G_1$, we obtain the following two cases.
The first case is given by
\begin{equation}
\begin{cases}
 C^{(43)}=d_2(C^{(14)}-d_2d_3C^{(17)}),\quad
 C^{(22)}=\dfrac{C^{(14)}-{d_2}^3d_3C^{(17)}}{{d_2}^2d_3},\\
 C^{(44)}=-{d_2}^2d_3(d_2C^{(15)}+d_3C^{(18)}+C^{(24)}),\quad
 C^{(38)}=-d_2d_3C^{(18)},\\
 C^{(35)}=-{d_2}^2C^{(15)},\quad
 C^{(36)}=-{d_2}^2d_3C^{(16)},
\end{cases}
\end{equation}
which gives the conditions \eqref{eqn:typeIII_cond2_3} and \eqref{eqn:typeIII_cond2_4}.
The second case is given by
\begin{equation}
\begin{cases}
 C^{(43)}=d_2C^{(14)},\quad
 C^{(22)}=\dfrac{C^{(14)}}{{d_2}^2d_3},\quad
 C^{(44)}=-\dfrac{d_3(d_2^4C^{(15)}+d_2^3C^{(24)}+C^{(35)})}{d_2},\\
 C^{(38)}=-\dfrac{C^{(35)}}{{d_2}^2},\quad
 C^{(17)}=0,\quad
 C^{(36)}=-\dfrac{d_2d_3(d_2C^{(15)}+C^{(24)})C^{(35)}}{C^{(14)}},\\
 C^{(18)}=-\dfrac{C^{(15)}}{d_2d_3},\quad
 C^{(16)}=-\dfrac{d_2C^{(15)}(d_2C^{(15)}+C^{(24)})}{C^{(14)}},
\end{cases}
\end{equation}
which gives the conditions \eqref{eqn:typeIII_cond2_5} and \eqref{eqn:typeIII_cond2_6}. 
We can verify by computation that under the conditions \eqref{eqn:typeIII_cond2_3}--\eqref{eqn:typeIII_cond2_6} the CACO property and square property hold.

\subsubsection*{\quad Case \eqref{eqn:typeIII_proof_c_22}.}
In this case, solving $F_1G_2=F_2G_1$, we obtain the following four conditions:
\begin{align}
&\begin{cases}
C^{(15)}=\dfrac{C^{(14)}C^{(17)}}{C^{(12)}},\quad
C^{(16)}=\dfrac{C^{(14)}C^{(18)}}{C^{(12)}},\quad
C^{(24)}=\dfrac{C^{(14)}C^{(22)}}{C^{(12)}},\\
C^{(35)}=\dfrac{C^{(14)}(d_3C^{(12)}C^{(43)}-C^{(14)}C^{(33)})}{{d_3}^2{C^{(12)}}^2},\quad
C^{(36)}=\dfrac{d_2C^{(14)}C^{(38)}}{C^{(12)}},\\
C^{(37)}=\dfrac{d_3C^{(12)}C^{(43)}-C^{(14)}C^{(33)}}{d_2{d_3}^2C^{(12)}},\quad
C^{(44)}=\dfrac{C^{(14)}C^{(34)}+d_2{d_3}^2C^{(12)}C^{(38)}}{d_3C^{(12)}},
\end{cases}\label{eqn:case5_caseii-2_C1}\\
&\begin{cases}
C^{(15)}=\dfrac{C^{(14)}C^{(17)}}{C^{(12)}},\quad
C^{(16)}=-\dfrac{d_2C^{(17)}(d_2C^{(14)}C^{(17)}+C^{(12)}C^{(24)})}{{C^{(12)}}^2},\\
C^{(18)}=-\dfrac{d_2C^{(17)}(d_2C^{(17)}+C^{(22)})}{C^{(12)}},\quad
C^{(34)}=-\dfrac{d_2d_3(d_2C^{(17)}+C^{(22)})C^{(33)}}{C^{(12)}},\\
C^{(35)}=\dfrac{C^{(14)}(d_3C^{(12)}C^{(43)}-C^{(14)}C^{(33)})}{{d_3}^2{C^{(12)}}^2},\quad
C^{(37)}=\dfrac{C^{(43)}}{d_2d_3}-\dfrac{C^{(14)}C^{(33)}}{d_2{d_3}^2C^{(12)}},\\
C^{(36)}=-\dfrac{d_2(d_2C^{(14)}C^{(17)}+C^{(12)}C^{(24)})(d_3C^{(12)}C^{(43)}-C^{(14)}C^{(33)})}{d_3{C^{(12)}}^3},\\
C^{(38)}=-\dfrac{(d_2C^{(17)}+C^{(22)})(d_3C^{(12)}C^{(43)}-C^{(14)}C^{(33)})}{d_3{C^{(12)}}^2},\\
C^{(44)}=-\dfrac{d_2\Big((C^{(12)}C^{(24)}-C^{(14)}C^{(22)})C^{(33)}+d_3C^{(12)}(d_2C^{(17)}+C^{(22)})C^{(43)}\Big)}{{C^{(12)}}^2},
\end{cases}\label{eqn:case5_caseii-2_C2}\\
&\begin{cases}
C^{(34)}=-\dfrac{d_2d_3(d_2C^{(17)}+C^{(22)})C^{(33)}}{C^{(12)}},\quad
C^{(35)}=-\dfrac{d_2C^{(15)}C^{(33)}}{d_3C^{(12)}},\\
C^{(36)}=-\dfrac{d_2C^{(16)}C^{(33)}}{C^{(12)}},\quad
C^{(37)}=-\dfrac{C^{(17)}C^{(33)}}{d_3C^{(12)}},\quad
C^{(38)}=-\dfrac{C^{(18)}C^{(33)}}{C^{(12)}},\\
C^{(43)}=-\dfrac{(d_2d_3C^{(17)}-C^{(14)})C^{(33)}}{d_3C^{(12)}},\\
C^{(44)}=-\dfrac{d_2(d_2C^{(15)}+d_3C^{(18)}+C^{(24)})C^{(33)}}{C^{(12)}},
\end{cases}\label{eqn:case5_caseii-2_C3}\\
&\begin{cases}
C^{(16)}=-\dfrac{d_2C^{(15)}(d_2C^{(17)}+C^{(22)})}{C^{(12)}},\quad
C^{(18)}=-\dfrac{d_2C^{(17)}(d_2C^{(17)}+C^{(22)})}{C^{(12)}},\\
C^{(24)}=\dfrac{C^{(14)}(d_2C^{(17)}+C^{(22)})}{C^{(12)}}-d_2C^{(15)},\quad
C^{(35)}=-\dfrac{d_2C^{(15)}C^{(33)}}{d_3C^{(12)}},\\
C^{(36)}=-\dfrac{d_2C^{(15)}C^{(34)}}{d_3C^{(12)}},\quad
C^{(37)}=-\dfrac{C^{(17)}C^{(33)}}{d_3C^{(12)}},\quad
C^{(38)}=-\dfrac{C^{(17)}C^{(34)}}{d_3C^{(12)}},\\
C^{(43)}=\dfrac{(C^{(14)}-d_2d_3C^{(17)})C^{(33)}}{d_3C^{(12)}},\quad
C^{(44)}=\dfrac{(C^{(14)}-d_2d_3C^{(17)})C^{(34)}}{d_3C^{(12)}}.
\end{cases}\label{eqn:case5_caseii-2_C4}
\end{align}
Since under the conditions \eqref{eqn:typeIII_proof_c_2} and \eqref{eqn:typeIII_proof_c_22}
the following hold:
\begin{subequations}
\begin{align}
 &C^{(15)}-\dfrac{C^{(14)}C^{(17)}}{C^{(12)}}\notag\\
 &\quad=-\dfrac{A^{(1)}_3B^{(1)}_3{A^{(2)}_1}^2(A^{(3)}_1A^{(3)}_4-A^{(3)}_2A^{(3)}_3)(B^{(3)}_1B^{(3)}_4-B^{(3)}_2B^{(3)}_3){d_2}^2{d_3}^2}{A^{(2)}_1A^{(3)}_1B^{(3)}_3-A^{(2)}_3A^{(3)}_1B^{(3)}_1d_3+A^{(2)}_1A^{(3)}_2B^{(3)}_1d_2d_3}
\neq0,\\
 &C^{(18)}+\dfrac{d_2C^{(17)}(d_2C^{(17)}+C^{(22)})}{C^{(12)}}\notag\\
 &\quad=-\dfrac{A^{(1)}_2{A^{(3)}_1}^2B^{(1)}_3(A^{(2)}_1A^{(2)}_4-A^{(2)}_2A^{(2)}_3)(B^{(3)}_1B^{(3)}_4-B^{(3)}_2B^{(3)}_3)d_2d_3}{A^{(2)}_1A^{(3)}_1B^{(3)}_3-A^{(2)}_3A^{(3)}_1B^{(3)}_1d_3+A^{(2)}_1A^{(3)}_2B^{(3)}_1d_2d_3}
\neq0,
\end{align}
\end{subequations}
the conditions \eqref{eqn:case5_caseii-2_C1}, \eqref{eqn:case5_caseii-2_C2} and \eqref{eqn:case5_caseii-2_C4} are inadequate.
Then, from the condition \eqref{eqn:case5_caseii-2_C3} 
we obtain the conditions \eqref{eqn:typeIII_cond2_7}--\eqref{eqn:typeIII_cond2_9} and 
\begin{equation}\label{eqn:case5_caseii-2_4}
\begin{cases}
B^{(1)}_2=-\dfrac{A^{(1)}_3B^{(1)}_3{d_2}^2d_3}{A^{(1)}_2},\\
A^{(1)}_4=\dfrac{A^{(1)}_3A^{(3)}_3B^{(1)}_3d_2(1+{d_2}^2d_3)-A^{(1)}_2A^{(3)}_1B^{(1)}_4(1-d_2d_3)}{A^{(3)}_1B^{(1)}_3d_2(1+d_2)d_3},\\
A^{(2)}_2=\dfrac{A^{(2)}_1(A^{(1)}_2A^{(3)}_1B^{(1)}_4(1+d_3)-A^{(1)}_3A^{(3)}_3B^{(1)}_3d_2(1-d_2d_3)}{A^{(1)}_2A^{(3)}_1B^{(1)}_3(1+d_2)},\\
A^{(2)}_4=\dfrac{A^{(2)}_3(A^{(1)}_2A^{(3)}_1B^{(1)}_4(1+d_3)-A^{(1)}_3A^{(3)}_3B^{(1)}_3d_2(1-d_2d_3)}{A^{(1)}_2A^{(3)}_1B^{(1)}_3(1+d_2)},\\
A^{(3)}_2=-\dfrac{A^{(2)}_3A^{(3)}_1(1-d_2d_3)}{A^{(2)}_1d_2(1+d_2)d_3},\quad
A^{(3)}_4=-\dfrac{A^{(2)}_3A^{(3)}_3(1-d_2d_3)}{A^{(2)}_1d_2(1+d_2)d_3},\quad
 d_2\neq-1.
\end{cases}
\end{equation}
We can verify by computation that under the conditions \eqref{eqn:typeIII_cond2_7}--\eqref{eqn:typeIII_cond2_9}
the CACO property and square property hold.
Under the condition \eqref{eqn:case5_caseii-2_4} the CACO property holds, but since $v_4$ is given by
\begin{equation}
 v_4=-\dfrac{A^{(3)}_3}{A^{(3)}_1}-\dfrac{d_3}{U_4},
\end{equation}
the square property does not hold. 
Therefore, the condition \eqref{eqn:case5_caseii-2_4} is inadequate.

\subsubsection*{\quad Case \eqref{eqn:typeIII_proof_c_23}.}
Under the condition \eqref{eqn:typeIII_proof_c_2} the following relations hold:
\begin{subequations}
\begin{align}
 &A^{(3)}_3C^{(12)}-A^{(3)}_1C^{(14)}
 =-A^{(1)}_2A^{(2)}_1B^{(1)}_2B^{(3)}_1d_2d_3(A^{(3)}_1A^{(3)}_4-A^{(3)}_2A^{(3)}_3)\neq0,
 \label{eqn:type3_proof_ii3_1}\\
 &A^{(1)}_3B^{(1)}_3B^{(3)}_2d_2d_3C^{(12)}-A^{(1)}_2B^{(1)}_2B^{(3)}_1C^{(17)}\notag\\
 &\quad=-A^{(1)}_2A^{(1)}_3A^{(2)}_1A^{(3)}_1B^{(1)}_2B^{(1)}_3d_2d_3(B^{(3)}_1B^{(3)}_4-B^{(3)}_2B^{(3)}_3)\neq0,
 \label{eqn:type3_proof_ii3_2}\\
 &A^{(3)}_3C^{(33)}-A^{(3)}_1C^{(34)}
 =A^{(1)}_2A^{(2)}_1B^{(1)}_2B^{(3)}_1d_2{d_3}^2(A^{(3)}_1A^{(3)}_4-A^{(3)}_2A^{(3)}_3)\neq0,\\
 &A^{(1)}_3B^{(1)}_3B^{(3)}_2d_2C^{(33)}+A^{(1)}_2B^{(1)}_2B^{(3)}_1C^{(37)}\notag\\
 &\quad=-A^{(1)}_2A^{(1)}_3A^{(2)}_1A^{(3)}_1B^{(1)}_2B^{(1)}_3{d_2}^2d_3(B^{(3)}_1B^{(3)}_4-B^{(3)}_2B^{(3)}_3)\neq0,\\
 &A^{(1)}_3A^{(3)}_3B^{(1)}_3d_2C^{(33)}+A^{(1)}_2A^{(3)}_1B^{(1)}_2C^{(42)}\notag\\
 &\quad=A^{(1)}_2A^{(1)}_3A^{(2)}_1B^{(1)}_2B^{(1)}_3B^{(3)}_1{d_2}^2{d_3}^2(A^{(3)}_1A^{(3)}_4-A^{(3)}_2A^{(3)}_3)\neq0.
\end{align}
\end{subequations}
Because of the condition \eqref{eqn:typeIII_proof_c_23}, we obtain
\begin{equation}
 C^{(14)},C^{(17)},C^{(34)},C^{(37)},C^{(42)}\neq0.
\end{equation}
From $F_1G_2=F_2G_1$ we obtain the following two conditions:
\begin{align}
 &\begin{cases}\label{eqn:case5_caseii3_C1}
C^{(12)}=C^{(33)}=0,\quad
C^{(16)}=-\dfrac{A^{(1)}_2B^{(1)}_2(A^{(1)}_2B^{(1)}_2+A^{(1)}_3B^{(1)}_3{d_2}^2d_3)C^{(15)}C^{(18)}}{{A^{(1)}_3}^2{B^{(1)}_3}^2{d_2}^3d_3C^{(14)}},\\
C^{(17)}=-\dfrac{{A^{(1)}_3}^2{B^{(1)}_3}^2{d_2}^3d_3C^{(14)}}{A^{(1)}_2B^{(1)}_2(A^{(1)}_2B^{(1)}_2+A^{(1)}_3B^{(1)}_3{d_2}^2d_3)},\\
C^{(22)}=-\dfrac{A^{(1)}_3B^{(1)}_3{d_2}^2C^{(14)}}{A^{(1)}_2B^{(1)}_2+A^{(1)}_3B^{(1)}_3{d_2}^2d_3},\\
C^{(24)}=-\dfrac{-A^{(1)}_2B^{(1)}_2d_2C^{(15)}+A^{(1)}_2B^{(1)}_2d_3C^{(18)}+A^{(1)}_3B^{(1)}_3{d_2}^2{d_3}^2C^{(18)}}{A^{(1)}_3B^{(1)}_3{d_2}^2d_3},\\
C^{(35)}=-\dfrac{A^{(1)}_2B^{(1)}_2(A^{(1)}_2B^{(1)}_2+A^{(1)}_3B^{(1)}_3{d_2}^2d_3)C^{(15)}C^{(37)}}{{A^{(1)}_3}^2{B^{(1)}_3}^2{d_2}^2d_3C^{(14)}},\\
C^{(36)}=-\dfrac{A^{(1)}_2B^{(1)}_2(A^{(1)}_2B^{(1)}_2+A^{(1)}_3B^{(1)}_3{d_2}^2d_3)C^{(15)}C^{(38)}}{{A^{(1)}_3}^2{B^{(1)}_3}^2{d_2}^2d_3C^{(14)}},\\
C^{(43)}=-\dfrac{A^{(1)}_2B^{(1)}_2C^{(37)}}{A^{(1)}_3B^{(1)}_3d_2},\\
C^{(44)}=\dfrac{A^{(1)}_2B^{(1)}_2C^{(15)}C^{(34)}+A^{(1)}_3B^{(1)}_3{d_2}^2d_3C^{(15)}C^{(34)}-A^{(1)}_2B^{(1)}_2d_3C^{(14)}C^{(38)}}{A^{(1)}_3B^{(1)}_3d_2d_3C^{(14)}},
\end{cases}\\
&\begin{cases}\label{eqn:case5_caseii3_C2}
C^{(12)}=C^{(33)}=0,\quad
C^{(16)}=\dfrac{C^{(14)}C^{(36)}}{C^{(34)}},\quad
C^{(17)}=\dfrac{d_2d_3C^{(14)}C^{(37)}}{C^{(34)}},\\
C^{(22)}=-\dfrac{{d_2}^2C^{(14)}(A^{(1)}_3B^{(1)}_3C^{(34)}+A^{(1)}_2B^{(1)}_2d_3C^{(37)})}{A^{(1)}_2B^{(1)}_2C^{(34)}},\quad
C^{(35)}=\dfrac{C^{(15)}C^{(34)}}{d_3C^{(14)}},\\
C^{(38)}=\dfrac{C^{(18)}C^{(34)}}{d_2C^{(14)}},\quad
C^{(43)}=\dfrac{d_2(A^{(1)}_3B^{(1)}_3C^{(34)}+A^{(1)}_2B^{(1)}_2d_3C^{(37)})}{A^{(1)}_2B^{(1)}_2},\\
C^{(44)}=\dfrac{(d_2C^{(15)}+d_3C^{(18)}+C^{(24)})C^{(34)}}{C^{(14)}}.
\end{cases}
\end{align}
Under the condition \eqref{eqn:case5_caseii3_C1} the CACO property holds but since $v_4$ is given by
\begin{align}
 v_4=&\dfrac{A^{(1)}_2B^{(1)}_2d_2d_3U_1-(A^{(1)}_2B^{(1)}_2+A^{(1)}_3B^{(1)}_3{d_2}^2d_3)U_4}{A^{(1)}_3B^{(1)}_3{d_2}^3d_3U_1U_4}\notag\\
 &+\dfrac{A^{(1)}_2B^{(1)}_2(A^{(1)}_2B^{(1)}_2+A^{(1)}_3B^{(1)}_3{d_2}^2d_3)C^{(15)}}{{A^{(1)}_3}^2{B^{(1)}_3}^2{d_2}^3d_3C^{(14)}},
\end{align}
the square property does not hold.
Therefore, the condition \eqref{eqn:case5_caseii3_C1} is inadequate.
From the condition \eqref{eqn:case5_caseii3_C2} 
we obtain the conditions \eqref{eqn:typeIII_cond2_10}--\eqref{eqn:typeIII_cond2_12},
\begin{align}
&\begin{cases}
 A^{(1)}_1=B^{(1)}_1=A^{(2)}_3=B^{(2)}_3=A^{(2)}_4=B^{(2)}_4=0,\quad
 A^{(1)}_4=\dfrac{A^{(1)}_2A^{(2)}_2}{A^{(2)}_1d_3},\\
 A^{(3)}_2=\dfrac{A^{(3)}_1B^{(3)}_3}{B^{(3)}_1d_3},\quad
 A^{(3)}_4=\dfrac{A^{(3)}_1(-A^{(1)}_2B^{(1)}_4B^{(3)}_3+A^{(1)}_2B^{(1)}_2B^{(3)}_4+A^{(1)}_3B^{(1)}_3B^{(3)}_4d_3)}{A^{(1)}_3B^{(1)}_3B^{(3)}_1d_3},\\
 B^{(3)}_2=\dfrac{A^{(1)}_3A^{(3)}_3B^{(1)}_3B^{(3)}_1+A^{(1)}_2A^{(3)}_1B^{(1)}_4B^{(3)}_1}{A^{(3)}_1(A^{(1)}_2B^{(1)}_2+A^{(1)}_3B^{(1)}_3d_3)},\quad
 d_2=-1,
\end{cases}\label{eqn:case5_caseii3_2}\\
\intertext{and}
&\begin{cases}
 A^{(1)}_1=B^{(1)}_1=0,\quad
 A^{(3)}_3=\dfrac{A^{(3)}_1\Big(A^{(1)}_4A^{(2)}_1d_3(1+d_3)+A^{(1)}_2(A^{(2)}_2-A^{(2)}_2d_3)\Big)}{2A^{(1)}_3A^{(2)}_1d_3},\\
 A^{(2)}_3=\dfrac{2A^{(2)}_1B^{(3)}_3}{B^{(3)}_1(1+d_3)},\quad
 A^{(2)}_4=\dfrac{2A^{(2)}_2B^{(3)}_3}{B^{(3)}_1(1+d_3)},\quad
 A^{(3)}_2=-\dfrac{A^{(3)}_1B^{(3)}_3(1-d_3)}{B^{(3)}_1d_3(1+d_3)},\\
 A^{(3)}_4=\dfrac{A^{(3)}_1}{2A^{(1)}_3A^{(2)}_1B^{(1)}_3{B^{(3)}_1}^2{d_3}^2(1+d_3)}\\
	\quad\Big(A^{(2)}_1B^{(1)}_3d_3(1+d_3)\Big(A^{(1)}_4B^{(3)}_1B^{(3)}_3(-1+d_3)+2A^{(1)}_3(B^{(3)}_2B^{(3)}_3-B^{(3)}_1B^{(3)}_4)d_3\Big)\\
	\qquad+A^{(1)}_2\Big(-A^{(2)}_2B^{(1)}_3B^{(3)}_1B^{(3)}_3(-1+d_3)^2\\
	\qquad+2A^{(2)}_1B^{(1)}_2(B^{(3)}_2B^{(3)}_3-B^{(3)}_1B^{(3)}_4)d_3(1+d_3)\Big)\Big),\\
 B^{(1)}_4=\dfrac{2A^{(2)}_1B^{(1)}_2B^{(3)}_2d_3+A^{(2)}_2B^{(1)}_3B^{(3)}_1(1+d_3)}{2A^{(2)}_1B^{(3)}_1d_3}\\
 	\quad+\dfrac{B^{(1)}_3\Big(2A^{(1)}_3B^{(3)}_2d_3+A^{(1)}_4(B^{(3)}_1-B^{(3)}_1d_3)\Big)}{2A^{(1)}_2B^{(3)}_1},\\
 d_2=1,\quad
 d_3\neq-1.
\end{cases}\label{eqn:case5_caseii3_5}
\end{align}
We can verify by computation that under the conditions \eqref{eqn:typeIII_cond2_10}--\eqref{eqn:typeIII_cond2_12}
the CACO property and square property hold.
Moreover, we can easily check that under the conditions \eqref{eqn:case5_caseii3_2} and \eqref{eqn:case5_caseii3_5},
the square property does not hold since $v_4$ does not depend on $v_1$.
Therefore, the conditions \eqref{eqn:case5_caseii3_2} and \eqref{eqn:case5_caseii3_5} are inadequate.

\subsubsection*{\quad Case \eqref{eqn:typeIII_proof_c_24}.}
From the conditions \eqref{eqn:type3_proof_ii3_1} and \eqref{eqn:type3_proof_ii3_2} and \eqref{eqn:typeIII_proof_c_24} 
we obtain
\begin{equation}
C^{(14)},C^{(17)}\neq0.
\end{equation}
From $F_1G_2=F_2G_1$ we obtain the following condition:
\begin{equation}
\begin{cases}
 C^{(12)}=0,\quad
 C^{(16)}=-\dfrac{A^{(1)}_2B^{(1)}_2C^{(15)}(A^{(1)}_2B^{(1)}_2C^{(15)}-A^{(1)}_3B^{(1)}_3d_2d_3C^{(24)})}{{A^{(1)}_3}^2{B^{(1)}_3}^2{d_2}^2{d_3}^2C^{(14)}},\\
 C^{(17)}=-\dfrac{{A^{(1)}_3}^2{B^{(1)}_3}^2{d_2}^3d_3C^{(14)}}{A^{(1)}_2B^{(1)}_2(A^{(1)}_2B^{(1)}_2+A^{(1)}_3B^{(1)}_3{d_2}^2d_3)},\\
 C^{(18)}=\dfrac{d_2(A^{(1)}_2B^{(1)}_2C^{(15)}-A^{(1)}_3B^{(1)}_3d_2d_3C^{(24)})}{d_3(A^{(1)}_2B^{(1)}_2+A^{(1)}_3B^{(1)}_3{d_2}^2d_3)},\\
 C^{(22)}=-\dfrac{A^{(1)}_3B^{(1)}_3{d_2}^2C^{(14)}}{A^{(1)}_2B^{(1)}_2+A^{(1)}_3B^{(1)}_3{d_2}^2d_3},\\
 C^{(35)}=-\dfrac{A^{(1)}_2B^{(1)}_2(A^{(1)}_2B^{(1)}_2+A^{(1)}_3B^{(1)}_3{d_2}^2d_3)C^{(15)}C^{(37)}}{{A^{(1)}_3}^2{B^{(1)}_3}^2{d_2}^2d_3C^{(14)}},\\
 C^{(36)}=-\dfrac{A^{(1)}_2B^{(1)}_2(A^{(1)}_2B^{(1)}_2+A^{(1)}_3B^{(1)}_3{d_2}^2d_3)C^{(15)}C^{(38)}}{{A^{(1)}_3}^2{B^{(1)}_3}^2{d_2}^2d_3C^{(14)}},\\
 C^{(43)}=\dfrac{(A^{(1)}_2B^{(1)}_2+A^{(1)}_3B^{(1)}_3{d_2}^2d_3)C^{(15)}C^{(33)}-A^{(1)}_2B^{(1)}_2d_3C^{(14)}C^{(37)}}{A^{(1)}_3B^{(1)}_3d_2d_3C^{(14)}},\\
 C^{(44)}=\dfrac{(A^{(1)}_2B^{(1)}_2+A^{(1)}_3B^{(1)}_3{d_2}^2d_3)C^{(15)}C^{(34)}-A^{(1)}_2B^{(1)}_2d_3C^{(14)}C^{(38)}}{A^{(1)}_3B^{(1)}_3d_2d_3C^{(14)}}.
\end{cases}
\end{equation}
From 
\begin{align}
 0&=C^{(16)}+\dfrac{A^{(1)}_2B^{(1)}_2C^{(15)}(A^{(1)}_2B^{(1)}_2C^{(15)}-A^{(1)}_3B^{(1)}_3d_2d_3C^{(24)})}{{A^{(1)}_3}^2{B^{(1)}_3}^2{d_2}^2{d_3}^2C^{(14)}}\notag\\
 &=-\dfrac{{A^{(1)}_2}^2B^{(1)}_2B^{(1)}_3{A^{(3)}_3}^2d_2d_3(A^{(2)}_1A^{(2)}_4-A^{(2)}_2A^{(2)}_3)(B^{(3)}_1B^{(3)}_4-B^{(3)}_2B^{(3)}_3)}{C^{(14)}},
\end{align}
we obtain
\begin{equation}
 A^{(3)}_3=B^{(3)}_3=0.
\end{equation}
Then, from the condition \eqref{eqn:caseN1_CAO_B2} we obtain
\begin{equation}
 A^{(2)}_2=B^{(2)}_2=0.
\end{equation}
From
\begin{equation}
 0=C^{(12)}=A^{(1)}_2B^{(1)}_2B^{(3)}_1d_3(A^{(2)}_1A^{(3)}_2d_2-A^{(2)}_3A^{(3)}_1),
\end{equation}
we obtain
\begin{equation}
A^{(2)}_3=\dfrac{A^{(2)}_1A^{(3)}_2d_2}{A^{(3)}_1},
\end{equation}
Moreover, from
\begin{align}
 0&=C^{(17)}+\dfrac{{A^{(1)}_3}^2{B^{(1)}_3}^2{d_2}^3d_3C^{(14)}}{A^{(1)}_2B^{(1)}_2(A^{(1)}_2B^{(1)}_2+A^{(1)}_3B^{(1)}_3{d_2}^2d_3)}\notag\\
 &=A^{(1)}_3A^{(2)}_1B^{(1)}_3d_2d_3
 \left(A^{(3)}_1B^{(3)}_4+\dfrac{A^{(1)}_3A^{(3)}_4B^{(1)}_3B^{(3)}_1{d_2}^3d_3}{A^{(1)}_2B^{(1)}_2+A^{(1)}_3B^{(1)}_3{d_2}^2d_3}\right),
\end{align}
we obtain
\begin{equation}
 B^{(3)}_4=-\dfrac{A^{(1)}_3A^{(3)}_4B^{(1)}_3B^{(3)}_1{d_2}^3d_3}{A^{(3)}_1(A^{(1)}_2B^{(1)}_2+A^{(1)}_3B^{(1)}_3{d_2}^2d_3)}.
\end{equation}
Then, we obtain the following inadequate condition:
\begin{equation}
 0=C^{(22)}+\dfrac{A^{(1)}_3B^{(1)}_3{d_2}^2C^{(14)}}{A^{(1)}_2B^{(1)}_2+A^{(1)}_3B^{(1)}_3{d_2}^2d_3}
=A^{(1)}_2A^{(2)}_4A^{(3)}_1B^{(1)}_3B^{(3)}_1d_2d_3\neq0.
\end{equation}
Therefore, the case \eqref{eqn:typeIII_proof_c_24} is inadequate.
\subsubsection{Case \eqref{eqn:typeIII_proof_c_3}.}
\label{subsubsection:caseN1_case3}
Because of the condition {\bf (iii)} and the forms of quad-equations $Q_i$, $i=1,\dots,9$,
we can without loss of generality let 
\begin{equation}
 A^{(i)}_1=B^{(i)}_1=1,\quad i=1,2,3.
\end{equation}
Set
\begin{equation}
 a^{(i)}_4=A^{(i)}_4-A^{(i)}_2A^{(i)}_3,\quad
 b^{(i)}_4=B^{(i)}_4-B^{(i)}_2B^{(i)}_3,\quad 
 i=1,2,3.
\end{equation}
Then, the quad-equations $Q_i$, $i=1,\dots,9$, can be rewritten as the following:
\begin{equation}\label{eqn:typeIII_c3_Q}
\begin{split}
 &Q_1=(u_1+d_2v_4+A^{(1)}_2)(u_5+d_1v_5+A^{(1)}_3)+a^{(1)}_4,\\
 &Q_2=(d_2u_4+v_1+A^{(1)}_2)(d_1u_2+v_2+A^{(1)}_3)+a^{(1)}_4,\\
 &Q_3=(u_5+d_3v_2+A^{(2)}_2)(u_3+d_2v_3+A^{(2)}_3)+a^{(2)}_4,\\
 &Q_4=(d_3u_2+v_5+A^{(2)}_2)(d_2u_6+v_6+A^{(2)}_3)+a^{(2)}_4,\\
 &Q_5=(u_3+d_1v_6+A^{(3)}_2)(u_1+d_3v_1+A^{(3)}_3)+a^{(3)}_4,\\
 &Q_6=(d_1u_6+v_3+A^{(3)}_2)(d_3u_4+v_4+A^{(3)}_3)+a^{(3)}_4,\\
 &Q_7=(u_2+d_2u_5+B^{(1)}_2)(d_1u_1+u_4+B^{(1)}_3)+b^{(1)}_4,\\
 &Q_8=(d_3u_3+u_6+B^{(2)}_2)(u_2+d_2u_5+B^{(2)}_3)+b^{(2)}_4,\\
 &Q_9=(d_1u_1+u_4+B^{(3)}_2)(d_3u_3+u_6+B^{(3)}_3)+b^{(3)}_4.
\end{split}
\end{equation}
Note that 
\begin{equation}
 a^{(i)}_4,b^{(i)}_4\neq0,\quad i=1,2,3,
\end{equation}
and the condition \eqref{eqn:caseN1_CAO_B2} can be rewritten as \eqref{eqn:caseN1_CAO_B2_2}.
For simplicity, we introduce the polynomials $C^{(51)},\dots,C^{(58)}$ given by
\begin{subequations}
\begin{align}
 &C^{(51)}=B^{(1)}_2-A^{(2)}_2d_2+A^{(1)}_3d_2d_3,
 &&C^{(52)}=B^{(3)}_3-A^{(2)}_3d_3+A^{(3)}_2d_2d_3,\\
 &C^{(53)}=A^{(2)}_2-B^{(1)}_2d_3+A^{(1)}_3d_2d_3,
 &&C^{(54)}=A^{(2)}_3-B^{(3)}_3d_2+A^{(3)}_2d_2d_3,\\
 &C^{(55)}=A^{(1)}_2b^{(1)}_4+a^{(1)}_4B^{(1)}_3d_2d_3,
 &&C^{(56)}=A^{(3)}_3b^{(3)}_4+a^{(3)}_4B^{(3)}_2d_2d_3,\\
 &C^{(57)}=A^{(1)}_2b^{(1)}_4-a^{(1)}_4B^{(1)}_3d_2,
 &&C^{(58)}=A^{(3)}_3b^{(3)}_4-a^{(3)}_4B^{(3)}_2d_3.
\end{align}
\end{subequations}

From the following relation:
\begin{equation}
 (a^{(1)}_4d_2d_3+A^{(1)}_2C^{(51)})C^{(11)}-C^{(51)}C^{(12)}
 =a^{(1)}_4a^{(2)}_4{d_2}^2{d_3}^2\neq0,
\end{equation}
we obtain 
\begin{equation}
 (C^{(11)},C^{(12)})\neq(0,0).
\end{equation}
Hence, from the following coefficients of $F_1G_2-F_2G_1=0$:
\begin{align*}
 {U_1}^0{U_4}^5{v_1}^0:\quad
 &d_2C^{(11)}C^{(31)}=0,\\
 {U_1}^0{U_4}^6{v_1}^0:\quad
 &C^{(12)}C^{(31)}=0,
\end{align*}
we obtain
\begin{equation}
 C^{(31)}=0.
\end{equation}
Moreover, from he following relation:
\begin{equation}
 (a^{(3)}_4d_2d_3+A^{(3)}_3C^{(54)})C^{(31)}-C^{(54)}C^{(32)}= a^{(2)}_4a^{(3)}_4d_2d_3\neq0,
\end{equation}
we obtain
\begin{equation}
 C^{(32)}\neq0,
\end{equation}
and then from the following coefficient of $F_1G_2-F_2G_1=0$:
\begin{equation}
 {U_1}^6{U_4}^0{v_1}^0:
 -{d_2}^6{d_3}^5C^{(11)}C^{(32)}=0,
\end{equation}
we obtain
\begin{equation}
 C^{(11)}=0.
\end{equation}
Therefore, from
\begin{equation}
 0=C^{(11)}=d_2d_3a^{(2)}_4+C^{(51)}C^{(52)},\quad
 0=C^{(31)}=a^{(2)}_4+C^{(53)}C^{(54)},
\end{equation}
and $a^{(2)}_4\neq0$, we obtain
\begin{equation}
 C^{(51)},C^{(52)},C^{(53)},C^{(54)}\neq0.
\end{equation}
From the following coefficients of $F_1G_2-F_2G_1=0$:
\begin{align*}
 {U_1}^6{U_4}^1{v_1}^0:
 &\dfrac{a^{(3)}_4{d_2}^7{d_3}^6C^{(53)}\Big(a^{(2)}_4d_2d_3(b^{(1)}_4+a^{(1)}_4d_3)-(b^{(3)}_4+a^{(3)}_4d_2){C^{(51)}}^2\Big)}{C^{(51)}}=0,\\
 {U_1}^1{U_4}^6{v_1}^0:
 &\dfrac{a^{(1)}_4a^{(2)}_4{d_2}^3{d_3}^3\Big(a^{(2)}_4(a^{(1)}_4+b^{(1)}_4d_3)-(a^{(3)}_4+b^{(3)}_4d_2){C^{(53)}}^2\Big)}{C^{(51)}C^{(53)}}=0,\\
 {U_1}^6{U_4}^3{v_1}^2:
 &\dfrac{a^{(1)}_4a^{(2)}_4{d_2}^5{d_3}^6}{C^{(51)}C^{(53)}}\Big(a^{(2)}_4b^{(1)}_4d_2d_3(1-{d_3}^2)-b^{(3)}_4{C^{(51)}}^2+b^{(3)}_4{d_2}^2{d_3}^2{C^{(53)}}^2\Big)\notag\\
 &+\dfrac{{d_2}^5{d_3}^6}{C^{(51)}C^{(53)}}
 \Big(a^{(2)}_4b^{(1)}_4d_3-b^{(3)}_4d_2{C^{(53)}}^2\Big)\notag\\
 &\hspace{5em}\Big(a^{(2)}_4d_2d_3(b^{(1)}_4+a^{(1)}_4d_3)-(b^{(3)}_4+a^{(3)}_4d_2){C^{(51)}}^2\Big)=0,
\end{align*}
we obtain
\begin{subequations}
\begin{align}
 &a^{(2)}_4d_2d_3(b^{(1)}_4+a^{(1)}_4d_3)-(b^{(3)}_4+a^{(3)}_4d_2){C^{(51)}}^2=0,\label{eqn:case5_ciii_cond_1}\\
 &a^{(2)}_4(a^{(1)}_4+b^{(1)}_4d_3)-(a^{(3)}_4+b^{(3)}_4d_2){C^{(53)}}^2=0,\label{eqn:case5_ciii_cond_2}\\
 &a^{(2)}_4b^{(1)}_4d_2d_3(1-{d_3}^2)-b^{(3)}_4{C^{(51)}}^2+b^{(3)}_4{d_2}^2{d_3}^2{C^{(53)}}^2=0.\label{eqn:case5_ciii_cond_3}
\end{align}
\end{subequations}
In the following, we consider by dividing into the following three cases.
\begin{align}
 &b^{(3)}_4+a^{(3)}_4d_2=0,\quad a^{(3)}_4+b^{(3)}_4d_2=0;\tag{iii:1}\label{eqn:typeIII_proof_c_31}\\
 &b^{(3)}_4+a^{(3)}_4d_2=0,\quad a^{(3)}_4+b^{(3)}_4d_2\neq0;\tag{iii:2}\label{eqn:typeIII_proof_c_32}\\
 &b^{(3)}_4+a^{(3)}_4d_2\neq0,\quad a^{(3)}_4+b^{(3)}_4d_2\neq0.\tag{iii:3}\label{eqn:typeIII_proof_c_33}
\end{align}
%
%
Note that considering the case \eqref{eqn:typeIII_proof_c_32} and considering the case
\begin{equation}
 b^{(3)}_4+a^{(3)}_4d_2\neq0,\quad a^{(3)}_4+b^{(3)}_4d_2=0,
\end{equation}
mean the same under the following symmetry of quad equations \eqref{eqn:typeIII_c3_Q}:
\begin{align}
 \iota:&~
 u_1\leftrightarrow u_4,\quad u_2\leftrightarrow u_5,\quad u_3\leftrightarrow u_6,\quad 
 v_1\leftrightarrow v_4,\quad v_2\leftrightarrow v_5,\quad v_3\leftrightarrow v_6,\notag\\
 &A^{(1)}_2\mapsto {d_2}^{-1}A^{(1)}_2,\quad
 A^{(1)}_3\mapsto {d_1}^{-1}A^{(1)}_3,\quad
 a^{(1)}_4\mapsto {d_1}^{-1}{d_2}^{-1}a^{(1)}_2,\notag\\
 &A^{(2)}_2\mapsto {d_3}^{-1}A^{(2)}_2,\quad
 A^{(2)}_3\mapsto {d_2}^{-1}A^{(2)}_3,\quad
 a^{(2)}_4\mapsto {d_2}^{-1}{d_3}^{-1}a^{(2)}_2,\notag\\
 &A^{(3)}_2\mapsto {d_1}^{-1}A^{(3)}_2,\quad
 A^{(3)}_3\mapsto {d_3}^{-1}A^{(3)}_3,\quad
 a^{(3)}_4\mapsto {d_3}^{-1}{d_1}^{-1}a^{(3)}_2,\notag\\
 &B^{(1)}_2\mapsto {d_2}^{-1}B^{(1)}_2,\quad
 B^{(1)}_3\mapsto {d_1}^{-1}B^{(1)}_3,\quad
 b^{(1)}_4\mapsto {d_1}^{-1}{d_2}^{-1}b^{(1)}_2,\notag\\
 &B^{(2)}_2\mapsto {d_3}^{-1}B^{(2)}_2,\quad
 B^{(2)}_3\mapsto {d_2}^{-1}B^{(2)}_3,\quad
 b^{(2)}_4\mapsto {d_2}^{-1}{d_3}^{-1}b^{(2)}_2,\notag\\
 &B^{(3)}_2\mapsto {d_1}^{-1}B^{(3)}_2,\quad
 B^{(3)}_3\mapsto {d_3}^{-1}B^{(3)}_3,\quad
 b^{(3)}_4\mapsto {d_3}^{-1}{d_1}^{-1}b^{(3)}_2,\notag\\
 &d_1\mapsto{d_1}^{-1},\quad
 d_2\mapsto{d_2}^{-1},\quad
 d_3\mapsto{d_3}^{-1}.
\end{align}

\begin{remark}
In what follows, for simplify, we let $c(i,j,k)$ be the coefficient of ${U_1}^i{U_4}^j{v_1}^k$ in $F_1G_2-F_2G_1$.
If the CACO property holds, then 
\begin{equation}
 c(i,j,k)=0
\end{equation}
for all $\{i,j,k\}$ hold.
\end{remark}

\subsubsection*{\quad Case \eqref{eqn:typeIII_proof_c_31}.}
From the condition \eqref{eqn:typeIII_proof_c_31} we obtain
\begin{equation}
 {d_2}^2=1.
\end{equation}
Moreover, from \eqref{eqn:case5_ciii_cond_1} and \eqref{eqn:case5_ciii_cond_2}, we obtain
\begin{equation}\label{eqn:typeIII_proof_c_31_same}
 b^{(1)}_4+a^{(1)}_4d_3=0,\quad
 a^{(1)}_4+b^{(1)}_4d_3=0,
\end{equation}
which give
\begin{equation}
 {d_3}^2=1.
\end{equation}
Therefore, we shall consider the following three cases:
\begin{align}
 &d_2=d_3=1;\tag{iii:1.1}\label{eqn:typeIII_proof_c_311}\\
 &d_2=1,\quad d_3=-1;\tag{iii:1.2}\label{eqn:typeIII_proof_c_312}\\
 &d_2=d_3=-1.\tag{iii:1.3}\label{eqn:typeIII_proof_c_313}
\end{align}
%
%
Note that considering the case \eqref{eqn:typeIII_proof_c_312} and considering the case 
\begin{equation}
 d_2=-1,\quad d_3=1
\end{equation}
mean the same under the following symmetry of quad equations \eqref{eqn:typeIII_c3_Q}:
\begin{align}\label{rqn:case5_s23}
 s_{23}:&~
 u_2\leftrightarrow u_6,\quad
 u_3\leftrightarrow u_5,\quad
 v_1\leftrightarrow v_4,\quad
 v_2\leftrightarrow v_3,\quad
 v_5\leftrightarrow v_6,\notag\\
 &A^{(1)}_2\leftrightarrow A^{(3)}_3,\quad
 A^{(1)}_3\leftrightarrow A^{(3)}_2,\quad
 a^{(1)}_4\leftrightarrow a^{(3)}_4,\quad
 A^{(2)}_2\leftrightarrow A^{(2)}_3,\notag\\
 &B^{(1)}_2\leftrightarrow B^{(3)}_3,\quad
 B^{(1)}_3\leftrightarrow B^{(3)}_2,\quad
 b^{(1)}_4\leftrightarrow b^{(3)}_4,\quad
 B^{(2)}_2\leftrightarrow B^{(2)}_3,\quad
 d_2\leftrightarrow d_3.
\end{align}
Note that the symmetry $s_{23}$ does not change the conditions \eqref{eqn:typeIII_proof_c_31}--\eqref{eqn:typeIII_proof_c_33}.
For example, applying $s_{23}$ to the condition \eqref{eqn:typeIII_proof_c_31}, 
we obtain the condition \eqref{eqn:typeIII_proof_c_31_same},
which is equivalent to the condition \eqref{eqn:typeIII_proof_c_31}
because of the conditions \eqref{eqn:case5_ciii_cond_1} and \eqref{eqn:case5_ciii_cond_3}.

\subsubsection*{\qquad Case \eqref{eqn:typeIII_proof_c_311}.}
From 
\begin{subequations}
\begin{align}
 0&=a^{(2)}_4b^{(1)}_4d_2d_3(1-{d_3}^2)-b^{(3)}_4{C^{(51)}}^2+b^{(3)}_4{d_2}^2{d_3}^2{C^{(53)}}^2
 =4A^{(1)}_3a^{(3)}_4(B^{(1)}_2-A^{(2)}_2),\\
 0&\neq C^{(51)}=A^{(1)}_3+B^{(1)}_2-A^{(2)}_2,
\end{align}
\end{subequations}
we obtain the following two cases:
\begin{subequations}
\begin{align}
 &A^{(1)}_3=B^{(1)}_3=0,\quad B^{(1)}_2\neq A^{(2)}_2;\label{eqn:typeIII_proof_c_311_1}\\
 &A^{(1)}_3,B^{(1)}_3\neq0,\quad B^{(1)}_2=A^{(2)}_2.\label{eqn:typeIII_proof_c_311_2}
\end{align}
\end{subequations}

Assume \eqref{eqn:typeIII_proof_c_311_1}.
Then, from the condition \eqref{eqn:caseN1_CAO_B2_2} we obtain
\begin{equation}
 B^{(3)}_2=-1,\quad
 B^{(2)}_3=B^{(1)}_2-a^{(1)}_4,\quad
 B^{(3)}_3=B^{(2)}_2-a^{(3)}_4,\quad
 b^{(2)}_4=a^{(1)}_4a^{(3)}_4,
\end{equation}
and then we obtain the following inadequate condition:
\begin{equation}
 0=c(5,2,0)=2A^{(3)}_2a^{(1)}_4a^{(3)}_4(A^{(2)}_2-B^{(1)}_2)\neq0.
\end{equation}
Therefore, the condition \eqref{eqn:typeIII_proof_c_311_2} holds.

From
\begin{equation}
 0=c(5,2,0)=-2A^{(1)}_3a^{(1)}_4a^{(3)}_4(A^{(2)}_3-B^{(3)}_3),
\end{equation}
we obtain
\begin{equation}
 B^{(3)}_3=A^{(2)}_3.
\end{equation}
Moreover, from
\begin{equation}
 0=2c(6,2,0)-c(6,3,1)=4A^{(1)}_3A^{(3)}_2a^{(1)}_4a^{(3)}_4(A^{(1)}_2-A^{(3)}_3),
\end{equation}
we obtain the following two cases:
\begin{align}
 &A^{(3)}_2,B^{(3)}_2\neq0,\quad A^{(3)}_3=A^{(1)}_2;\label{eqn:typeIII_proof_c_311_21}\\
 &A^{(3)}_2=B^{(3)}_2=0.\label{eqn:typeIII_proof_c_311_22}
\end{align}
When \eqref{eqn:typeIII_proof_c_311_21}, we obtain the following inadequate condition:
\begin{equation}
 0=2a^{(1)}_4A^{(3)}_2c(6,2,0)+A^{(1)}_3a^{(3)}_4c(3,5,0)=-4A^{(1)}_3A^{(3)}_2{a^{(1)}_4}^2{a^{(3)}_4}^2\neq0.
\end{equation}
Moreover, when \eqref{eqn:typeIII_proof_c_311_22}, from
\begin{equation}
 0=c(6,2,0)=-A^{(1)}_3{a^{(3)}_4}^2\Big(a^{(1)}_4+A^{(1)}_3(A^{(1)}_2-A^{(3)}_3)\Big),
\end{equation}
we obtain 
\begin{equation}
 a^{(1)}_4=-A^{(1)}_3(A^{(1)}_2-A^{(3)}_3),
\end{equation}
and then we obtain the following inadequate condition:
\begin{equation}
 0=c(6,3,0)
 ={A^{(1)}_3}^2{a^{(3)}_4}^2(A^{(1)}_2-A^{(3)}_3)^2
 =(a^{(1)}_4a^{(3)}_4)^2\neq0.
\end{equation}
Therefore, the case \eqref{eqn:typeIII_proof_c_311} is inadequate.
\subsubsection*{\qquad Case \eqref{eqn:typeIII_proof_c_312}.}
From 
\begin{subequations}
\begin{align}
 0&=a^{(2)}_4b^{(1)}_4d_2d_3(1-{d_3}^2)-b^{(3)}_4{C^{(51)}}^2+b^{(3)}_4{d_2}^2{d_3}^2{C^{(53)}}^2\notag\\
 &=4A^{(2)}_2a^{(3)}_4(A^{(1)}_3-B^{(1)}_2),\\
 0&\neq C^{(51)}=A^{(2)}_2-A^{(1)}_3+B^{(1)}_2,
\end{align}
\end{subequations}
we obtain the following two cases:
\begin{subequations}
\begin{align}
 &A^{(2)}_2=B^{(2)}_2=0,\quad B^{(1)}_2\neq A^{(1)}_3;\label{eqn:typeIII_proof_c_312_1}\\
 &A^{(2)}_2,B^{(2)}_2\neq0,\quad B^{(1)}_2=A^{(1)}_3.\label{eqn:typeIII_proof_c_312_2}
\end{align}
\end{subequations}

Assume \eqref{eqn:typeIII_proof_c_312_1}.
Then, from the condition \eqref{eqn:caseN1_CAO_B2_2} we obtain
\begin{equation}
 B^{(3)}_2=-1+B^{(1)}_3,\quad
 B^{(3)}_3=-a^{(3)}_4+B^{(2)}_2,\quad
 B^{(2)}_3=a^{(1)}_4+B^{(1)}_2,\quad
 b^{(2)}_4=-a^{(1)}_4a^{(3)}_4.
\end{equation}
From
\begin{equation}
 0=c(5,2,0)=-2a^{(1)}_4a^{(3)}_4(A^{(1)}_3-B^{(1)}_2)(A^{(2)}_3-A^{(3)}_2),
\end{equation}
we obtain
\begin{equation}
 A^{(3)}_2=A^{(2)}_3.
\end{equation}
From
\begin{equation}
 0=c(5,3,0)-c(6,3,1)=8a^{(1)}_4A^{(3)}_3B^{(3)}_3a^{(3)}_4(A^{(1)}_3-B^{(1)}_2),
\end{equation}
we obtain
\begin{equation}
 A^{(3)}_3=B^{(3)}_3=0,
\end{equation}
and then from
\begin{equation}
 0=c(6,2,0)={a^{(3)}_4}^2(A^{(1)}_3-B^{(1)}_2)\Big(a^{(1)}_4+(A^{(1)}_3-B^{(1)}_2)(1+A^{(1)}_2-B^{(1)}_3)\Big),
\end{equation}
we obtain
\begin{equation}
 a^{(1)}_4=-(A^{(1)}_3-B^{(1)}_2)(1+A^{(1)}_2-B^{(1)}_3).
\end{equation}
Then, we obtain the following inadequate condition:
\begin{equation}
 0=c(6,3,0)={a^{(1)}_4}^2{a^{(3)}_4}^2\neq0.
\end{equation}
Therefore, the conditions \eqref{eqn:typeIII_proof_c_312_2} holds.

From the condition \eqref{eqn:caseN1_CAO_B2_2} we obtain
\begin{equation}
 B^{(3)}_2=-1+B^{(1)}_3,\quad
 B^{(3)}_3=-a^{(3)}_4+B^{(2)}_2,\quad
 B^{(2)}_3=A^{(1)}_3+a^{(1)}_4,\quad
 b^{(2)}_4=-a^{(1)}_4a^{(3)}_4.
\end{equation}
From 
\begin{equation}
 0=c(5,2,0)=2a^{(1)}_4A^{(2)}_2a^{(3)}_4(B^{(2)}_2-a^{(3)}_4),
\end{equation}
we obtain
\begin{equation}
 B^{(2)}_2=a^{(3)}_4,
\end{equation}
which and $B^{(3)}_3=-a^{(3)}_4+B^{(2)}_2$ gives
\begin{equation}
 A^{(3)}_3=B^{(3)}_3=0.
\end{equation}
From 
\begin{equation}
 0=2a^{(1)}_4(A^{(2)}_3-A^{(3)}_2)c(6,2,0)-A^{(2)}_2a^{(3)}_4c(3,5,0)
 =4{a^{(1)}_4}^2A^{(2)}_2 {a^{(3)}_4}^2(A^{(2)}_3-A^{(3)}_2),
\end{equation}
we obtain
\begin{equation}
 A^{(3)}_2=A^{(2)}_3.
\end{equation}
From 
\begin{equation}
 0=c(6,2,0)=A^{(2)}_2{a^{(3)}_4}^2\Big(a^{(1)}_4-A^{(2)}_2(1+A^{(1)}_2-B^{(1)}_3)\Big),
\end{equation}
we obtain
\begin{equation}
 a^{(1)}_4=A^{(2)}_2(1+A^{(1)}_2-B^{(1)}_3).
\end{equation}
Then, we obtain the following inadequate condition:
\begin{equation}
 0=c(6,3,0)=-{a^{(1)}_4}^2{a^{(3)}_4}^2\neq0.
\end{equation}
Therefore, the case \eqref{eqn:typeIII_proof_c_312} is inadequate.
\subsubsection*{\qquad Case \eqref{eqn:typeIII_proof_c_313}.}
Then, the following relations hold:
\begin{equation}
 C^{(53)}=C^{(51)},\quad
 C^{(54)}=C^{(52)},\quad
 C^{(57)}=C^{(55)},\quad
 C^{(58)}=C^{(56)}.
\end{equation}
From $c(6,2,0)=0$
we obtain
\begin{equation}\label{eqn:case5_ciii_1_3_eqn1}
 C^{(56)}=-\dfrac{a^{(1)}_4a^{(3)}_4+(a^{(1)}_4A^{(3)}_3+C^{(55)})C^{(52)}}{C^{(51)}}-A^{(1)}_2a^{(3)}_4,
\end{equation}
and then from $c(6,3,0)=0$ we obtain
\begin{align}\label{eqn:case5_ciii_1_3_eqn2}
 &A^{(1)}_2
 =-\dfrac{1}{a^{(1)}_4a^{(3)}_4C^{(51)}}
 \Big({a^{(1)}_4}^2a^{(3)}_4+a^{(1)}_4a^{(3)}_4B^{(1)}_3C^{(51)}+{a^{(1)}_4}^2A^{(3)}_3C^{(52)}\notag\\
 &\qquad+a^{(1)}_4A^{(3)}_3B^{(1)}_3C^{(51)}C^{(52)}-a^{(1)}_4A^{(3)}_3B^{(3)}_2C^{(51)}C^{(52)}-a^{(3)}_4C^{(51)}C^{(55)}\notag\\
 &\qquad+a^{(1)}_4C^{(52)}C^{(55)}+B^{(1)}_3C^{(51)}C^{(52)}C^{(55)}-B^{(3)}_2C^{(51)}C^{(52)}C^{(55)}\Big).
\end{align}
From 
\begin{equation}
 0=C^{(11)}=a^{(2)}_4+(A^{(1)}_3+A^{(2)}_2+B^{(1)}_2)(A^{(2)}_3+A^{(3)}_2+B^{(3)}_3),
\end{equation}
we obtain
\begin{equation}
 a^{(2)}_4=-(A^{(1)}_3+A^{(2)}_2+B^{(1)}_2)(A^{(2)}_3+A^{(3)}_2+B^{(3)}_3),
\end{equation}
and then from \eqref{eqn:case5_ciii_1_3_eqn1} we obtain
\begin{equation}
 a^{(3)}_4=\dfrac{(1+A^{(1)}_2+A^{(3)}_3+B^{(3)}_2)(A^{(2)}_3+A^{(3)}_2+B^{(3)}_3)}{A^{(1)}_2+A^{(3)}_3+B^{(3)}_2},
\end{equation}
and from \eqref{eqn:case5_ciii_1_3_eqn2} we obtain
\begin{equation}
 a^{(1)}_4=-\dfrac{(A^{(1)}_3+A^{(2)}_2+B^{(1)}_2)(A^{(1)}_2+A^{(3)}_3+B^{(3)}_2)}{1+A^{(1)}_2+A^{(3)}_3+B^{(3)}_2}.
\end{equation}
Therefore, we obtain the condition \eqref{eqn:typeIII_cond3_1}.
We can verify by computation that under this condition the CACO property and square property hold. 
\subsubsection*{\quad Case \eqref{eqn:typeIII_proof_c_32}.}
From \eqref{eqn:case5_ciii_cond_1}, we obtain
\begin{equation}
 b^{(1)}_4=-a^{(1)}_4d_3.
\end{equation}
From 
\begin{subequations}
\begin{align}
 0&\neq a^{(3)}_4+b^{(3)}_4d_2=a^{(3)}_4(1-{d_2}^2),\\
 0&\neq (a^{(3)}_4+b^{(3)}_4d_2){C^{(53)}}^2
 =a^{(2)}_4(a^{(1)}_4+b^{(1)}_4d_3)
 =a^{(1)}_4(1-{d_3}^2),
\end{align}
\end{subequations}
we obtain
\begin{equation}
 {d_2}^2,{d_3}^2\neq1.
\end{equation}
Then, from \eqref{eqn:case5_ciii_cond_2}, we obtain
\begin{equation}
 a^{(2)}_4=\dfrac{a^{(3)}_4(1-{d_2}^2){C^{(53)}}^2}{a^{(1)}_4(1-{d_3}^2)},
\end{equation}
and then from \eqref{eqn:case5_ciii_cond_3}, we obtain
\begin{align}
 0&=a^{(2)}_4b^{(1)}_4d_2d_3(1-{d_3}^2)-b^{(3)}_4{C^{(51)}}^2+b^{(3)}_4{d_2}^2{d_3}^2{C^{(53)}}^2\notag\\
 &=a^{(3)}_4d_2({C^{(51)}}^2-{d_3}^2{C^{(53)}}^2).
\end{align}
Therefore, we shall consider the following two cases:
\begin{align}
 &C^{(51)}=d_3C^{(53)};\tag{iii:2.1}\label{eqn:typeIII_proof_c_321}\\
 &C^{(51)}=-d_3C^{(53)}.\tag{iii:2.2}\label{eqn:typeIII_proof_c_322}
\end{align}
%
\subsubsection*{\qquad Case \eqref{eqn:typeIII_proof_c_321}.}
From $c(6,2,0)=0$ we obtain
\begin{equation}
 C^{(56)}=\dfrac{A^{(3)}_3a^{(3)}_4(1-{d_2}^2)}{1-{d_3}^2}-A^{(1)}_2a^{(3)}_4-\dfrac{a^{(1)}_4a^{(3)}_4d_2d_3}{C^{(53)}}+\dfrac{a^{(3)}_4(1-{d_2}^2)C^{(55)}}{a^{(1)}_4(1-{d_3}^2)},
\end{equation}
and then from $c(5,3,0)=0$ we obtain
\begin{align}
 C^{(58)}=&\dfrac{a^{(1)}_4a^{(3)}_4A^{(3)}_3(1+d_2)(1-2{d_3}^2+d_2{d_3}^2)}{a^{(1)}_4d_2(1-{d_3}^2)}
 -\dfrac{A^{(1)}_2a^{(3)}_4(1+d_3-{d_3}^2)}{d_2}\notag\\
 &-\dfrac{a^{(1)}_4a^{(3)}_4d_3(1-2{d_3}^2)}{C^{(53)}}
 +\dfrac{a^{(3)}_4({d_2}^2-{d_3}^2)C^{(55)}}{a^{(1)}_4d_2(1-{d_3}^2)}+\dfrac{a^{(3)}_4d_3C^{(57)}}{a^{(1)}_4d_2}.
\end{align}
Under these conditions
we obtain the conditions of the parameters $A^{(i)}_j$ and $B^{(i)}_j$:
\eqref{eqn:typeIII_cond3_2}--\eqref{eqn:typeIII_cond3_4}.
We can verify by computation that under these conditions the CACO property and square property hold.

\subsubsection*{\qquad Case \eqref{eqn:typeIII_proof_c_322}.}
From $c(6,2,0)=0$ we obtain
\begin{equation}
 C^{(56)}=\dfrac{a^{(3)}_4\Big(A^{(3)}_3(1-{d_2}^2)-A^{(1)}_2(1-{d_3}^2)\Big)}{1-{d_3}^2}
 -\dfrac{a^{(1)}_4a^{(3)}_4d_2d_3}{C^{(53)}}+\dfrac{a^{(3)}_4(1-{d_2}^2)C^{(55)}}{a^{(1)}_4(1-{d_3}^2)},
\end{equation}
and then from $c(5,3,0)=0$ we obtain
\begin{align}
 C^{(58)}=&-\dfrac{a^{(3)}_4}{a^{(1)}_4d_2(1-{d_3}^2)C^{(53)}}
 \Big({a^{(1)}_4}^2d_2d_3(1-{d_3}^2)(1-2{d_3}^2)\notag\\
 &+\Big(a^{(1)}_4\big(A^{(1)}_2(1-{d_3}^2)(1+(1-d_3)d_3)-A^{(3)}_3(1+d_2)(1-(2-d_2){d_3}^2)\big)\notag\\
 &-({d_2}^2-{d_3}^2)C^{(55)}-d_3(1-{d_3}^2)C^{(57)}\Big)C^{(53)}\Big).
\end{align}
Under these conditions
we obtain the conditions of the parameters $A^{(i)}_j$ and $B^{(i)}_j$:
\eqref{eqn:typeIII_cond3_5}--\eqref{eqn:typeIII_cond3_9}.
We can verify by computation that under these conditions the CACO property and square property hold.

\subsubsection*{\quad Case \eqref{eqn:typeIII_proof_c_33}.}
From \eqref{eqn:case5_ciii_cond_1} and \eqref{eqn:case5_ciii_cond_2}, we obtain
\begin{equation}
 {C^{(51)}}^2=\dfrac{a^{(2)}_4d_2d_3(b^{(1)}_4+a^{(1)}_4d_3)}{b^{(3)}_4+a^{(3)}_4d_2},\quad
 {C^{(53)}}^2=\dfrac{a^{(2)}_4(a^{(1)}_4+b^{(1)}_4d_3)}{a^{(3)}_4+b^{(3)}_4d_2},
\end{equation}
respectively.
Then, from \eqref{eqn:case5_ciii_cond_3} we obtain
\begin{align}
 0&=a^{(2)}_4b^{(1)}_4d_2d_3(1-{d_3}^2)-b^{(3)}_4{C^{(51)}}^2+b^{(3)}_4{d_2}^2{d_3}^2{C^{(53)}}^2\notag\\
 &=-\dfrac{a^{(2)}_4a^{(3)}_4d_2d_3\Big(a^{(1)}_4b^{(3)}_4d_3(1-{d_2}^2)-a^{(3)}_4b^{(1)}_4d_2(1-{d_3}^2)-b^{(1)}_4b^{(3)}_4({d_2}^2-{d_3}^2)\Big)}{(b^{(3)}_4+a^{(3)}_4d_2)(a^{(3)}_4+b^{(3)}_4d_2)}.
\end{align} 
Note that from the equation above and 
the conditions \eqref{eqn:case5_ciii_cond_2} and \eqref{eqn:typeIII_proof_c_33},
if ${d_2}^2=1$ then ${d_3}^2=1$, and vice versa.
In the following, we consider by dividing into the following four cases.
\begin{align}
 &d_2=d_3=1;\tag{iii:3.1}\label{eqn:typeIII_proof_c_331}\\
 &d_2=1,\quad d_3=-1;\tag{iii:3.2}\label{eqn:typeIII_proof_c_332}\\
 &d_2=d_3=-1;\tag{iii:3.3}\label{eqn:typeIII_proof_c_333}\\
 &{d_2}^2,{d_3}^2\neq1.\tag{iii:3.4}\label{eqn:typeIII_proof_c_334}
\end{align}
%
%
%
Note that considering the case \eqref{eqn:typeIII_proof_c_332} and considering the case 
\begin{equation}
 d_2=-1,\quad d_3=1,
\end{equation}
mean the same under the symmetry $s_{23}$ \eqref{rqn:case5_s23}.
\subsubsection*{\qquad Case \eqref{eqn:typeIII_proof_c_331}.}
From $C^{(51)}\neq0$, we have 
\begin{equation}
 a^{(1)}_4+b^{(1)}_4\neq0.
\end{equation}
From \eqref{eqn:case5_ciii_cond_1} we obtain
\begin{equation}
 a^{(2)}_4=\dfrac{(A^{(1)}_3-A^{(2)}_2+B^{(1)}_2)^2(a^{(3)}_4+b^{(3)}_4)}{a^{(1)}_4+b^{(1)}_4},
\end{equation}
and then from \eqref{eqn:case5_ciii_cond_2} and \eqref{eqn:case5_ciii_cond_3}, we obtain
\begin{equation}
 A^{(1)}_3(A^{(2)}_2-B^{(1)}_2)=0,
\end{equation}
which give the following two cases:
\begin{subequations}
\begin{align}
 &A^{(1)}_3=B^{(1)}_3=0,\quad B^{(1)}_2\neq A^{(2)}_2;\label{eqn:typeIII_proof_c_331_1}\\
 &A^{(1)}_3,B^{(1)}_3\neq0,\quad B^{(1)}_2=A^{(2)}_2.\label{eqn:typeIII_proof_c_331_2}
\end{align}
\end{subequations}
Note that the condition 
\begin{equation}
 A^{(1)}_3=B^{(1)}_3=A^{(2)}_2-B^{(1)}_2=0
\end{equation}
is inconsistent with $a^{(2)}_4\neq0$.

Assume \eqref{eqn:typeIII_proof_c_331_1}.
Then, from $C^{(11)}=C^{(31)}=0$ we obtain
\begin{equation}
 b^{(3)}_4=-a^{(3)}_4-\dfrac{(a^{(1)}_4+b^{(1)}_4)(A^{(2)}_3-A^{(3)}_2-B^{(3)}_3)}{A^{(2)}_2-B^{(1)}_2},\quad
 A^{(3)}_2=B^{(3)}_2=0.
\end{equation}
The condition $B^{(1)}_3=B^{(3)}_2=0$ is inconsistent with the condition \eqref{eqn:caseN1_CAO_B2_2}.
Therefore, the condition \eqref{eqn:typeIII_proof_c_331_2} holds.

From $C^{(11)}=C^{(31)}=0$ we obtain
\begin{equation}
 b^{(3)}_4=-a^{(3)}_4+\dfrac{(a^{(1)}_4+b^{(1)}_4)(A^{(2)}_3-A^{(3)}_2-B^{(3)}_3)}{A^{(1)}_3},\quad
 B^{(3)}_3=A^{(2)}_3,
\end{equation}
and then from the condition \eqref{eqn:caseN1_CAO_B2_2} we obtain
\begin{equation}
 a^{(2)}_4=-A^{(1)}_3A^{(3)}_2,
\end{equation}
which gives
\begin{equation}
 A^{(1)}_3,A^{(3)}_2\neq0.
\end{equation}
If 
\begin{equation}
 a^{(1)}_4A^{(3)}_2+A^{(1)}_3a^{(3)}_4=0,
\end{equation}
then we obtain
\begin{equation*}
 c(6,3,1)=-\dfrac{2a^{(1)}_4{A^{(3)}_2}^2b^{(1)}_4(A^{(1)}_3+b^{(1)}_4)}{A^{(1)}_3},\quad
 c(5,4,1)=\dfrac{2a^{(1)}_4{A^{(3)}_2}^2b^{(1)}_4(3A^{(1)}_3+b^{(1)}_4)}{A^{(1)}_3},
\end{equation*}
which cannot be simultaneously zero. 
Therefore, we obtain
\begin{equation}
 a^{(1)}_4A^{(3)}_2+A^{(1)}_3a^{(3)}_4\neq0.
\end{equation}
Then, from $c(6,2,0)=0$ we obtain
\begin{align}
 B^{(1)}_3
 =&\dfrac{(A^{(1)}_2A^{(1)}_3+a^{(1)}_4-A^{(1)}_3A^{(3)}_3)(a^{(1)}_4A^{(3)}_2-A^{(1)}_3a^{(3)}_4)}{A^{(1)}_3(a^{(1)}_4A^{(3)}_2+A^{(1)}_3a^{(3)}_4)}\notag\\
 &+\dfrac{A^{(3)}_2b^{(1)}_4(A^{(1)}_3+2a^{(1)}_4+b^{(1)}_4)}{A^{(1)}_3(a^{(1)}_4A^{(3)}_2+A^{(1)}_3a^{(3)}_4)}
 +1,
\end{align}
and then from $c(5,3,0)=0$ we obtain
\begin{equation}
 b^{(1)}_4=-A^{(1)}_2A^{(1)}_3-a^{(1)}_4+A^{(1)}_3A^{(3)}_3.
\end{equation}
Moreover, from
\begin{equation}
 0=c(4,4,0)=2a^{(1)}_4A^{(3)}_2\Big(a^{(1)}_4+A^{(1)}_3(A^{(1)}_2-A^{(3)}_3)\Big)\Big(a^{(3)}_4-A^{(3)}_2(A^{(1)}_2-A^{(3)}_3)\Big),
\end{equation}
we obtain
\begin{equation}
 a^{(1)}_4=-A^{(1)}_3(A^{(1)}_2-A^{(3)}_3)\quad
 \text{or}\quad
 a^{(3)}_4=A^{(3)}_2(A^{(1)}_2-A^{(3)}_3),
\end{equation}
which give inadequate conditions 
\begin{equation}
 b^{(1)}_4=0\quad
 \text{or}\quad
 b^{(3)}_4=0,
\end{equation}
respectively.
Therefore, the case \eqref{eqn:typeIII_proof_c_331} is inadequate.
\subsubsection*{\qquad Case \eqref{eqn:typeIII_proof_c_332}.}
From $C^{(51)}\neq0$, we have 
\begin{equation}
 a^{(1)}_4-b^{(1)}_4\neq0.
\end{equation}
From \eqref{eqn:case5_ciii_cond_1} we obtain
\begin{equation}
 a^{(2)}_4=\dfrac{(A^{(1)}_3+A^{(2)}_2-B^{(1)}_2)^2(a^{(3)}_4+b^{(3)}_4)}{a^{(1)}_4-b^{(1)}_4},
\end{equation}
and then from \eqref{eqn:case5_ciii_cond_2} and \eqref{eqn:case5_ciii_cond_3}, we obtain
\begin{subequations}
\begin{align}
 0&=a^{(2)}_4(a^{(1)}_4+b^{(1)}_4d_3)-(a^{(3)}_4+b^{(3)}_4d_2){C^{(53)}}^2
 =4A^{(2)}_2(A^{(1)}_3-B^{(1)}_2)(a^{(3)}_4+b^{(3)}_4),\\
 0&=a^{(2)}_4b^{(1)}_4d_2d_3(1-{d_3}^2)-b^{(3)}_4{C^{(51)}}^2+b^{(3)}_4{d_2}^2{d_3}^2{C^{(53)}}^2
 =4A^{(2)}_2(B^{(1)}_2-A^{(1)}_3)b^{(3)}_4,
\end{align}
\end{subequations}
which give the following two conditions:
\begin{subequations}
\begin{align}
 &A^{(2)}_2=B^{(2)}_2=0,\quad B^{(1)}_2\neq A^{(1)}_3;\label{eqn:typeIII_proof_c_332_1}\\
 &A^{(2)}_2,B^{(2)}_2\neq0,\quad B^{(1)}_2=A^{(1)}_3.\label{eqn:typeIII_proof_c_332_2}
\end{align}
\end{subequations}
Note that the condition $ A^{(2)}_2=B^{(2)}_2=A^{(1)}_3-B^{(1)}_2=0$ is inconsistent with $a^{(2)}_4\neq0$.

Assume \eqref{eqn:typeIII_proof_c_332_1}.
Then, from $C^{(11)}=C^{(31)}=0$ we obtain
\begin{equation}
 b^{(3)}_4=-a^{(3)}_4-\dfrac{(a^{(1)}_4-b^{(1)}_4)(A^{(2)}_3-A^{(3)}_2+B^{(3)}_3)}{A^{(1)}_3-B^{(1)}_2},\quad
 A^{(3)}_2=A^{(2)}_3.
\end{equation}
From the condition \eqref{eqn:caseN1_CAO_B2_2} we obtain
\begin{equation}
 B^{(3)}_2=B^{(1)}_3-1,\quad
 a^{(3)}_4=-\dfrac{(A^{(1)}_3+a^{(1)}_4-B^{(1)}_2-b^{(1)}_4)B^{(3)}_3}{A^{(1)}_3-B^{(1)}_2},
\end{equation}
which gives
\begin{equation}
 A^{(1)}_3+a^{(1)}_4-B^{(1)}_2-b^{(1)}_4\neq0,\quad 
 A^{(3)}_3,B^{(3)}_3\neq0.
\end{equation}
Then, we obtain the following inadequate condition:
\begin{equation}
 0=c(6,3,1)-2c(6,2,0)=4a^{(1)}_4A^{(3)}_3{B^{(3)}_3}^2(A^{(1)}_3+a^{(1)}_4-B^{(1)}_2-b^{(1)}_4)\neq0.
\end{equation}
Therefore, the condition \eqref{eqn:typeIII_proof_c_332_2} holds.

From $C^{(11)}=C^{(31)}=0$ we obtain
\begin{equation}
 b^{(3)}_4=-a^{(3)}_4-\dfrac{(a^{(1)}_4-b^{(1)}_4)(A^{(2)}_3-A^{(3)}_2)}{A^{(2)}_2},\quad
 A^{(3)}_3=B^{(3)}_3=0,
\end{equation}
and then from the condition \eqref{eqn:caseN1_CAO_B2_2} we obtain
\begin{equation}
 B^{(3)}_2=B^{(1)}_3-1.
\end{equation}
Then, we have
\begin{equation}
 a^{(2)}_4=-A^{(2)}_2(A^{(2)}_3-A^{(3)}_2),\quad
 a^{(3)}_4=B^{(2)}_2-\dfrac{(A^{(2)}_3-A^{(3)}_2)(a^{(1)}_4-b^{(1)}_4)}{A^{(2)}_2},
\end{equation}
which give
\begin{equation}
 A^{(2)}_2\neq0,\quad
 A^{(2)}_3\neq A^{(3)}_2,\quad
 A^{(2)}_2B^{(2)}_2-(A^{(2)}_3-A^{(3)}_2)(a^{(1)}_4-b^{(1)}_4)\neq0.
\end{equation}
Then, we obtain the following inadequate condition:
\begin{align}
 0&=c(6,3,1)-2c(6,2,0)\notag\\
 &=\dfrac{4a^{(1)}_4(A^{(2)}_3-A^{(3)}_2)(a^{(1)}_4-b^{(1)}_4)\Big(A^{(2)}_2B^{(2)}_2-(A^{(2)}_3-A^{(3)}_2)(a^{(1)}_4-b^{(1)}_4)\Big)}{A^{(2)}_2}\neq0.
\end{align}
Therefore, the case \eqref{eqn:typeIII_proof_c_332} is inadequate.
\subsubsection*{\qquad Case \eqref{eqn:typeIII_proof_c_333}.}
From $C^{(51)}\neq0$, we have 
\begin{equation}
 a^{(1)}_4-b^{(1)}_4\neq0.
\end{equation}
From \eqref{eqn:case5_ciii_cond_1} we obtain
\begin{equation}
 a^{(2)}_4=\dfrac{(A^{(1)}_3+A^{(2)}_2+B^{(1)}_2)^2(a^{(3)}_4-b^{(3)}_4)}{a^{(1)}_4-b^{(1)}_4}.
\end{equation}
Then, from $C^{(11)}=C^{(31)}=0$ and the condition \eqref{eqn:caseN1_CAO_B2_2}, we obtain
\begin{subequations}
\begin{align}
 &a^{(2)}_4=\dfrac{(A^{(1)}_3+A^{(2)}_2+B^{(1)}_2)^2(a^{(3)}_4-b^{(3)}_4)}{a^{(1)}_4-b^{(1)}_4},\quad
 B^{(3)}_2=B^{(1)}_3-1,\\
 &B^{(3)}_3=-\dfrac{(A^{(2)}_3+A^{(3)}_2)(a^{(1)}_4-b^{(1)}_4)+(A^{(1)}_3+A^{(2)}_2+B^{(1)}_2)(a^{(3)}_4-b^{(3)}_4)}{a^{(1)}_4-b^{(1)}_4}.
\end{align}
\end{subequations}
We here consider by dividing into the following two cases:
\begin{subequations}
\begin{align}
 &a^{(3)}_4=\dfrac{a^{(1)}_4b^{(3)}_4}{b^{(1)}_4};\label{eqn:typeIII_proof_c_333_1}\\
 &a^{(3)}_4\neq\dfrac{a^{(1)}_4b^{(3)}_4}{b^{(1)}_4}.\label{eqn:typeIII_proof_c_333_2}
\end{align}
\end{subequations}

We first consider the case \eqref{eqn:typeIII_proof_c_333_1}. 
Then, from $c(6,2,0)=0$ we obtain
\begin{equation}
 b^{(1)}_4=-A^{(1)}_3+2a^{(1)}_4-A^{(2)}_2-B^{(1)}_2,
\end{equation}
and then from $c(6,6,2)=0$ we obtain
\begin{equation}
 B^{(1)}_3=\dfrac{1-2A^{(1)}_2-2A^{(3)}_3}{2}.
\end{equation}
Therefore, we obtain the condition \eqref{eqn:typeIII_cond3_10}.
We can verify by computation that under this condition the CACO property and square property hold. 

We next consider the case \eqref{eqn:typeIII_proof_c_333_2}.
From $c(6,2,0)=0$ we obtain
\begin{equation*}
 B^{(1)}_3
 =-A^{(1)}_2-A^{(3)}_3-\dfrac{b^{(3)}_4(a^{(1)}_4-b^{(1)}_4)}{a^{(3)}_4b^{(1)}_4-a^{(1)}_4b^{(3)}_4}
 +\dfrac{(a^{(1)}_4-b^{(1)}_4)(a^{(3)}_4b^{(1)}_4+a^{(1)}_4b^{(3)}_4-b^{(1)}_4b^{(3)}_4)}{(A^{(1)}_3+A^{(2)}_2+B^{(1)}_2)(a^{(3)}_4b^{(1)}_4-a^{(1)}_4b^{(3)}_4)}.
\end{equation*}
If 
\begin{equation}
 b^{(1)}_4=-A^{(1)}_3+a^{(1)}_4-A^{(2)}_2-B^{(1)}_2,
\end{equation}
then we obtain the following inadequate condition:
\begin{equation}
 0=c(6,3,0)=\dfrac{a^{(1)}_4{a^{(3)}_4}^2b^{(1)}_4b^{(3)}_4(A^{(1)}_3+A^{(2)}_2+B^{(1)}_2)}{A^{(1)}_3a^{(3)}_4+a^{(3)}_4(A^{(2)}_2+B^{(1)}_2)-a^{(1)}_4(a^{(3)}_4-b^{(3)}_4)}\neq0.
\end{equation}
Therefore, we obtain
\begin{equation}
 b^{(1)}_4\neq-A^{(1)}_3+a^{(1)}_4-A^{(2)}_2-B^{(1)}_2.
\end{equation}
From $c(6,3,0)=0$ we obtain
\begin{align}
 &b^{(3)}_4=a^{(3)}_4-\dfrac{a^{(1)}_4a^{(3)}_4(a^{(1)}_4-b^{(1)}_4)}{(A^{(1)}_3-a^{(1)}_4+A^{(2)}_2+B^{(1)}_2+b^{(1)}_4)^2}.
\end{align} 
Therefore, we obtain the condition \eqref{eqn:typeIII_cond3_11}.
We can verify by computation that under this condition the CACO property and square property hold. 

\subsubsection*{\qquad Case \eqref{eqn:typeIII_proof_c_334}.}
Eliminating $C^{(51)}$ and $C^{(53)}$ from \eqref{eqn:case5_ciii_cond_3} by using \eqref{eqn:case5_ciii_cond_1} and \eqref{eqn:case5_ciii_cond_2}, we obtain
\begin{equation}
 a^{(1)}_4=\dfrac{b^{(1)}_4\Big(a^{(3)}_4d_2(1-{d_3}^2)+b^{(3)}_4({d_2}^2-{d_3}^2)\Big)}{b^{(3)}_4d_3(1-{d_2}^2)}.
\end{equation}
From \eqref{eqn:case5_ciii_cond_1} we obtain
\begin{equation}
 b^{(1)}_4=\dfrac{b^{(3)}_4(1-{d_2}^2)(B^{(1)}_2-A^{(2)}_2d_2+A^{(1)}_3d_2d_3)^2}{a^{(2)}_4d_2d_3(1-{d_3}^2)},
\end{equation}
from $C^{(11)}=0$ we obtain
\begin{equation}
 a^{(2)}_4=-\dfrac{(B^{(1)}_2-A^{(2)}_2d_2+A^{(1)}_3d_2d_3)(B^{(3)}_3-A^{(2)}_3d_3+A^{(3)}_2d_2d_3)}{d_2d_3},
\end{equation}
from the condition \eqref{eqn:caseN1_CAO_B2_2} we obtain
\begin{equation}
 B^{(3)}_2=B^{(1)}_3-1,
\end{equation}
and from \eqref{eqn:case5_ciii_cond_2} we obtain
\begin{align}
 0=&{d_3}^2\Big(a^{(2)}_4(a^{(1)}_4+b^{(1)}_4d_3)-(a^{(3)}_4+b^{(3)}_4d_2){C^{(53)}}^2\Big)\notag\\
 =&(a^{(3)}_4+b^{(3)}_4d_2)\Big(A^{(1)}_3d_2d_3(1+d_3)+B^{(1)}_2(1-{d_3}^2)-A^{(2)}_2(d_2-d_3)\Big)\notag\\
 &\Big(A^{(1)}_3d_2d_3(1-d_3)+B^{(1)}_2(1+{d_3}^2)-A^{(2)}_2(d_2+d_3)\Big).
\end{align}
Therefore, in the following, we consider by dividing into the following two cases:
\begin{align}
 &A^{(1)}_3=-\dfrac{B^{(1)}_2(1-{d_3}^2)-A^{(2)}_2(d_2-d_3)}{d_2d_3(1+d_3)};\tag{iii:3.4.1}\label{eqn:typeIII_proof_c_3341}\\
 &A^{(1)}_3=-\dfrac{B^{(1)}_2(1+{d_3}^2)-A^{(2)}_2(d_2+d_3)}{d_2d_3(1-d_3)}.\tag{iii:3.4.2}\label{eqn:typeIII_proof_c_3342}
\end{align}
In the following, we will consider the cases \eqref{eqn:typeIII_proof_c_3341} and \eqref{eqn:typeIII_proof_c_3342} separately.
%
\subsubsection*{\hspace{2.5em} Case \eqref{eqn:typeIII_proof_c_3341}.}
From $C^{(31)}=0$ we obtain
\begin{equation}
 B^{(3)}_3=-\dfrac{A^{(2)}_3(d_2-d_3)+A^{(3)}_2d_2d_3(1+d_2)}{1-{d_2}^2}.
\end{equation}
From
\begin{align}
 0&=c(6,3,1)+\dfrac{a^{(3)}_4{d_2}^4{d_3}^3(1-{d_3}^2)}{a^{(3)}_4d_2(1-{d_3}^2)+b^{(3)}_4({d_2}^2-{d_3}^2)}c(3,6,1)\notag\\
 &=\dfrac{2a^{(3)}_4b^{(3)}_4{d_2}^5{d_3}^7(1-{d_2}^2)(A^{(2)}_2-B^{(1)}_2+A^{(2)}_2d_2-B^{(1)}_2d_3)^2}{(1+d_3)^3(1-d_3)(A^{(2)}_3-A^{(2)}_3d_2d_3+A^{(3)}_2d_2d_3+A^{(3)}_2{d_2}^2d_3)}\notag\\
 &\Big(b^{(3)}_4(1-{d_2}^2)(1-{d_2}^2{d_3}^2)+d_2(1-{d_3}^2)\Big(A^{(3)}_2d_2(1+d_2)d_3+A^{(2)}_3(1-d_2d_3)\Big)\Big),
\end{align}
we obtain
\begin{equation}
 A^{(3)}_2=-\dfrac{A^{(2)}_3(1-d_2d_3)}{d_2d_3(1+d_2)}-\dfrac{b^{(3)}_4(1-d_2)(1-{d_2}^2{d_3}^2)}{{d_2}^2d_3(1-{d_3}^2)}.
\end{equation}
Note that we have
\begin{equation}
 b^{(1)}_4=\dfrac{(1-{d_2}^2)d_3\Big(A^{(2)}_2(1+d_2)-B^{(1)}_2(1+d_3)\Big)}{(1+d_3)(1-{d_2}^2{d_3}^2)},
\end{equation}
which implies
\begin{equation}
 {d_2}^2{d_3}^2\neq1.
\end{equation}
From $c(6,3,1)=0$ we obtain
\begin{equation}
 a^{(3)}_4=-\dfrac{b^{(3)}_4(A^{(1)}_2-A^{(3)}_3d_2+B^{(1)}_3d_2d_3)({d_2}^2-{d_3}^2)}{{d_2}^2d_3(1-{d_3}^2)}+\dfrac{b^{(3)}_4d_2(1-{d_2}^2)}{1-{d_2}^2{d_3}^2},
\end{equation}
from $c(6,2,0)=0$ we obtain
\begin{align}
 &A^{(3)}_3=\dfrac{A^{(1)}_2(1+d_3)}{1+d_2}-\dfrac{d_2d_3(1-d_2)(1-{d_3}^2)}{1-{d_2}^2{d_3}^2},
\end{align}
and from $c(6,4,1)=0$ we obtain
\begin{align}
 &B^{(1)}_3=-\dfrac{A^{(1)}_2(1-d_2d_3)}{d_2d_3(1+d_2)}+\dfrac{(1-d_2)(1+d_2{d_3}^2)}{1-{d_2}^2{d_3}^2}.
\end{align}
Therefore, we obtain the condition \eqref{eqn:typeIII_cond3_12}.
We can verify by computation that under this condition the CACO property and square property hold. 
\subsubsection*{\hspace{2.5em} Case \eqref{eqn:typeIII_proof_c_3342}.}
From $C^{(31)}=0$ we obtain
\begin{equation}
A^{(3)}_2=-\dfrac{B^{(3)}_3(1+{d_2}^2)-A^{(2)}_3(d_2+d3_)}{(1-d_2)d_2d_3}.
\end{equation}
Note that since $a^{(1)}_4\neq0$ we have
\begin{equation}
 b^{(3)}_4({d_2}^2-{d_3}^2)+a^{(3)}_4d_2(1-{d_3}^2)\neq0.
\end{equation}
Then, from
\begin{align}
 0=&c(6,3,1)+\dfrac{a^{(3)}_4{d_2}^4{d_3}^3(1-{d_3}^2)}{b^{(3)}_4({d_2}^2-{d_3}^2)+a^{(3)}_4d_2(1-{d_3}^2)}c(3,6,1)\notag\\
 =&-2a^{(2)}_4a^{(3)}_4b^{(1)}_4{d_2}^6{d_3}^6\left(d_2(1-{d_3}^2)+\dfrac{b^{(3)}_4(1-d_2)(1+{d_2}^2{d_3}^2)}{B^{(3)}_3(1+d_2)-A^{(2)}_3(1+d_3)}\right),
\end{align}
we obtain
\begin{equation}
 B^{(3)}_3=\dfrac{A^{(2)}_3(1+d_3)}{1+d_2}-\dfrac{b^{(3)}_4(1-d_2)(1+{d_2}^2{d_3}^2)}{d_2(1+d_2)(1-{d_3}^2)}.
\end{equation}

Consider the case 
\begin{equation}
 {d_2}^2={d_3}^2.
\end{equation}
Solving $F_1G_2-F_2G_1=0$ under the condition
\begin{equation}
 d_2=d_3,
\end{equation}
we obtain the condition \eqref{eqn:typeIII_cond3_13}, 
and solving $F_1G_2-F_2G_1=0$ under the condition
\begin{equation}
 d_2=-d_3,
\end{equation}
we obtain the condition \eqref{eqn:typeIII_cond3_14}.
We can verify by computation that under these conditions the CACO property and square property hold.

Also, we consider the case 
\begin{equation}
 {d_2}^2\neq{d_3}^2.
\end{equation}
From $c(6,3,1)=0$ we obtain
\begin{equation}
 B^{(1)}_3=-\dfrac{A^{(1)}_2-A^{(3)}_3d_2}{d_2d_3}
 -\dfrac{d_2(1-{d_3}^2)}{b^{(3)}_4({d_2}^2-{d_3}^2)}
 \left(a^{(3)}_4-\dfrac{d_2(b^{(3)}_4+2a^{(3)}_4d_2+b^{(3)}_4{d_2}^2)}{1+{d_2}^2{d_3}^2}\right),
\end{equation}
and then from $c(6,2,0)=0$ we obtain
\begin{equation}
 a^{(3)}_4=-\dfrac{b^{(3)}_4(1+{d_2}^2)}{2d_2}
+\dfrac{A^{(1)}_2b^{(3)}_4(1+{d_2}^2{d_3}^2)}{2{d_2}^2d_3(1-d_3)}
-\dfrac{A^{(3)}_3b^{(3)}_4(1+d_2)(1+{d_2}^2{d_3}^2)}{2{d_2}^2d_3(1-{d_3}^2)}.
\end{equation}
Note that since
\begin{equation}
 a^{(2)}_4=-\dfrac{b^{(3)}_4(1+{d_2}^2{d_3}^2)\Big(A^{(2)}_2(1+d_2)-B^{(1)}_2(1+d_3)\Big)}{d_2(1+d_3)(1-d_3)^2},
\end{equation}
we have
\begin{equation}
 1+{d_2}^2{d_3}^2\neq0,\quad
 A^{(2)}_2(1+d_2)-B^{(1)}_2(1+d_3)\neq0.
\end{equation}
Then, we obtain
\begin{align}
 0=&c(6,3,0)\notag\\
 =&\dfrac{a^{(3)}_4b^{(3)}_4{d_2}^5{d_3}^5(1-d_2)(1+d_2)^3\Big(A^{(2)}_2(1+d_2)-B^{(1)}_2(1+d_3)^2\Big)}{4(1-d_3)^2({d_2}^2-{d_3}^2)}\notag\\
 &\left(A^{(3)}_3-\dfrac{A^{(1)}_2(1+d_3)}{1+d_2}+\dfrac{d_2d_3(1+{d_2}^2)}{(1+d_2)(1+{d_2}^2{d_3}^2)}+\dfrac{d_2{d_3}^3(1-d_2)}{1+{d_2}^2{d_3}^2}\right)\notag\\
 &\left((1-{d_2}^2{d_3}^2)A^{(3)}_3-\dfrac{A^{(1)}_2(1+d_3)(1-{d_2}^2{d_3}^2)}{1+d_2}+\dfrac{d_2d_3(1+{d_2}^2)(1+{d_3}^2)}{1+d_2}\right),
\end{align}
which gives the following two cases:
\begin{subequations}
\begin{align}
 &A^{(3)}_3=\dfrac{A^{(1)}_2(1+d_3)}{1+d_2}-\dfrac{d_2d_3(1+{d_2}^2)}{(1+d_2)(1+{d_2}^2{d_3}^2)}-\dfrac{d_2{d_3}^3(1-d_2)}{1+{d_2}^2{d_3}^2};\label{eqn:typeIII_proof_c_3342_1}\\
 &(1-{d_2}^2{d_3}^2)A^{(3)}_3-\dfrac{A^{(1)}_2(1+d_3)(1-{d_2}^2{d_3}^2)}{1+d_2}+\dfrac{d_2d_3(1+{d_2}^2)(1+{d_3}^2)}{1+d_2}=0.\label{eqn:typeIII_proof_c_3342_2}
\end{align}
\end{subequations}

In the case \eqref{eqn:typeIII_proof_c_3342_1},
from $c(5,4,0)=0$ we obtain
\begin{equation}
 {d_2}^2+{d_3}^2=-2.
\end{equation}
Therefore, we obtain the condition \eqref{eqn:typeIII_cond3_15}.
We can verify by computation that under this condition the CACO property and square property hold.

Consider the condition \eqref{eqn:typeIII_proof_c_3342_2}.
If 
\begin{equation}
 {d_2}^2 {d_3}^2=1,
\end{equation} 
then from the condition \eqref{eqn:typeIII_proof_c_3342_2} we obtain 
\begin{equation}
 (1+{d_2}^2)(1+{d_3}^2)=0,
\end{equation} 
which gives the inadequate condition:
\begin{equation}
 {d_2}^2={d_3}^2.
\end{equation}
Therefore,
\begin{equation}
 {d_2}^2 {d_3}^2\neq 1.
\end{equation}
Then, the condition \eqref{eqn:typeIII_proof_c_3342_2} can be solved by $A^{(3)}_3$ as the following:
\begin{equation}
 A^{(3)}_3=\dfrac{A^{(1)}_2(1+d_3)}{1+d_2}-\dfrac{d_2d_3(1+{d_2}^2)(1+{d_3}^2)}{(1+d_2)(1-{d_2}^2{d_3}^2)}.
\end{equation}
Then, we obtain the condition \eqref{eqn:typeIII_cond3_16}.
We can verify by computation that under this condition the CACO property and square property hold.

Therefore, we have completed the proof.
\def\cprime{$'$} \def\cprime{$'$}

\end{document}